\documentclass[11pt, a4paper]{article}
\usepackage[left=2cm, right=2cm, top=2cm, bottom=2cm]{geometry}
\usepackage[utf8]{inputenc}
\usepackage{amsmath}
\usepackage{amssymb}
\usepackage{natbib}
\usepackage{hyperref}
\usepackage{amsthm}
\usepackage{verbatim}
\usepackage{graphicx}
\usepackage{mathtools}
\mathtoolsset{showonlyrefs}

\usepackage{subcaption}
\usepackage{color}
\usepackage{algorithmicx,algorithm,algpseudocode}

\usepackage[makeroom]{cancel}
\usepackage{multirow}
\usepackage{booktabs}
\usepackage{tikz}
\usetikzlibrary{calc}
\usepackage{bbm}
\usepackage{array}
\usepackage{enumitem}
\usepackage{placeins}
\newcolumntype{x}[1]{>{\centering\arraybackslash}p{#1}}

\newcommand{\tikzmark}[1]{\tikz[overlay,remember picture] \node (#1) {};}
\newcommand{\data}{x}
\newcommand{\dataObs}{x^0}
\newcommand{\datasim}{x^{\mbox{sim}}}

\newcommand{\distance}{d}

\newcommand{\parameter}{\theta}

\newcommand{\statistics}{s}

\newcommand{\logit}{\text{logit}}
\newcommand{\Var}{\mathbb{V}}
\DeclareMathOperator*{\argmin}{arg\,min}

\newcommand{\Hessian}{\operatorname{H}}
\newcommand{\E}{\mathbb{E}}

\newcommand{\R}{\mathbb{R}}

\newcommand{\X}{\mathcal{X}}
\newcommand{\s}{s}

\newcommand{\lr}{\text{lr}}
\newcommand{\dsim}{x}
\newcommand{\dobs}{y}

\newcommand{\Dsim}{X}
\newcommand{\Dobs}{Y}

\newcommand{\sumjk}{\sum_{\substack{j,k=1\\k\neq j}}^m}

\newcommand{\SE}{S_{\operatorname{E}}}
\newcommand{\DE}{D_{\operatorname{E}}}

\newcommand\independent{\protect\mathpalette{\protect\independenT}{\perp}}
\def\independenT#1#2{\mathrel{\rlap{$#1#2$}\mkern2mu{#1#2}}}

\makeatletter
\newcommand\footnoteref[1]{\protected@xdef\@thefnmark{\ref{#1}}\@footnotemark}
\makeatother

\newtheorem{theorem}{Theorem}
\newtheorem{corollary}{Corollary}
\newtheorem{lemma}{Lemma}

\theoremstyle{definition}
\newtheorem{definition}{Definition}
\newtheorem{remark}{Remark}

\title{Score Matched Neural Exponential Families \\for Likelihood-Free Inference\thanks{Code for reproducing the experiments is available \href{https://github.com/LoryPack/SM-ExpFam-LFI}{here}.}}
\author{}

\author{
	Lorenzo Pacchiardi$^1$\thanks{Corresponding author: lorenzo.pacchiardi@stats.ox.ac.uk.}, Ritabrata Dutta$^2 $\\
	{\em \small $^1$Department of Statistics, University of Oxford, UK}\\
	{\em \small $^2$Department of Statistics, University of Warwick, UK}\\
}
\date{31st January 2022}

\begin{document}
	
	\maketitle		
	
	\begin{abstract}
		Bayesian Likelihood-Free Inference (LFI) approaches allow to obtain posterior distributions for stochastic models with intractable likelihood, by relying on model simulations. In Approximate Bayesian Computation (ABC), a popular LFI method, summary statistics are used to reduce data dimensionality. ABC algorithms adaptively tailor simulations to the observation in order to sample from an approximate posterior, whose form depends on the chosen statistics. In this work, we introduce a new way to learn ABC statistics: we first generate parameter-simulation pairs from the model independently on the observation; then, we use Score Matching to train a neural conditional exponential family to approximate the likelihood. The exponential family is the largest class of distributions with fixed-size sufficient statistics; thus, we use them in ABC, which is intuitively appealing and has state-of-the-art performance. In parallel, we insert our likelihood approximation in an MCMC for doubly intractable distributions to draw posterior samples. We can repeat that for any number of observations with no additional model simulations, with performance comparable to related approaches. We validate our methods on toy models with known likelihood and a large-dimensional time-series model. 
	\end{abstract}

	\section{Introduction}\label{sec:intro}
	Stochastic simulator models are used to simulate realizations of physical phenomena; usually, a set of parameters $ \theta $ govern the simulation output. As the model is stochastic, repeated simulations with fixed $\theta $ yield different outputs; their distribution is the model's \textit{likelihood}. Simple models provide an analytic expression of the likelihood, which is however unavailable for more complex ones. 
	
	Upon observing a real-world realization of the phenomenon the model is describing, researchers may want to obtain a \textit{posterior distribution} over parameters.
	If the likelihood is known, standard Bayesian inference tools (such as MCMC or variational inference) allow to get posterior samples; if otherwise the likelihood is missing, Likelihood-Free Inference (LFI) techniques are the only viable solution. LFI has been applied in several domains, including genomics \citep{tavare1997inferring}, biological science \citep{dutta2018parameter}, meteorology \citep{hakkarainen2012closure}, geological science \citep{pacchiardi2020}, genomics \citep{tavare1997inferring, toni2009approximate, marttinen2015recombination}, and epidemiology \citep{mckinley2018approximate, minter2019approximate, dutta2021using}.
	
	In some LFI techniques \citep{wood2010statistical, thomas2020likelihood, price2018bayesian}, the missing likelihood is replaced with an explicit approximation built from model simulations at each parameter value.
	Alternatively, Approximate Bayesian Computation (ABC) \citep{marin2012approximate, lintusaari2017} circumvents the unavailability of the likelihood by comparing model simulations with the observation according to some notion of distance. 
	To reduce data dimensionality, ABC usually relies on summary statistics, whose choice is not straightforward. Recently, the expressive capabilities of Neural Networks (NNs) have been leveraged to learn ABC statistics \citep{jiang2017learning, wiqvist2019partially, pacchiardi2020, aakesson2020convolutional}. These techniques train NNs parametrizing the statistics by minimizing a suitable loss on a training set of parameter-simulation pairs generated from the model.

	In the present work, we propose a new way to learn ABC statistics with NNs. Most previous works \citep{jiang2017learning, wiqvist2019partially, aakesson2020convolutional} trained a single NN with the regression loss introduced in \cite{fearnhead_constructing_2012}. Instead, we consider an exponential family with two NNs parametrizing respectively the sufficient statistics and natural parameters and fit this to the model.
	As the exponential family is the most general class of distributions with sufficient statistics of fixed size (Appendix~\ref{app:suff_stats}), it makes intuitive sense to use the learned sufficient statistics in ABC. Indeed, this approach empirically yields superior or equivalent performance with respect to state-of-the-art approaches.
	
	As in previous approaches, we consider a training set of parameter-simulation pairs; however, we train the NNs with Score Matching (SM) or its Sliced approximation (SSM), which do not require evaluating the normalizing constant of the exponential family. We extend the SM and SSM objectives to the setting of conditional densities, thus fitting the likelihood approximation for all values of data and parameters.

	In contrast to related approaches \citep{jiang2017learning, wiqvist2019partially, pacchiardi2020, aakesson2020convolutional}, our method provides a full likelihood approximation. We test therefore direct sampling from the corresponding posterior (in place of ABC) with an MCMC algorithm for doubly intractable distributions. This approach achieves satisfactory performance with no additional model simulations and is therefore a valid alternative to standard LFI schemes for expensive simulator models. The computational gain is even larger when inference is performed for several observations, as the same likelihood approximation can be used.

	The rest of our paper is organized as follows. In Section~\ref{sec:LFI}, we briefly review some LFI methods. Next, in Section~\ref{sec:lik_app} we introduce the neural conditional exponential family and show how to fit it with SM or SSM. In Section~\ref{sec:inference_strategy} we discuss how to exploit the exponential family to extract ABC statistics or for direct sampling. We extensively validate our proposed approaches in 	Section~\ref{sec:experiments}.
	Finally, we discuss related works in Section~\ref{sec:related_works} and conclude in Section~\ref{sec:conclusion}, where we also highlight directions for future research. 	
	\subsection{Notation}
	We set here notation for the rest of our manuscript. We will denote respectively by $ \mathcal{X} \subseteq \R^d $ and $ \Theta \subseteq \R^p $ the data and parameter space. 
	Upper case letters will denote random variables while lower case ones will denote observed (fixed) values. Subscripts will denote vector components and superscript in brackets sample indices, while
	$ \|\cdot \| $ will denote the $ \ell_2 $ norm.
	\section{Likelihood-Free Inference}\label{sec:LFI}
	
	Let us consider a model which allows to generate a simulation $ x \in \X $ at any parameter value $ \theta\in\Theta $, but for which it is not possible to evaluate the likelihood $ p_0(x|\theta) $.
	Given an observation $ x^0 $ and a prior on the parameters $ \pi(\theta) $, Bayesian inference obtains the posterior distribution $ \pi_0(\theta|x^0) = \frac{\pi(\theta) p_0(x^0|\theta)}{p_0(x^0)} $. However, obtaining that explicitly (or even sampling from it with Markov Chain Monte Carlo, MCMC) is impossible without having access to the likelihood, at least up to a normalizing constant that is independent on $ x $. 	
	
	Likelihood-Free Inference techniques yield approximations of the posterior distribution by replacing likelihood evaluations with model simulations. Broadly, they can be split into two kinds of approaches: Surrogate Likelihood, which explicitly builds a likelihood function, and Approximate Bayesian Computation (ABC), which instead uses discrepancy between simulated and observed data. We detail those two approaches in the following.
	
	\subsection{Approximate Bayesian Computation (ABC)}\label{sec:ABC}
	Approximate Bayesian Computation (ABC, \citealt{marin2012approximate,lintusaari2017}) algorithms sample from an approximate posterior distribution by finding parameter values that yield simulated data resembling the observations to a sufficient degree. 	
	Usually, ABC methods rely on statistics function to reduce the dimensionality of simulated and observed data and thus improve the computational efficiency of the algorithm. Therefore, the similarity between simulated and observed data is normally defined via a distance function between the corresponding statistics.
	ABC algorithms sample from:
	\begin{equation}\label{Eq:ABC_post}
		\pi^\epsilon (\theta|\statistics(\dataObs)) \propto \pi (\theta) \int_\X \mathbbm{1}[\distance(\statistics(\dataObs), \statistics(\datasim)) \le \epsilon]\  p_0(\datasim|\theta) d\datasim	\end{equation}
	where $ s $ and $ d $ denote respectively the statistics function and the discrepancy measure, while $ \mathbbm 1[\cdot] $ is an indicator function and $ \epsilon $ is an acceptance threshold.\footnote{In general, $ \mathbbm{1}(\distance(\statistics(\dataObs), \statistics(\datasim))) $ can be replaced with $ \mathbb K_\epsilon (\distance(\statistics(\dataObs), \statistics(\datasim))) $, where $ \mathbb K_\epsilon $ is a rejection kernel which concentrates around 0 and goes to 0 for argument going to $ \pm \infty $, and whose width decreases when $ \epsilon $ decreases. We restrict here on $ \mathbbm{1}$, which is the most common choice.}
	
	It is easy to sample from $ \pi^\epsilon (\theta|\statistics(\dataObs)) $ with rejection, by generating $ (\theta^{(j)}, x^{(j)}) \sim \pi(\theta) p_0(x|\theta) $ and accepting $ \theta^{(j)} $ if $ \distance(\statistics(\dataObs), \statistics(x^{(j)})) \le \epsilon $. This is inefficient when $ \epsilon $ is small, as most simulations end up being rejected; more efficient ABC algorithms have been developed, based for instance on Markov Chain Monte Carlo \citep{marjoram2003markov}, Population Monte Carlo \citep{beaumont2010approximate} or Sequential Monte Carlo \citep{del2012adaptive}, which reach lower values of $ \epsilon $ with the same computational budget. 
	
	If $ d $ is a distance measure and the statistic $ s $ is sufficient for $ p(x|\theta) $, then the ABC posterior $ \pi^\epsilon (\theta|\statistics(\dataObs)) $ converges to the true one as $ \epsilon \to 0 $ (\citealt{lintusaari2017}; see Appendix~\ref{app:suff_stats} for definition of sufficient statistics). 
	
	In general, however, sufficient statistics are unavailable so that ABC can at best give an approximation of the exact posterior even with $ \epsilon\to0 $. Hand-picking the best summary statistics for a given model limits the applicability of ABC. Therefore, some approaches to automatically build them have been developed \citep{fearnhead_constructing_2012,nunes2010optimal, joyce2008approximately, aeschbacher2012novel, blum2013comparative, pacchiardi2020}. We describe here the approach that is most commonly used.
	
	\paragraph{ABC with semi-automatic statistics.}
	The summary statistic $ s^\star(x) = \E[\parameter| \data] $ yields an approximate posterior $ \pi^\epsilon(\theta|s^\star (\dataObs)) $ whose mean minimizes the expected squared error loss from the true parameter value, as $ \epsilon \to 0 $ \citep{fearnhead_constructing_2012}. Of course, $ \E[\parameter| \data] $ is unknown; therefore, \cite{fearnhead_constructing_2012} suggested to consider a set of functions $ s_\beta(x) $ and find the one that best approximates $ \E[\parameter|\data] $. Empirically, this corresponds to finding: 
	\begin{equation}\label{Eq:FP_loss}
		\hat \beta = \argmin_\beta  \frac{1}{N} \sum_{j=1}^N  \left\|s_\beta(\data^{(j)}) - \parameter^{(j)}\right\|_2^2,
	\end{equation}			
	where $(\parameter^{(j)}, \data^{(j)})_{j=1}^{N}, \ \parameter^{(j)} \sim \pi(\theta), \ \data^{(j)} \sim p_0(\cdot|\parameter^{(j)}) $ are N data-parameter pairs generated before performing ABC.
	
	In practice, \cite{fearnhead_constructing_2012} considered $ s_\beta(x)  $ to be a linear function of the weights $ \beta $, such that solving the minimization problem in Eq.~\eqref{Eq:FP_loss} amounts to linear regression. \cite{jiang2017learning, wiqvist2019partially} and \cite{aakesson2020convolutional} instead used Neural Networks to parametrize $ s_\beta $ and train them to minimize the loss in Eq.~\eqref{Eq:FP_loss}.

	\subsection{Surrogate Likelihood}\label{sec:surr_lik}
	
	Surrogate likelihood approaches exploit simulations to build an explicit likelihood approximation, which is then inserted in standard likelihood-based sampling schemes (say, MCMC; \citealt{wood2010statistical, thomas2020likelihood, fasiolo2018extended, an2019accelerating, an2020robust, price2018bayesian}).
	Here, we discuss two methods that fall under this framework.

	\paragraph{Synthetic Likelihood.}
	
	Synthetic Likelihood (SL, \citealt{wood2010statistical}) replaces the exact likelihood with a normal distribution for the summary statistics $ s=s(x) $; specifically, it assumes $S|\theta \sim \mathcal{N}(\mu_\theta, \Sigma_\theta)$, where the mean vector $ \mu_\theta $ and covariance matrix $ \Sigma_\theta $ depend on $ \theta $. For each $ \theta $, an estimate of $ \mu_\theta $ and $ \Sigma_\theta $ is built with model simulations, and the likelihood of the observed statistics $ s(\dataObs) $ is evaluated.
	
	Using the statistics likelihood to define a posterior distribution yields a Bayesian SL (BSL). In \cite{price2018bayesian}, MCMC is used to sample from the BSL posterior; as in ABC, this approach requires generating simulations for each considered $ \theta $. However, one single simulation per parameter value is usually sufficient in ABC, while estimating $ \mu_\theta, \Sigma_\theta $ in BSL requires multiple simulations. Nevertheless, a parametric likelihood approximation seemingly allows scaling to larger dimensional statistics space \citep{price2018bayesian}. 	
	
	\paragraph{Ratio Estimation.} 	
	
	\cite{thomas2020likelihood} tackle LFI by estimating the ratio between likelihood and data marginal: $r(x, \theta) =  p(x|\theta) /p(x)$. As $ r(x,\theta) \propto p(x|\theta) $ with respect to $ \theta $, this can be inserted in a likelihood-based sampling scheme in order to sample from the posterior, similarly to BSL. This method is termed Ratio Estimation (RE). 
	
	In practice, for each value of $\theta$, an estimate $ \hat r(x,\theta) $ is built by generating simulations from the model $ p_0(x|\theta) $ and from the marginal\footnote{Simulating from the marginal can be done by drawing $\theta^j \sim p(\theta)$, $x^j \sim p(x|\theta^j)$ and discarding $\theta^j$.} $p(x)$, and then fitting a logistic regression discriminating between the two sets of samples. In fact, logistic regression attempts to find a function $\hat h(x; \theta)$ for which $e^{\hat h(x;\theta)} \approx r(x;\theta)$. 
	
	In \cite{thomas2020likelihood}, $ \hat h $ is chosen in the class of functions $h_\beta(x;\theta) = \sum_{i=1}^k \beta_i(\theta) \psi_i(x)=  \beta(\theta)^T \psi(x)$, where $\psi$ is a vector of summary statistics and $ \beta(\theta) $ are coefficients determined independently for each $ \theta $; this therefore boils down to a linear logistic regression. Viewed differently, this approach approximates the likelihood with an exponential family, as in fact it assumes $ p(x|\theta) \propto \hat r(x,\theta) = \exp(\hat \beta(\theta)^T \psi(x)) $, for some coefficients $ \hat{\beta} $ determined by data; it is thus a more general likelihood assumption than SL (as the normal distribution belongs to the exponential family). 
	\begin{remark}[Reducing the number of simulations.]
		Several approaches have attempted to reduce the number of simulations required in SL, RE and ABC. Specifically, Gaussian Processes can be exploited to replace evaluations of simulation-based quantities with emulated values, while at the same time providing an uncertainty quantification used to guide the next model simulations; that has been done for both ABC \citep{meeds2014gps, wilkinson2014accelerating, gutmann2016bayesian, jarvenpaa2019efficient, jarvenpaa2020batch} and SL \citep{meeds2014gps, wilkinson2014accelerating, moores2015pre, jarvenpaa2021parallel}. These approaches all consider a fixed observation $ x $ and emulate over parameter values.
		
		In \cite{hermans2019likelihood}, a classifier more powerful than logistic regression is used to learn the likelihood ratio estimate for all $ (x, \theta) $ using parameter-simulation pairs (similarly to what we propose in Section.~\ref{sec:exp_fish_div_likelihood}), amortizing therefore RE with respect to the observation.
	\end{remark}

	\section{Likelihood approximation with the exponential family}\label{sec:lik_app}
	
	In this Section, we first introduce the parametric family which we use to approximate the likelihood. Then, we describe Score Matching (SM) and Sliced Score Matching (SSM), which allow us to fit distributions with unknown normalizing constants to data. Finally, we discuss how to extend SM and SSM to the setting of conditional densities, to obtain a likelihood approximation valid for all data and parameter values.

	\subsection{Conditional exponential family}\label{sec:exp_fam}

	A probability distribution belongs to the exponential family if it has a density of the following form:\footnote{As we are concerned with continuous random variables, across the work we use the Lebesgue measure as a base measure, without explicitly referring to it.}

	\begin{equation}\label{Eq:expon_family_general}
		p(x|\theta) = \frac{e^{\eta(\theta) ^T f(x)} h(x)}{Z(\theta)},
	\end{equation}
	where $ x  \in \X ,\  \theta \in \Theta$. Here, $ f: \X \to \R^{d_s} $ is a function of the data (sufficient statistics), $ \eta: \Theta \to \R^{d_s} $ is a function of the parameters (natural parameters) and $ h(x): \X \to \R $ is a scalar function of data (base measure). The normalizing constant $ Z(\theta) = \int_\X \exp \{\eta(\theta) ^T f(x)\} h(x) dx $ is intractable and assumed to be finite $\forall\ \theta \in \Theta $; we will discuss later (Sec.~\ref{sec:exp_fish_div_likelihood}) why this assumption is not an issue. 
	
	Across this work, we refer to Eq.~\eqref{Eq:expon_family_general} as \textit{conditional} exponential family to stress that we learn an approximation valid for all $ x $'s and $ \theta $'s by selecting functions $ f $, $ \eta $ and $ h $. 
	
	In the following, we rewrite Eq.~\eqref{Eq:expon_family_general} as: $ p(x|\theta) =  \exp (\bar \eta(\theta)^T \bar f(x)) /Z(\theta ) $, where $\bar \eta (\theta )= (\eta(\theta), 1)\in \R^{d_s+1} $ and $\bar f(x) = (f(x), \log h(x))\in \R^{d_s+1} $.  For simplicity, we will drop the bar notation, using as convention that $ \eta(\theta )$ contains the natural parameters of the exponential family plus an additional constant term, and that the last component of $ f(x) $ is $ \log h(x) $. 
	
	\paragraph{Neural conditional exponential family.}
	Let now $ f_w $ and $ \eta_w $ denote two Neural Networks (NNs) with collated weights $ w $ (in practice the two NNs do not share parameters); then, the neural conditional exponential family is defined as:
	\begin{equation}\label{Eq:expon_family_NNs}
		p_w(x|\theta) =  \frac{e^{\eta_w(\theta) ^T f_w(x)}}{Z_w(\theta)}.	\end{equation}

	\begin{remark} Ratio Estimation (Sec.~\ref{sec:surr_lik}) parametrizes the likelihood as an exponential family with user-specified statistics and natural parameters $ \beta $ independently learned for each $ \theta \in \Theta $. In contrast, here we learn both the $ f_w(\data) $ and the $ \eta_w(\theta) $ transformations over all $ \mathcal{X} $ and $ \Theta $ at once by selecting the best $ w $.
	\end{remark}

	\begin{remark}[Identifiability]
		For a family of distributions indexed by a parameter $ \phi $, identifiability means that $ p_\phi(x|\theta) = p_{\phi'}(x|\theta) \ \forall \ x, \theta \implies \phi = \phi' $. The weights $ w $ in the neural conditional exponential family are not identifiable for two reasons: first, NNs have many intrinsic symmetries. Secondly, replacing $ f_w $ in Eq.~\eqref{Eq:expon_family_NNs} with $ A \cdot f_w$ and $ \eta_w $ with $ (A^T)^{-1}\cdot \eta_w $ does not change the probability density. In \cite{khemakhem2020ice}, two suitable concepts of \textit{function} identifiability up to some linear transformations are defined. $ f_w $ and $ \eta_w $ in Eq.~\eqref{Eq:expon_family_NNs} are identifiable according to those definitions, under strict conditions on the architectures of the Neural Networks parametrizing them \citep{khemakhem2020ice}. Moreover, they empirically verify that NNs not satisfying the above assumptions result in approximately identifiable $ f_w $ and $ \eta_w $, according to their definition. More details are given in Appendix~\ref{app:identifiability}.
	\end{remark}

	\begin{remark}[Universal approximation]
		Using larger $ {d_s} $ in Eq.~\eqref{Eq:expon_family_general} increases the expressive power of the approximating family. \cite{khemakhem2020ice} proved that, by considering a freely varying $ {d_s} $ and generic $ f $ and $ \eta $, the conditional exponential family has universal approximation capabilities for the set of conditional probability densities; we give more details in Appendix~\ref{app:universal_appr}.
		This result does not consider the practicality of fitting the approximating family to data, which arguably becomes harder when $ {d_s} $ increases.	
	\end{remark}

	\subsection{Score Matching}\label{sec:density_est}

	In order to fit a parametric density $ p_w $ to data, the standard approach is finding the Maximum Likelihood Estimator (MLE) for $ w $; in the limit of infinite data, that corresponds to minimizing the Kullback-Leibler divergence from the data distribution. MLE requires however the normalizing constant of $ p_w $ to be known, which is not the case for our approximating family (Eq.~\ref{Eq:expon_family_NNs}). 
	
	Score Matching (SM, \citealt{hyvarinen2005estimation}) is a possible way to bypass the intractability of the normalizing constant. In this Subsection, we review SM for unconditional densities, discuss a faster version and provide an extension for distributions with bounded domain.

	\paragraph{The original Score Matching (SM)}

	Let us discard now the conditional dependency on $ \theta $ and, following \cite{hyvarinen2005estimation}, consider a random variable $ X $ distributed according to $ p_0(x) $. We want to use samples from $ p_0 $ to fit a generic model $ p_w(x) = \tilde p_w(x) / Z_w $, where $ \tilde p_w $ is unnormalized and $ Z_w $ is intractable. 
	
	\begin{definition} 
		Score Matching (SM) corresponds to finding: 
		\begin{equation}\label{}
			\argmin_{w} D_F(p_0\|p_w),
		\end{equation}	where the Fisher Divergence $ D_F $ is defined as: 
		\begin{equation}\label{Eq:Fisher_div}
			D_F(p_0\|p_w) = \frac{1}{2} \int_{\mathcal{X}} p_0(x) \|\nabla_x\log p_0(x) - \nabla_x \log p_w(x) \|^2 dx.
		\end{equation} 
	\end{definition}	
	
	The Fisher divergence depends only on the the logarithmic derivatives $ \nabla_x \log p_0(x) $ and $ \nabla_x \log p_w(x) $, which are termed \textit{scores}\footnote{In most of the statistics literature, \textit{score} usually refer to the derivative of the log-likelihood with respect to the parameter; here, the nomenclature is slightly different.}; computing $ D_F $ does not require therefore knowing the normalizing constant, as in fact:
	\begin{equation}\label{}
		\nabla_x \log p_w(x) = \nabla_x (\log \tilde p_w(x) - \log Z_w) = \nabla_x \log \tilde p_w(x).
	\end{equation}
	
	Nevertheless, a Monte Carlo estimate of $ D_F $ in Eq.~\eqref{Eq:Fisher_div} using samples from $ p_0  $ would require knowing the logarithmic gradient of $ p_0 $; for this reason, Eq.~\eqref{Eq:Fisher_div} is usually termed \textit{implicit} Fisher divergence. If $ \mathcal{X} = \R^d $ and under mild conditions, Theorem 1 in \cite{hyvarinen2005estimation} obtains an equivalent form for $ D_F $ in which $ p_0 $ only appears as the distribution over which expectation needs to be computed. We give here a similar result holding for more general $ \mathcal{X}  = \bigotimes_{i=1}^d (a_i,b_i) $, where $ a_i, b_i \in \R \cup \{\pm \infty\}$; specifically, we consider the following assumptions:

	\begin{enumerate}[label=\textbf{A\arabic*}]
		\item\label{ass:A1} $ p_0(x) \frac{\partial}{\partial x_i} \log p_w(x) \to 0 $ when $ x_i \searrow  a_i $ and $x_i \nearrow  b_i, \forall w, 	i $,
		\item \label{ass:A2} $ \E_{p_0}[\|\nabla_x \log p_0(X)\|^2] <\infty$, $ \E_{p_0}[\|\nabla_x \log p_w(X)\|^2] <\infty, \forall \ w$,
		\item \label{ass:A3} $\E_{p_0} \left| \frac{\partial^2}{\partial x_i \partial x_j} \log p_w(X)  \right| < \infty, \forall w, \forall i,j=1,\ldots, d $.\footnote{We remark that this assumption was not present in \cite{hyvarinen2005estimation}, but it is necessary to apply Fubini-Tonelli theorem in the proof (see Appendix~\ref{app:th_FD_explicit_proof}), as already discussed in \cite{yu2018generalized}.}
	\end{enumerate}
	With the above, we can state the following: 
	
	\begin{theorem}\label{Th:FD_explicit}
		Let $ \mathcal X = \bigotimes_{i=1}^d (a_i,b_i) $, where $	 a_i, b_i \in \R \cup \{\pm \infty\} $. If $ p_0(x) $ is differentiable and $ p_w(x) $ doubly differentiable over $\mathcal{X}$, 
		then, under assumptions \ref{ass:A1}, \ref{ass:A2}, \ref{ass:A3}, the objective in Eq.~\eqref{Eq:Fisher_div} can be rewritten as: 
		
		\begin{equation}\label{Eq:FD_explicit}
			D_F(p_0\|p_w) =  \int_{\X} p_0(x) \sum_{i=1}^d \left[ \frac{1}{2} \left( \frac{\partial \log p_w(x)}{\partial x_i}  \right)^2 +  \left( \frac{\partial^2 \log p_w(x)}{\partial x_i^2}  \right) \right] dx + C.
		\end{equation}
		Here, $C  $ is a constant w.r.t. $p_w$ and $x_i$ is the \textit{i}-th coordinate of $x$. 	\end{theorem}
	A proof is given in Appendix~\ref{app:th_FD_explicit_proof}. We refer to Eq.~\eqref{Eq:FD_explicit} as the \textit{explicit} Fisher Divergence, as it is now immediate to build an unbiased Monte Carlo estimate using samples from $ p_0 $.

	From Eq.~\eqref{Eq:Fisher_div}, it is clear that $ D_F(p_0\|p_w) \ge0\ \forall\ w$, and that choosing $ p_w=p_0 $ implies $ D_F(p_0\|p_w)=0 $; however, the converse can not be said in general unless $ p_0 $ is supported over all $ \X $, as stated in the following theorem:

	\begin{theorem}[Theorem 2 in \citealp{hyvarinen2005estimation}]\label{Th:FD}
		Assume $ p_0(x)>0 \ \forall x \in \mathcal{X}$. Then: 
		\begin{equation}\label{}
			D_F(p_0\|p_w) = 0 \iff p_0(x) = p_w(x) \ \forall x \in \mathcal{X}.
		\end{equation}
	\end{theorem}
	
	We prove Theorem~\ref{Th:FD} in Appendix~\ref{app:th_FD_proof}. If $ p_{0}(x) $ is zero for some $ x \in \mathcal{X} $, there could be a distribution $ p_w $ such that $ D_F(p_0\|p_w)=0  $ even if $p_w\neq p_0 $ (as discussed in Appendix~\ref{app:Fish_div_prop_disjoint_supp}).

	\paragraph{Sliced Score Matching (SSM)}
	
	In the explicit Fisher Divergence (Eq.~\ref{Eq:FD_explicit}), the second derivatives of the log-density with respect to all components of $ x $ are required. When using the neural conditional exponential family (Eq.~\ref{Eq:expon_family_NNs}), this amounts to evaluating the second derivatives of $ f_w $ with respect to its input (see Appendix \ref{app:sm_exp_fam}). {Practically}, automatic differentiation libraries to obtain the derivatives effortlessly; however, the computational cost of the second derivatives is substantial, as it requires a number of forward and backward passes proportional to the dimension of the data $ x $ (see Appendix~\ref{app:SM_computational_complexity})\footnote{Computing the derivatives during the forward pass offers some speedup, as automatic differentiation repeats some computations several times. However, this approach requires custom NN implementation (see Appendix~\ref{app:forward_der_comp}) and is thus not scalable to complex architectures. More importantly, the computational gain is limited with respect to what is achieved by the Sliced SM version introduced next.}.

	Some approaches to reducing the computational burden have been proposed; we review some in Section~\ref{sec:related_works}. Here, we consider Sliced Score Matching (SSM, \citealt{song2019sliced}), which considers projections of the scores on random directions; matching the projections on \textit{all} random directions ensures the two distributions are identical (under the same conditions as the original SM). More precisely, let $ v \in \mathcal V \subseteq \R^d $ be a noise vector with a distribution $ q$; SSM is defined as follows:

	\begin{definition} 
		Sliced Score Matching (SSM) corresponds to finding: 
		\begin{equation}\label{}
			\argmin_{w} D_{FS}(p_0\|p_w),
		\end{equation} where the Sliced Fisher Divergence is defined as: 
		\begin{equation}\label{Eq:Fisher_div_sliced}
			D_{FS}(p_0\|p_w) = \frac{1}{2} \int_{\mathcal{V}} q(v) \int_{\mathcal{X}} p_0(x) (v^T \nabla_x\log p_0(x) - v^T \nabla_x \log p_w(x) )^2 dx dv.
		\end{equation}
	\end{definition}

	We will require the noise distribution $ q(v) $ to satisfy the following Assumption:
	\begin{enumerate}[label=\textbf{A4}]
		\item\label{ass:A4} For the random vector $ V\sim q $, the covariance matrix $ \E[V V^T] \succ 0$ is positive definite and $ \E_{}\|V\|_2^2<\infty $.
	\end{enumerate}
	
	As for SM, we can obtain an \textit{explicit} formulation from the \textit{implicit} one in Eq.~\eqref{Eq:Fisher_div_sliced}.
	This is done in the following Theorem, which we prove for convenience in Appendix~\ref{app:th_FDS_explicit_proof}:
	
	\begin{theorem}[Theorem 1 in \cite{song2019sliced}]\label{Th:FDS_explicit}
		Let $ \mathcal X = \bigotimes_{i=1}^d (a_i,b_i) $, where $ a_i, b_i \in \R \cup \{\pm \infty\} $.  If $ p_0(x) $ is differentiable and $ p_w(x) $ doubly differentiable over $\mathcal{X}$, 
		then, under assumptions \ref{ass:A1}, \ref{ass:A2}, \ref{ass:A3}, \ref{ass:A4}, the objective in Eq.~\eqref{Eq:Fisher_div_sliced} can be rewritten as: 
		\begin{equation}\label{Eq:FDS_explicit}
			D_{FS}(p_0\|p_w) = \int_{\mathcal{V}} q(v)  \int_{\X} p_0(x) \left[ v^T (\Hessian_x \log p_w(x)) v +  \frac{1}{2} \left(v^T \nabla_x\log p_w(x)  \right)^2 \right] dx dv + C,
		\end{equation}
		where $ \Hessian_x \log p_w(x) $ denotes the Hessian matrix of $ \log p_w(x) $ with respect to components of $ x $ and $C $ is a constant w.r.t. $p_w$.
	\end{theorem}
	
	Assumption \ref{ass:A4} is satisfied by, among others, Gaussian and Rademacher random variables \citep{song2019sliced}. Additionally, these two distributions allow to computes explicitly the expectation with respect to $ v $ of $ \left(v^T \nabla_x\log p_0(x)  \right)^2  $ in Eq.~\eqref{Eq:FDS_explicit}, which leads to:
	\begin{equation}\label{Eq:FDS_explicit_vr}
		D_{FS}(p_0\|p_w) = \int_{\X} p_0(x) \left\{ \int_{\mathcal{V}} q(v) \left[ v^T (\Hessian_x \log p_w(x)) v \right]dv +  \frac{1}{2} \left\|\nabla_x\log p_w(x) \right\|_2^2 \right\} dx + C;
	\end{equation}
	in \cite{song2019sliced}, the Monte Carlo estimate of the latter expression is found to perform better than the one for Eq.~\eqref{Eq:FDS_explicit}; across this work, we will therefore consider Eq.~\eqref{Eq:FDS_explicit_vr} with the Rademacher noise when using SSM.
	
	Analogously to SM, a non-negative $ p_0 $ ensures that the sliced Fisher divergence is zero if and only if $ p_w=p_0 $:
	
	\begin{theorem}[Lemma 1 in \cite{song2019sliced}]\label{Th:FDS}
		Assume Assumption \ref{ass:A4} holds and that $ p_0(x)>0 \ \forall x \in \mathcal{X}$. Then: 
		\begin{equation}\label{}
			D_{FS}(p_0\|p_w) = 0 \iff p_0(x) = p_w(x) \ \forall x \in \mathcal{X}.
		\end{equation}
	\end{theorem}
	The above result is proven in Appendix~\ref{app:th_FDS_proof}.

	SSM has a reduced computational cost with respect to SM when automatic differentiation is used. In fact, computing the second derivatives in Eq.~\eqref{Eq:FD_explicit} requires $ d $ backward propagations, while the quadratic form involving the Hessian in Eq.~\eqref{Eq:FDS_explicit_vr} only requires two backward propagations independently on the dimension of $ x $ (see Appendix~\ref{app:SM_computational_complexity}).

	\paragraph{SM and SSM over transformed domain.}	
	If the domain $ \X $ is unbounded ($ a_i=-\infty $, $ b_i = +\infty \ \forall i $), \ref{ass:A1} is usually satisfied. Instead, it is easy for that assumption to be violated if $ \X $ is bounded: for instance, if $ p_0(x) $ converges to a constant at the boundary of $ \X $, \ref{ass:A1} requires $ \nabla_x \log p_w(x) $ to go to 0. Further, if $ p_0(x) $ diverges, $ \nabla_x \log p_w(x) $ has to converge to 0 faster than some rate, which is in general a strong requirement.

	To apply SM to distributions with bounded domain $ \X $ under looser conditions, we then transform $ \X $ to $ \mathcal{Y} = \R^d $ with a suitable bijection mapping $ t $; this creates a new random variable $ Y = t(X) $, with distributions $ p^Y_0(y)  $ and $ p^Y_w(y) $ induced by $ p_0 $ and $ p_w $ on $ \X $. 
	
	\begin{definition}
		We define respectively as Transformed Score Matching (TranSM) and Transformed Sliced Score Matching (TranSSM) the procedures: 
		\begin{equation}\label{}
			\argmin_{w} D_F(p^Y_0\|p^Y_w) \qquad \text{ and } \qquad \argmin_{w} D_{FS}(p^Y_0\|p^Y_w).
		\end{equation}
	\end{definition}
	
	TranSM and TranSSM enjoy similar properties as SM and SSM, as stated in the following Theorem, which mirrors Theorem~\ref{Th:FD} and Theorem~\ref{Th:FDS}:
	
	\begin{theorem}\label{Th:FD_y}
		Let $ Y=t(X) \in \mathcal Y$ for a bijection $ t $.
		Assume Assumption \ref{ass:A4} is satisfied, $ p_0(x) > 0\ \forall x \in \mathcal{X}$, and let $ D $ denote either $ D_F $ or $ D_{FS}$; then:
		\begin{equation}\label{}
			D(p^Y_0\|p^Y_w) = 0 \iff p_w(x) = p_0(x) \ \forall x \ \in \X.
		\end{equation}
	\end{theorem}
	
	We prove this Theorem (in Appendix~\ref{app:th_FD_y_proof}) by relying on the equivalence between distributions for the random variables $ Y $ and $ X $.

	Motivated by Theorem~\ref{Th:FD_y}, across this work we will apply TranSM and TranSSM whenever $ \X  $ is bounded; precisely, we adopt the same bijections as in the \texttt{Stan} package for statistical modeling (Appendix~\ref{app:transformations}, \citealt{carpenter2017stan}). 
	
	\begin{remark}
		Another way to extend SM to distributions with bounded support involves multiplying the scores in the implicit Fisher divergence by a correction factor that allows to obtain an explicit for under looser assumptions. This Corrected SM (CorrSM) approach was first proposed for non-negative random variables in \cite{hyvarinen2007some}, which also remarked how CorrSM and TranSM are equivalent (Appendix \ref{App:corrected_transformed_SM_equivalence}). TranSM is however more practically viable, as a single SM implementation is needed, while CorrSM requires separate implementations for different kinds of domains. Additionally, the transformations can be straightforwardly applied to SSM, while (to the best of our knowledge) a correction approach for SSM has not yet been proposed. 
	\end{remark}
	
	\begin{remark}
		Across this work we restrict to domains $ \X $ defined by independent constraints across the coordinates, i.e. $ \mathcal X = \bigotimes_{i=1}^d (a_i,b_i) $. However, TranSM, TranSSM and CorrSM can be potentially applied to distributions with more general supports. We briefly review this and other extensions in Appendix~\ref{app:SM_extensions}.
	\end{remark}

	\subsection{Score Matching for conditional densities}\label{sec:exp_fish_div_likelihood}
	
	We now go back to conditional densities $p_0(\cdot|\theta)$ and define an expected Fisher divergence  \citep{arbel2017kernel}
	by considering $\theta \sim \pi(\theta)$: 
	\begin{equation}\label{Eq:exp_fisher_div}
		\begin{aligned}
			D_F^E(p_0\|p_w) &= \int_{\Theta} \pi(\theta) D_F(p_0(\cdot|\theta)\| p_w(\cdot|\theta)) d\theta
			\\&= \frac{1}{2} \int_{\Theta} \int_{\mathcal{X}} p_0(x|\theta) \pi(\theta) \|\nabla_x\log p_0(x|\theta) - \nabla_x \log p_w(x|\theta) \|^2 dx d\theta.
		\end{aligned}
	\end{equation}
	Analogously, we define an expected sliced Fisher divergence $ D_{FS}^E(p_0\|p_w) $ from $ D_{FS} $.
	
	Note that $ D_F^E(p_0\| p_w) \ge 0 $; moreover, under the assumption that $p_0(x|\theta) $ is supported on the whole domain $\mathcal X$ $\forall \ \theta$, the above objective is equal to 0 if and only if $p_0(\cdot|\theta) =  p_w(\cdot|\theta) $ $\pi(\theta)$-almost everywhere \citep{arbel2017kernel}. In fact, $D_F(p_0(\cdot|\theta)\| p_w(\cdot|\theta)) \ge 0\ \forall \ \theta$, with equality holding if and only if the two conditional distributions are the same, as long as they are both supported on the whole $\mathcal X$ (by Theorem~\ref{Th:FD}). Exploiting Theorem~\ref{Th:FDS}, a similar result can be obtained for $ D_{FS}^E $.

	Requiring $ p_0(x|\theta) >0\ \forall\ x \in \X,\ \forall\ \theta \in \Theta$ means that the model is capable of generating all possible values of $ x \in \mathcal{X} $ with non-zero probability for all $ \theta \in \Theta $; otherwise, there may be distributions $ p_w \neq p_0$ minimizing the objective (see Appendix~\ref{app:Fish_div_prop_disjoint_supp}).

	Analogously to Eqs.~\eqref{Eq:FD_explicit} and \eqref{Eq:FDS_explicit}, we can obtain explicit formulations\footnote{\label{Note:assumptions_Fish_divergence}Requiring $ p_0(x|\theta) $ to be differentiable and $ p_w(x|\theta) $ doubly differentiable over $\mathcal{X}$ for all $ \theta $ and:
		\begin{enumerate}[label=\textbf{A\arabic*.}]
			\item $ p_0(x|\theta) \frac{\partial}{\partial x_i} \log p_w(x|\theta) \to 0 $ for $ x_i \searrow  a_i $ and $x_i \nearrow  b_i, \forall w, i, \theta $,
			\item$ \E_{p_0}[\|\nabla_x \log p_0(X|\theta)\|^2] < \infty,\ \E_{p_0}[\|\nabla_x \log p_w(X|\theta)\|^2] < \infty, \forall w, \theta$,
			\item $\E_{p_0} \left| \frac{\partial^2}{\partial x_i \partial x_j} \log p_w(X|\theta)  \right| < \infty, \forall w,\theta, \ \forall i,j=1,\ldots, d $.
		\end{enumerate}.
		
	}of $ D^E_F $ and $ D_{FS}^E $:
	\begin{equation}\label{Eq:exp_fisher_div_exp}
		\begin{aligned}
			&D_F^E(p_0\|p_w) =\\  &\qquad \underbrace{\int_{\Theta} \int_{\X} p_0(x|\theta) \pi(\theta) \sum_{i=1}^d \left[ \frac{1}{2} \left( \frac{\partial \log p_w(x|\theta)}{\partial x_i}  \right)^2 +  \left( \frac{\partial^2 \log p_w(x|\theta)}{\partial x_i^2}  \right) \right] dx d\theta}_{J(w)} + C,\\
			&D_{FS}^E(p_0\|p_w) =\\ &\quad\underbrace{\int_{\mathcal{V}}\int_{\Theta} \int_{\X}  q(v) p_0(x|\theta) \pi(\theta) \left[ v^T (\Hessian_x \log p_w(x|\theta)) v +  \frac{1}{2} \left\|\nabla_x\log p_w(x|\theta) \right\|_2^2  \right] dx d\theta dv}_{J_S(w)} + C,
		\end{aligned}
	\end{equation}
	where $ C $ denotes constants with respect to $ w $. For these two expressions, it is immediate to obtain Monte Carlo estimates using samples $(\theta^{(j)}, x^j)_{j=1}^N, \theta^{(j)} \sim \pi$ and $x^{(j)} \sim p(\cdot|\theta^{(j)})$, and draws from the noise model $ \{v^{(j,k)}\}_{1\le j\le N, 1\le k\le M}$ : 
	\begin{equation}\label{Eq:FD_expected_explicit_MC}
		\begin{aligned}
			\hat J(w) &=\frac{1}{N} \sum_{j=1}^N \left[ \sum_{i=1}^d \left(  \frac{1}{2} \left( \frac{\partial \log p_w(x^{(j)}|\theta^{(j)})}{\partial x_i}  \right)^2 +  \left( \frac{\partial^2 \log p_w(x^{(j)}|\theta^{(j)})}{\partial x_i^2}  \right) \right) \right] , \\
			\hat J_S(w) &=\\ &\frac{1}{NM} \sum_{j=1}^N \sum_{k=1}^M \left[ v^{(j,k),T} (\Hessian_x \log p_w(x^{(j)}|\theta^{(j)})) v^{(j,k)} +  \frac{1}{2} \left\|\nabla_x\log p_w(x^{(j)}|\theta^{(j)}) \right\|_2^2 \right].
		\end{aligned}
	\end{equation}
	
	As we discussed before, computing the square bracket term in $ \hat J_S(w)  $ only requires two backward propagations, while the square bracket term in $ \hat J(w)  $ requires $ d $. However, $ \hat J_S(w)  $ sums over $ M\cdot N $ terms, while $ \hat J(w) $ sums over $ N $. Nevertheless, when using $ \hat J_S(w)  $ to train a neural network, good results can be obtained by using a single different noise realization for each training sample $ (\theta^{(j)}, x^j) $ at each epoch, thus leading to smaller computational cost with respect to $ \hat J(w) $ \citep{song2019sliced}.

	\paragraph{Score Matching for conditional exponential family.}
	Both the implicit and explicit versions $ D_F^E(p_0\|p_w) $ and $ D_{FS}^E(p_0\|p_w) $ are well defined only if $ p_w  $ is a proper density, which requires $ Z_w(\theta) <\infty$. In practice, we are interested in:
	\begin{equation}\label{Eq:J_min_problem}
		\hat w = \argmin \hat J(w)\quad \text{ or } \quad \hat w = \argmin \hat J_S(w). 
	\end{equation}
	When we compute $ \hat J(w) $ or $ \hat J_S(w) $,  we only need to evaluate $ f_w $ in a finite set of points $ \{x^{(j)}\}_{j=1}^N $ which belong to a bounded subset of $ \X $, say $ A\subset \X $. We can thus redefine the approximating family as follows: 
	
	\begin{equation}\label{}
		p'_w(x|\theta) = \frac{\tilde p'_w(x|\theta)}{Z'_w(\theta)}, \quad \tilde p'_w(x|\theta) = \begin{cases}
			&\exp (\eta_w(\theta)^T f_w(x)) \quad \text{ if }x \in A \\
			& g_w(x,\theta) \quad \text{ otherwise,}
		\end{cases}
	\end{equation}
	where $ g_w $ is such that $ Z'_w(\theta) < \infty $, and can always be chosen such that $ \tilde p'_w(x|\theta) $ is a smooth and continuous function of $ x $.
	
	Replacing $ p_w $ in Eq.~\eqref{Eq:FD_expected_explicit_MC} with $ p'_w $ does not change the value of $ \hat J(w) $ and $ \hat J_S(w) $; however, as $ p'_w $ is normalized, Eqs.~\eqref{Eq:exp_fisher_div} and \eqref{Eq:exp_fisher_div_exp} are now well defined for all $ w $. Additionally, we do not need to specify $ A $ or $ g_w $ explicitly as we never evaluate the normalizing constant $ Z'_w(\theta) $ (which depends on them). In the following we will thus use interchangeably $ p_w $ and $ p'_w $.

	\begin{remark}[\textbf{Notation}] For notational simplicity, in the rest of the work we drop the hat in $ \hat w $, denoting by $ p_w $ the approximation obtained by one of the strategies in Eq.~\eqref{Eq:J_min_problem}, and by $ f_w $ and $ \eta_w $ the corresponding sufficient statistics and natural parameters networks.	
	\end{remark}	
	
	\section{Inference using the likelihood approximation}\label{sec:inference_strategy}

	By fitting the neural conditional exponential family (Eq.~\ref{Eq:expon_family_NNs}) with SM or SSM to parameter-simulation pairs generated from the model, we obtain an approximation of the likelihood up to the normalization constant $ Z_w(\theta) $: for each fixed $ \theta $, the function $ x \mapsto \exp(\eta_w(\theta)^T f_w(x)) $ is approximately proportional to $ p_0(x|\theta) $. We exploit $ p_w $ in two ways: first, by using $ f_w $ as summaries in ABC; secondly, using the full approximation to draw samples from the posterior with MCMC for doubly intractable distributions. Both approaches are illustrated in Figure~\ref{fig:diaram} and discussed next.

	\begin{figure}[!tb]
		\centering
		\includegraphics[width=1\linewidth]{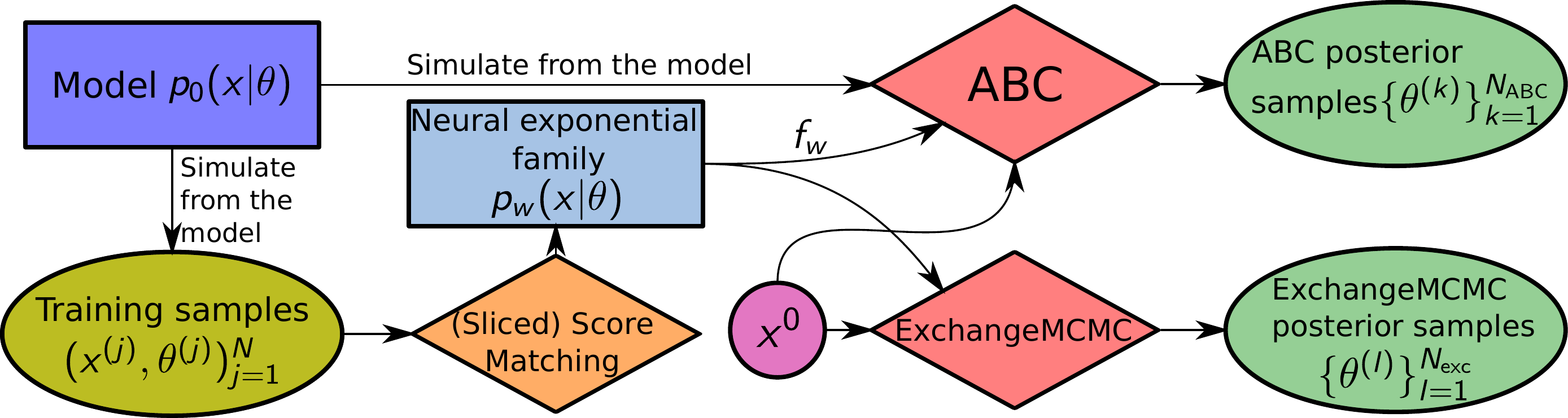}
		\caption{\textbf{Diagram showing the roles of the different components in the inference routines}. SSM or SM are used to train the neural exponential family on model simulations. Then, the sufficient statistics $ f_w $ are used in ABC or, alternatively, the neural exponential family is inserted in ExchangeMCMC. Notice how additional model simulation are needed for ABC but not for ExchangeMCMC.}
		\label{fig:diaram}
	\end{figure}

	\subsection{ABC with neural conditional exponential family statistics}\label{sec:inf_ABC_exp_fam}
	
	The exponential family is the most general set of distributions with sufficient statistics of a given size (see Appendix~\ref{app:suff_stats}). Using $ f_w $ as summaries in ABC is therefore intuitively appealing: $ f_w $ represents in fact the sufficient statistics of the best exponential family approximation to $ p_0 $, according to the (sliced) Fisher divergence. If $ p_0 $ belongs to the exponential family, $ D_F^E(p_0\|p_w) =  D_{FS}^E(p_0\|p_w) = 0 $ if $ f_w $ and $ \eta_w $ are sufficient statistics and natural parameters for $ p_0 $.
	
	To use $ f_w$ in ABC, some practicalities are needed: first, we discard the last component of $ f_w $, which represents the base measure $ h_w(x) $. Secondly, as discussed in Section~\ref{sec:exp_fam}, $ f_w $ is identifiable only up to a scale parameter, so that the magnitude of the different components of $ f_w $ can vary significantly. Before using ABC, then, we rescale the different components by their standard deviation on new samples generated from the model, to prevent the ones with larger magnitude from dominating the ABC distance.
	
	In the following, we will call the approach described here ABC with Score Matching statistics, for short ABC-SM, or ABC-SSM in the Sliced Score Matching case.

	\subsection{ExchangeMCMC with neural conditional exponential family}\label{sec:inf_doubly_intractable}
	
	In contrast to other statistics learning methods \citep{jiang2017learning, wiqvist2019partially, pacchiardi2020, aakesson2020convolutional}, our technique provides a full likelihood approximation; it is therefore tempting to sample from the corresponding posterior directly, bypassing in this way ABC (and the additional model simulations required) altogether.
	
	Unfortunately, $ p_w $ is known only up to the normalizing constant $ Z_w(\theta) $; therefore, standard MCMC cannot be directly exploited , and methods for \textit{doubly-intractable} distributions are required (see \citealp{park2018bayesian} for a review). Here, we use ExchangeMCMC (\cite{murray2012mcmc}, Algorithm~\ref{alg:ExchangeMCMC} in Appendix~\ref{app:exchange}), an MCMC where, for each proposed $ \theta' $, a simulation from the distribution $ p_w(x|\theta') $ is used within a Metropolis-Hastings rejection step. 
	
	In our case, we cannot generate samples from $ p_w(\cdot|\theta') $ directly; to circumvent such issue, \cite{murray2012mcmc} suggested to run an MCMC targeting $ p_w(\cdot|\theta') $ for each ExchangeMCMC step; if the chain is long enough, the last step can be considered as a draw from $ p_w $. Empirical results \citep{caimo2011bayesian, alquier2016noisy, everitt2012bayesian, liang2010double} show that a relatively small number of inner MCMC steps are enough for good performance and that initializing the inner chain at the observation improves convergence; we employ therefore this strategy in our work.

	Further, a variant of ExchangeMCMC employs Annealed Importance Sampling to improve the mixing of the chain \citep{murray2012mcmc}. Specifically, a sequence of $ K $ auxiliary variables are drawn from Metropolis-Hastings kernels targeting some intermediate distributions, and the ExchangeMCMC acceptance rate is modified accordingly. This approach, termed \textit{bridging}, decreases the number of rejections in ExchangeMCMC due to poor auxiliary variables. See Appendix~\ref{app:exchange} for more details.

	In the following, we will refer to using ExchangeMCMC with our likelihood approximation as Exc-SM or Exc-SSM, according to whether SM or SSM are used to obtain the likelihood approximation.
	Exc-SM and Exc-SSM avoid additional model simulations (beyond the ones required to train $ p_w $) at the cost of running a nested MCMC. However, the number of steps in the inner MCMC required to obtain good performance increases with the dimension of $ \mathcal{X} $. Exc-SM and Exc-SSM are therefore ideally applied to models with moderate dimension $ x $ (up to a few hundred), for which simulation is expensive. The same likelihood approximation can be used to perform inference on several observations, which makes the computational gain even greater in this case.
	
	\begin{remark}[Model misspecification]
		The neural exponential family approximation is only valid close to where training data was distributed. Specifically, if $ \dataObs $ is far from that region, $ p_{\hat w}(\dataObs|\theta) $ may be assigned a large value rather than a (correct) small one. This could happen when the model $ p_0 $ is unable to generate data resembling the observation for any parameter value, such that standard Bayesian inference in presence of the likelihood would also perform poorly. To get better inference in such scenarios, we could use the generalized posterior introduced in \cite{matsubara2021robust}, which is robust to outliers and suitable for doubly-intractable distributions (as the exponential family).
		
		In Exc-SM and Exc-SSM, we may wonder whether running MCMC over $ x $ targeting $ p_{w}(x|\theta) $ for a fixed $ \theta $ can fail due to what we discussed above.
		We believe this is not the case: if the MCMC is initialized close to training data and $ p_{ w}(\cdot|\theta) $ is a good representation of $ p_0(\cdot|\theta) $ in that region, $ p_{w}(\cdot|\theta) $ is small for values of $ x $ close to the boundary of the region where training data was distributed. Then, the MCMC is ``trapped'' inside that region and has no way of reaching regions of $ \X $ where $ p_{w} $ may behave irregularly.
	\end{remark}

	\section{Simulation Studies}\label{sec:experiments} 
	
	We perform simulation studies with our proposed approaches (Exc-SM, Exc-SSM, ABC-SM and ABC-SSM) and we compare with:
	
	\begin{itemize}
		\item ABC with statistics learned with the rejection approach discussed in Section~\ref{sec:ABC} \citep{fearnhead_constructing_2012, jiang2017learning}, which we term ABC-FP.
		\item Population Monte Carlo (PMC, \citealt{cappe2004population}) with Ratio Estimation (PMC-RE).
		\item PMC with Synthetic Likelihood, using the robust covariance matrix estimator developed in \cite{ledoit2004well} to estimate $ \Sigma_\theta $. We will denote this as PMC-SL.
	\end{itemize}
	
	Specifically, we consider three exponential family models and two time-series models (AR(2) and MA(2)) for which the exact likelihood is available, as well as a large-dimensional model with unknown likelihood (the Lorenz96 model, \citealt{lorenz1996predictability, wilks2005effects}).
	
	Exc-SM and Exc-SSM do not require additional simulations and run an MCMC, in contrast to sequential algorithms for the other methods, which we run with parallel computing. Comparing the computational cost is therefore not easy; in Appendix~\ref{app:evolution_WD_iter}, we report the number of simulations needed by the different methods to reach the same performance achieved by Exc-SM for the models with known likelihood.

	\paragraph{Choice of neural network architecture} 
	In our exponential family approximation, we fix $ d_s $ to the number of parameters in each model. We added a Batch Normalization layer (see Appendix~\ref{app:batch_norm}) on top of the neural network representing $ \eta_w $ to reduce the unidentifiability introduced by the dot product between $ f_w(x) $ and $ \eta_w(\theta) $ (as discussed in Appendix~\ref{app:sm_exp_fam}). For all techniques, we use $ 10^4 $ training samples. All NNs use SoftPlus nonlinearity (NNs using the more common ReLU nonlinearity cannot be used with SM and SSM as they have null second derivative with respect to the input).
	
	For all models, $ \eta_w $ is represented by fully connected NNs. For the exponential family models, fully connected NNs are also used for $ f_w $ and $ s_\beta $.
	For the time series and Lorenz96 models, we parametrize $ f_w $ and $ s_\beta $ with Partially Exchangeable Networks (PENs, \citealt{wiqvist2019partially}). The output of an $ r $-PEN is invariant to input permutations which do not change the probability density of data distributed according to a Markovian model of order $ r $ (see Appendix \ref{app:PENs}).
	As AR(2) is a 2-Markovian model, we use a $ 2$-PEN; similarly, a 1-PEN is used for the Lorenz96 model, which is 1-Markovian. Finally, we use a $ 10 $-PEN for the MA(2) model; albeit not being Markovian, \cite{wiqvist2019partially} argued that MA(2) can be considered as ``almost'' Markovian so that the loss of information in imposing a PEN structure of high enough order is negligible.
	Further details are given in Appendix~\ref{app:experim_details}. 
	
	\paragraph{Choice of inferential parameters.}
	For ABC, we employ Simulated annealing ABC (SABC, \citealt{albert_2015}), which considers a set of particles and updates their position in parameter space across several iterations. We use 100 iterations and 1000 particles (posterior samples), corresponding to 1000 model simulations per iteration. 
	
	In Exc-SM and Exc-SSM, we run Exchange MCMC for 20000 steps, of which the first 10000 are burned-in. During burn-in, at intervals of 100 outer steps, we tune the proposal sizes according to the acceptance rate in the previous 100 steps.
	Our implementation of the Exchange algorithm is detailed in Appendix~\ref{app:exchange_implementation}.
	For the exponential family and time-series models, we test different numbers of inner MCMC steps, and eventually use $ 30 $ for the former and $ 100 $ for the latter, above which there was no noticeable performance improvement (more details in Appendix~\ref{app:inner_steps_exp_fam_models} and Appendix~\ref{app:inner_MCMC_arma}).

	In PMC-SL and PMC-RE, we run the PMC algorithm with 10 iterations, 1000 posterior samples and respectively 100 (with SL) and 1000 (with RE) simulations per parameter value in order to estimate the approximate likelihood; such a large number of simulations (respectively $ 10^5 $ and $ 10^6 $ for each iteration) is required for the likelihood estimate to be numerically stable. For the exponential family models, we use the true sufficient statistics; for AR(2) and MA(2), we instead use autocovariances with lag 1 and 2 (as for instance in \citealt{marin2012approximate}). For PMC-RE, the cross-product of the statistics is also added to the list of statistics, as PMC-SL implicitly uses it.

	\subsection{Exponential family models}
	First, we consider three models for which a sample is defined as a 10-dimensional independent and identical distributed (i.i.d.) draw from either a Gaussian, Gamma or Beta distribution (all belonging to the exponential family). We put uniform priors on the parameters, with bounds given in Table~\ref{Tab:priors}. These models have different data ranges: unbounded for Gaussian, and respectively lower bounded by 0 and bounded in $ [0,1] $ for Gamma and Beta. Therefore, we directly apply SM and SSM to the Gaussian model, while TranSM and TranSSM are applied to Gamma and Beta.

	\begin{table}[!tb]
		\centering
		
		\begin{tabular}{ lx{0.46cm}x{0.46cm}x{0.46cm}x{0.46cm}x{0.46cm}x{0.46cm}x{0.46cm}x{0.46cm}x{0.46cm}x{0.46cm}x{0.46cm}x{0.46cm}x{0.46cm}x{0.46cm}}
			\toprule
			\textit{Model} & \multicolumn{2}{c}{Beta}  &  \multicolumn{2}{c} {Gamma\footnote{}} &  \multicolumn{2}{c} {Gaussian} &  \multicolumn{2}{c} {AR(2)} &  \multicolumn{2}{c} {MA(2)} &  \multicolumn{4}{c} {Lorenz96} \\ \cmidrule(r){2-3} \cmidrule(r){4-5} \cmidrule(r){6-7} \cmidrule(r){8-9} \cmidrule(r){10-11} \cmidrule(r){12-15} \textit{Parameter}
			& $\alpha$ & $ \beta $ & $ k $ & $ \theta $ & $ \mu $ & $ \sigma $ & $ \theta_1 $ & $ \theta_2 $ &$ \theta_1 $ & $ \theta_2 $ & $ b_0 $ & $ b_1 $ & $ \sigma_e $ & $\phi$ \\ 
			\midrule
			\textit{Lower bound} & 1 & 1 & 1 & 1 &$-10 $& $ 1 $ & $ -1 $ & $ -1 $ & $ -1 $ & 0 & 1.4 & 0 & $ 1.5 $ & $ 0 $ \\ 
			\textit{Upper bound} & 3 & 3 & 3 & 3 & $ 10 $ & $ 10 $ & $ 1 $ & $ 0 $ & $ 1 $ & 1 & 2.2 & 1 & $ 2.5 $ & $ 1 $ \\ 
			\bottomrule
		\end{tabular}

		\caption{\textbf{Bounds of uniform priors for the considered models.}}
		\label{Tab:priors}	
	\end{table}
	
	\footnotetext{The scale parameter of the Gamma distribution is usually called $ \theta $; in contrast, in all the rest of this work we use $ \theta  $ to denote all parameters. Please beware of the difference.}

	\paragraph{Inferred sufficient statistics and natural parametrization.}
	With these models, our exponential family approximation is well specified. Thus, we compare the learned $ f_w $ and $ \eta_w $ with the exact sufficient statistics and natural parameters using the Mean Correlation Coefficient (MCC, Appendix~\ref{app:MCC}) metric, which ranges in $ [0,1] $ and measures how well a multivariate function is recovered. We report results obtained with SM in Table~\ref{Tab:MCC_toy}; the values are close to $ 1 $, indicating that the sufficient statistics and natural parameters are recovered quite well by our method. In Appendix~\ref{app:exp_fam_models}, we report values corresponding to SSM (Table~\ref{Tab:MCC_toy_SSM}), as well as comparisons of exact and learned embeddings in Figures~\ref{fig:statistics_exp_fam_models} and \ref{fig:statistics_exp_fam_models_SSM}.
	\begin{table}[htbp]
		\centering
		\resizebox{\linewidth}{!}{\begin{tabular}{lx{2.73cm}x{2.73cm}x{2.73cm}x{2.73cm}}
				\toprule
				\textit{Model} & \textit{MCC weak in} & \textit{MCC weak out} & \textit{MCC strong in} & \textit{MCC strong out} \\
				\midrule
				Beta (statistics) & 0.964 & 0.958 & 0.723 & 0.723 \\
				Beta (nat. par.) & 0.990 & 0.991 & 0.807 & 0.812 \\
				\midrule
				Gamma (statistics) & 0.911 & 0.924 & 0.894 & 0.883 \\
				Gamma (nat. par.) & 0.967 & 0.967 & 0.872 & 0.873 \\
				\midrule
				Gaussian (statistics) & 0.944 & 0.937 & 0.808 & 0.824 \\
				Gaussian (nat. par.) & 0.974 & 0.974 & 0.970 & 0.972 \\
				\bottomrule
		\end{tabular}}
		\caption{\textbf{MCC for exponential family models between exact embeddings and those learned with SM}. We show weak and strong MCC values; MCC is between 0 and 1 and measures how well an embedding is recovered up to permutation and rescaling of its components (strong) or linear transformation (weak); the larger, the better. ``in'' denotes MCC on training data used to find the best transformation, while ``out'' denote MCC on test data. We used 500 samples in both training and test data sets.}
		\label{Tab:MCC_toy}
	\end{table}

	\paragraph{Inferred posterior distribution.}
	Figure~\ref{fig:posteriors_exp_fam_models} shows posteriors obtained with the proposed methods, for a possible observation for each model; {we see that all approximate posteriors are remarkably close to the exact one; moreover, the results with SM and SSM are indistinguishable.}
	For all methods, we also estimate the Wasserstein distance between true and approximate posterior and compute the Root Mean Squared Error (RMSE) between mean of the true and approximate posterior; this is repeated for 100 simulated observations, with results reported in Figure~\ref{fig:boxplots_exp_fam_models}. ABC-FP is here the worst method; additionally, ABC-SM and ABC-SSM perform similarly to PMC-SL and PMC-RE. Finally, Exc-SM and Exc-SSM are marginally worse.

	\begin{figure}[!tb]
		\centering
		\begin{subfigure}{\textwidth}
			\centering
			\begin{subfigure}{0.2\textwidth}
				\centering
				\includegraphics[width=\linewidth]{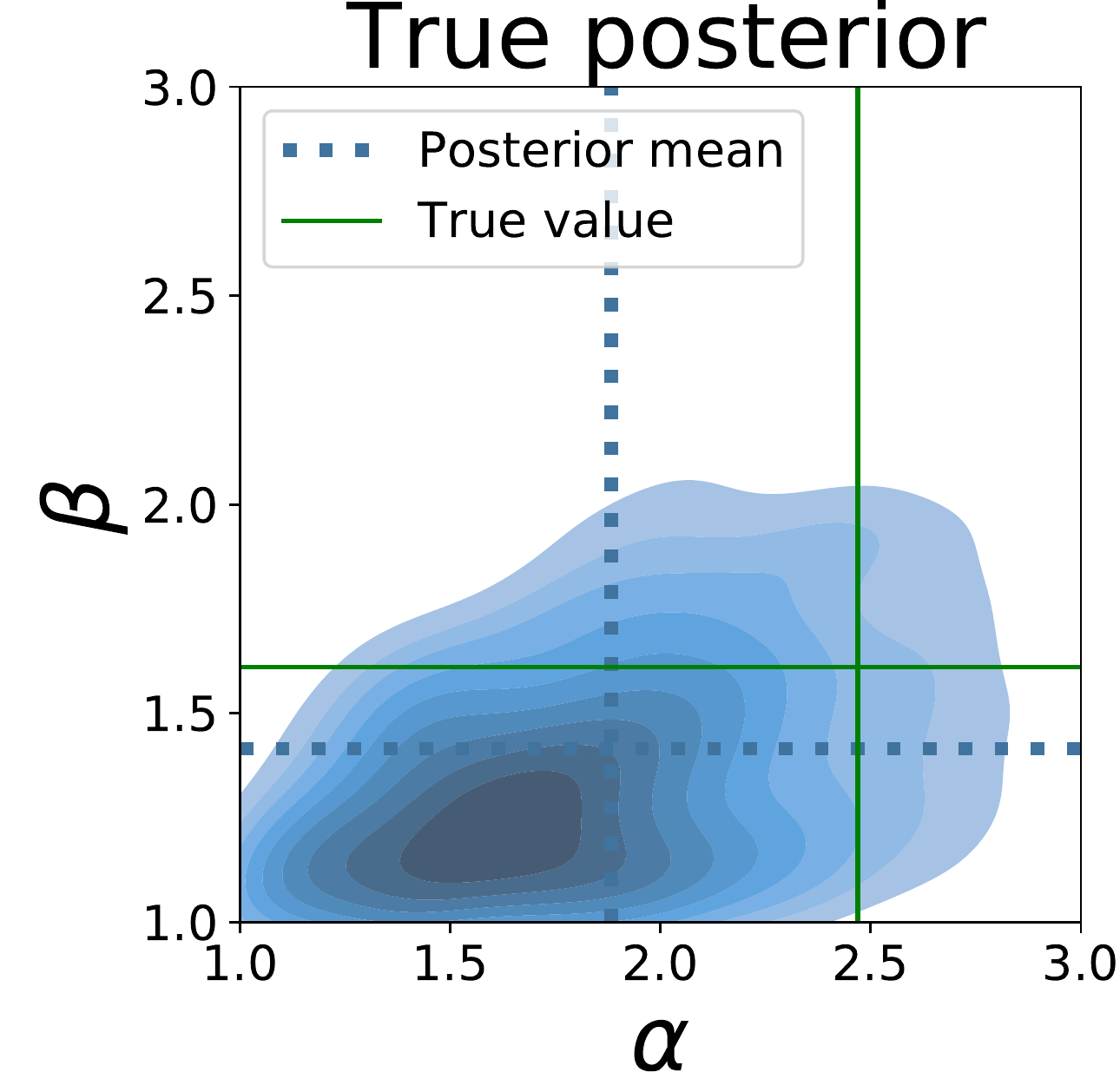}
			\end{subfigure}\begin{subfigure}{0.2\textwidth}
				\centering
				\includegraphics[width=\linewidth]{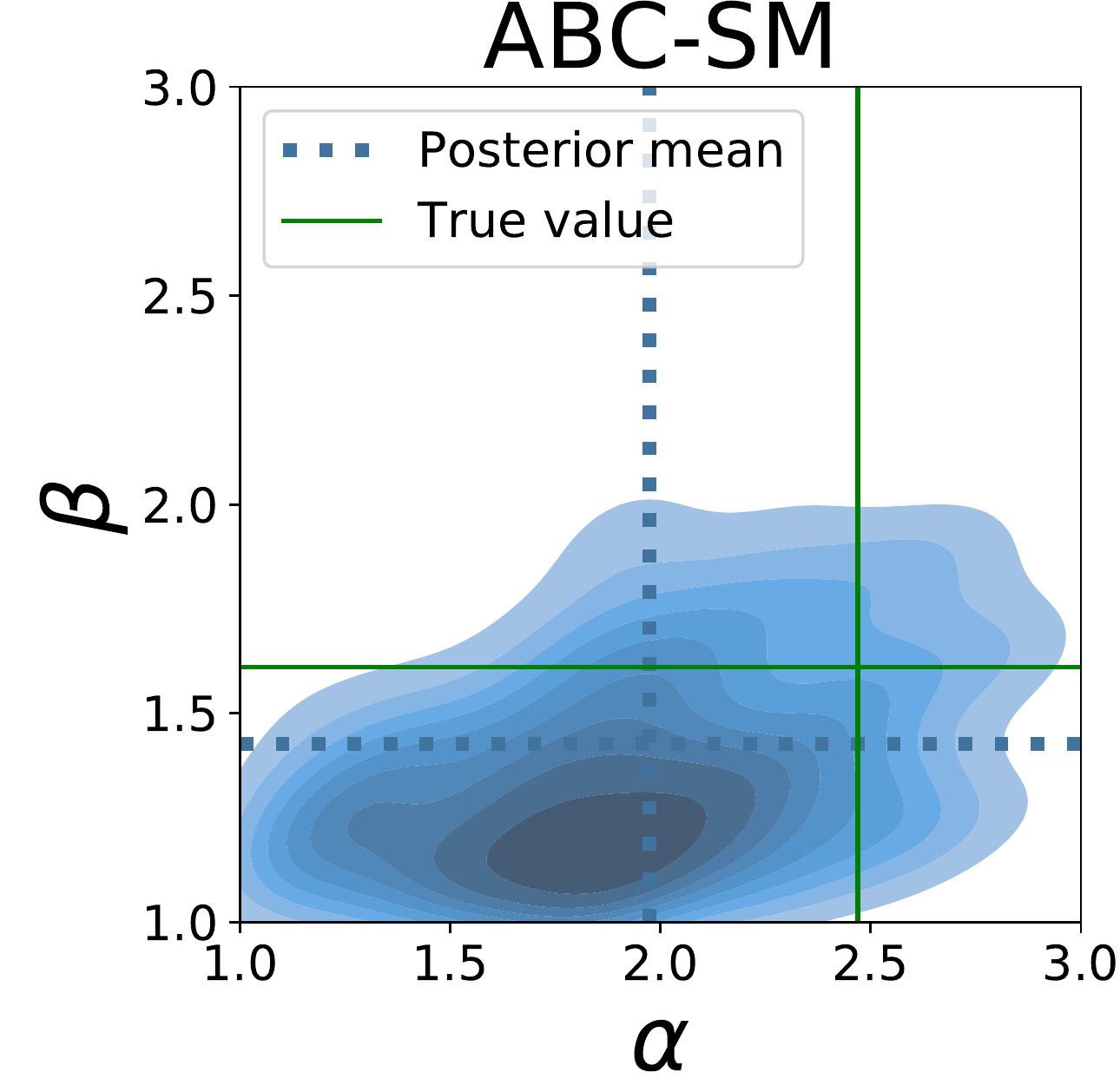}
			\end{subfigure}\begin{subfigure}{0.2\textwidth}
				\centering
				\includegraphics[width=\linewidth]{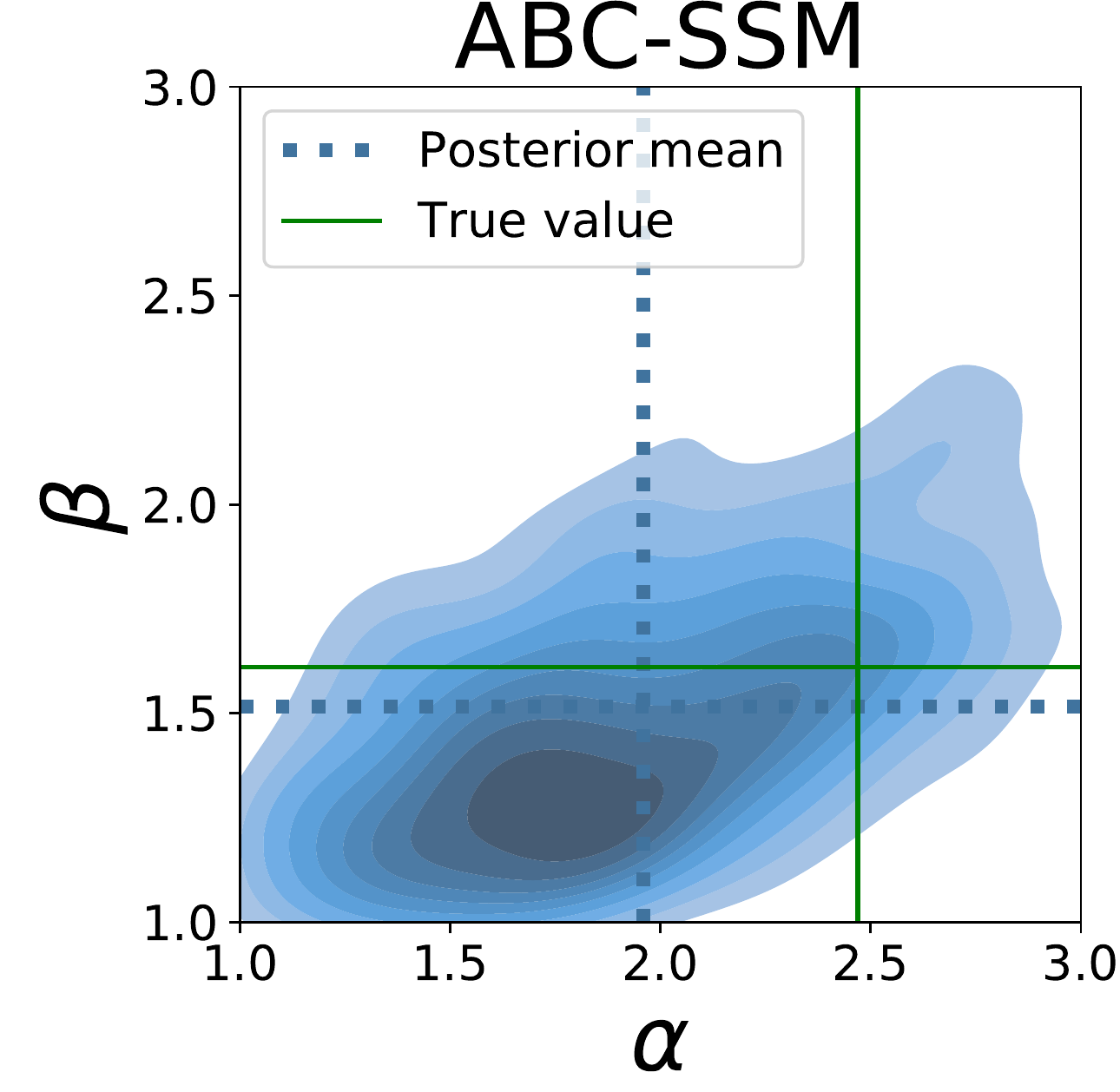}
			\end{subfigure}\begin{subfigure}{0.2\textwidth}
				\centering
				\includegraphics[width=\linewidth]{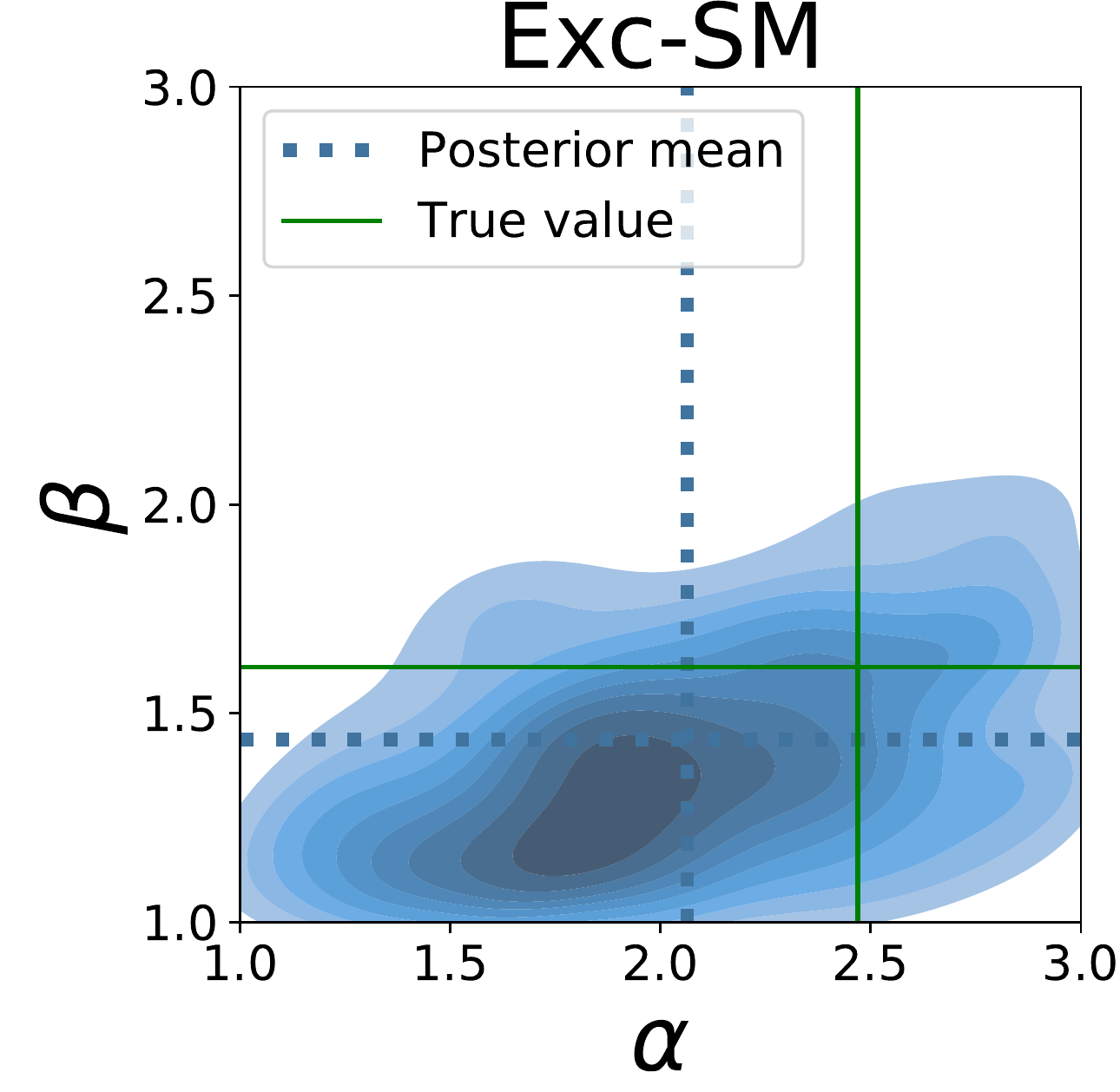}
			\end{subfigure}\begin{subfigure}{0.2\textwidth}
				\centering
				\includegraphics[width=\linewidth]{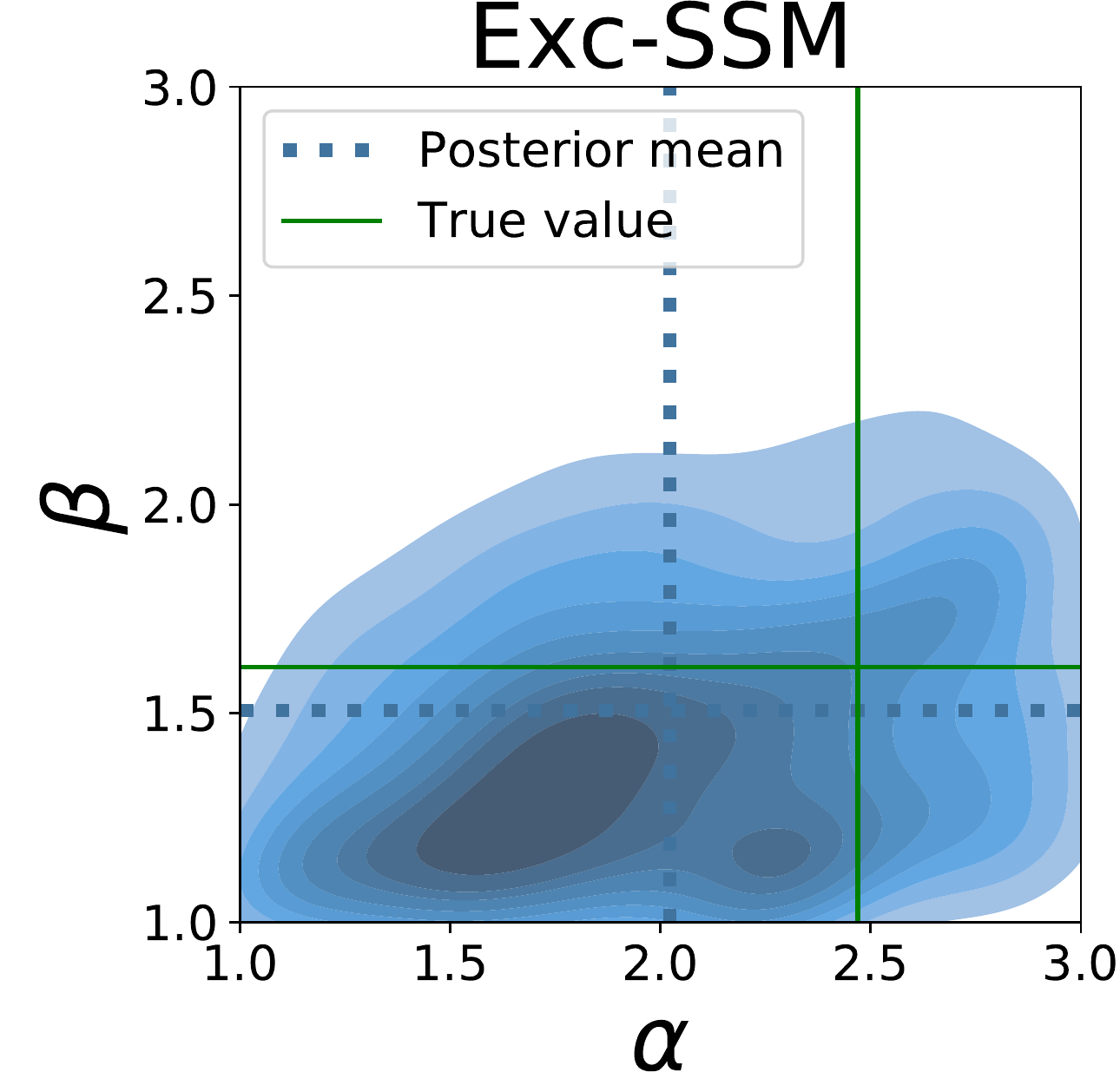}
			\end{subfigure}
			\caption{Beta}
		\end{subfigure}\\
		\begin{subfigure}{\textwidth}
			\centering
			\begin{subfigure}{0.2\textwidth}
				\centering
				\includegraphics[width=\linewidth]{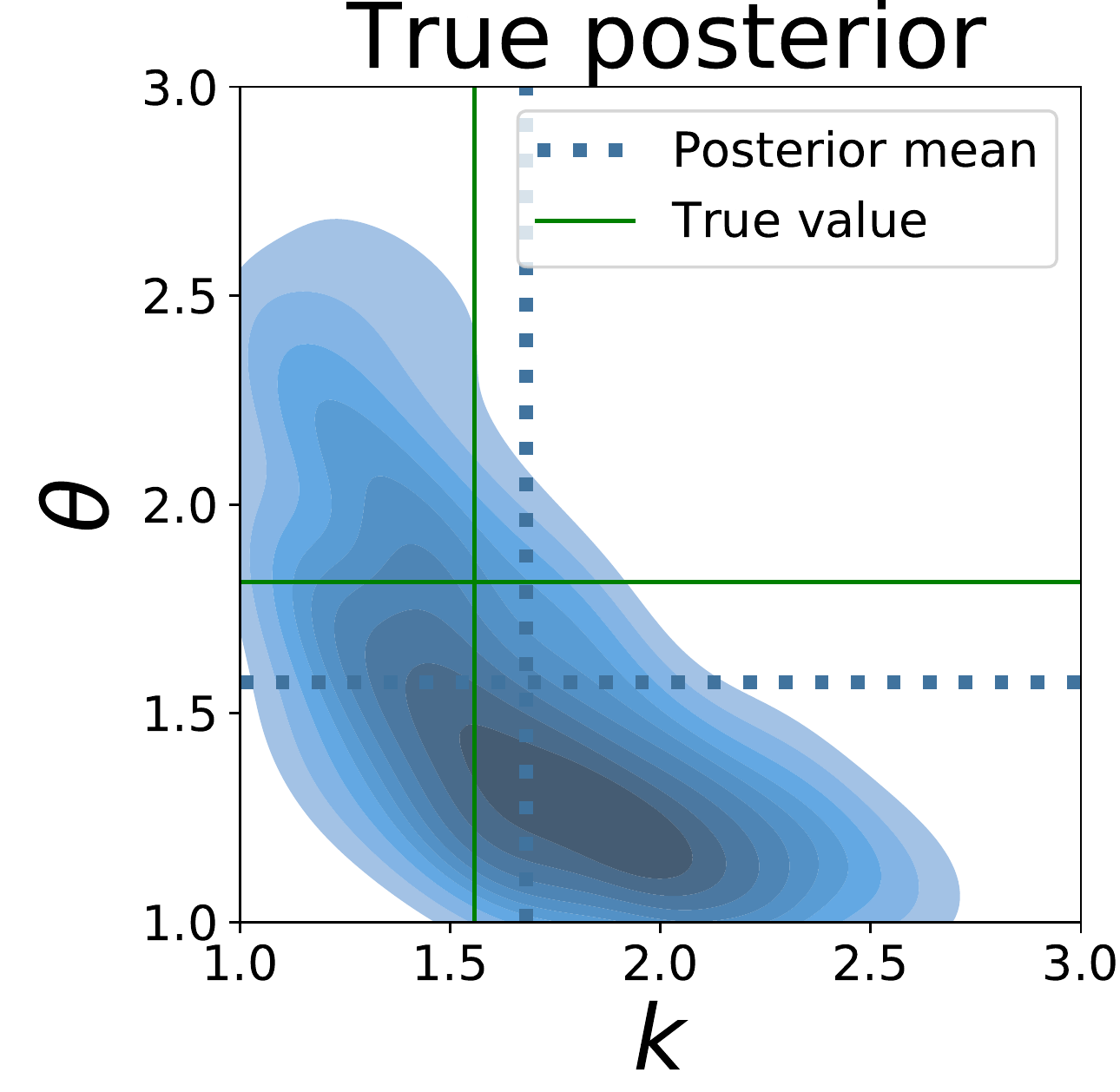}
			\end{subfigure}\begin{subfigure}{0.2\textwidth}
				\centering
				\includegraphics[width=\linewidth]{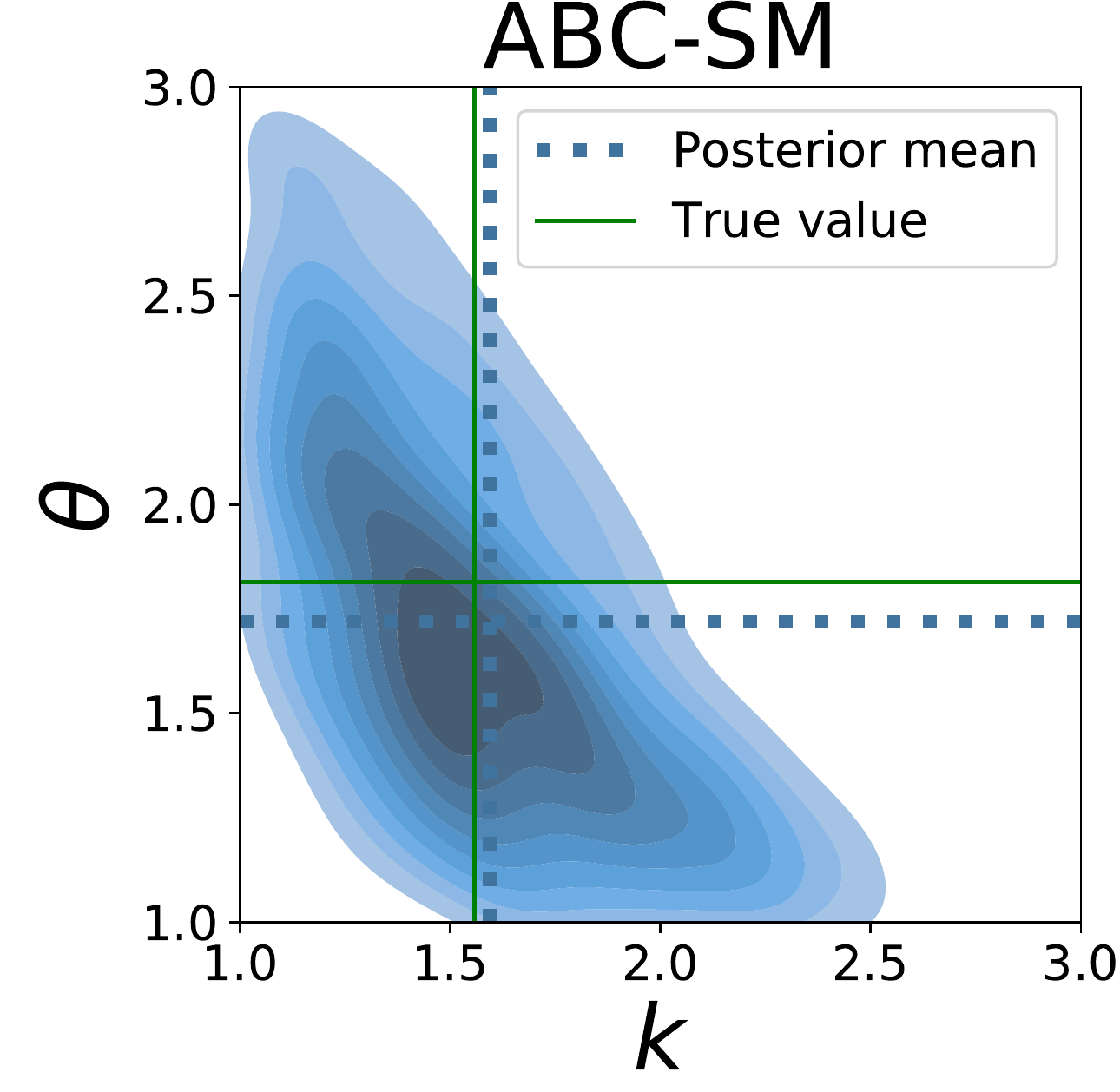}
			\end{subfigure}\begin{subfigure}{0.2\textwidth}
				\centering
				\includegraphics[width=\linewidth]{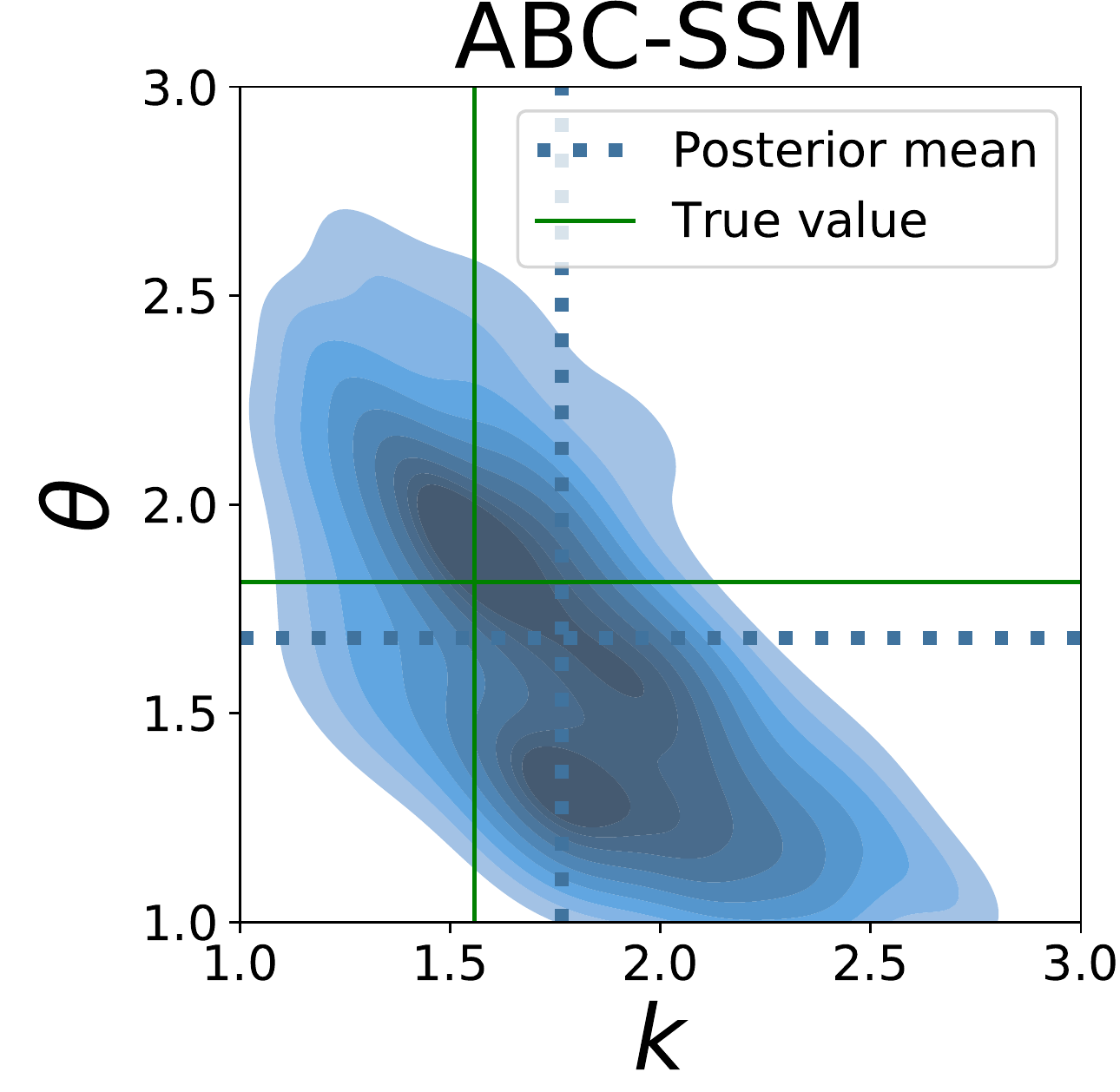}
			\end{subfigure}\begin{subfigure}{0.2\textwidth}
				\centering
				\includegraphics[width=\linewidth]{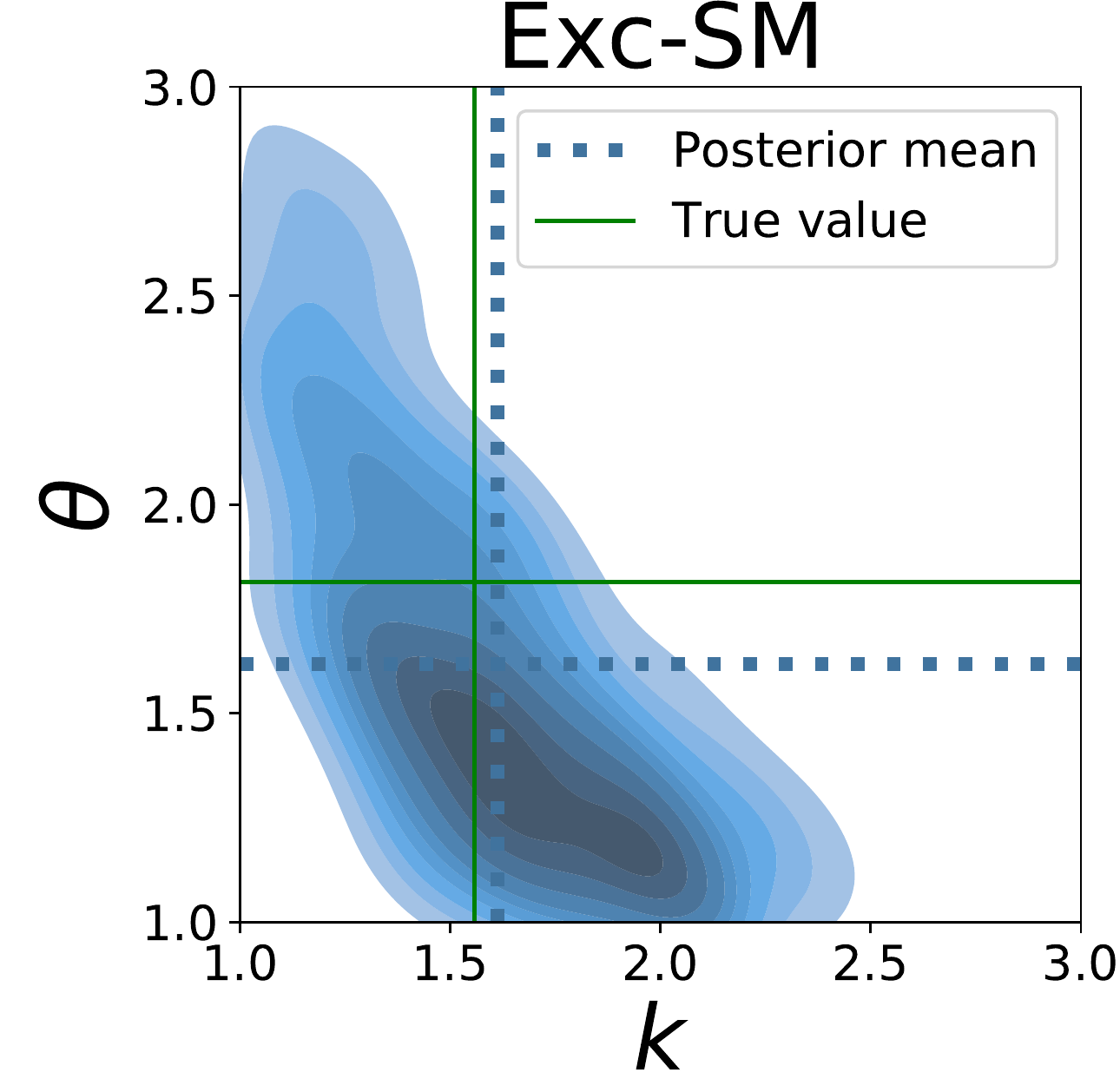}
			\end{subfigure}\begin{subfigure}{0.2\textwidth}
				\centering
				\includegraphics[width=\linewidth]{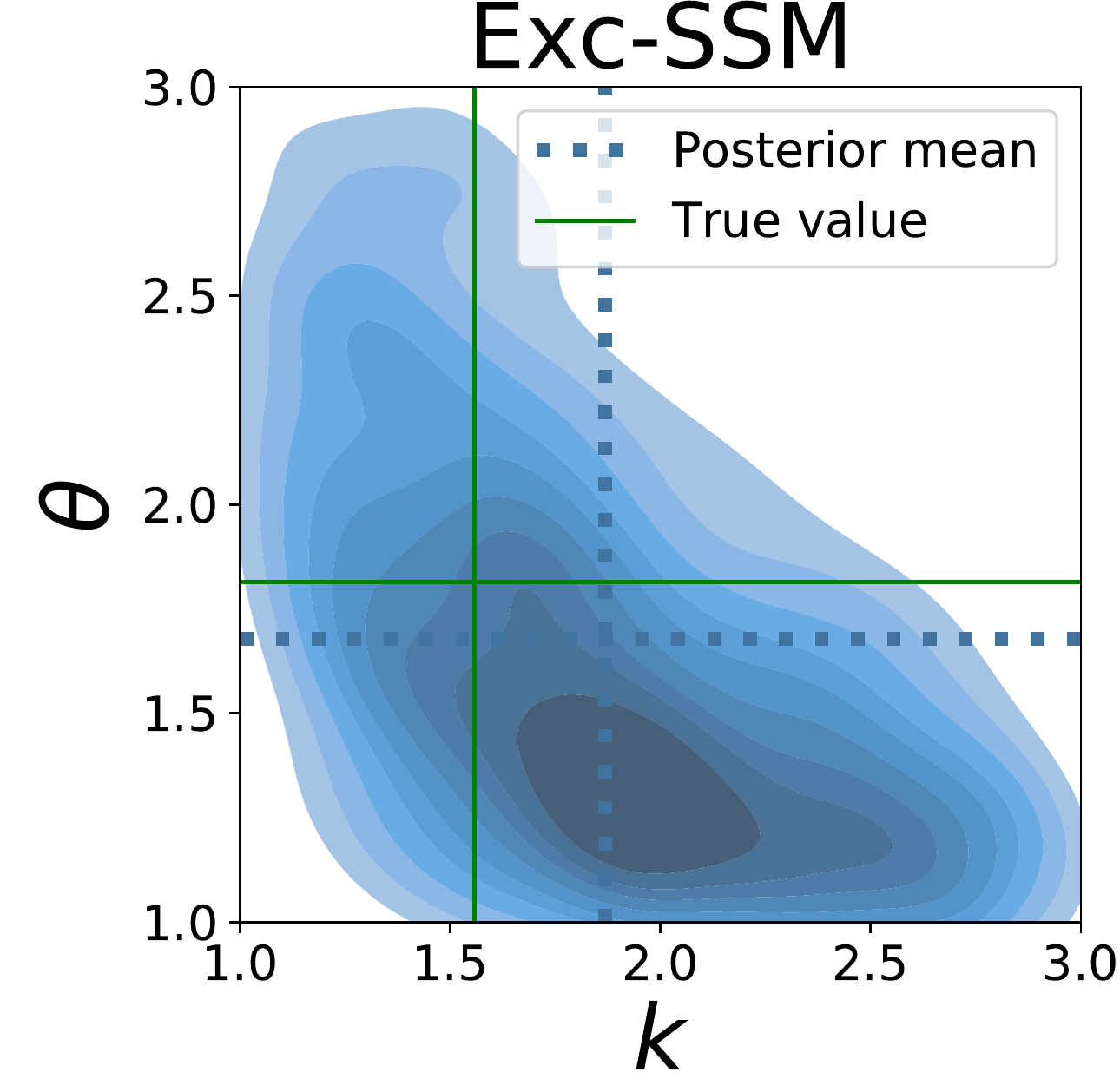}
			\end{subfigure}
			\caption{Gamma}
		\end{subfigure}\\
		\begin{subfigure}{\textwidth}
			\centering
			\begin{subfigure}{0.2\textwidth}
				\centering
				\includegraphics[width=\linewidth]{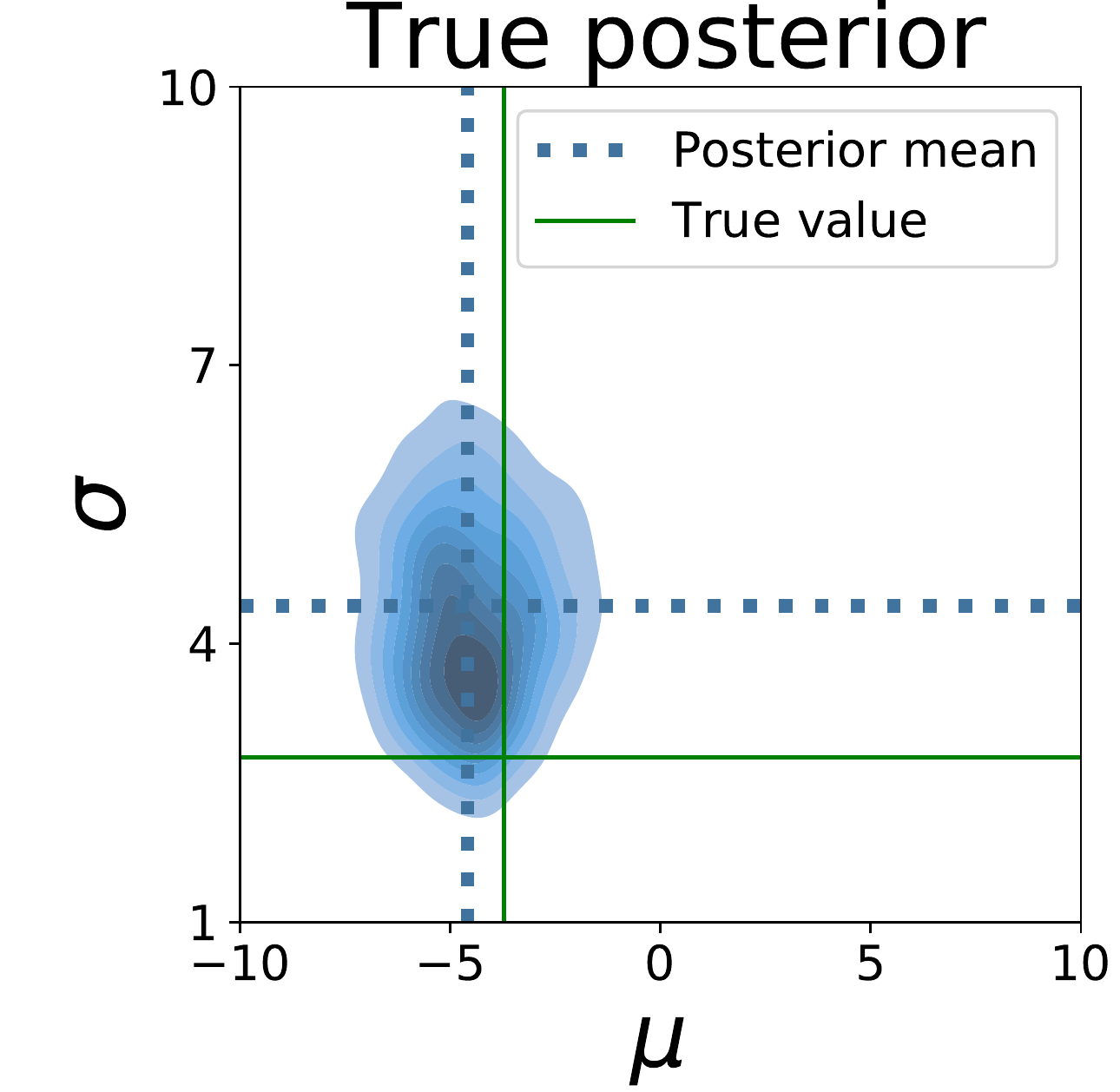}
			\end{subfigure}\begin{subfigure}{0.2\textwidth}
				\centering
				\includegraphics[width=\linewidth]{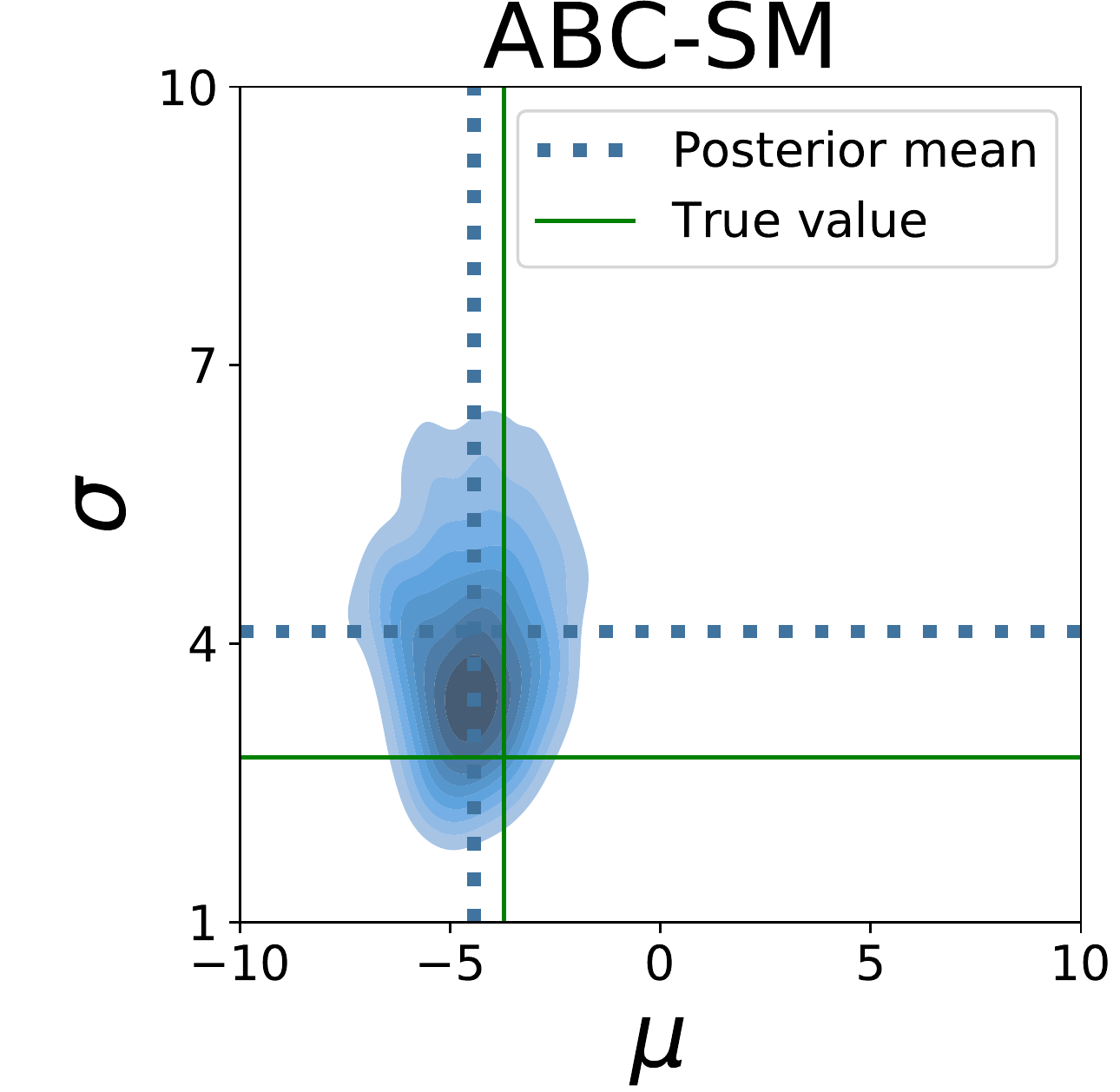}
			\end{subfigure}\begin{subfigure}{0.2\textwidth}
				\centering
				\includegraphics[width=\linewidth]{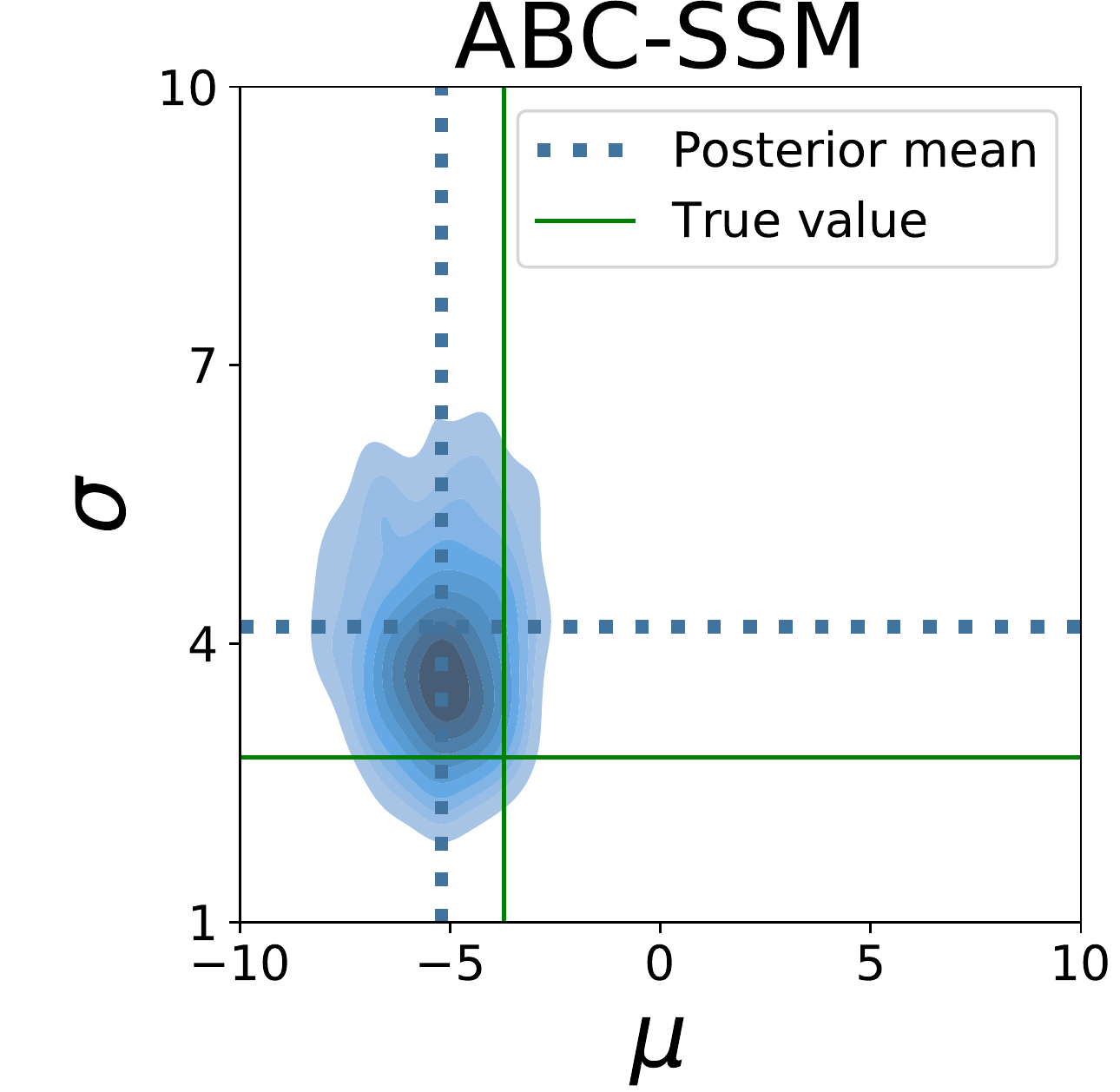}
			\end{subfigure}\begin{subfigure}{0.2\textwidth}
				\centering
				\includegraphics[width=\linewidth]{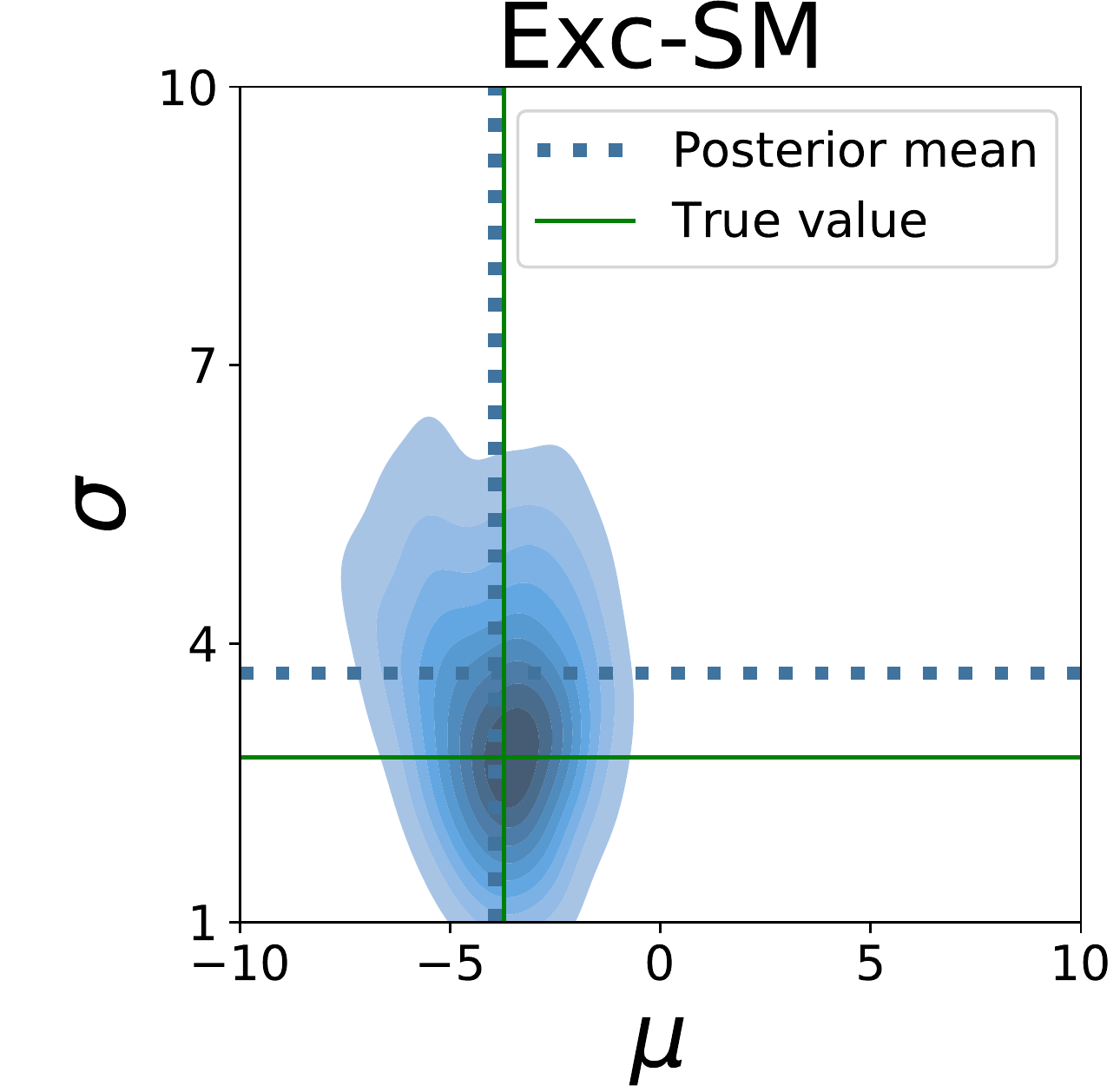}
			\end{subfigure}\begin{subfigure}{0.2\textwidth}
				\centering
				\includegraphics[width=\linewidth]{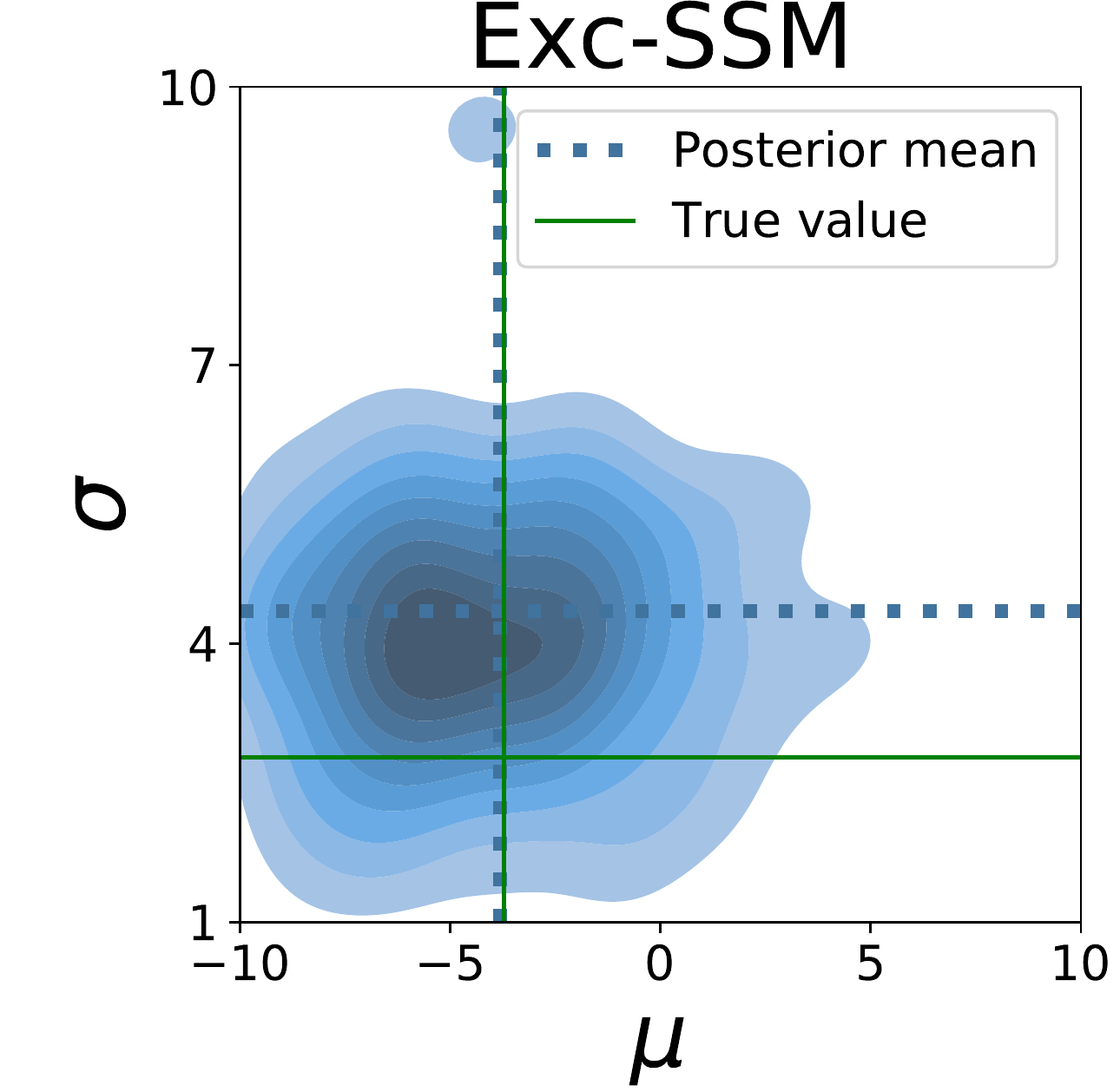}
			\end{subfigure}
			\caption{Gaussian}
		\end{subfigure}\\
		\caption{\textbf{True and approximate posteriors for exponential family models,} for a single observation per model. Dashed line represents posterior mean, while green solid line represents the exact parameter value.}		\label{fig:posteriors_exp_fam_models}
	\end{figure}
	
	\begin{figure}[!tb]
		\centering
		\begin{subfigure}{0.32\textwidth}
			\centering
			\includegraphics[width=1\linewidth]{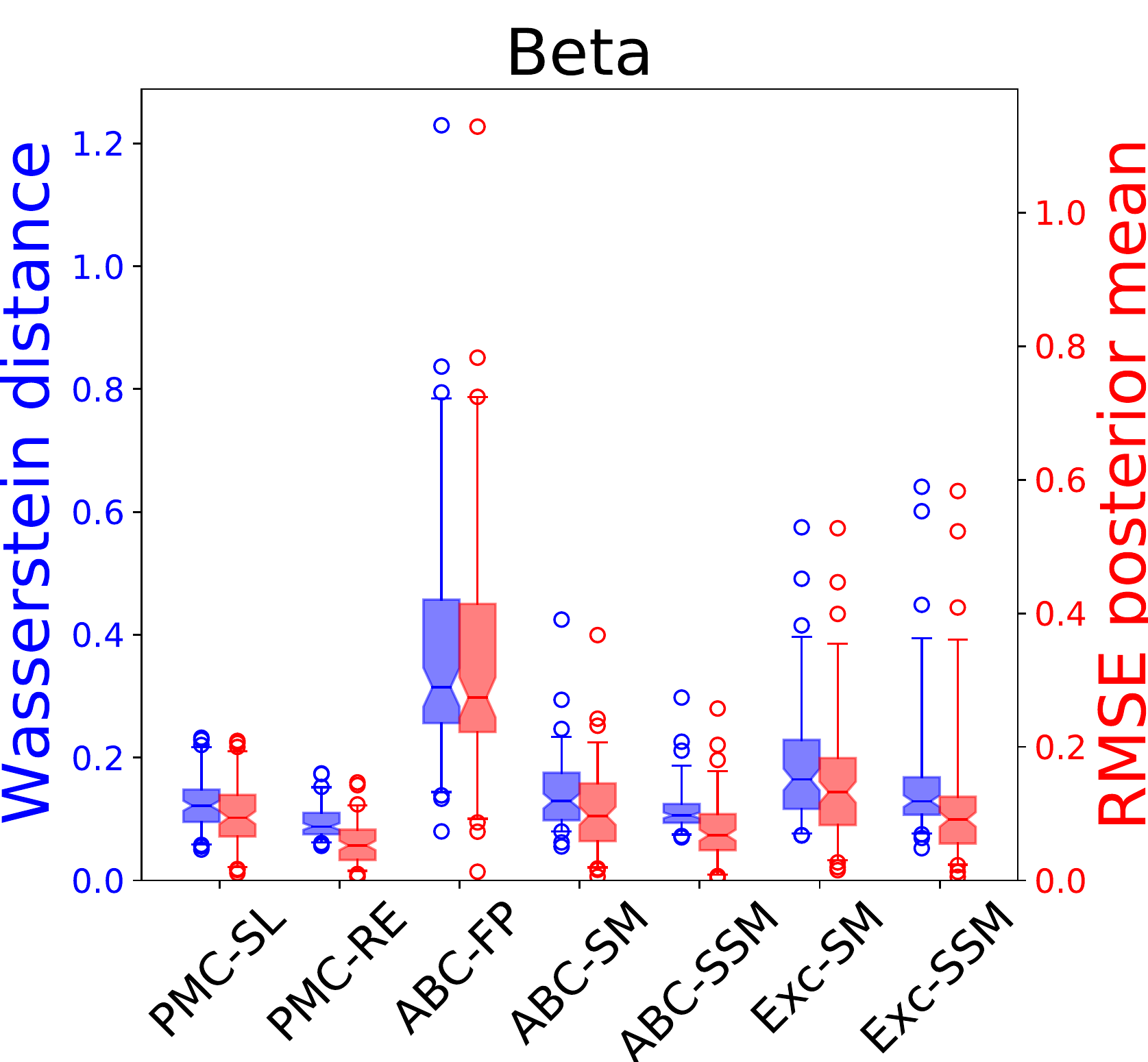}
			
		\end{subfigure}~
		\begin{subfigure}{0.32\textwidth}
			\centering
			\includegraphics[width=1\linewidth]{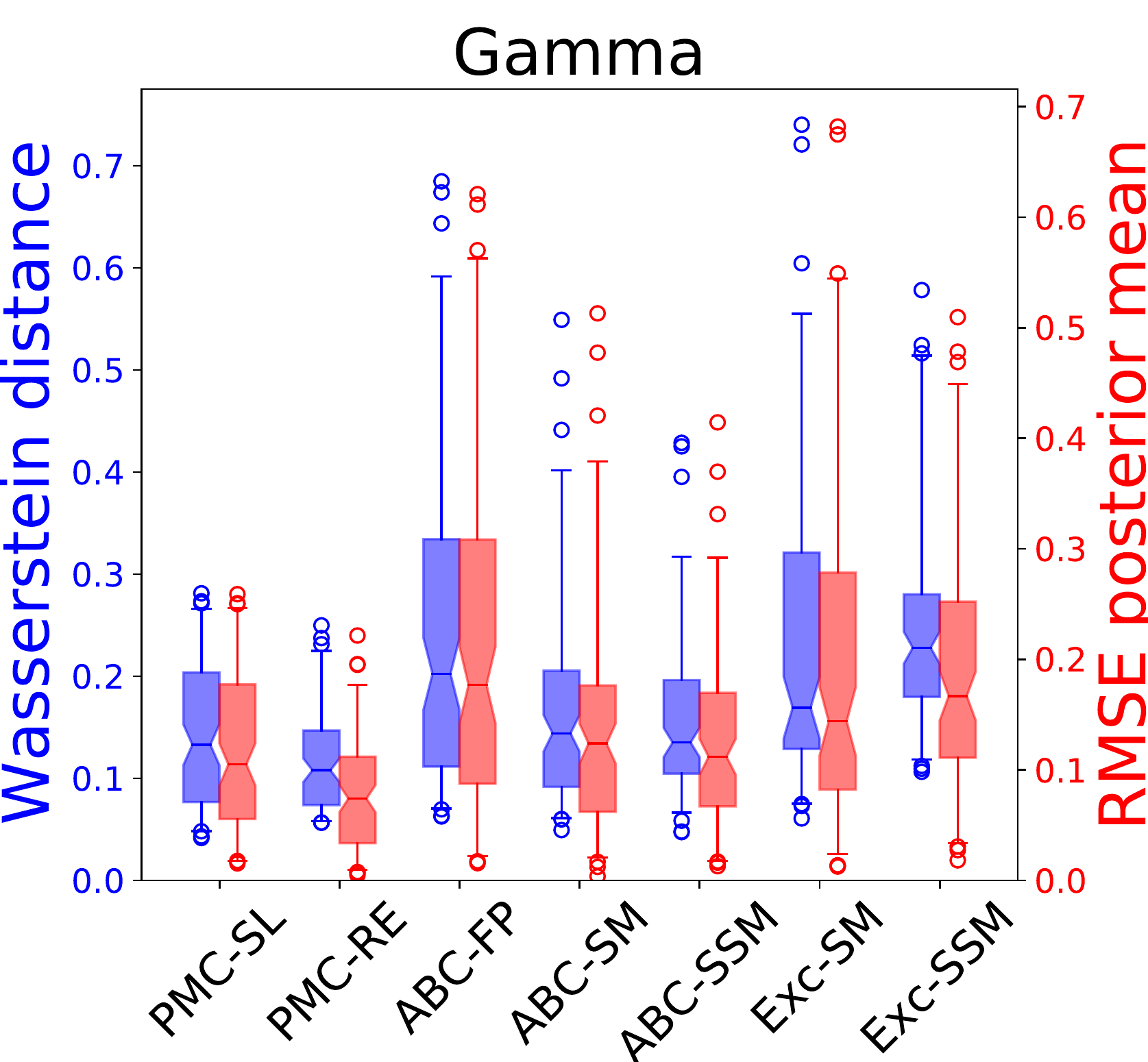}
			
		\end{subfigure}~
		\begin{subfigure}{0.32\textwidth}
			\centering
			\includegraphics[width=1\linewidth]{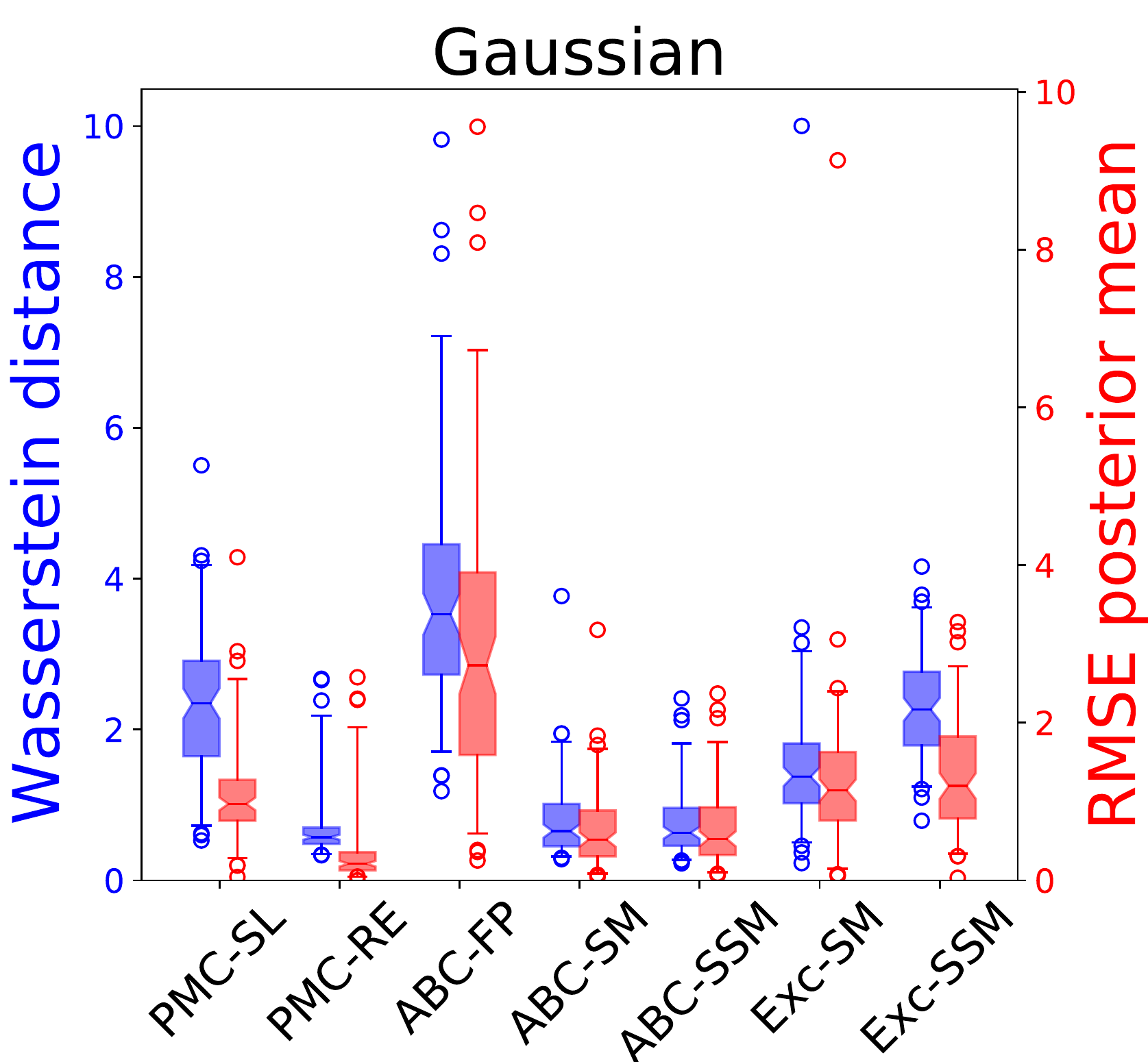}
			
		\end{subfigure}
		\caption{\textbf{Performance of the different techniques for exponential family models.} Wasserstein distance from the exact posterior and RMSE between exact and approximate posterior means are reported for 100 observations using boxplots. Boxes span from 1st to 3rd quartile, whiskers span 95\% probability density region and horizontal line denotes median. The numerical values are not comparable across examples, as they depend on the range of parameters. Here, SL and RE used the true sufficient statistics.}
		\label{fig:boxplots_exp_fam_models}
	\end{figure}

	\subsection{Time series models}
	The Moving Average model of order 2, or MA(2), and the AutoRegressive model of order 2, or AR(2), are special cases of the ARMA time-series model. The MA(2) model is defined by: 
	\begin{equation}
		X_1 = \xi_1, \quad  X_2 = \xi_2 + \theta_1 \xi_{1}, \quad X_j = \xi_j + \theta_1 \xi_{j-1} + \theta_2 \xi_{j-2}, \quad j=3,\ldots,t,
	\end{equation}
	while the AR(2) is defined as: 
	\begin{equation}
		X_1 = \xi_1, \quad  X_2 = \xi_2 + \theta_1 X_1, \quad X_j = \xi_j + \theta_1 X_{j-1} + \theta_2 X_{j-2}, \quad j=3,\ldots,t;
	\end{equation}	
	in both, $ \xi_j $'s are i.i.d. standard normal error terms (unobserved). We take here $ t=100 $ and we put uniform priors on the parameters of the two models, with bounds given in Table~\ref{Tab:priors}. For these models, the true likelihood can be evaluated, but they do not belong to the exponential family\footnote{More precisely, they cannot be written as an exponential family with embedding dimension fixed with data size; in fact, MA(2) can be written as a Gaussian distribution with $ t\times t $ covariance matrix, which is an exponential family whose embedding dimension increases with data size.}. 
	
	\paragraph{Inferred posterior distribution.}
	
	Figure~\ref{fig:posteriors_AR2_MA2} shows the posterior obtained with our proposed methods, for a possible observation for each model; again, {our approximations are close to the exact posterior, with Exc-SM and Exc-SSM leading to slightly broader posteriors.}
	Again, we assess performance with the Wasserstein distance from the true posterior and the RMSE between the means of true and approximate posterior for all methods, over 100 fresh observations. The results are reported in Figure~\ref{fig:boxplots_AR2_MA2}; here, ABC-FP is the best method, with ABC-SM and ABC-SSM marginally worse. Exc-SM and Exc-SSM follow and perform better or comparably to PMC-RE and PMC-SL.
	
	\begin{figure}[!tb]
		\centering
		\begin{subfigure}{\textwidth}
			\centering
			\begin{subfigure}{0.2\textwidth}
				\centering
				\includegraphics[width=\linewidth]{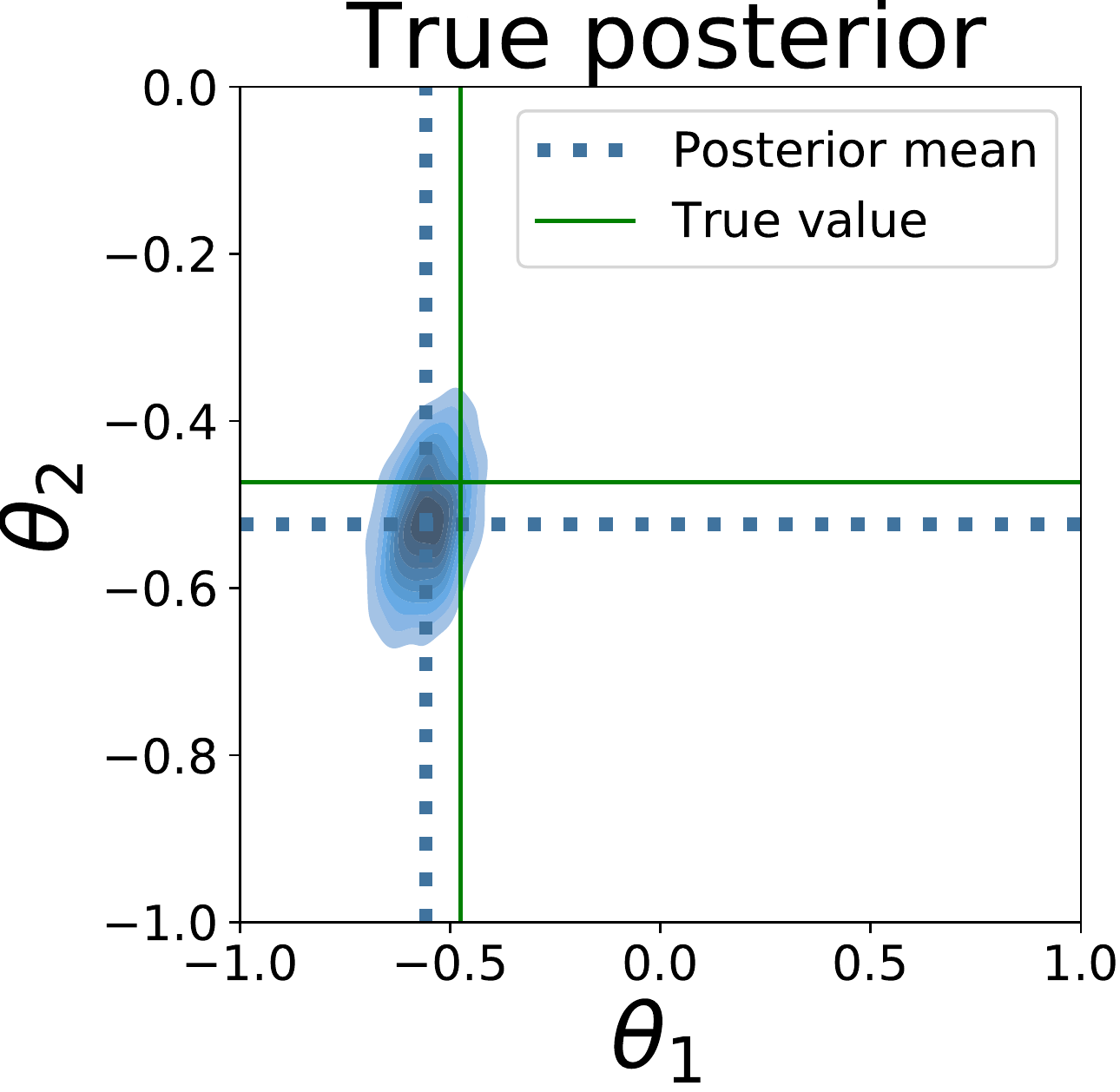}
			\end{subfigure}\begin{subfigure}{0.2\textwidth}
				\centering
				\includegraphics[width=\linewidth]{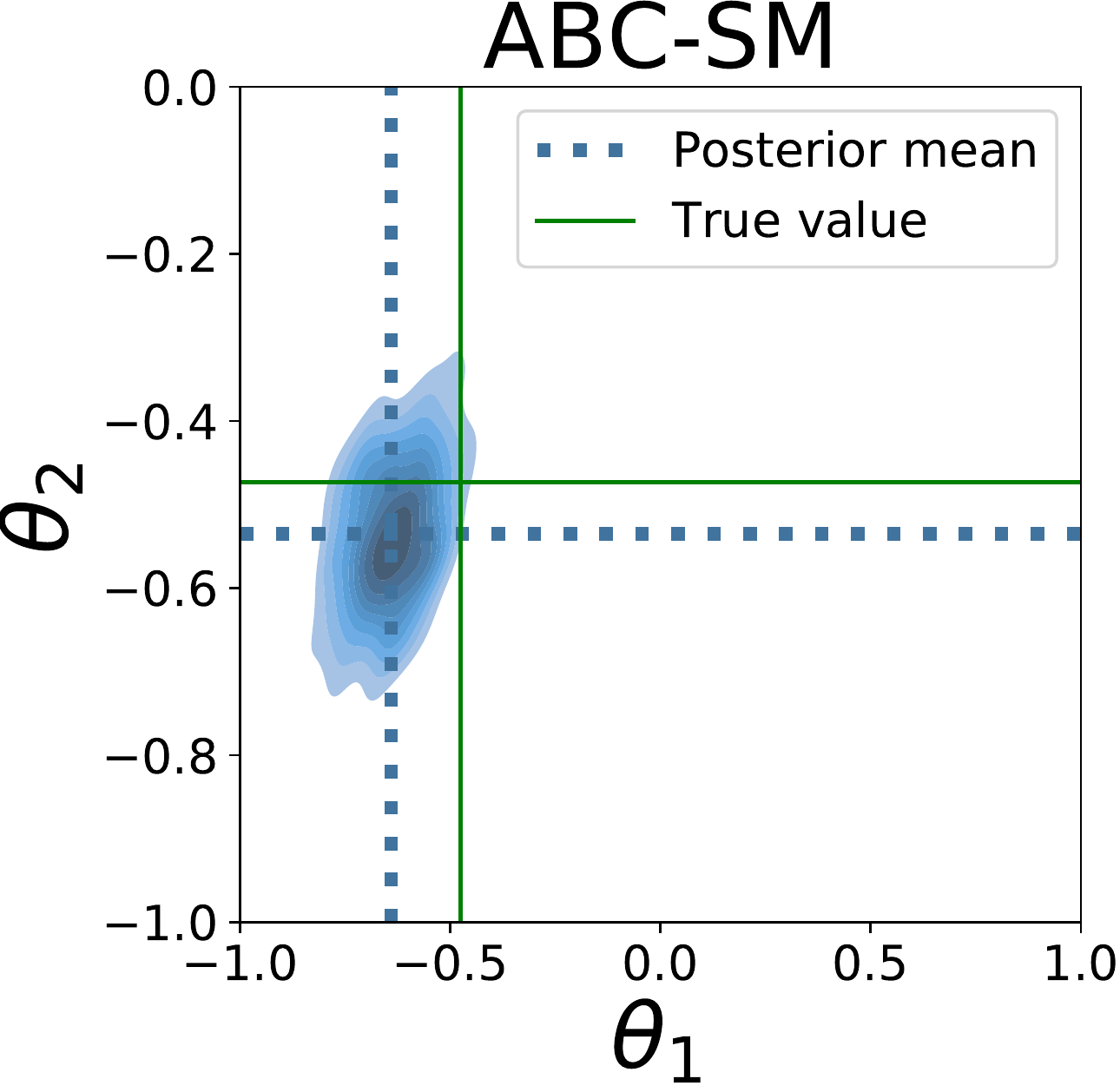}
			\end{subfigure}\begin{subfigure}{0.2\textwidth}
				\centering
				\includegraphics[width=\linewidth]{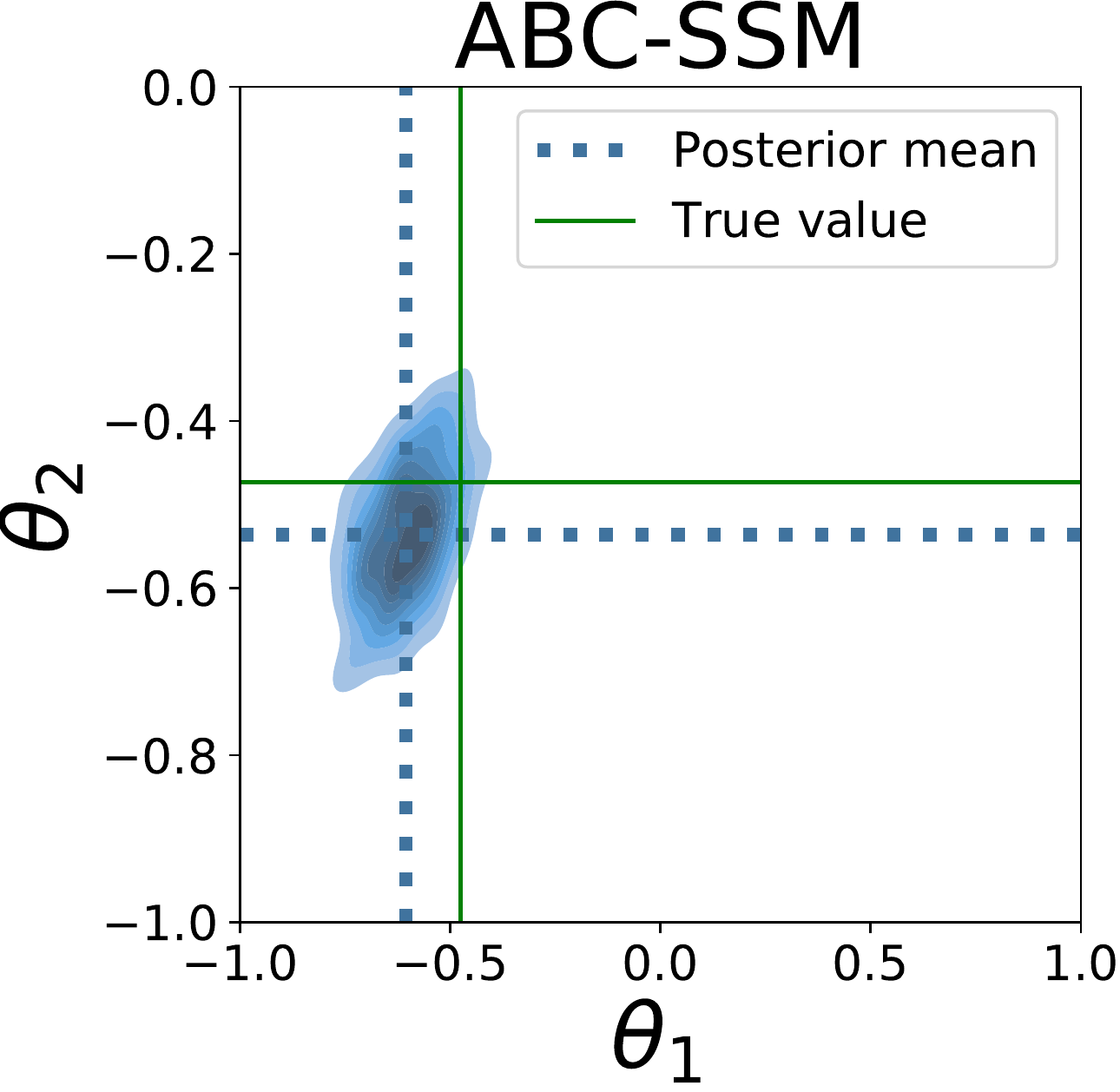}
			\end{subfigure}\begin{subfigure}{0.2\textwidth}
				\centering
				\includegraphics[width=\linewidth]{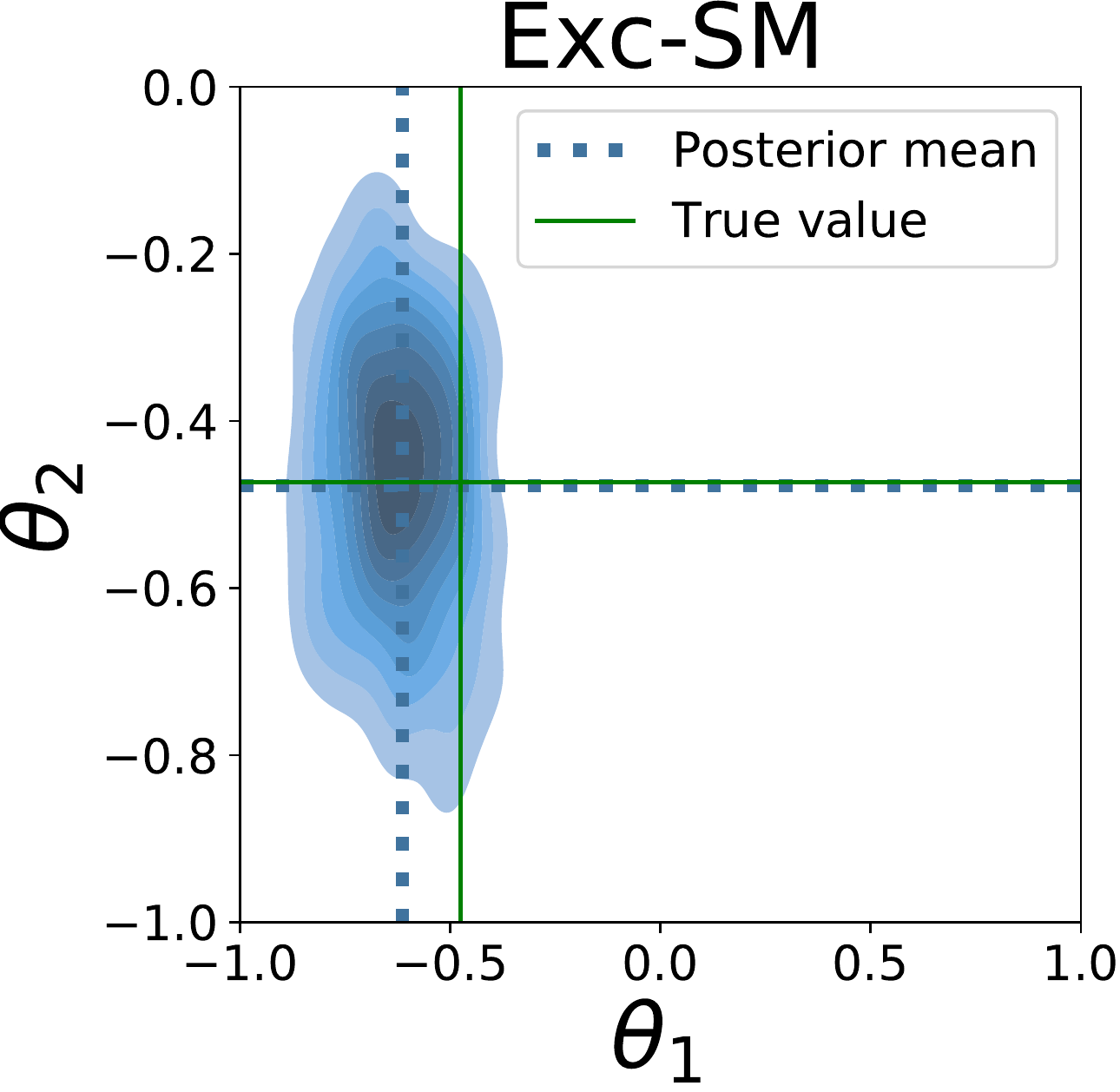}
			\end{subfigure}\begin{subfigure}{0.2\textwidth}
				\centering
				\includegraphics[width=\linewidth]{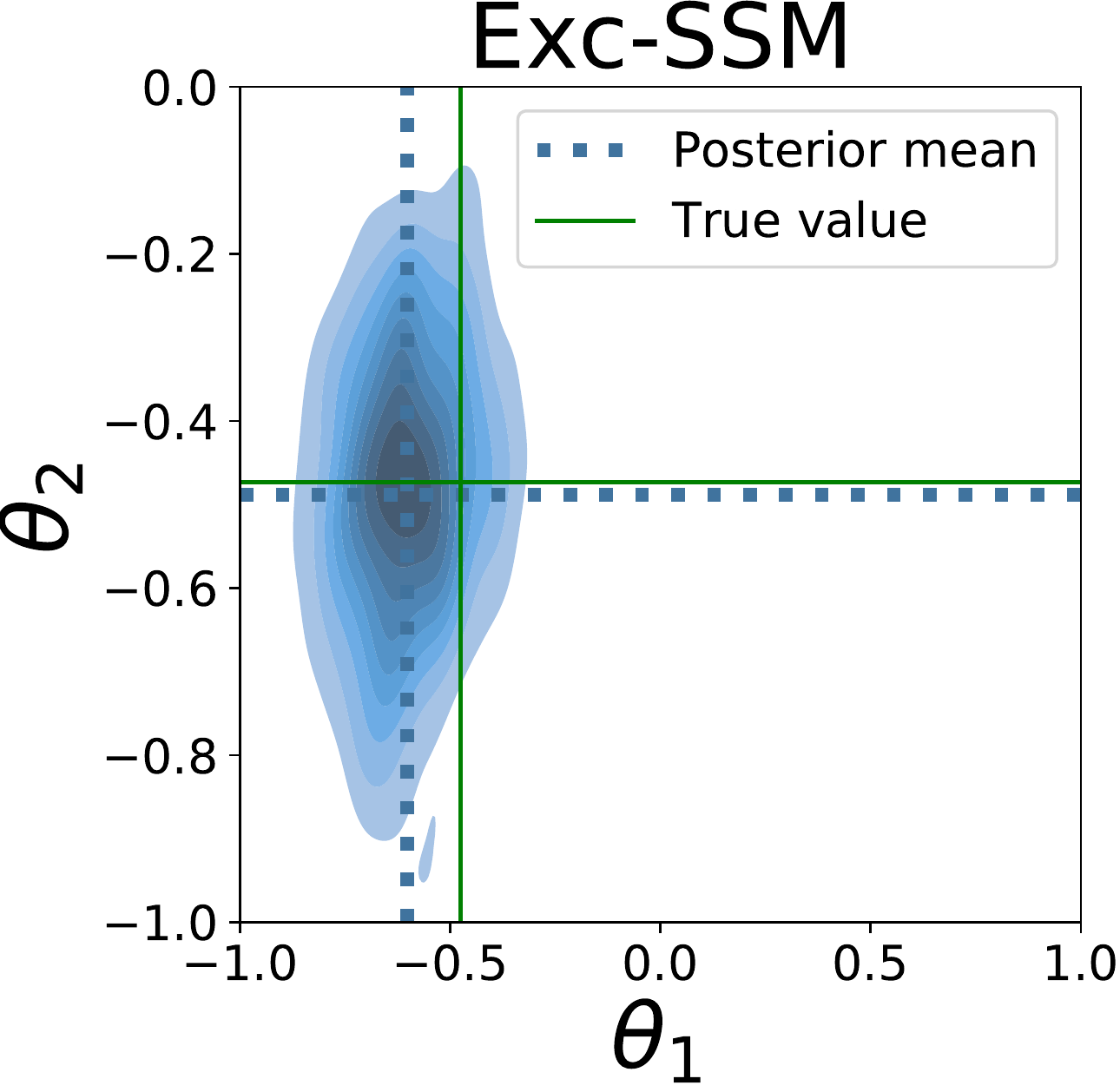}
			\end{subfigure}
			\caption{AR(2)}
		\end{subfigure}\\
		\begin{subfigure}{\textwidth}
			\centering
			\begin{subfigure}{0.2\textwidth}
				\centering
				\includegraphics[width=\linewidth]{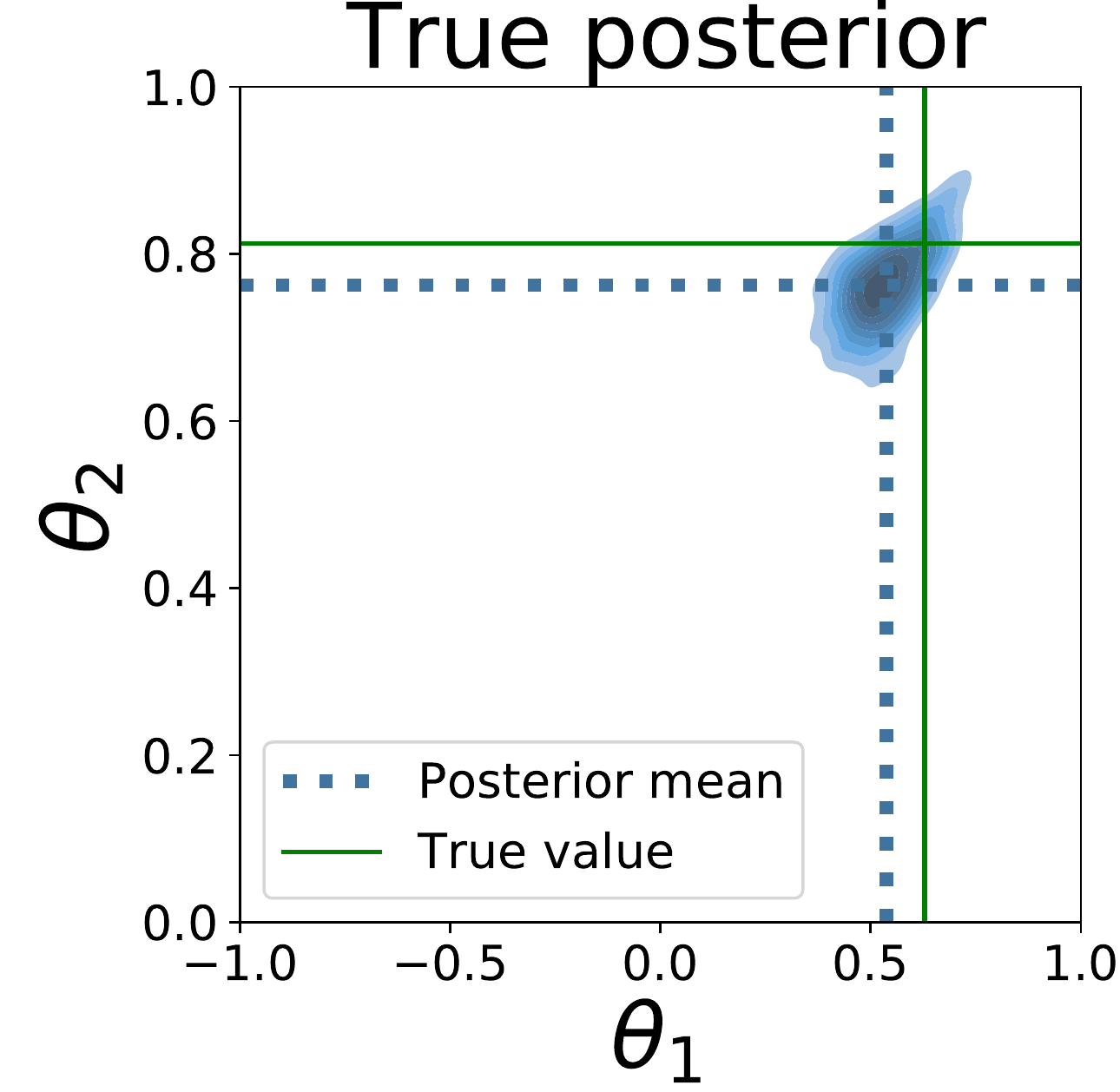}
			\end{subfigure}\begin{subfigure}{0.2\textwidth}
				\centering
				\includegraphics[width=\linewidth]{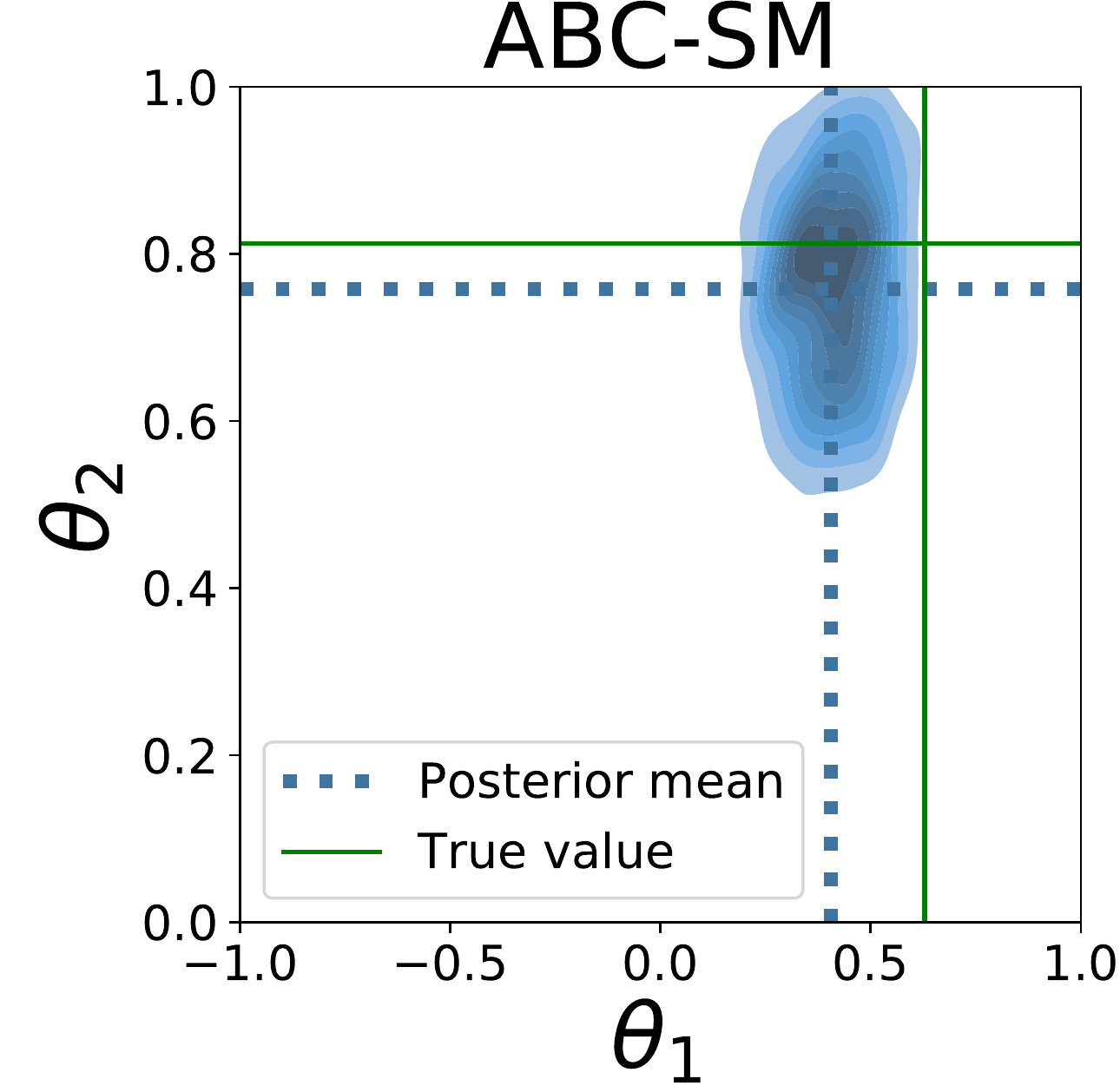}
			\end{subfigure}\begin{subfigure}{0.2\textwidth}
				\centering
				\includegraphics[width=\linewidth]{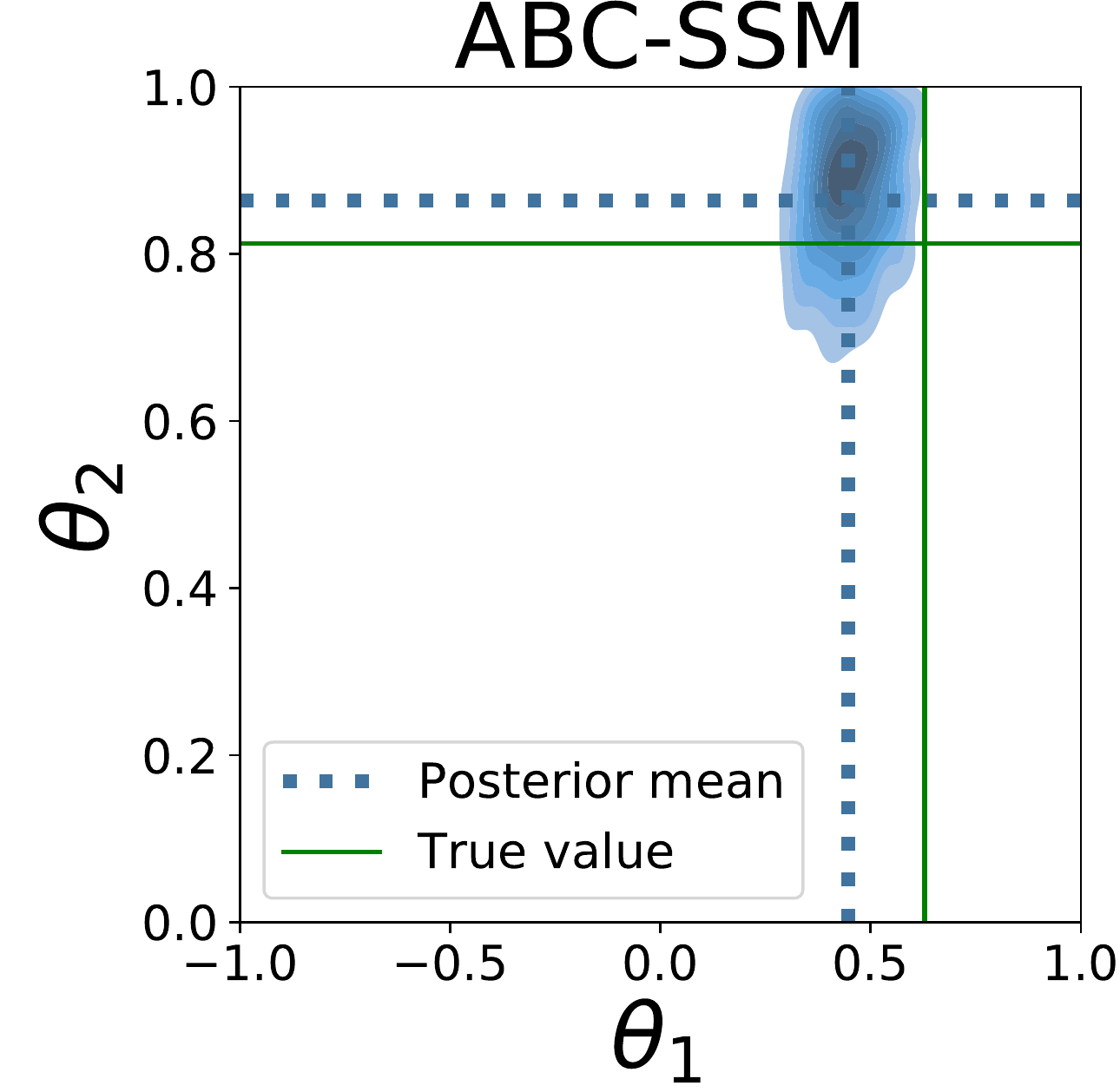}
			\end{subfigure}\begin{subfigure}{0.2\textwidth}
				\centering
				\includegraphics[width=\linewidth]{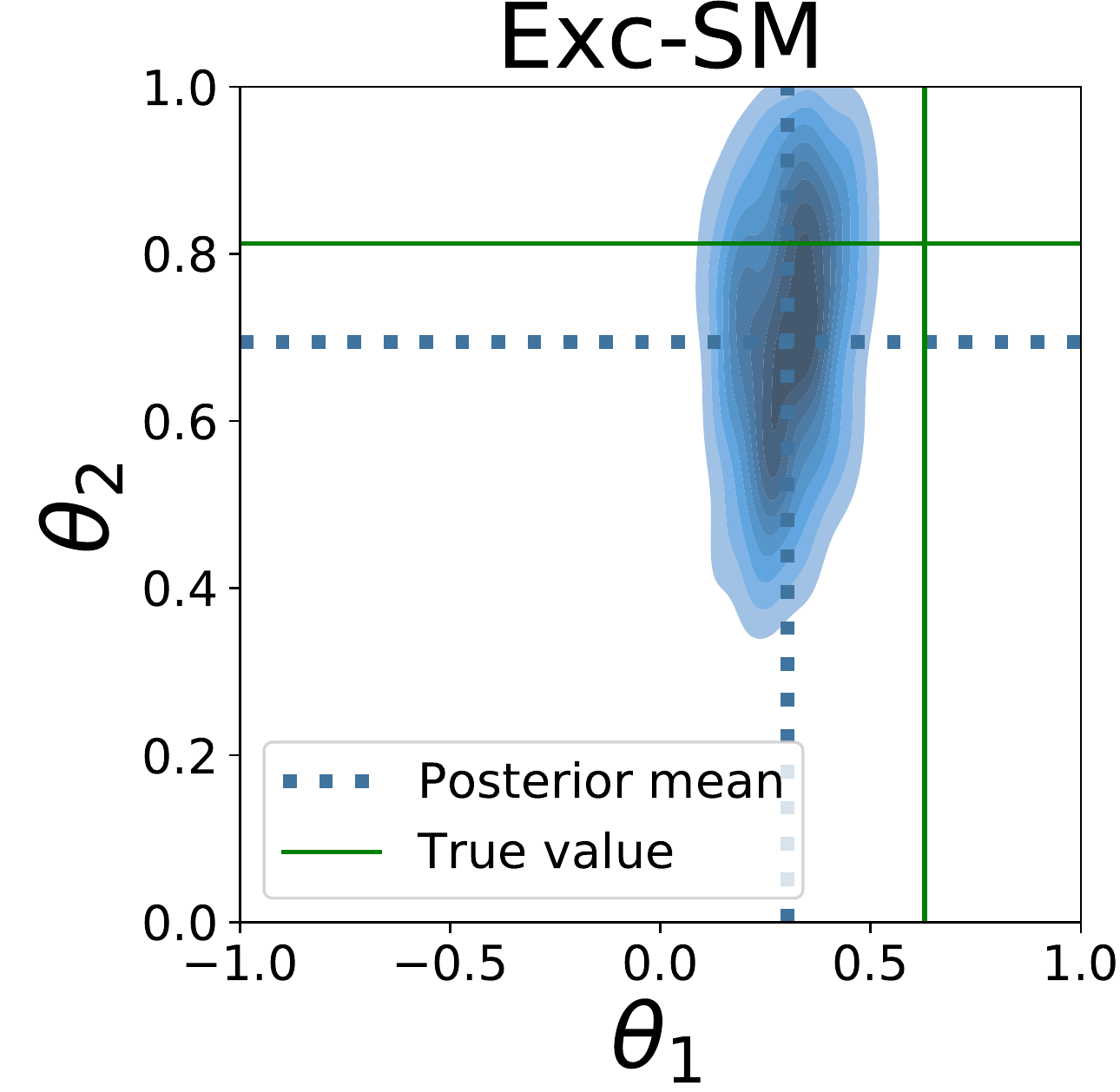}
			\end{subfigure}\begin{subfigure}{0.2\textwidth}
				\centering
				\includegraphics[width=\linewidth]{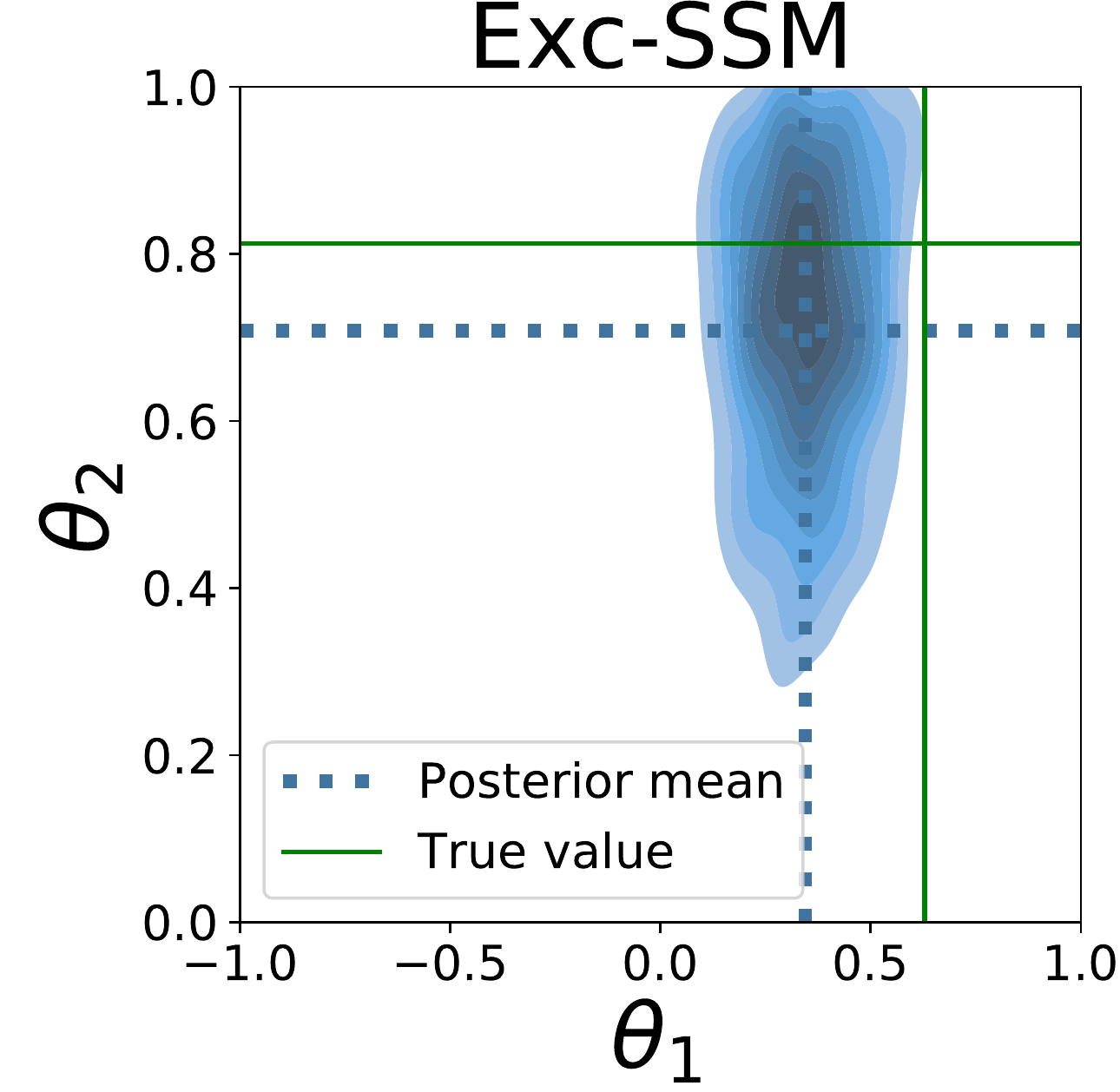}
			\end{subfigure}
			\caption{MA(2)}
		\end{subfigure}\\
		\caption{\textbf{True and approximate posteriors for AR(2) and MA(2) models,} for a single observation per model. Dashed line represents posterior mean, while green solid line represents the exact parameter value.}
		\label{fig:posteriors_AR2_MA2}
	\end{figure}
	
	\begin{figure}[!tb]
		\centering
		\begin{subfigure}{0.32\textwidth}
			\centering
			\includegraphics[width=1\linewidth]{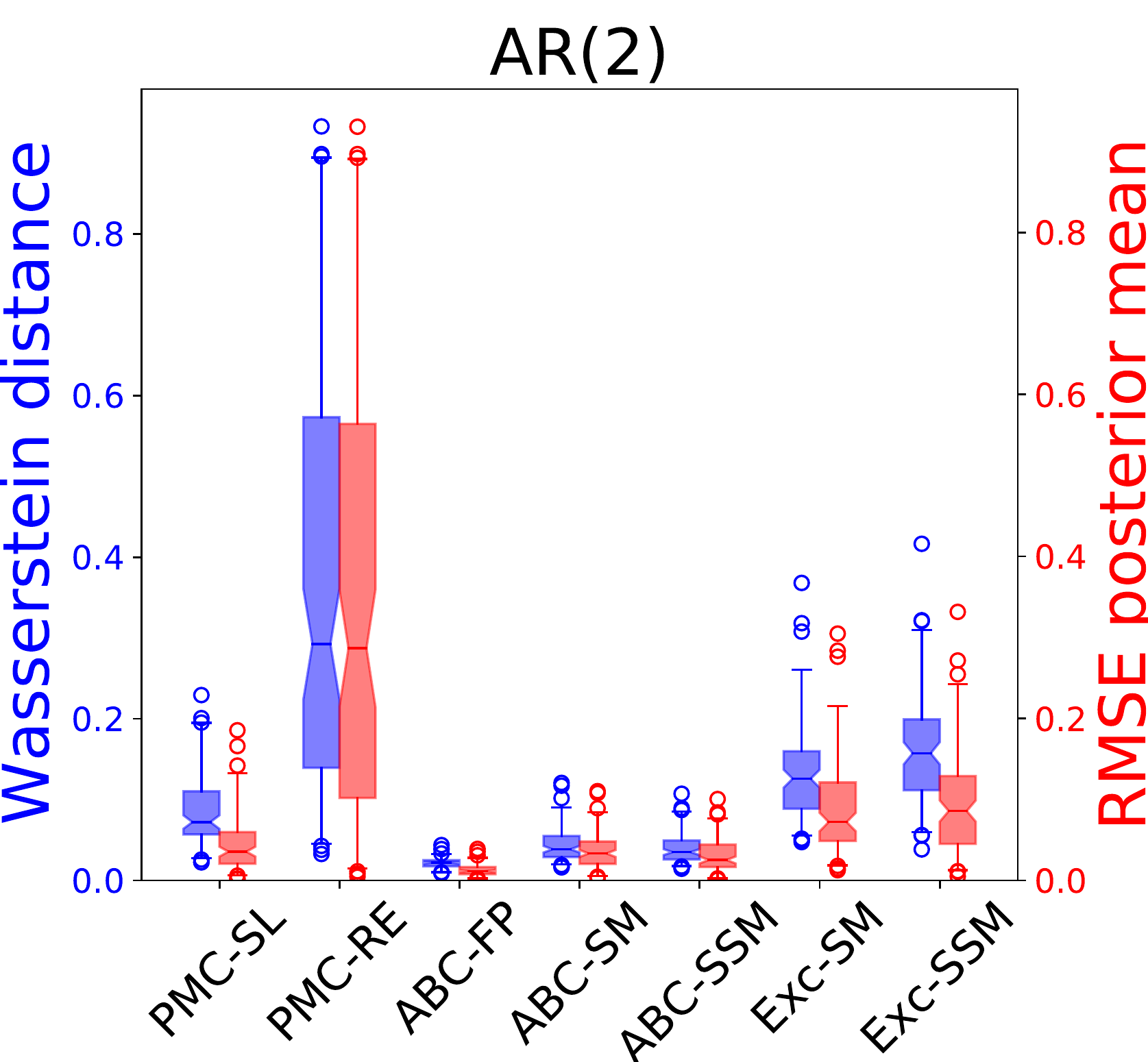}
		\end{subfigure}~
		\begin{subfigure}{0.32\textwidth}
			\centering
			\includegraphics[width=1\linewidth]{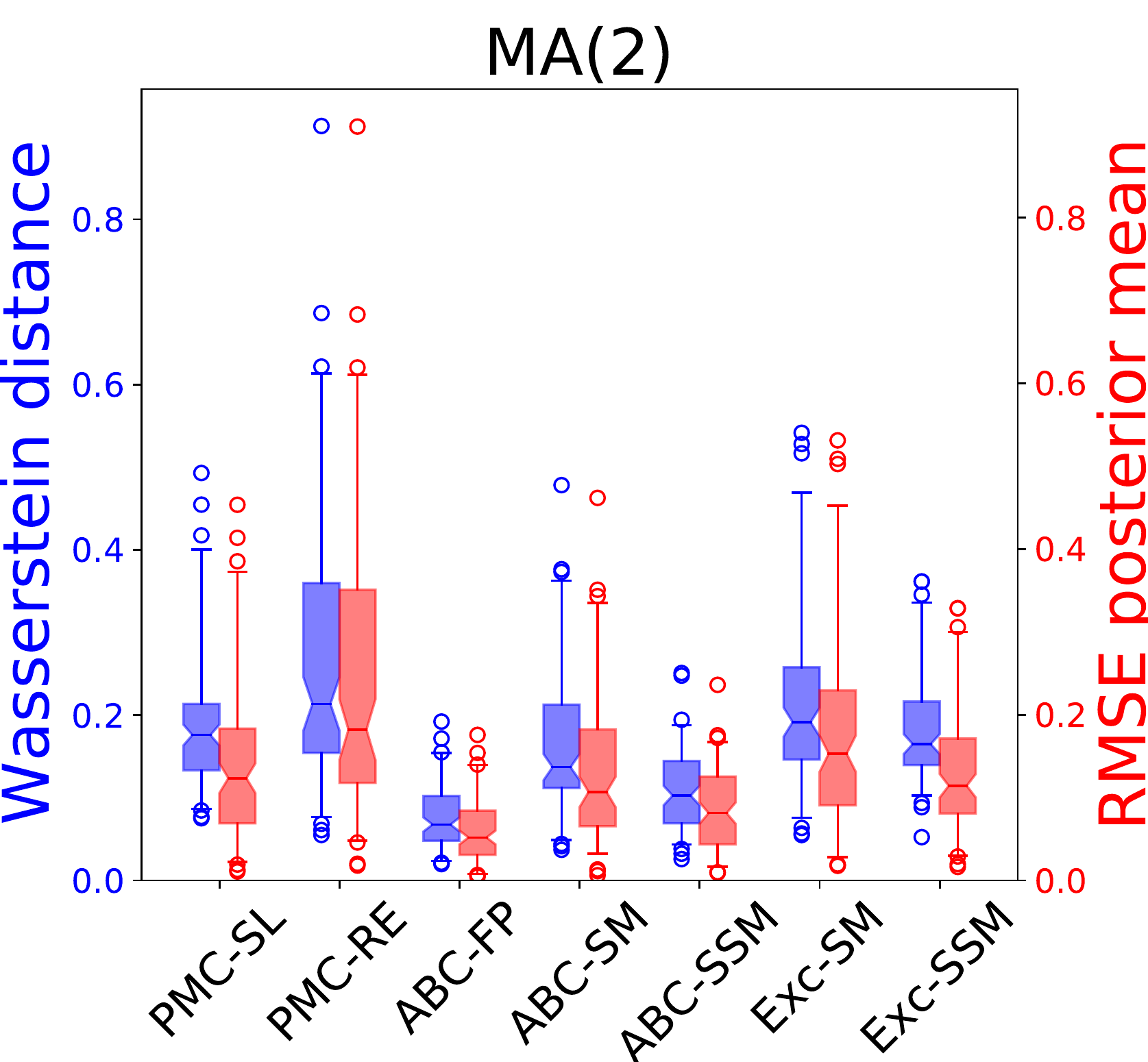}
			
		\end{subfigure}
		\caption{\textbf{Performance of the different techniques for AR(2) and MA(2) models.} Wasserstein distance from the exact posterior and RMSE between exact and approximate posterior means are reported for 100 observations using boxplots. Boxes span from 1st to 3rd quartile, whiskers span 95\% probability density region and horizontal line denotes median.}
		\label{fig:boxplots_AR2_MA2}
	\end{figure}

	\subsection{Lorenz96 meteorological model}\label{sec:Lorenz}
	The Lorenz96 model \citep{lorenz1996predictability} is a toy model of chaotic atmospheric behaviour, including interacting slow and fast variables. We use here a modified version \citep{wilks2005effects} where the effect of the fast variables on the slow ones is replaced by a stochastic effective term depending on a set of parameters. Specifically, this model is defined by the following coupled Differential Equations (DEs):
	
	\begin{equation}\label{Eq:Lorenz}
		\frac{dy_k}{dt} = - y_{k-1}(y_{k-2}-y_{k+1}) -y_k + 10 - g(y_k, t; \parameter); \quad k=1, \ldots, K, 
	\end{equation}
	where $y_k(t)$ is the value at time $ t $ of the \textit{k}-th variable and indices are cyclic (index $ K+1 $ corresponds to $ 1 $, and so on). The effective term $ g $ depends on $ \parameter = (b_0, b_1, \sigma_e, \phi) $, and is defined upon discretizing the DEs with a timestep $ \Delta t$: 
	\begin{equation*}
		g(y, t; \parameter) = \overbrace{b_0 + b_1 y}^{\text{linear deterministic term}} + \overbrace{\phi \cdot \eta(t -\Delta t) + \sigma_e (1 - \phi^2)^{1/2} \eta(t)}^{\text{zero mean AR(1)}},
	\end{equation*}
	where $ \eta(t) \sim \mathcal{N}(0,1) $. We numerically integrate the model using 4th order Runge-Kutta method with $ \Delta t = 1/30 $ in an interval $ [0,T] $; further, we fix the initial condition to a value $ y(0) $ which is independent on parameters. Here, the true likelihood is unaccessible, so that sampling from the exact posterior is impossible.
	
	We consider the model in a small and large configuration: in the small one, we take $ K=8 $ and $ T=1.5 $, which lead to 45 timesteps and a data dimension of $ d=45\cdot 8 = 360 $. In the large one, we take instead $ K=40 $ and $ T=4 $, corresponding to 120 timesteps and a data dimension of $ d=120 \cdot 40 = 4800 $.

	\paragraph{Inferred posterior distribution.}
	For both configurations, we perform inference with ABC-SSM and ABC-FP; additionally, we use Exc-SSM for the small configuration (as running the exchange algorithm in the large one is too costly). In ExchangeMCMC, we used 500 inner MCMC steps and 200 bridging steps for each outer step. Figure~\ref{fig:Lorenz_posteriors} reports posteriors for a single observation in both setups. The ABC-FP posterior is narrower than the one for ABC-SSM, but concentrated around similar parameter values. The posterior for Exc-SSM looks slightly different: it concentrates on similar parameter values as the other two for $ \theta_1 $ and $ \theta_2 $, but is broader for $ \sigma_e $ and $ \phi $.

	\begin{figure}[!tb]
		\centering
		\begin{subfigure}{\textwidth}
			\begin{subfigure}{0.32\textwidth}
				\centering
				\includegraphics[width=1\linewidth]{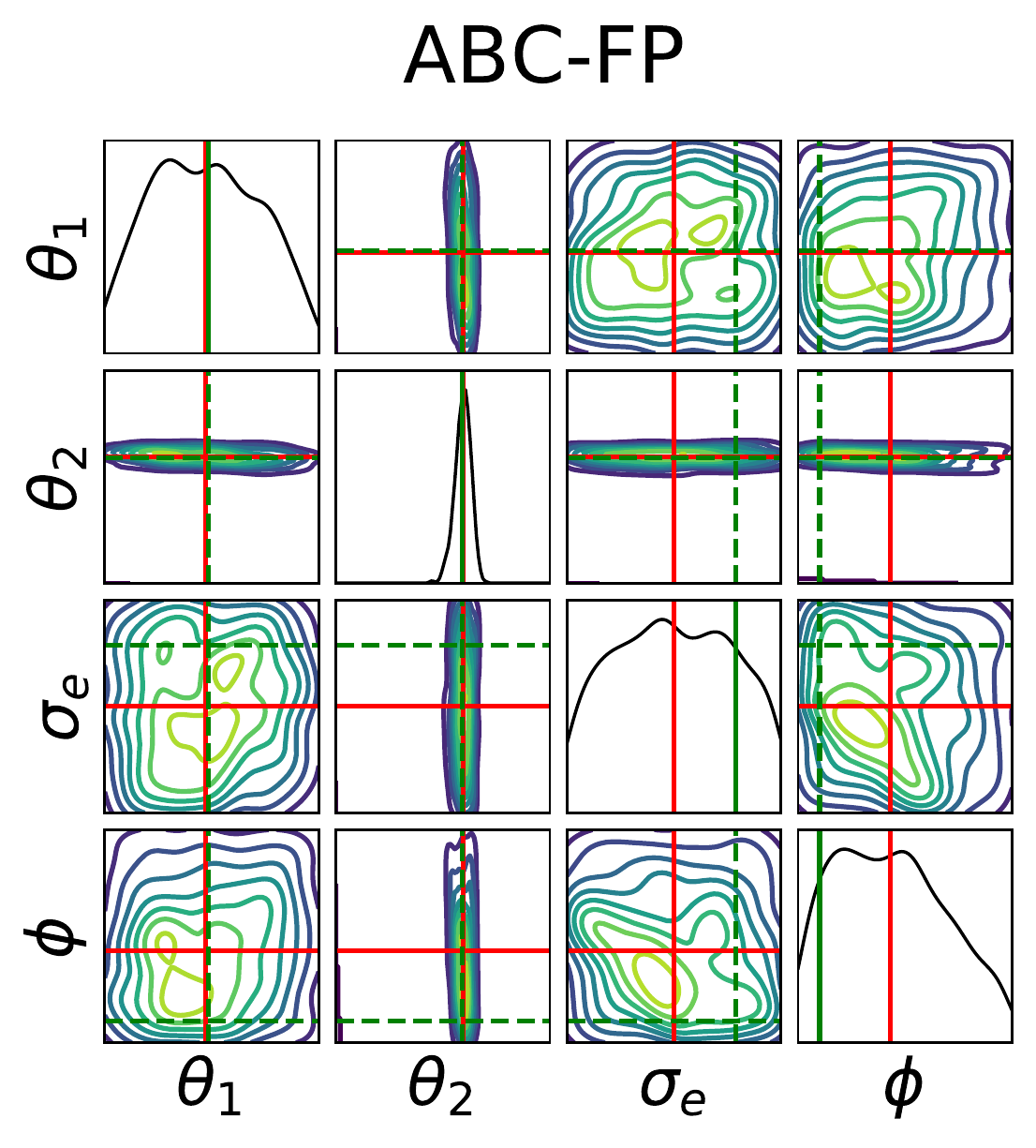}
			\end{subfigure}~
			\begin{subfigure}{0.32\textwidth}
				\centering
				\includegraphics[width=1\linewidth]{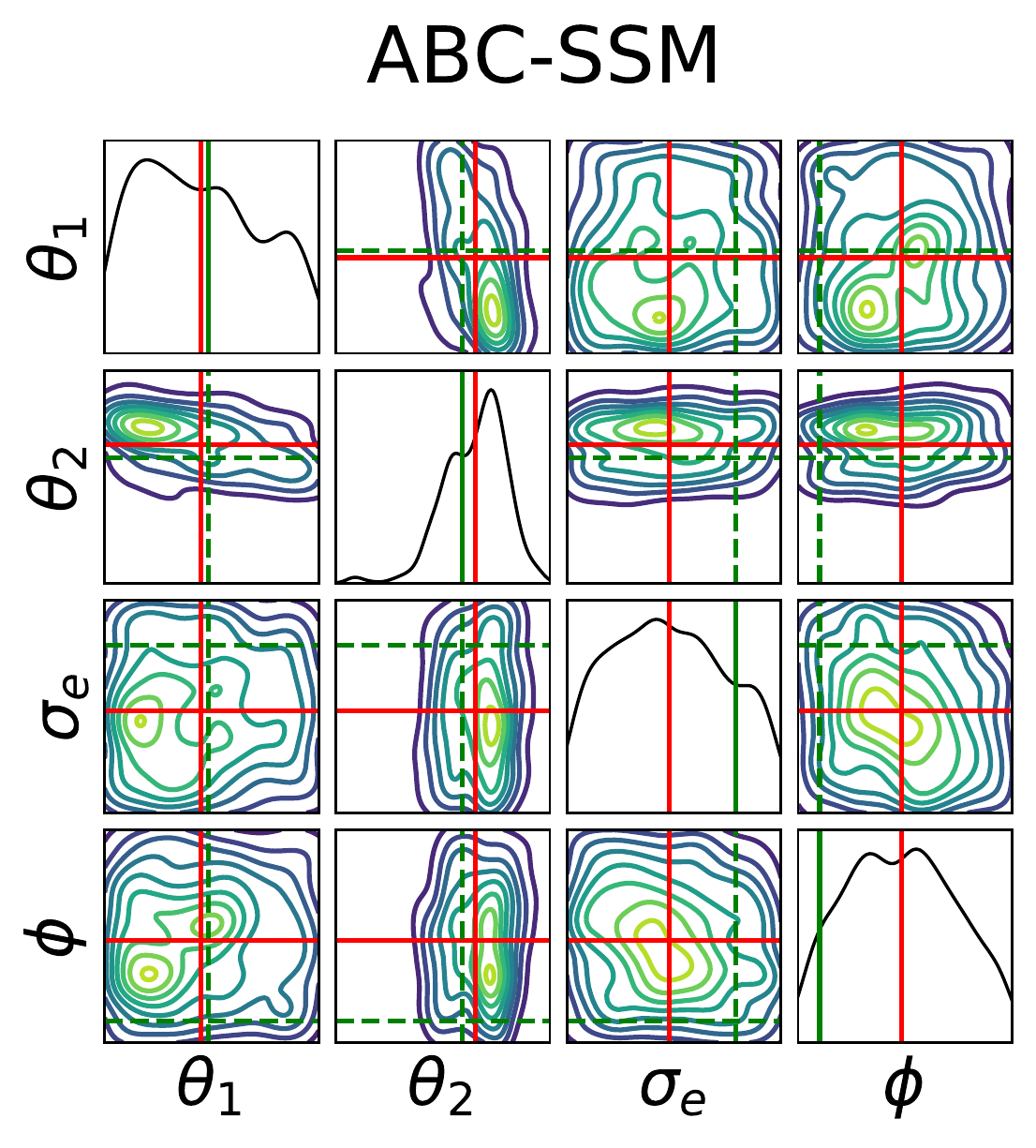}
			\end{subfigure}~
			\begin{subfigure}{0.32\textwidth}
				\centering
				\includegraphics[width=1\linewidth]{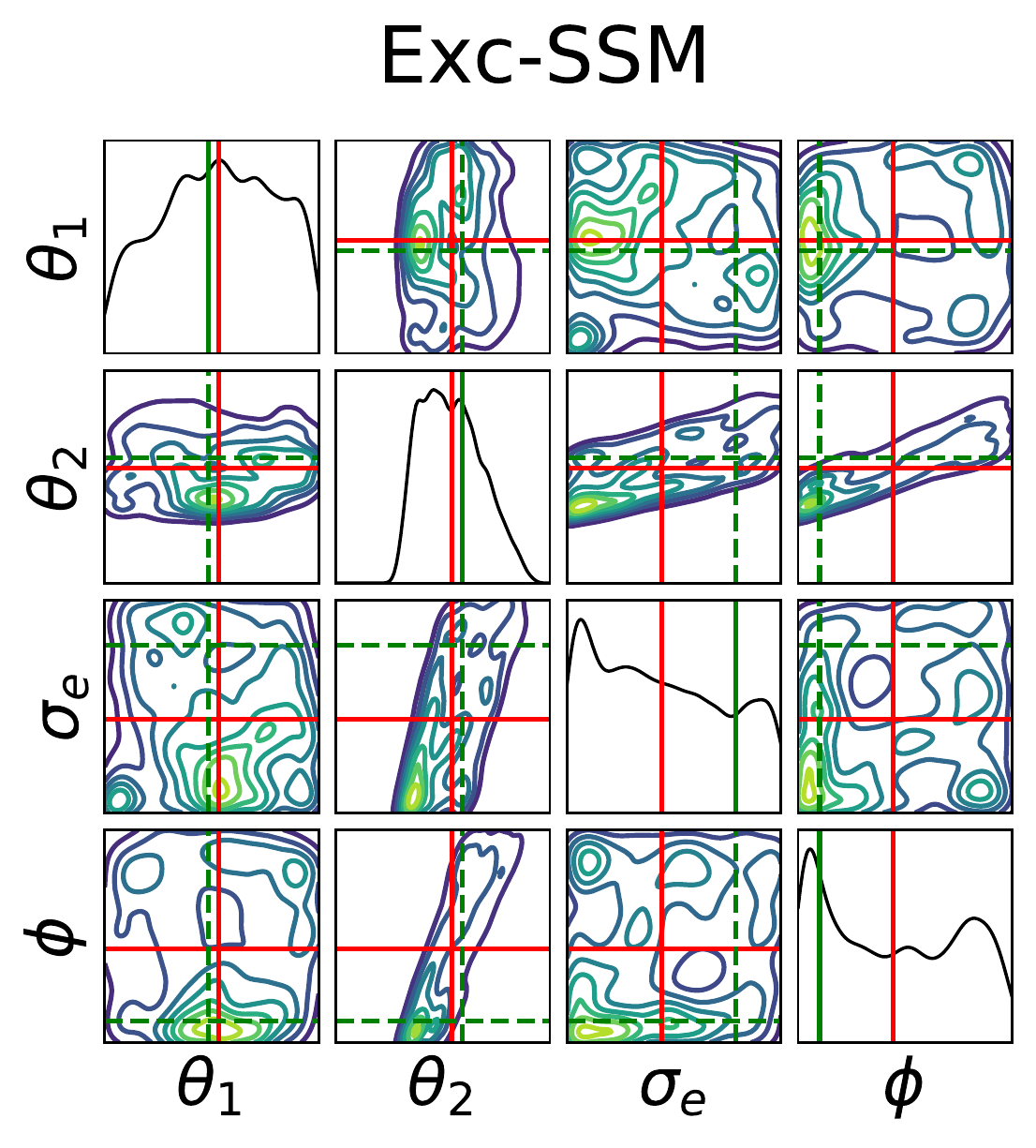}
			\end{subfigure}
			\caption{\textbf{Lorenz96 - Small configuration}}
		\end{subfigure}\\
		\begin{subfigure}{\textwidth}
			\centering
			\begin{subfigure}{0.32\textwidth}
				\centering
				\includegraphics[width=1\linewidth]{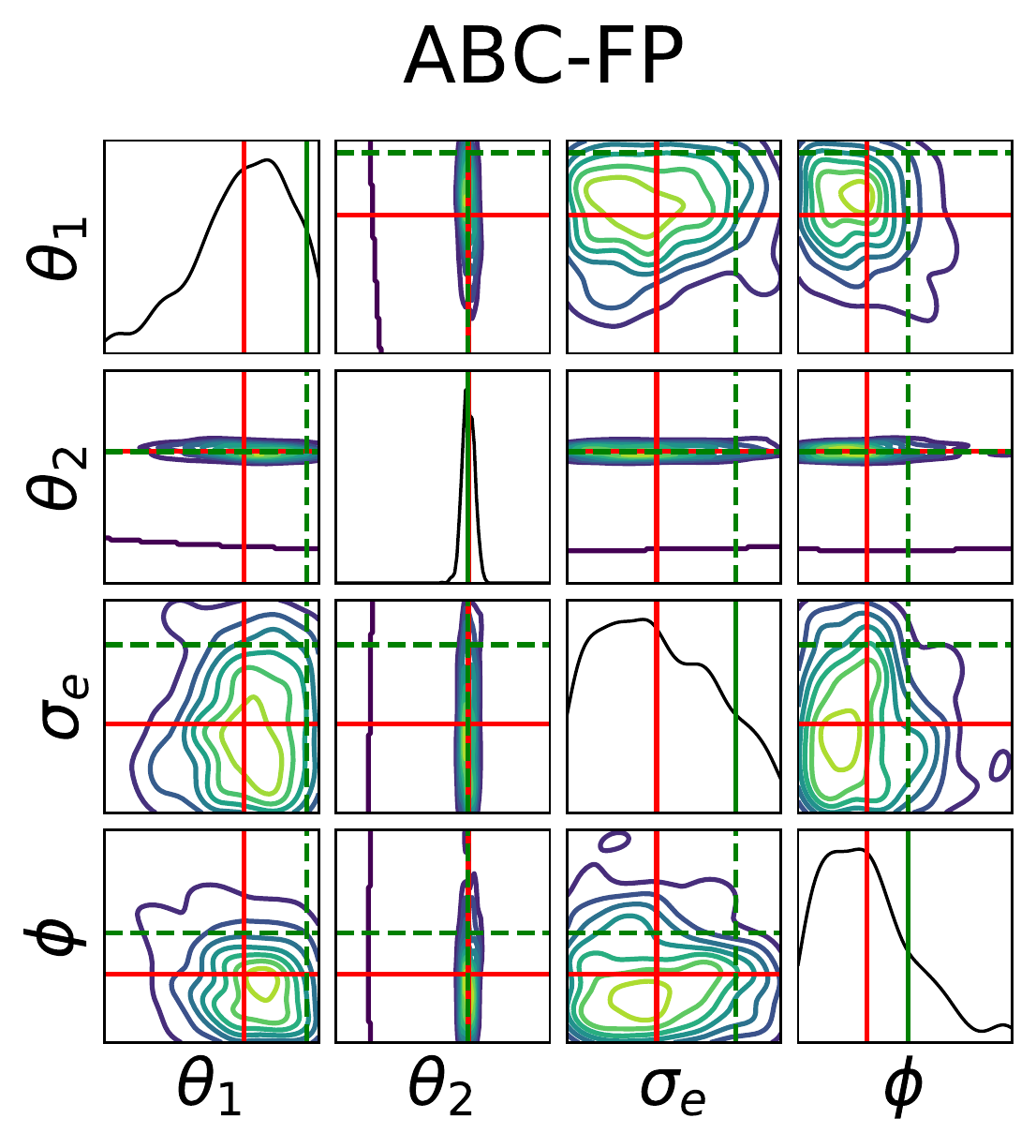}
			\end{subfigure}	~ \begin{subfigure}{0.32\textwidth}
				\centering
				\includegraphics[width=1\linewidth]{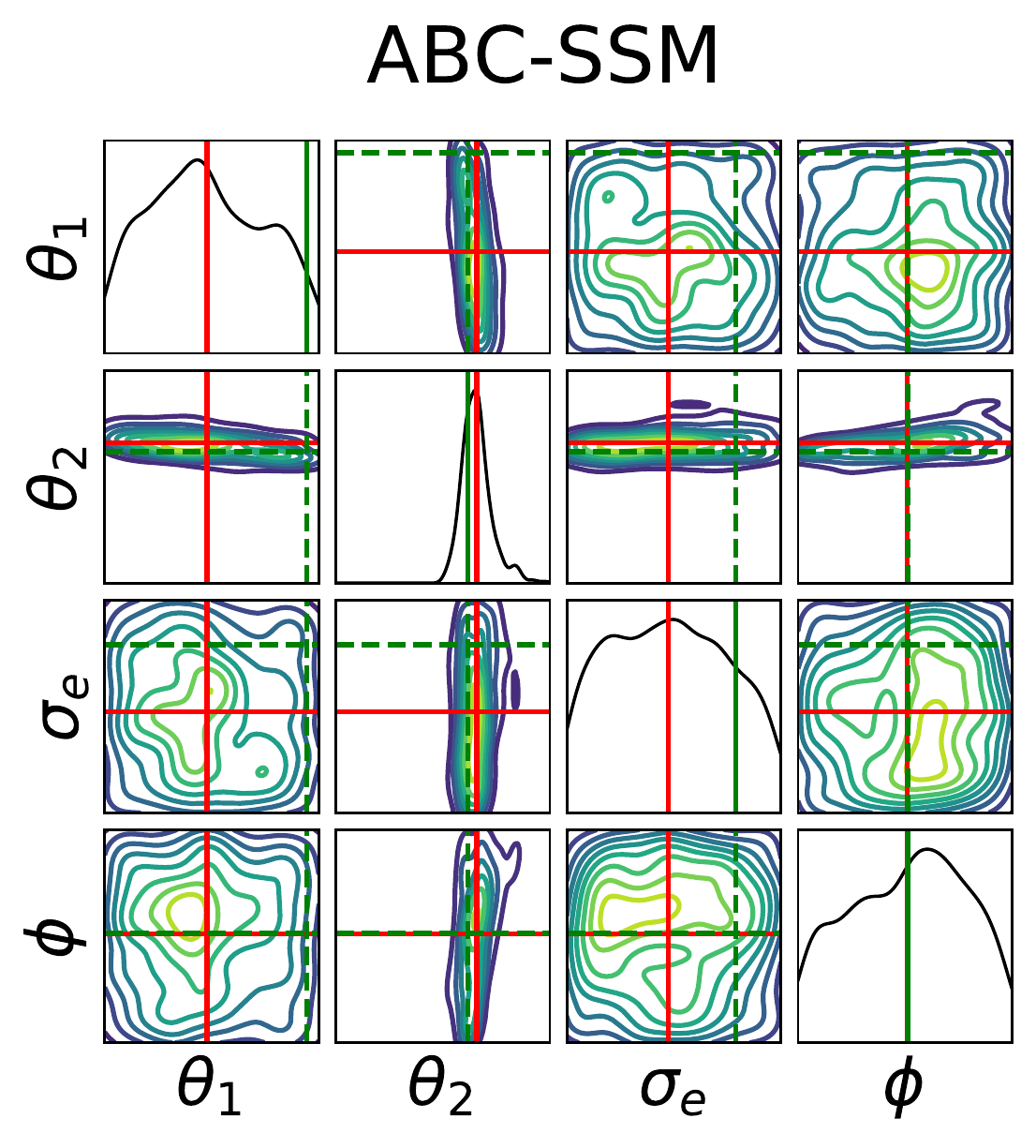}
			\end{subfigure}
			\caption{\textbf{Lorenz96 - Large configuration}}
		\end{subfigure}
		\caption{\textbf{Example of posterior inference for a single Lorenz96 observation with different approaches,} for both small (first row) and large (second row) configurations. In each panel, the diagonal plots represent the univariate marginal distributions, while the off-diagonal ones are bivariate density contour plots. Moreover, the green and red lines represent respectively exact parameter value and posterior mean. All axes span full prior range (see Table~\ref{Tab:priors}).}
		\label{fig:Lorenz_posteriors}
	\end{figure}

	\paragraph{Posterior predictive validation.}
	We assess the performance of the posterior predictive distribution
	$ p(x|\dataObs) = \int p(x|\theta) \pi(\theta|\dataObs) d\theta, $
	in which $ \pi(\theta|y) $ is the posterior obtained with one of the considered techniques. Specifically, we use the Kernel and Energy Scoring Rules (SRs) to evaluate how well the predictive $ p(x|\dataObs) $ predicts the observation $ \dataObs $. A SR \citep{gneiting2007strictly} is a function of a distribution and an observation, and assesses the mismatch between them (the smaller, the better). More details on Scoring Rules are given in Appendix~\ref{app:SRs}, together with the definition of those used here.

	As the Lorenz96 model returns a (multivariate) time-series, we compare predictive distribution and observation at each timestep independently. 
	We repeat this validation with 100 different observations and corresponding posteriors, for all techniques and both model configurations. We obtain cumulative scores by summing the scores over timesteps; median and some quantile values over the 100 observations are shown in Fig.~\ref{fig:SRs}. In both model configurations, the posterior predictive generated by ABC-SSM and Exc-SSM marginally outperform ABC-FP, according to both SRs. Score values at each timestep, together with more details on the validation technique, are reported in Appendix~\ref{app:Lorenz_validation_results}.

	\begin{figure}[!tb]
		\centering
		\includegraphics[height=6cm,keepaspectratio]{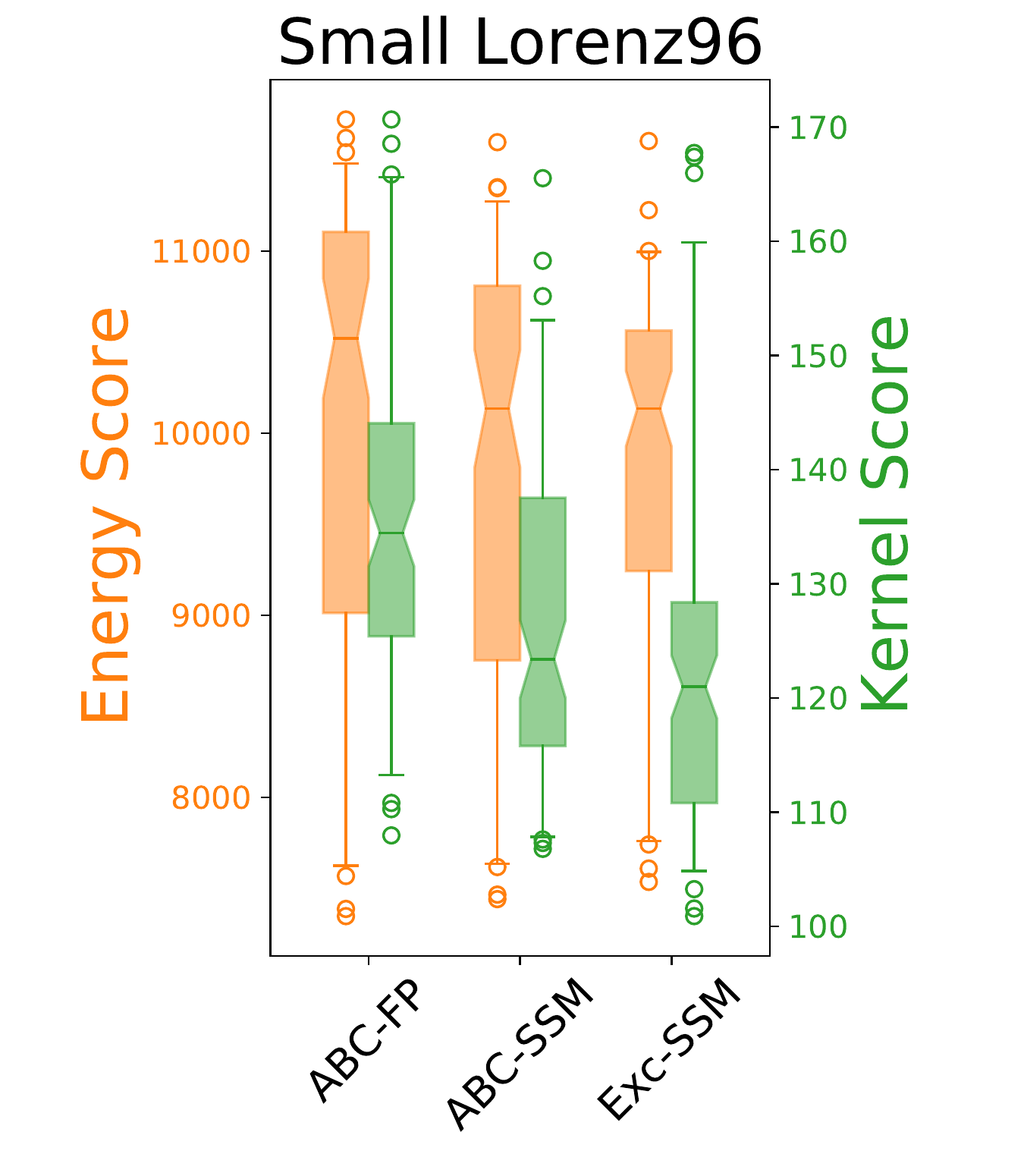}
		\includegraphics[height=6cm,keepaspectratio]{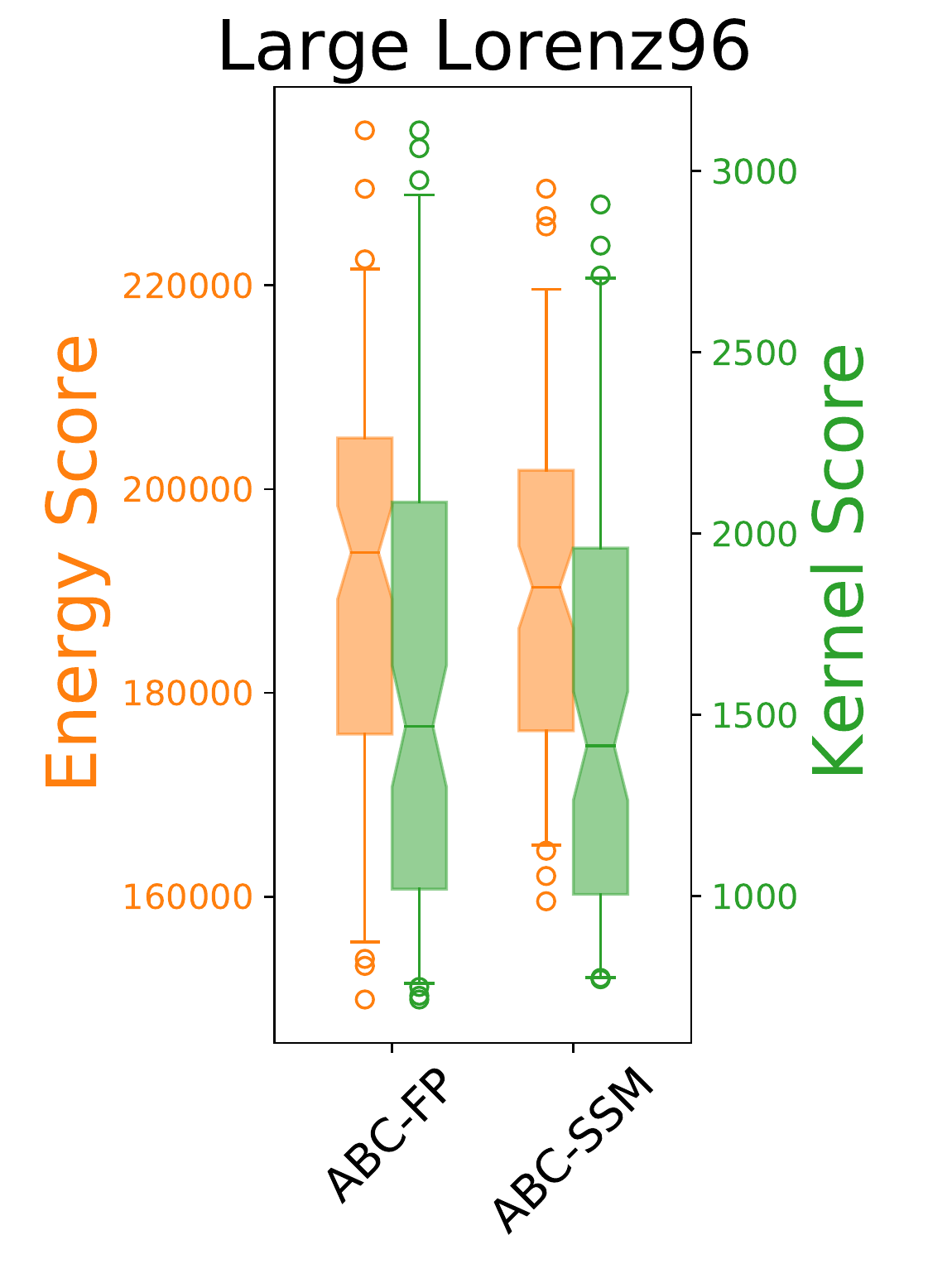}
		\caption{\textbf{Predictive performance of the different methods according to the Kernel and Energy Scoring Rules.} Each boxplot represents cumulative (i.e., summed over the time index) scoring rule value for a given method for the small (left) and large (right) Lorenz96 configuration. Boxes span from 1st to 3rd quartile, whiskers span 95\% probability density region and horizontal line denotes median.}
		\label{fig:SRs}
	\end{figure}

	\section{Related works}\label{sec:related_works}

	\paragraph{ABC statistics with NN.} Using NNs for learning statistics for ABC has been previously suggested, first with the regression approach discussed in Sec.~\ref{sec:ABC} (used in \citealt{jiang2017learning, wiqvist2019partially, aakesson2020convolutional}). In \cite{pacchiardi2020}, a NN is trained so that the distance between pairs of simulated statistics best reproduces the distance between the corresponding parameter values. \cite{chen2020neural} instead fit a NN by maximizing an estimate of the mutual information between statistics and parameters.  
	\paragraph{ABC statistics from auxiliary models.}
	Some similarities with our work can be found in \cite{gleim2013approximate} and \cite{ruli2016approximate}, which consider respectively an auxiliary model and a composite likelihood alongside the simulator model and obtain summary statistics from them; similarly, our approach can be seen as building an auxiliary exponential family model with easily accessible summary statistics. 
	
	\paragraph{LFI with NNs.}
	The idea of sampling directly from the approximate posterior defined by $ p_w $ is related to a suite of LFI methods using NNs. A large part of these works use normalizing flows (NNs implementing invertible transformations, suitable for efficiently representing probability densities; see \citealp{papamakarios2019normalizing} for a review); specifically, \cite{papamakarios2018sequential, lueckmann2019likelihood} use them to learn an approximation of the likelihood, while \cite{papamakarios2016fast, lueckmann2017flexible, greenberg2019automatic, radev2020bayesflow} learn the posterior. Most of the above approaches focus on inference for a single observation, tailoring the simulations to better approximate the posterior for the relevant parameter values at lower computational cost; instead, \cite{radev2020bayesflow} and \cite{lueckmann2019likelihood} propose to amortize across observations, similarly to our approach. Besides normalizing flows, \cite{klein2020marginally} casts the LFI problem as a distributional regression one. Finally, \cite{tabak2020conditional} does not employ NNs but rather solves a Wasserstein barycenter problem to model conditional maps which allow sampling from the distribution of an observation conditional on some covariates.

	\paragraph{Fitting unnormalized models.} Several techniques besides Score Matching have been proposed for fitting unnormalized models: MCMC-MLE \citep{geyer1991markov} exploits MCMC to estimate the normalizing constant for different values of the parameter and uses that in MLE. Contrastive Divergence (CD, \citealt{hinton2002training}) instead uses MCMC to obtain a stochastic approximation of the gradient of the log-likelihood; this requires a smaller number of MCMC steps with respect to MCMC-MLE, but the stochastic gradient estimate is biased. Minimum Probability Flow (MPF, \citealt{sohl2011new}) considers a dynamics from data to model distribution, and minimizes the Kullback Leibler divergence between the data distribution and the one obtained by running the dynamics for a short time; however, the efficacy of MPF depends significantly on the considered dynamics. Noise Contrastive Estimation (NCE, \citealt{gutmann2012noise}) converts the parameter estimation problem to ratio estimation between data distribution and a suitable noise distribution; in practice, NCE uses logistic regression to discriminate the observed data and data generated from the noise; in the loss, the normalizing constant appears explicitly and can be estimated independently. To be effective, NCE requires the noise distribution to overlap well with the data density while being easy to sample from and to evaluate, which is not easy to get. Finally, some works \citep{dai2019kernel, dai2019exponential} use the dual formulation of MLE to avoid estimating the normalizing constant at the price of introducing dual variables to be jointly estimated.

	\paragraph{Fitting unnormalized conditional models.} The approaches above running dynamics to estimate quantities (CD, MCMC-MLE, MPF) cannot easily be applied to the conditional setting. NCE has been instead used in \cite{ton2019noise} with a single noise distribution for all values of the conditioning variable $ \theta $; it however requires an independent NN to parametrize the normalizing constant $ Z(\theta) $. SM can be instead easily applied to the conditional setting, as previously done in \cite{arbel2017kernel} and further demonstrated in our work. To the best of our knowledge, SSM was not applied to a conditional setting before, although the extension straightforwardly follows what is done for SM.

	\paragraph{Fast approximations of SM.}
	Besides SSM, some approximations to SM have been investigated, all of which only require first derivatives. For instance, Denoising Score Matching \citep{vincent2011connection} computes the Fisher divergence between the model and a Kernel Density Estimate of the data distribution, which is equal to a quantity independent on the second derivative. Alternatively, Kernel Stein Discrepancy \citep{liu2016stein} intrinsically depends on first derivatives only. For both techniques, however, several samples for each $ \theta $ would be required in the conditional setting.
	Finally, \cite{wang2020wasserstein} exploits a connection between the Fisher divergence and gradient flows in the 2-Wasserstein space to develop an approximation that only relies on first-order derivatives.

	\paragraph{Kernel Conditional Exponential Families (KCEFs).}
	In KCEFs \citep{arbel2017kernel}, the summary statistics and natural parameters are functions in a Reproducing Kernel Hilbert Space, whose properties allow to evaluate the density although an infinite-dimensional embedding space is used (using the \textit{kernel trick}). In \cite{arbel2017kernel}, SM was used to fit KCEFs instead of our neural conditional exponential families. They build on \cite{sriperumbudur2017density}, which first used SM to perform density estimation with (non-conditional) Kernel Exponential Families (KEFs). In \cite{wenliang2018learning} NNs are used to parametrize the kernel in a KEF and trained with the SM loss.
	However, as KCEFs do not have finite-dimensional sufficient statistics, they are unsuitable for learning ABC statistics. Additionally, KCEFs have a worse complexity in terms of training dataset size with respect to NN-based methods. KCEFs have also been used with a dual MLE objective in \cite{dai2019kernel}.

	\section{Conclusions and extensions}\label{sec:conclusion}
	
	We proposed a technique to approximate the likelihood using a neural conditional exponential family, trained via (Sliced) Score Matching to handle the intractable normalizing constant.

	We tested this approximation in two setups: first, by using the exponential family sufficient statistics as ABC statistics, which is intuitively appealing as the exponential family is the largest class of distributions with fixed-size sufficient statistics. We empirically showed this to be comparable or outperform ABC with summaries built via the state-of-the-art regression approach \citep{fearnhead_constructing_2012, jiang2017learning}.
	
	Secondly, we used MCMC for doubly-intractable distributions to sample from the posterior corresponding to the likelihood approximation, which we found to have performance comparable to the other approaches. This can be repeated for any new observation without additional model simulations, making it advantageous for expensive simulator models.
	
	Our proposed direct sampling approach based on exponential family likelihood approximation and ExchangeMCMC could be improved as follows:

	\begin{itemize}

		\item  we used ExchangeMCMC \citep{murray2012mcmc} to handle double intractability, but other algorithms could be more efficient, (for instance the one in \cite{liang2016adaptive}, which makes use of parallel computing). Alternatively, we could exploit the generalized posterior introduced in \cite{matsubara2021robust}, which allows standard MCMC to be used for double intractable distributions and is robust to outliers.
		
		\item To infer the posterior for a single observation, approximating the likelihood for all $ x $'s and $ \theta $'s as we do now is suboptimal. Similar performance may in fact be obtained with fewer simulations tailored to the observation. Sequential schemes implementing such ideas have been introduced for LFI using normalizing flows (see for instance \citealt{papamakarios2018sequential}) and could be extended to our setup.
		
		\item The motivation for the current work was learning ABC statistics, hence the exponential family formulation. However, the dot-product structure between $ f_w $ and $ \eta_w $ is not beneficial for the direct sampling approach. An \textit{energy-based model}, which employs a single NN with input $ (x,\theta) $ in the exponent, may be more expressive and easier to train.

		\item An energy-based posterior approximation $ \pi_w(\theta|x) $ could be trained by minimizing the expectation over the data marginal $ p(x) $ of the (sliced) Fisher divergence with respect to the true posterior. This is complementary to the strategy employed in this work to fit $ p_w(x|\theta) $. Interestingly, 	
		$ \pi_w(\theta|x) $ would be known up to a normalizing constant depending on $ x $ only, making use of standard MCMC possible.

	\end{itemize}

	\bibliographystyle{abbrvnat}

	\appendix
	
	\section{Sufficient statistics}\label{app:suff_stats}
	
	Consider a conditional probabilistic model $ p(x|\theta) $; moreover, abusing notation, we will also denote as $ p(x|t) $ the density $ p_{X|T}(X=x|T=t) $, as well as $ p(x|t; \theta) $ the density $ p_{X|T,\Theta} (X=x|T=t; \Theta=\theta) $, where here $ \Theta $ denotes the random variable which takes values $ \theta $. Finally, we use $ \pi $ to denote distributions over the parameter values $ \theta $; specifically, $ \pi(\theta|x) $ denotes the standard posterior, and $ \pi(\theta|t) $ is an abuse of notation for the density $ \pi_{\Theta|T}(\Theta=\theta|T=t) $.  
	
	\begin{definition}
		A statistic $ t = t(x) $ is sufficient if $ p(x|t; \theta) = p(x|t) $, where $ \theta $ is a parameter of the distribution. Alternatively, we have, in the Bayesian setting: 
		\begin{equation}\label{eq:suff_bayesian}
			\pi(\theta|x) = \pi(\theta|t(x)),
		\end{equation}
		for any (non-degenerate) choice of prior distribution $ \pi(\theta) $.
	\end{definition}
	
	The existence of sufficient statistics implies a precise form of the distribution: 
	
	\begin{theorem}(\textbf{Fisher-Neyman factorization theorem}):
		A statistic is sufficient $ \iff  p(x|\theta) = h(x) g(t(x) | \theta)$, where $ h $ and $ g $ are non-negative functions. 		
	\end{theorem}	
	It is clear from the above theorem that $ f(x) $ in the exponential family is sufficient. 
	
	A stronger result regarding exponential family also exists, which goes under the name of Pitman–Koopman–Darmois theorem. This theorem says that the exponential family is the most general family of distributions for which there is a sufficient statistics whose size is fixed with the number of samples, provided that the domain of the probability distribution does not vary with the parameter $ \theta $ \citep{koopman1936distributions}.

	\section{Properties of the conditional exponential family}
	\subsection{Linear identifiability}\label{app:identifiability}
	
	Let us consider the exponential family model $ p_w(x|\theta) =  \exp (\eta_w(\theta)^T  f_w(x))/ Z_w(\theta ) $, where here $ \eta_w $ and $ f_w $ are not restricted to be Neural Networks. \cite{khemakhem2020ice} studies the identifiability properties of the above family; specifically, they consider identifiability properties of the feature extractors $ \eta_w $ and $ f_w $. Identifiability of representations is useful as it means that two different models from the above family learn similar representations when trained with different initialization on the same data set. In the framework of our likelihood-free inference task, it also means that if the true model belongs to the family we are considering, it is theoretically possible to recover the natural parameters and the true sufficient statistics.
	
	If we consider now $ \eta_w $ and $ f_w $ to be Neural Networks, this causes problems as they are not identifiable in the standard sense. In fact, many different parameter configurations lead to the same function (as there are many symmetries in how the transformations in a Neural Network layer are defined). Therefore, \cite{khemakhem2020ice} introduces, following previous works, two more general notions of identifiable representations. Subsequently, let $ \mathcal{W} $ denote the space of the possible Neural Network weights $ w $.
	
	\begin{definition}\textbf{Weak identifiability (Section 2.2  in \citealp{khemakhem2020ice}).}
		Let $ \sim^f_w $ and  $ \sim^\eta_w $ be equivalence relations on $ \mathcal{W} $ defined as: 
		\begin{equation}\label{}
			\begin{aligned}
				&w \sim^f_w w' \iff f_w(x) = A f_{w'}(x) +c \\
				&w \sim^\eta_w w' \iff \eta_w(\theta) = B \eta_{w'}(\theta) + e
			\end{aligned}
		\end{equation}
		where $ A $ and $ B $ are $ (d_s \times d_s) $-matrices of rank at least $ \min(d_s, d) $ and $\min(d_s, p)$ respectively, and $ c  $ and $ e $ are vectors.
	\end{definition}
	
	\begin{definition}\textbf{Strong identifiability (Section 2.2  in \citealp{khemakhem2020ice}).}
		Let $ \sim^f_s $ and  $ \sim^\eta_s $ be equivalence relations on $ \mathcal{W} $ defined as: 
		\begin{equation}\label{}
			\begin{aligned}
				&w \sim^f_s w' \iff \forall i, f_{i,w}(x) = a_i f_{\sigma(i),w'}(x) +c_i \\
				&w \sim^\eta_s w' \iff \forall i,  \eta_{i,w}(\theta) = b_i \eta_{\gamma(i),w'}(\theta) + e_i
			\end{aligned}
		\end{equation}
		where $ \sigma $ and $ \gamma $ are permutations of $ [[1,n]] $, $ a_i $ and $ b_i $ are non-zero scalars and $ c_i $ and $ e_i $ are scalars.
	\end{definition}
	
	Weak identifiability means that two parameters are equivalent if the corresponding feature extractors are the same up to linear transformation. Strong identifiability is a specific case of the weak one, in which the linear transformation is restricted to be a scaled permutation. 
	
	After introducing these concepts, \cite{khemakhem2020ice} provides two theorems (Theorem 1 and 2) in which weak or strong identifiability of the representations are implied if different parameter values lead to same distributions, i.e. $ p_w(x|\theta) = p_{w'}(x|\theta) \forall \ x, \theta \implies w \sim^f_w w' $ or $ p_w(x|\theta) = p_{w'}(x|\theta) \forall \ x, \theta \implies w \sim^f_s w' $ (and similar for $ \sim^\eta_w $ and $ \sim^\eta_s $). These two results hold under some strict conditions on the functional form of the feature extractors $ f $ and $ \eta $ (concerning differentiability, rank of the Jacobian and other properties). 
	
	Then, they verify explicitly these conditions for a simple fully connected Neural Network architecture, in which they restrict the activation functions ot be LeakyReLUs and the layer widths to be monotonically increasing or decreasing. Of course, this architecture is quite restrictive; it is in fact impossible to study theoretically the properties of more complex architectures and to show that they satisfy the necessary conditions for identifiability to hold. However, they show empirically that the conditional structure, even with more complex architectures, is helpful in increasing identifiability of the representations (as computed by the MCC, see Appendix~\ref{app:MCC}). 
	
	Therefore, even if the architectures which are used throughout this work do not satisfy the assumptions needed to explicitly show identifiability of the representations, we argue that the presence of such results is an hint towards the fact that identifiability (according to the above definitions) is actually achievable.

	\subsubsection{Mean Correlation Coefficient (MCC)}\label{app:MCC}
	
	In order to evaluate empirically whether the above identifiability results are satisfied, we need a way to measure similarity between two embeddings of the same set of data. By following \cite{khemakhem2020ice}, we describe here the Mean Correlation Coefficient (MCC), which is a simple measure to do that. In the main text, we used this technique to assess how well our approximating family recovered the exact sufficient statistics and natural parameters, in the case where the data generating model $ p_0 $ belongs to the exponential family.

	In the following, we describe two versions of MCC, which are directly linked to the weak and strong identifiability definitions in Appendix~\ref{app:identifiability}.

	\paragraph{Strong MCC.}
	
	Let us consider two sets of embeddings $ \{y_i \in \R^{d_s}\}_{i=1}^n $ and $ \{z_i \in \R^{d_s}\}_{i=1}^n $, which can be thought of as samples from two multivariate random variables $ Y = t_1(X) $ and $ Z=t_2(X) $, for two functions $ t_1, t_2 $ and a random variable $ X $; in general, the order of the components of these two vectors is arbitrary, so that we cannot say, for instance, that the 1st component of $ Y $ corresponds to the 1st of $ Z $. Then, MCC computes all the correlation coefficients between each pair of components of $ Y $ and $ Z $. Next, it solves a linear sum assignment problem which identifies each component of $ Y $ with a component of $ Z $ aiming to maximize the sum of the absolute value of the corresponding correlation coefficient. In this way, it tries to couple the embeddings which are most linear one to the other. Finally, the MCC is computed as the mean of the absolute correlation coefficients after the right permutation of elements of the vector. 
	
	MCC is therefore a metric between 0 and 1 which measures how well each component of the original embedding (say, $ Y $) has been recovered independently by the other one (say, $ Z $). Moreover, as it relies on the correlation coefficient, it is not sensitive to rescaling or translation of each of the embeddings (in fact, the correlation coefficient of two univariate random variables is $ \pm 1 $ when a deterministic linear relationship between the two exists, unless the relationship is is perfectly horizontal or vertical).
	
	Finally, in order to get a better estimate of the MCC, we can split the set of embeddings in two: one which will be used to determine the permutation and which will give an in-sample estimate of the MCC, and the other one to which the previous permutation will be applied and will give an out-of-sample estimate of MCC.
	
	\paragraph{Weak MCC.}
	
	The above definition of MCC aims to estimate how well each single element of the embedding is recovered, independently on the other; we call it, following \cite{khemakhem2020ice}, strong MCC. However, there may be cases where the recovered embedding is equal to the correct one up to a more generic linear transformation $ A $. In that case, we would like to have a way to measure, up to a linear transformation, how far are two embeddings. Following \cite{khemakhem2020ice}, we therefore apply Canonical Correlation Analysis (CCA) \cite{hotelling1936relations} to learn $ A $, and after that compute MCC, which we will call weak MCC. Specifically, CCA is a way to compute linear transformation $ A $  so that the correlations between corresponding components of $ A \cdot Y $ and $ Z $ is maximized. The so-defined weak MCC is therefore a number between 0 and 1 which measures how close is $ Y $ to $ Z $ after the best linear transformation is applied to Y. Similarly as before, a part of data is used to learn the best embedding; we can therefore use fresh samples to get an out-of-sample estimate along the in-sample one.

	\subsection{Universal approximation capability}\label{app:universal_appr}
	
	\cite{khemakhem2020ice} provide a result (Theorem 3) in which they prove universal approximation capability of the conditional exponential family. Specifically, they show that, considering the dimension of the representations $ d_s $ as an additional parameter, it is possible to find an arbitrarily good approximation of any conditional probability density $ p_0(x|\theta) $, provided that $ \mathcal X $ and $ \Theta $ are compact Hausdorff spaces. In practice, good approximation may be achieved with a value of $ d_s $ larger than $ d $ or $ p $. Moreover, as remarked in the main text, this result is not concerned with the way the approximating family is fit; indeed, we expect that this task becomes (both statistically and computationally) more challenging when $ d_s $ increases.

	\section{Some properties of Score Matching}\label{app:Fish_div_prop}
	
	\subsection{Proof of Theorem \ref{Th:FD_explicit}}\label{app:th_FD_explicit_proof}
	
	Our Theorem~\ref{Th:FD_explicit} extends Theorem 1 in \cite{hyvarinen2005estimation}, that considers the case $ \X = \R^d $, and which is recovered in the case $ a_i = - \infty $ and $ b_i = + \infty $ $ \forall \ i $. Our proof follows \cite{hyvarinen2005estimation}, which however explicitly stated only Assumptions \ref{ass:A1} and \ref{ass:A2}. Following \cite{yu2018generalized}, we add the additional Assumption \ref{ass:A3} which is required for Fubini-Tonelli theorem to apply.
	
	In order to prove Theorem \ref{Th:FD_explicit}, Assumption \ref{ass:A3} can be weakened to $\E_{p_0} \left| \frac{\partial^2}{\partial x_i^2} \log p_w(X)  \right| < \infty, \forall w, \forall i=1,\ldots, d $. We state the more general one in the main text as that allows Theorem~\ref{Th:FDS_explicit} for SSM to be proved as well.
	
	\begin{proof}
		Let $ \s_0(x) = \nabla_{x}\log p_0(x) $ denote the score of the data distribution, and analogously $ \s(x;w) = \nabla_{x}\log p_w(x) $. Then, Eq.~\eqref{Eq:Fisher_div} can be rewritten as: 
		\begin{equation*}
			\begin{aligned}
				D_F(p_0\|p_w) &= \frac{1}{2} \int_{\mathcal{X}} p_0(x) \|\s_0(x) - \s(x;w) \|^2 dx \\&= \frac{1}{2} \int_{\mathcal{X}} p_0(x) \left[\|\s_0(x)\|^2 + \|\s(x;w) \|^2 -2\s_0(x)^T \s(x;w) \right] dx \\&= \underbrace{\frac{1}{2} \int_{\mathcal{X}} p_0(x) \|\s_0(x)\|^2 dx}_{C} +  \underbrace{\frac{1}{2} \int_{\mathcal{X}} p_0(x)\|\s(x;w) \|^2dx}_{A} \underbrace{- \int_{\mathcal{X}} p_0(x)\s_0(x)^T \s(x;w)  dx}_{B}
			\end{aligned}
		\end{equation*}
		Note that we can split the integral into the three parts as the first two are assumed to be finite in \ref{ass:A2}; as a consequence, $ B $ is also finite thanks to $ |2ab| \le a^2 + b^2 $. 
		
		The first element does not depend on $ w $, so that we can safely ignore that. The second one appears as it is in the final Eq.~\eqref{Eq:FD_explicit}. Therefore, we will focus on the last term, which we can write as: 
		\begin{equation}\label{Eq:FD_explicit_proof_sum}
			B = - \sum_{i=1}^d \int_{\X}  p_0(x) s_{0,i}(x) s_i(x;w) dx;
		\end{equation}
		Let's now consider a single $ i $, which we can rewrite as: 
		\begin{equation*}
			- \int_{\X}  p_0(x) \frac{\partial \log p_0(x) }{\partial x_i} s_i(x;w) dx = - \int_{\X}  \frac{p_0(x)}{p_0(x)} \frac{\partial p_0(x) }{\partial x_i} s_i(x;w) dx = - \int_{\X} \frac{\partial p_0(x) }{\partial x_i} s_i(x;w) dx.
		\end{equation*}
		Now, let's consider the integral over $ x_i $ first, and apply partial integration to it. Doing this relies on the Fubini-Tonelli's theorem, which can be safely applied due to $ B $ being finite as discussed above. 
		We now apply partial integration to the integral over $ x_i $: 
		\begin{equation*}
			\begin{aligned}
				- &\int_{a_1}^{b_1} \int_{a_2}^{b_2} \ldots \int_{a_d}^{b_d} \frac{\partial p_0(x) }{\partial x_i} s_i(x;w) dx = - \int_{a_1}^{b_1} \ldots \int_{a_{i-1}}^{b_{i-1}} \int_{a_{i+1}}^{b_{i+1}} \ldots \int_{a_{d}}^{b_{d}}  \Bigg[ \left. p_0(x) s_i(x;w) \right|_{x_i \searrow  a_i}^{x_i \nearrow  b_i} \\ &- \int_{a_i}^{b_i} p_0(x) \frac{\partial s_i(x;w)}{\partial x_i} dx_i \Bigg] dx_1 \ldots dx_{i-1} dx_{i+1} \ldots  dx_{d} \\ &= \int_{a_1}^{b_1} \ldots \int_{a_{i-1}}^{b_{i-1}} \int_{a_{i+1}}^{b_{i+1}} \ldots  \int_{a_{d}}^{b_{d}} \Bigg[ \int_{a_i}^{b_i} p_0(x) \frac{\partial s_i(x;w)}{\partial x_i} dx_i \Bigg] dx_1 \ldots dx_{i-1} dx_{i+1} \ldots  dx_{d} \\ &= \int_{\X}  p_0(x) \frac{\partial s_i(x;w) }{\partial x_i} dx.
			\end{aligned}
		\end{equation*}
		where the second equality holds thanks to Assumption \ref{ass:A1}. For stating the last equality rigorously, we need to invoke Fubini-Tonelli's theorem again; this relies on the assumption:
		
		\begin{equation*}
			\int_{\X}  \left|p_0(x) \frac{\partial s_i(x;w) }{\partial x_i} \right| dx = \E_{p_0} \left|\frac{\partial s_i(X;w) }{\partial x_i} \right| < \infty,
		\end{equation*}
		which is equivalent to Assumption \ref{ass:A3}.
		
		By repeating this argument for all terms in the sum in Eq.~\eqref{Eq:FD_explicit_proof_sum}, we obtain that: 
		\begin{equation*}
			B = \int_{\X}  p_0(x) \sum_{i=1}^d  \frac{\partial s_i(x;w) }{\partial x_i} dx =  \int_{\X}  p_0(x) \sum_{i=1}^d  \frac{\partial^2 \log p_w(x)}{\partial x_i^2} dx ,
		\end{equation*}
		which concludes our proof. 
	\end{proof}

	\subsection{Proof of Theorem \ref{Th:FD}}\label{app:th_FD_proof}
	We give here an extended version of Theorem \ref{Th:FD}, which we then prove following \cite{hyvarinen2005estimation}.
	
	\begin{theorem}\label{Th:FD_extended}
		Assume $ \exists w^\star : p_0(\cdot) = p_{w^\star}(\cdot) $. Then: 
		\begin{equation}\label{}
			w = w^\star \implies D_F(p_0\|p_w) = 0.
		\end{equation}
		Further, if $ p_0(x)>0 \ \forall x \in \mathcal{X}$, you also have: 
		\begin{equation}\label{}
			D_F(p_0\|p_w) = 0 \implies p_w(\cdot) = p_0(\cdot).
		\end{equation}
		Finally, if no other value $ w \neq w^\star $ gives a pdf $ p_w $ that is equal to $ p_{w^\star} $:
		\begin{equation}\label{}
			D_F(p_0\|p_w) = 0 \implies w = w^\star.
		\end{equation}
	\end{theorem}

	\begin{proof}
		The first statement is straightforward. For the second one, if $ D_F(p_0\|p_w) = 0  $ and $ p_0(x) > 0\ \forall\ x \in \mathcal{X} $, then $ \nabla_x \log p_w(x) = \nabla_x \log p_0(x) $; this in turn implies that $ \log p_w(x) = \log p_0(x) + c \ \forall\ x \in \mathcal{X}$ for a constant $ c $, which however is 0 as both $ p_0 $ and $ p_w $ are pdf's.

		The third statement follows from the second as, for the additional assumption, $ w^\star $ is the only choice of $ w $ which gives $ p_w = p_0 $. 
	\end{proof}
	
	\subsection{Proof of Theorem \ref{Th:FDS_explicit}}\label{app:th_FDS_explicit_proof}

	Similarly to the SM case, our Theorem~\ref{Th:FDS_explicit} extends Theorem 1 in \cite{song2019sliced}, that considers the case $ \X = \R^d $, and which is recovered in the case $ a_i = - \infty $ and $ b_i = + \infty $ $ \forall \ i $. Our proof follows \cite{song2019sliced}, which however explicitly stated only Assumptions \ref{ass:A1}, \ref{ass:A2} and \ref{ass:A4}. Following \cite{yu2018generalized}, we add the additional Assumption \ref{ass:A3} which is required for Fubini-Tonelli theorem to apply. The strategy of the proof is very similar to Theorem~\ref{Th:FD_explicit}.

	\begin{proof}
		As before, let $ \s_0(x) = \nabla_{x}\log p_0(x) $ denote the score of the data distribution, and analogously $ \s(x;w) = \nabla_{x}\log p_w(x) $. Then, Eq.~\eqref{Eq:Fisher_div_sliced} can be rewritten as: 
		\begin{equation*}
			\begin{aligned}
				D_{FS}(p_0\|p_w) &= \frac{1}{2} \int_{\mathcal{V}} \int_{\mathcal{X}} q(v) p_0(x) [v^T \s_0(x) - v^T \s(x;w)]^2 dx dv \\
				&= \frac{1}{2} \int_{\mathcal{V}} \int_{\mathcal{X}} q(v) p_0(x) \left[(v^T \s_0(x))^2 + (v^T \s(x;w))^2 -2 (v^T \s_0(x))(v^T \s(x;w))\right] dx dv \\
				&= \int_{\mathcal{V}} \int_{\mathcal{X}} q(v) p_0(x) \left[\frac{1}{2} (v^T \s(x;w))^2 - (v^T \s_0(x))(v^T \s(x;w))\right] dx dv + C 
			\end{aligned}
		\end{equation*}
		Note that we can split the integral into the three parts thanks to Assumptions \ref{ass:A2} and \ref{ass:A4}, which ensure that the expectation of each term is bounded. Additionally, we have absorbed in the constant $ C $ the term which does not depend on $ w $. 
		
		In the last row of the above Equation, the first term in the square brackets appears in Eq.~\eqref{Eq:FDS_explicit}, which is what we want to prove. We focus therefore on the second term:
		\begin{equation}\label{Eq:FDS_explicit_proof_sum}
			\begin{aligned}
				- &\int_{\mathcal{V}} \int_{\mathcal{X}} q(v) p_0(x) \left[(v^T \s_0(x))(v^T \s(x;w))\right] dx dv \\
				= - &\int_{\mathcal{V}} \int_{\mathcal{X}} q(v) p_0(x) \left[(v^T \nabla_x  \log p_0(x))(v^T  \nabla_x  \log p_w(x))\right] dx dv \\
				= - &\int_{\mathcal{V}} \int_{\mathcal{X}} q(v) \left[(v^T \nabla_x    p_0(x))(v^T  \nabla_x \log p_w(x))\right] dx dv \\
				= -&\sum_{i=1}^d \int_{\mathcal{V}} \int_{\mathcal{X}} q(v) \left[v_i\frac{\partial p_0(x)}{\partial x_i}(v^T  \nabla_x \log p_w(x))\right] dx dv \\
			\end{aligned}
		\end{equation}
		We will now consider one single element of the sum for a chosen $ i $, and consider the integral over $ x_i $ first. Swapping the order of integration relies on Fubini-Tonelli's theorem, which can be safely applied due to the above quantity being finite as discussed above.
		We then apply partial integration to the integral over $ x_i $: 
		\begin{equation*}
			\begin{aligned}
				- &\int_{\mathcal{V}} \int_{a_1}^{b_1} \int_{a_2}^{b_2} \ldots \int_{a_d}^{b_d} q(v) \left[v_i\frac{\partial p_0(x)}{\partial x_i}(v^T  \nabla_x \log p_w(x))\right] dx dv \\= - &\int_{\mathcal{V}} \int_{a_1}^{b_1} \ldots \int_{a_{i-1}}^{b_{i-1}} \int_{a_{i+1}}^{b_{i+1}} \ldots \int_{a_{d}}^{b_{d}}  q(v) \Bigg[ \left. v_i p_0(x)(v^T  \nabla_x \log p_w(x)) \right|_{x_i \searrow  a_i}^{x_i \nearrow  b_i} \\ &\qquad \quad - \int_{a_i}^{b_i} v_i p_0(x) \left(v^T  \frac{\partial }{\partial x_i}\nabla_x \log p_w(x)\right) dx_i \Bigg]dx_{-i},
			\end{aligned}
		\end{equation*}
		where $ dx_{-i}=dx_1 \ldots dx_{i-1} dx_{i+1} \ldots  dx_{d}  $. 
		The first element in the square bracket goes to 0 thanks to Assumption \ref{ass:A1}. We are left therefore with:
		\begin{equation}\label{Eq:FDS_explicit_step}
			\begin{aligned}
				\int_{\mathcal{V}} \int_{a_1}^{b_1} \ldots \int_{a_{i-1}}^{b_{i-1}} \int_{a_{i+1}}^{b_{i+1}} \ldots \int_{a_{d}}^{b_{d}} \Bigg[\int_{a_i}^{b_i} &q(v) p_0(x) v_i \left(v^T  \frac{\partial }{\partial x_i}\nabla_x \log p_w(x)\right)  dx_i \Bigg] dx_{-i}\\
				=\int_{\mathcal{V}} \int_{\X} &q(v) p_0(x)v_i\left(v^T  \frac{\partial }{\partial x_i}\nabla_x \log p_w(x)\right) dx dv.
			\end{aligned}	
		\end{equation}
		The last equality again exploits Fubini-Tonelli theorem, which in this case requires:
		\begin{equation*}
			\int_{\mathcal{V}} \int_{\X} q(v) p_0(x)\left| v_i \left(v^T  \frac{\partial }{\partial x_i}\nabla_x \log p_w(x)\right) \right| dx dv =\E_{V\sim q} \E_{X\sim p_0} \left| V_i \left(V^T  \frac{\partial }{\partial x_i}\nabla_x \log p_w(X)\right) \right|< \infty.
		\end{equation*}
		This is satisfied by combining Assumptions \ref{ass:A3} and \ref{ass:A4}, as in fact: 
		\begin{equation*}
			\begin{aligned}
				&\E_{V\sim q} \E_{X\sim p_0} \left| V_i \left(V^T  \frac{\partial }{\partial x_i}\nabla_x \log p_w(X)\right) \right|	= \E_{V\sim q} \E_{X\sim p_0} \left| V_i \sum_{j=1}^{d} V_j  \frac{\partial^2 }{\partial x_i \partial x_j} \log p_w(X)\right|\\
				\le& \sum_{j=1}^{d} \E_{V\sim q} \E_{X\sim p_0} \left| V_i V_j  \frac{\partial^2 }{\partial x_i \partial x_j} \log p_w(X)\right| = \sum_{j=1}^{d} \E_{V\sim q}  \left| V_i V_j \right|\cdot \E_{X\sim p_0}\left|\frac{\partial^2 }{\partial x_i \partial x_j} \log p_w(X)\right| \\
				\le &\sum_{j=1}^{d} \sqrt{\E_{V\sim q} V_i^2 \cdot \E_{V\sim q} V_j^2  }\cdot \E_{X\sim p_0}\left|\frac{\partial^2 }{\partial x_i \partial x_j} \log p_w(X)\right|,
			\end{aligned}
		\end{equation*}
		where the last inequality holds thanks to Cauchy-Schwarz inequality; Assumptions \ref{ass:A3} and \ref{ass:A4} ensure the last row is $ <\infty $.
		
		Back to Eq.~\eqref{Eq:FDS_explicit_proof_sum}, we have therefore: 
		\begin{equation}\label{}
			\begin{aligned}
				-&\sum_{i=1}^d \int_{\mathcal{V}} \int_{\mathcal{X}} q(v) \left[v_i\frac{\partial p_0(x)}{\partial x_i}(v^T  \nabla_x \log p_w(x))\right] dx dv \\
				=&\int_{\mathcal{V}} \int_{\X} \sum_{i=1}^d q(v) p_0(x)v_i\left(v^T  \frac{\partial }{\partial x_i}\nabla_x \log p_w(x)\right) dx dv\\
				=&\int_{\mathcal{V}} \int_{\X} \sum_{i=1}^d \sum_{j=1}^d q(v) p_0(x)v_i v_j  \frac{\partial^2 }{\partial x_i\partial x_j} \log p_w(x) dx dv\\
				=&\int_{\mathcal{V}}\int_{\X} q(v) p_0(x)\left[ v^T  \left(\Hessian_x \log p_w(x)\right) v\right]\ dx dv.
			\end{aligned}
		\end{equation}
	\end{proof}

	\subsection{Proof of Theorem \ref{Th:FDS}}\label{app:th_FDS_proof}
	We give here an extended version of Theorem \ref{Th:FDS}. This is a version of Lemma 1 in \cite{song2019sliced}, whose proof we adapt.
	
	\begin{theorem}\label{Th:FDS_extended}
		Assume $ \exists w^\star : p_0(\cdot) = p_{w^\star}(\cdot) $. Then: 
		\begin{equation}\label{}
			w = w^\star \implies D_{FS}(p_0\|p_w) = 0.
		\end{equation}
		Further, if $ p_0(x)>0 \ \forall x \in \mathcal{X}$, you also have: 
		\begin{equation}\label{}
			D_{FS}(p_0\|p_w) = 0 \implies p_w(\cdot) = p_0(\cdot).
		\end{equation}
		Finally, if no other value $ w \neq w^\star $ gives a pdf $ p_w $ that is equal to $ p_{w^\star} $:
		\begin{equation}\label{}
			D_{FS}(p_0\|p_w) = 0 \implies w = w^\star.
		\end{equation}
	\end{theorem}

	\begin{proof}
		The first statement is straightforward. 
		
		For the second one, if $ D_F(p_0\|p_w) = 0  $ and $ p_0(x) > 0\ \forall\ x \in \mathcal{X} $, then:
		
		\begin{equation}\label{}
			\begin{aligned}
				&\int_{\mathcal{V}} q(v) (v^T (\nabla_x\log p_0(x) -  \nabla_x \log p_w(x) ))^2 dv=0 \\
				\iff &\int_{\mathcal{V}} q(v) v^T (\nabla_x\log p_0(x) -  \nabla_x \log p_w(x) )(\nabla_x\log p_0(x) -  \nabla_x \log p_w(x) )^T v dv=0 \\
				\iff &(\nabla_x\log p_0(x) -  \nabla_x \log p_w(x) )^T E[VV^T] (\nabla_x\log p_0(x) -  \nabla_x \log p_w(x) )=0 \\
				\stackrel{(\star)}{\iff} &\nabla_x\log p_0(x) -  \nabla_x \log p_w(x) =0 \\
				\iff &\log p_w(x) = \log p_0(x) + c \ \forall\ x \in \mathcal{X},
			\end{aligned}	
		\end{equation}
		where in the third line above $ V $ is a random variable distributed according to $ q $, for which therefore $ \E[VV^T] $ is positive definite by Assumption \ref{ass:A4}, which ensures equivalence $ (\star) $ holds. Additionally, as both $ p_0 $ and $ p_w $ are pdf's (and therefore normalized), the constant $ c =0$.
		
		The third statement follows from the second as, for the additional assumption, $ w^\star $ is the only choice of $ w $ which gives $ p_w = p_0 $. 
	\end{proof}

	\subsection{Proof of Theorem \ref{Th:FD_y}}\label{app:th_FD_y_proof}
	We give here an extended version of Theorem \ref{Th:FD_y}, which we then prove.

	\begin{theorem}\label{Th:FD_y_extended}
		Let $ Y=t(X) \in \mathcal Y$ for a bijection $ t $, and denote by $ p_0^Y $ and $ p_w^Y $ the distributions on $ \mathcal{Y} $ induced by the distributions $ p_0 $ and $ p_w $ on $ \X $.	
		Assume $ \exists w^\star : p_0(\cdot) = p_{w^\star}(\cdot) $, and let $ D $ denote either $ D_F $ or $ D_{FS} $. Then: 
		\begin{equation}\label{}
			w = w^\star \implies D(p^Y_0\|p^Y_w) = 0.
		\end{equation}
		Further, if $ p_0(x) > 0\ \forall x \in \mathcal{X}$ and Assumption \ref{ass:A4} holds, you also have: 
		\begin{equation}\label{}
			D(p^Y_0\|p^Y_w) = 0 \implies p_w(\cdot) = p_0(\cdot).
		\end{equation}
		Finally, if no other value $ w \neq w^\star $ gives a pdf $ p_w $ that is equal to $ p_{w^\star} $:
		\begin{equation}\label{}
			D(p^Y_0\|p^Y_w) = 0 \implies w = w^\star.
		\end{equation}
	\end{theorem}
	\begin{proof}
		The proof relies on the equivalence between distributions for the random variables $ Y $ and $ X $; in fact, by fixing $ y=t(x) $ and denoting by  $ |J_t(x)| $ the determinant of the Jacobian matrix of $ t $ evaluated in $ x $, we have that $ p_0^Y(y) = \frac{p_0(x)}{|J_t(x)|} $ and $ p_w^Y(y) = \frac{p_w(x)}{|J_t(x)|} $, so that $ p_0(\cdot) = p_{w}(\cdot) \iff p_0^Y(\cdot) = p_{w}^Y(\cdot)$. The first and third statements follow directly from this fact by applying Theorem~\ref{Th:FD_extended} (if $ D $ is chosen to be $ D_F $) or Theorem~\ref{Th:FDS_extended} (if $ D = D_{FS} $)  to $ D(p^Y_0\|p^Y_w) $; for the second, notice also that $ p_0(x) > 0\ \forall\ x \in \X \implies p^Y_0(y)\ \forall\ y \in \mathcal{Y} > 0 $ as $ t $ is a bijection; then, Theorem~\ref{Th:FD_extended} or Theorem~\ref{Th:FDS_extended} imply that $ p^Y_0(\cdot) = p^Y_w(\cdot) \implies  p_0(\cdot) = p_w(\cdot) $.  		\end{proof}
	
	\subsection{Discussion on the positivity condition for SM}\label{app:Fish_div_prop_disjoint_supp}
	
	We follow here the discussion in Appendix D of \cite{arbel2017kernel}. 
	
	Consider again the Fisher divergence in Eq.~\eqref{Eq:Fisher_div}; we want to understand the conditions under which this is a divergence between probability measures, which essentially means it is 0 if and only if the probability distributions are the same.
	
	Besides the fact that both $p_0$ and $p_w$ need to be continuous (otherwise the gradient would be a delta function), it turns out that a necessary condition is that $p_0$ is positive on the whole space to which the random variable belong (let's say $\mathcal X$), if we don't put any restrictions on $p_w$. Otherwise, the following scenario may happen (see Appendix D in \citealp{arbel2017kernel}): consider the case in which $p_0(x)$ is a mixture of two densities supported on disjoint sets: $p_0(x) = \alpha_A p_A(x) + \alpha_B p_B(x)$, with non-negative weights $\alpha_A, \alpha_B$, such that 
	$$ p_A(x) > 0 \iff x \in \mathcal{X}_A, \quad p_B(x) > 0 \iff x \in \mathcal{X}_B, \quad  \mathcal{X}_A, \mathcal{X}_B \in \mathcal{X};  \quad \mathcal{X}_A \cap \mathcal{X}_B = \varnothing. $$
	
	Note that this implies $\exists x \in \mathcal{X} : p_0(x)=0 $. In this setting, any $ p_w $ of the form $p_w(x) = \beta_A p_A(x) + \beta_B p_B(x)$ will give $ D_F(p_0\|p_w) = 0 $. 
	This can be seen by computing the Fisher divergence directly:
	
	\begin{equation}\label{}
		\begin{aligned}
			D_F(p_0\|p_w) &= \frac{1}{2} \int_\mathcal{X} p_0(x) \|\nabla_x\log p_0(x) - \nabla_x \log p_w(x) \|^2 dx \\ 
			&=\frac{1}{2} \int_\mathcal{X} p_0(x) \|\nabla_x\log (\alpha_A p_A(x) + \alpha_B p_B(x)) - \nabla_x \log (\beta_A p_A(x) + \beta_B p_B(x)) \|^2 dx \\
			&=\frac{1}{2} \int_{\mathcal{X}_A} p_0(x) \|\nabla_x\log (\alpha_A p_A(x)) - \nabla_x \log (\beta_A p_A(x)) \|^2 dx +\\
			&\qquad \frac{1}{2} \int_{\mathcal{X}_B} p_0(x) \|\nabla_x\log (\alpha_B p_B(x)) - \nabla_x \log (\beta_B p_B(x)) \|^2 dx \\ 
			&=\frac{1}{2} \int_{\mathcal{X}_A} p_0(x) \underbrace{\|\nabla_x\log p_A(x) - \nabla_x \log p_A(x) \|^2}_{=0} dx +\\ &\qquad \frac{1}{2} \int_{\mathcal{X}_B} p_0(x) \underbrace{\|\nabla_x\log p_B(x) - \nabla_x \log p_B(x) \|^2}_{=0} dx \\ &= 0,
		\end{aligned}
	\end{equation}
	where the third equality relies on splitting the integration domain over the two subsets on which $ p_A $ and $ p_B $ are supported (and the other is 0) and the fourth equality relies on the presence of the logarithmic derivatives, that removes the importance of the mixture weights.

	Due to the above, in the case of the conditional Fisher divergence in Eq.~\eqref{Eq:exp_fisher_div}, $ p_0(x|\theta) $ needs to be supported on the whole $ \mathcal{X} $ for each $ \theta $ in order for $ D_F^E(p_0|p_w) = 0 \iff p_0(\cdot|\theta) = p_w(\cdot|\theta)$ $ \pi(\theta) $-almost everywhere. This can be seen by considering the case of univariate parameter $	\theta $  and by choosing $ p_0(x|\theta) = p_A(x) H(\theta) + (1-H(\theta)) p_B(x)$, where $ p_A $ and $ p_B $ are as above and $ H(\cdot) $ represents the Heaviside function. In this case, choosing $ p_w(x|\theta) = q(x) $, where $ q(x)  = \beta_A p_A(x) + \beta_B p_B(x) $ will give $ D_F^E(p_0|p_w) = 0 $, as for each fixed $ \theta $, $ D_F(p_0(\cdot|\theta)|p_w(\cdot|\theta)) $ falls in the case considered above.

	\subsection{Equivalence of Correction Factor and Transformed Score Matching}\label{App:corrected_transformed_SM_equivalence}

	As discussed in the main text (Section~\ref{sec:density_est}), the first extension of score matching to non-negative random variables was given in \cite{hyvarinen2007some}:
	\begin{equation}\label{Eq:Fisher_div_non_neg}
		D_F^+(p_0\|p_w) = \frac{1}{2} \int_{\R^d_+} p_0(x) \|\nabla_x\log p_0(x) \odot x - \nabla_x \log p_w(x) \odot x \|^2 dx,
	\end{equation}
	where $ \odot $ denotes element wise product between vectors and $ \R^d_+ $ is the positive octant of $ \R^d $. The correction factor $ x $ relaxes assumption \ref{ass:A1} to $ p_0(x) x_i^2 \frac{\partial}{\partial x_i} \log p_w(x) \to 0 $, so that it is possible to get an explicit form of the above with looser assumptions. This is an example of Corrected Score Matching (CorrSM, \ref{sec:density_est}), in which the issue arising due to distribution having a compact support is fixed by introducing a correction factor in the formulation of the objective. 
	
	\cite{yu2018generalized} proposed a more general score matching for non-negative random variables by allowing freedom of choice in the factor that is used in the integrand to correct for the integration by parts step (Appendix~\ref{app:th_FD_explicit_proof}), leading to the following objective:
	\begin{equation}\label{Eq:generalized_nonneg_Fish_div}
		D_F^+(p_0\|p_w) = \frac{1}{2} \int_{\mathbb{R}^n_+} p_0(x) \|(\nabla_x\log p_0(x))  \odot \sqrt{h(x)} - (\nabla_x \log p_w(x))  \odot \sqrt{h(x)}| \|^2 dx,
	\end{equation}
	where $ h(x) $ has the same dimension as $ x $, and has positive elements almost surely. 
	
	The explicit formulation associated to Eq.~\eqref{Eq:generalized_nonneg_Fish_div} can be obtained under the following assumptions:
	\begin{enumerate}[label=\textbf{A\arabic*b}]
		\item \label{ass:A1b} $ p_0(x) h_j(x_j) \partial_j \log p_w(x) \to 0 $ for $ x_i \searrow 0$ and $ x_i \nearrow \infty,\  \forall \ w, i$,
		
		\item\label{ass:A2b} $ \E_{p_0} \|\nabla_x \log p_0(X) \odot h^{1/2}(X)\|_2^2 < \infty, $ $ \E_{p_0} \|\nabla_x \log p_w(X) \odot h^{1/2}(X)\|_2^2 < \infty\ \forall w$,
		\item\label{ass:A3b} $\E_{p_0} \| (\nabla_x \log p_w(X) \odot h(X)))' \|_1 < \infty \ \forall w $, where the prime symbol denotes element-wise differentiation.
	\end{enumerate}

	Under the above assumption, Eq.~\eqref{Eq:generalized_nonneg_Fish_div} is equal to:
	\begin{equation}\label{Eq:generalized_Fish-div_nonnegative_partial_integr}
		\begin{aligned}
			D_F^+(p_0\|p_w) =  \int_{\mathbb{R}^n_+} p_0(x) \sum_{i=1}^d \Bigg[ &\frac{1}{2} h_i(x) \left( \frac{\partial \log p_w(x)}{\partial x_i}  \right)^2 + \\&h_i(x) \left( \frac{\partial^2 \log p_w(x)}{\partial x_i^2}  \right)  + h_i'(x)  \frac{\partial \log p_w(x)}{\partial x_i} \Bigg] dx + C,
		\end{aligned}
	\end{equation}
	where $ C$ is a constant with respect to $ p_w $.
	
	In Proposition 2 in \cite{yu2018generalized}, they give a result similar to our Theorems~\ref{Th:FD} and \ref{Th:FD_y} guaranteeing that minimization of $ D_F^+(p_0\|p_w) $ is a valid procedure for estimating a probabilistic model. When considering the finite-sample estimate of \ref{Eq:generalized_Fish-div_nonnegative_partial_integr}, different choices of $ h(x) $ may allow to focus more on smaller/larger values of $x$, which may in practice have better properties than the original form for non-negative data in \cite{hyvarinen2007some}, which is recovered for $ h(x) = x^2 $ (where the square is applied element-wise).

	This formulation in Eq.~\eqref{Eq:generalized_nonneg_Fish_div}, albeit originally considered for non-negative random variables only, can be extended to random variables with any bounded domain, by choosing a suitable function $ h $ and modifying \ref{ass:A1b} to hold for $ x_i $ going to the limits of the domain. In the next Sections, we therefore compare this approach with TranSM without specifying the domain; we will show that both the implicit and explicit formulation are the same with both TranSM and CorrSM, implying that the two are equivalent (we will show this in the specific case in which the transformation and the function $ h $ act independently on the different coordinates, but we believe this to be the case more in general; see Appendix~\ref{app:SM_extensions}).

	\subsubsection{Equivalence of the implicit form}
	
	As mentioned in the main text (Section~\ref{sec:density_est}), another approach to apply Fisher divergence to distributions with bounded domain (on one side or both) is to transform the data space to the real line and then apply standard Fisher divergence; we called this Transformation Score Matching (TranSM). Let us denote $ t $ such a transformation, which we assume to be bijective. Then, starting from $ p_0(x) $ and $ p_w(x) $ we get $ p_0(y) = \frac{p_0(x)}{|J_t(x)|} $ and $ p_w(y) = \frac{p_w(x)}{|J_t(x)|} $ for $ y=t(x) $, where $ J_t $ is the Jacobian matrix of $ t $ and $ |\cdot| $ denotes here the determinant; here, differently from the main text, we use a lighter notation where $ p_0(y) $ and $ p_0(x) $ are two different densities associated to the different name of the argument (same for $ p_w $). We investigate what is the Fisher divergence between the densities of the transformed distributions. Recall that $ p_0(y) dy = p_0(x) dx $, for $ y=t(x) $. Moreover, we also have that: 
	
	\begin{equation}\label{}
		\nabla_y g(y) = J_{t^{-1}}(t(x)) \nabla_x g(t(x)) = (J_{t}(x))^{-1} \nabla_x g(t(x)),
	\end{equation}
	where the second equality comes from the fact that $ J_{t^{-1}}(t(x)) = (J_{t}(x))^{-1} $ due to the inverse function theorem. Then, we can compute the Fisher Divergence between $ p_0(y) $ and $ p_w(y) $ (corresponding to the TranSM objective):
	\begin{equation}\label{}
		\begin{aligned}
			D_F(&p_0(y)\|p_w(y)) = \frac{1}{2} \int p_0(y) \|\nabla_y\log p_0(y) - \nabla_y \log p_w(y) \|^2 dy\\
			&= \frac{1}{2} \int p_0(x) \|(J_{t}(x))^{-1} [\nabla_x \log p_0(x) - \nabla_x \log |J_{t}(x)| - \nabla_x \log p_w(x) + \nabla_x \log |J_{t}(x)| ]\|^2 dx \\
			&= 	\frac{1}{2} \int p_0(x) \|(J_{t}(x))^{-1} [\nabla_x \log p_0(x) - \nabla_x \log p_w(x) ]\|^2 dx.
		\end{aligned}		
	\end{equation}
	
	In the rather common case in which the transformation $ t $ acts on each component independently, the Jacobian matrix is diagonal; in this case, the latter is equivalent to Eq.~\eqref{Eq:generalized_nonneg_Fish_div} upon defining $ h(x) $ to be a vector containing the squares of the diagonal elements of the Jacobian, i.e. putting $ \sqrt{h_i(x)} = (J_{t}(x))^{-1}_{ii} $. 
	
	\subsubsection{Equivalence of the explicit form}
	
	For both TranSM and CorrSM it is possible to get an explicit form of the objective (Eqs.~\ref{Eq:FD_explicit} and \ref{Eq:generalized_Fish-div_nonnegative_partial_integr}), in which the integrand does not depend on the data distribution $ p_0 $. In case in which the transformation $ t $ acts on the different components independently, the original explicit divergence for the transformed variable $ Y=t(X) $ is equivalent to the corrected explicit divergence for the original $ X $, analogously to the implicit Fisher divergence form. In fact, by applying the definition of explicit Fisher divergence (Eq.~\ref{Eq:FD_explicit}) to the transformed $ Y $, you get: 
	\begin{equation}\label{}
		D_F(p_0(y)\|p_w(y)) = \int p_0(y) \sum_{i=1}^d \underbrace{\left[ \frac{1}{2} \left( \frac{\partial \log p_w(y)}{\partial y_i}  \right)^2 +  \left( \frac{\partial^2 \log p_w(y)}{\partial y_i^2}  \right) \right]}_{\star} dy + C 
	\end{equation}
	where $ C $ is a constant with respect to $ p_w $. Considering only the term in square brackets, denoting $ \partial_i = \frac{\partial}{\partial x_i} $ and setting $ \sqrt{h_i(x)} = (J_{t}(x))^{-1}_{ii} $, we get: 
	\begin{equation}\label{}
		\begin{aligned}
			\star &= \frac{1}{2} h_i(x) \left[  \left( \partial_i \log p_w(x) \right)^2 +  \frac{1}{4} \left( \partial_i \log h_i(x) \right)^2 + \partial_i \log p_w(x)  \cdot \partial_i \log h_i(x) \right] \\ &+ \frac{1}{2} h_i'(x) \cdot \partial_i \log p_w(x) + \frac{1}{4} h_i'(x) \partial_i \log h_i(x) + h_i(x) \partial_i^2 \log p_w(x) + \frac{1}{2} h_i(x) \partial_i^2 \log h_i(x) \\
			&= \textcolor{blue}{\frac{1}{2} h_i(x) \left( \partial_i \log p_w(x) \right)^2} +  \frac{1}{8} h_i(x) \left( \partial_i \log h_i(x) \right)^2 \\ &+ \textcolor{blue}{h_i'(x) \cdot \partial_i \log p_w(x)} + \frac{1}{4} h_i'(x) \partial_i \log h_i(x) + \textcolor{blue}{h_i(x) \partial_i^2 \log p_w(x)} + \frac{1}{2} h_i(x) \partial_i^2 \log h_i(x).
		\end{aligned}
	\end{equation}	
	
	The blue terms are the same that appear in the CorrSM explicit formulation (Eq.~\ref{Eq:generalized_Fish-div_nonnegative_partial_integr}), while all other terms are constants with respect to $ p_w $. 
	
	We have shown therefore that CorrSM and TranSM are equivalent in both the explicit and implicit formulation if the transformation is applied independently on the elements of $ x $. Therefore, the two approaches are completely equivalent when it comes to minimizing them.

	\subsection{Specific formulation of TranSM}\label{app:transformations}

	We discuss here the transformations we apply in this work in TranSM; specifically, we only consider the case in which the support for the multivariate $ x $ is defined by an intersection of element-wise inequalities, i.e. $ x \in \bigotimes_{i=1}^d (a_i,b_i) $, where $ a_i, b_i $ can take on the values $ \pm \infty $ as well. In this case, then, a transformation can be applied independently on each element of $ x $. We consider here the following transformations (which are also used in the \texttt{Stan} package \citealp{carpenter2017stan}):

	\begin{itemize}
		\item When $ X \in [0, \infty)^d $, the transformation we use is $ y_i = \log (x_i) \in \R^d$. This corresponds to diagonal Jacobian with elements $ (J_{t}(x))^{-1}_{ii} = x_i $, so that the above expression becomes the same as the original Fisher divergence for non-negative random variables discussed in Eq.~\eqref{Eq:Fisher_div_non_neg}. 
		
		\item More generally, if $ x_i \in [a_i, +\infty) $  for $ |a_i| < \infty $, we can transform the data as $ y_i = \log (x_i - a_i) \in \R $, while if $ x_i \in (-\infty, b_i] $  for $ |b_i| < \infty $, we simply reverse the transformation: $ y_i = \log (b_i - x_i) \in \R $. These correspond to $ (J_{t}(x))^{-1}_{ii} = x_i - a_i$ and $ (J_{t}(x))^{-1}_{ii} = b_i - x_i$.

		\item Finally, if $ x_i \in (a_i,b_i) $ for $ |a_i|, |b_i| < \infty $, we can use the transformation defined as: $ y_i = t(x_i) = \logit \left(\frac{x_i - a_i}{b_i - a_i} \right) $ with inverse transformation $x_i = t^{-1} (y_i) = a + (b-a) \frac{e^{y_i}}{e^{y_i} + 1}$. This corresponds to $ (J_{t}(x))^{-1}_{ii} = \frac{(x_i - a_i)(b_i - x_i)}{b_i-a_i}$.

	\end{itemize}

	\subsection{Score matching for distributions with more general domain}\label{app:SM_extensions}
	
	As highlighted in the main text, across this work we are concerned with applying score matching to distributions whose support is defined by independent constraints on the different coordinates, as for instance $ \mathcal X = \bigotimes_{i=1}^d (a_i,b_i) $. That is arguably the most common case in the literature. However, there have been some recent works which applied SM to more general cases. For instance, \cite{mardia2016score} devised a way to apply it to a directional distribution defined on an oriented Riemannian manifold (for instance, a sphere). It is interesting how their derivation of the explicit form from the implicit one relies on the classical divergence theorem (also known as Stokes' theorem), of which the partial integration trick used in Theorem~\ref{Th:FD_explicit} is a specific case. \cite{liu2019estimating} introduced instead a way to apply score matching for a distribution on Euclidean space with complex truncation boundaries; their approach boils down to introducing a smart correction factor which goes to 0 at the boundary (thus allowing partial integration) but still being tractable; again, they need a more general version of Theorem~\ref{Th:FD_explicit} to obtain an objective for which the integrand does not depend on the data distribution.

	In Appendix~\ref{App:corrected_transformed_SM_equivalence}, we established that CorrSM and TranSM are equivalent if the transformation is applied independently on the elements of $ x $, which requires the domain to be defined by independent constraints on the coordinates. In the more general case of a irregular subset of Euclidean space (as in the setup of \citealp{liu2019estimating}), it is not clear whether it is always possible to associate a correction factor to a transformation which maps the space to $ \R^d $. That seems to be plausible if the domain satisfies some regularity conditions which may be related to convexity (for instance think of a triangle in $ \R^2 $, which can be easily stretched to cover the full space). We are not aware however of any work investigating this.

	\subsection{Score matching with exponential family}\label{app:sm_exp_fam}
	
	We consider here the exponential family:  
	$$p_w(x|\theta) = \exp (\eta_w(\theta)^T f_w(x))/Z_w(\theta ),$$
	and want to find the value of $w$ minimizing either $D_F^E(p_0\|p_w)$ or $D_{FS}^E(p_0\|p_w)$, which are defined in Eq.~\eqref{Eq:exp_fisher_div_exp}.

	Under the conditions discussed in Section~\ref{sec:exp_fish_div_likelihood}, if $\pi(\theta) > 0\ \forall\ \theta$, then $D_F^E(p_0\|p_w)=0$ and $D_{FS}^E(p_0\|p_w)=0 \iff p_w(x|\theta)= p_0(x|\theta) $ $ \pi(\theta) $-almost everywhere. In this case, if $ f_w $ and $ \eta_w $ satisfy the conditions required for the theorems mentioned in Appendix~\ref{app:identifiability} to hold, then $ f_w $ and $ \eta_w $ are respectively sufficient statistics and natural parameters of $ p_0 $. 
	
	In order to find the value of the empirical estimate of the explicit for of both $D_F^E(p_0\|p_w)$ and $D_{FS}^E(p_0\|p_w)$, we insert the definition of the exponential family with in Eq.~\ref{Eq:FD_expected_explicit_MC}, which leads to:
	\begin{equation}
		\begin{aligned}
			\hat J(w) &=\frac{1}{N} \sum_{j=1}^N \left[ \sum_{i=1}^d \left( \frac{1}{2} \left(\eta_w(\theta^{(j)})^T \frac{\partial}{\partial x_i} f_w(x^{(j)})\right)^2 + \eta_w(\theta^{(j)})^T  \frac{\partial^2}{\partial x_i^2} f_w(x) \right) \right] , \\
			\hat J_S(w) &=\frac{1}{NM} \sum_{j=1}^N \sum_{k=1}^M \left[ v^{(j,k),T} \Hessian_x (\eta_w(\theta^{(j)})^T f_w(x^{(j)}))  v^{(j,k)} +  \frac{1}{2}\sum_{i=1}^d \left(\eta_w(\theta^{(j)})^T \frac{\partial}{\partial x_i} f_w(x^{(j)})\right)^2 \right].
		\end{aligned}
	\end{equation}
	
	Note that the objective does not change if you set $f_w(x) $ to $c + f_w(x)$, for a constant vector $ c $; in fact, this constant gets absorbed into the normalizing constant in $p_w(x|\theta)$. 
	
	Similarly, $ \eta_w(\theta)^T f_w(x) = (1/c \cdot \eta_w(\theta))^T (c \cdot f_w(x) ) $ for some constant $c\neq0$. Therefore, statistics and corresponding parameters are only defined up to a scale with respect to one another; if you use two Neural Networks to learn both of them, different network initializations may lead to different learned statistics and natural parameters, but their product should be fixed (up to translation of $ f_w(x) $).

	However, this degeneracy may make training the approximate likelihood $ p_w $ with the score matching approach harder. In order to improve training, we usually add a Batch Normalization layer on top of the $ \eta_w $ network. Basically, Batch Normalization fixes the scale of the output of $ \eta_w $ over a training batch, therefore removing this additional degree of freedom and making training easier. We discuss in more detail this in Appendix~\ref{app:batch_norm}.

	\section{Scoring Rules}\label{app:SRs}
	
	A Scoring Rule (SR) $ S $ \citep{dawid2014theory, gneiting2007strictly} is a function of a probability distribution over $ \X $ and of an observation in $ \X $. In the framework of probabilistic forecasting, $ S(P, \dobs) $ represents the \textit{penalty} which you incur when stating a forecast $ P $ for an observation $ \dobs $.\footnote{Some authors \citep{gneiting2007strictly} use the convention of $ S(P,\dobs) $ representing a \textit{reward} rather than a {penalty}, which is equivalent up to change of sign.}
	
	If the observation $\dobs$ is a realization of a random variable $ \Dobs $ with distribution $ Q $, the expected Scoring Rule can be defined as: 
	\begin{equation}
		S(P,Q) := \E_{\Dobs \sim Q} S(P, \Dobs),
	\end{equation}
	where we overload notation in the second argument of $ S $.
	The Scoring Rule $ S $ is said to be \textit{proper} relative to a set of distributions $ \mathcal{P}(\X) $ over $ \X $ if $$ S(Q, Q) \le S(P,Q) \ \forall \ P,Q \in \mathcal{P}(\X),$$ i.e., if the expected Scoring Rule is minimized in $ P $ when $ P=Q $. Moreover, $ S $ is \textit{strictly proper} relative to $ \mathcal{P}(\X) $ if $ P = Q $ is the unique minimum: $$ S(Q,Q) < S(P,Q)  \ \forall \ P, Q \in \mathcal{P}(\X)  \text{ s.t. }  P\neq Q;$$ 
	when minimizing an expected strictly proper Scoring Rule, a forecaster would provide their true belief \citep{gneiting2007strictly}.

	By following \cite{dawid2014theory}, we define the divergence related to a proper Scoring Rule as $ D(P,Q) := S(P,Q) - S(Q,Q) \ge 0 $. Notice that $ P=Q \implies D(P,Q) = 0$, but there may be $ P\neq Q $ such that $ D(P,Q)=0 $. However, if $ S $ is strictly proper, $ D(P,Q) = 0 \iff P=Q $, which is the commonly used condition to define a statistical divergence (as for instance the Kullback-Leibler, or KL, divergence).
	Therefore, each strictly proper Scoring Rule corresponds to a statistical divergence between probability distributions

	In the following, we detail the two Scoring Rules which we consider in the main text (Section~\ref{sec:Lorenz}).
	
	\paragraph{Energy score}
	The energy score is given by: 
	\begin{equation}
		\SE^{(\beta)}(P, \dobs) = 2 \cdot \E \left[\| \Dsim - \dobs\|_2^\beta\right] - \E\left[\|\Dsim- \Dsim'\|_2^\beta\right] ,\quad  \Dsim \independent  \Dsim' \sim P, 
	\end{equation}
	where $ \beta \in (0,2) $ and $ \independent $ denotes independence between random variables.  
	This is a strictly proper Scoring Rule for the class $ \mathcal{P}_\beta (\X)$ of probability measures $ P $ such that $ \E_{X\sim P}\|\Dsim\|^\beta < \infty $ \citep{gneiting2007strictly}. The related divergence is the square of the energy distance, which is a metric between probability distributions (\citealt{rizzo2016energy})\footnote{The probabilistic forecasting literature \citep{gneiting2007strictly} use a different convention of the energy score and the subsequent kernel score, which amounts to multiplying our definitions by $ 1/2 $. We follow here the convention used in the statistical inference literature \citep{rizzo2016energy, cherief2020mmd, nguyen2020approximate}}:
	
	\begin{equation}
		\DE(P,Q) = 2\cdot \E \left[\| \Dsim - \Dobs \|_2^\beta \right]- \E\left[\|\Dsim- \Dsim'\|_2^\beta \right]  - \E \left[\|\Dobs- \Dobs'\|_2^\beta \right], 
	\end{equation}
	for $ \Dsim\independent \Dsim' \sim P $ and $ \Dobs \independent \Dobs' \sim Q$. 
	
	In our case of interest (Section~\ref{sec:Lorenz}, Appendix~\ref{app:Lorenz_validation_results}), we are unable to evaluate exactly $ \SE^{(\beta)} $ as we do not have a closed form for $ P $. Thus, we obtain samples $ \{\dsim_j\}_{j=1}^m$, $ \dsim_j\sim P $, and unbiasedly estimate the Energy Score with:
	\begin{equation}
		\hat S_{\text{E}}^{(\beta)}(\{\dsim_j\}_{j=1}^m, \dobs) =\frac{2}{m} \sum_{j=1}^m \left\| \dsim_j - \dobs\right\|_2^\beta - \frac{1}{m(m-1)}\sumjk \left\|\dsim_j-\dsim_k\right\|_2^\beta.
	\end{equation}
	In the main text (Section~\ref{sec:Lorenz}), we use $ \beta=1$, in which case we simplify notation $ S_{\text{E}}^{(1)} = \SE $.
	
	\paragraph{Kernel score}
	For a positive definite kernel $ k(\cdot, \cdot) $ (see \citealt{gretton2012kernel}), the kernel Scoring Rule for $ k $ is defined as \citep{gneiting2007strictly}: 
	\begin{equation}
		S_k(P, \dobs) = \E[k(\Dsim,\Dsim')] - 2\cdot\E [k(\Dsim, \dobs)],\quad  \Dsim \independent  \Dsim' \sim P. 
	\end{equation}
	The corresponding divergence is the squared Maximum Mean Discrepancy (MMD, \citealp{gretton2012kernel}) relative to the kernel $ k $:
	\begin{equation}
		D_k(P, Q) = \E [k(\Dsim, \Dsim')] + \E [k(\Dobs, \Dobs')]- 2\cdot \E [k(\Dsim,\Dobs)],
	\end{equation}
	for $ \Dsim \independent\Dsim' \sim P $ and $ \Dobs\independent \Dobs' \sim Q$. 
	
	The Kernel Score is proper for the class of probability distributions for which $ \E[k(\Dsim,\Dsim')] $ is finite (by Theorem 4 in \citealp{gneiting2007strictly}). Additionally, it is strictly proper under conditions which ensure that the MMD is a metric for probability distributions on $ \X $ (see for instance \citealp{gretton2012kernel}). These conditions are satisfied, among others, by the Gaussian kernel (which we use in this work):
	\begin{equation}\label{Eq:gau_k}
		k(x, y)=\exp \left(-\frac{\|x-y\|_{2}^{2}}{2 \gamma^{2}}\right);
	\end{equation}
	there, $ \gamma $ is a scalar bandwidth, which is tuned as described in Appendix~\ref{app:tuning_gamma}. 	As for the Energy Score, when the exact form of $ P $ is inaccessible and therefore $ S_k $ is impossible to be evaluated exactly, we use samples $ \{\dsim_j\}_{j=1}^m$, $ \dsim_j\sim P $, and unbiasedly estimate the Kernel Score $ S_k(P,\dobs) $ with:
	\begin{equation}
		\hat S_k(\{\dsim_j\}_{j=1}^m, \dobs) = \frac{1}{m(m-1)}\sumjk  k(\dsim_j,\dsim_k )-\frac{2}{m} \sum_{j=1}^m k(\dsim_j,\dobs).
	\end{equation}

	\section{Computational practicalities}

	\subsection{Computational cost of SM and SSM}\label{app:SM_computational_complexity}
	
	For SM, as discussed in Section~\ref{sec:density_est}, exploiting automatic differentiation libraries to compute the second derivatives of the log density requires $ d $ times more backward derivative computations with respect to the first derivatives. In fact, automatic differentiation libraries are able to compute derivatives of a scalar with respect to several variables at once. One single call is therefore sufficient to obtain $ \nabla_{x} \log p_w(x) $. However, $ d $ additional calls are required to obtain the second derivatives $ \frac{\partial^2}{\partial x_i^2} \log p_w(x),\ i=1, \ldots, d $, which are the diagonal elements of the Hessian matrix of $ \log p_w(x) $; each additional call computes the gradient of $ \frac{\partial}{\partial x_k} \log p_w(x) $ with respect to all components of $ x $, for some $ k\in [1,2,\ldots,d] $. Algorithm 2 in \cite{song2019sliced} gives a pseudocode implementation of this approach. Computational improvement can be obtained by implementing custom code which performs the gradient computation in the forward pass (i.e., along the computation of the neural network output for a given input $ x $). This avoids repeating some computations multiple times, which is done when performing the backward step repeatedly; however, the implementation is tricky and needs custom code for each different neural network type (we discuss how it can be done for a fully connected neural network in Appendix~\ref{app:forward_der_comp}). Additionally, the computational speed-up achievable in this way is limited with respect to what is offered by, for instance, SSM.

	SSM instead requires only two backward propagation steps independently on the input size of the network. This is possible by exploiting the vector-Hessian product structure and computing the linear products with $ v $ (which is independent on the input $  x$, thus can be swapped with gradient computation) after the gradient has been computed once, so that you only ever require the gradient of a scalar quantity. See Algorithm 1 in \cite{song2019sliced} for a precise description of how that can be done.

	\subsubsection{Forward computation of derivatives}\label{app:forward_der_comp}

	In the standard score matching approach, the first and second derivatives of Neural Network outputs with respect to the inputs are required, namely: 
	\begin{equation}\label{}
		\frac{\partial f_w(x)}{\partial x_i} \qquad \text{ and } \qquad \frac{\partial^2 f_w(x)}{\partial x_i^2}. 
	\end{equation}
	Neural network training is possible thanks to the use of autodifferentiation libraries, which allow to keep trace of the different operations and to automatically compute the gradients required for training. These libraries can be used to obtain the above derivatives. 
	
	However, as discussed previously, it is more efficient to compute the required derivatives during the forward pass of training (i.e. when the output of the Neural Network for a given input is computed). This requires additional coding effort specific to the chosen Neural Network architecture. For instance, \cite{avrutskiy2017backpropagation} provide formulas to compute derivatives of any order recursively for fully connected Neural Networks. For more complex NN architectures, this approach is not viable as the coding effort becomes substantial. Additionally, with large $ d $ the improvement obtained by forward derivatives computation is much smaller than what offered, for instance, by SSM.
	
	In the current work, the forward derivatives approach has been implemented for fully connected Neural Networks and Partially Exchangeable Networks (Appendix~\ref{app:PENs}). The computational advantage is evident for the computation of the second derivatives, as shown below. This approach allowed us to apply SM to relatively high dimensional data spaces (up to 100 dimensional for the MA(2) and AR(2) case), but not to the larger Lorenz96 model.

	\paragraph{Forward computation of derivatives of fully connected NNs}
	
	We revisit here the work in \cite{avrutskiy2017backpropagation}. Let us consider a fully connected Neural Networks with $ L $ layers, where the weights and biases of the $ l $-th layer are denoted by $ W_l $ and $ b_l $, $ l=1,\ldots,L $. Let us denote by $ x $ the input of the Neural Network, and by $ z_l $ the hidden values after the $ l $-th layer, before the activation function (denoted by $ \sigma $) is applied. Specifically, the Neural Network hidden values are determined by: 
	\begin{equation}\label{}
		z_1 = W_{1} \cdot x + b_1, \quad z_l = W_l \cdot \sigma(z_{l-1}) + b_l, \quad l=2,\ldots,L,
	\end{equation}
	where the activation function is applied element wise. Similar recursive expressions can be given for the first derivatives:
	
	\begin{equation}
		\frac{\partial z_{1}}{\partial x_i} = (W_1)_{i,\cdot}, \quad \frac{\partial z_l}{\partial x_i}  = W_l \cdot \left[\sigma'(z_{l-1}) \odot \frac{\partial z_{l-1}}{\partial x_i}\right], \quad l=2,\ldots,L,
	\end{equation}
	where $ (W_1)_{i,\cdot} $ denotes the $ i $-th column of $ W_1 $, and $ \odot $ denotes element wise multiplication. Expressions for the second derivatives are instead: 
	\begin{equation}
		\frac{\partial^2 z_{1}}{\partial x_i^2} = \mathbf{0}, \quad \frac{\partial^2 z_l}{\partial x_i^2}  = W_l \cdot \left[\sigma''(z_{l-1}) \odot \left(\frac{\partial z_{l-1}}{\partial x_i}\right)^2 + \sigma'(z_{l-1}) \odot \frac{\partial^2 z_{l-1}}{\partial x_i^2} \right], \quad l=2,\ldots,L,
	\end{equation}
	where here $ \mathbf{0} $ denotes a 0 vector with the same size as $ z_1 $. 
	
	When instead we are interested in cross terms of the form $ \frac{\partial z_l}{\partial x_i \partial x_j} $, we can apply the following: 
	\begin{equation}
		\frac{\partial z_1}{\partial x_i \partial x_j} = \mathbf{0}, \ \frac{\partial z_l}{\partial x_i \partial x_j}   = W_l \cdot \left[\sigma''(z_{l-1}) \odot \frac{\partial z_{l-1}}{\partial x_i}\odot \frac{\partial z_{l-1}}{\partial x_j} + \sigma'(z_{l-1}) \odot \frac{\partial^2 z_{l-1}}{\partial x_i \partial x_j} \right], \ l=2,\ldots,L.
	\end{equation}
	
	With respect to naively using auto-differentiation libraries, computing the derivatives in the forward step is much cheaper; specifically, for fully connected NNs, we found empirically the first to scale quadratically with the output size, while the second scales linearly (Figure~\ref{fig:NN_derivatives_time}).
	
	\begin{figure}[!tb]
		\centering
		\begin{subfigure}{0.49\textwidth}
			\centering
			\includegraphics[width=1\linewidth]{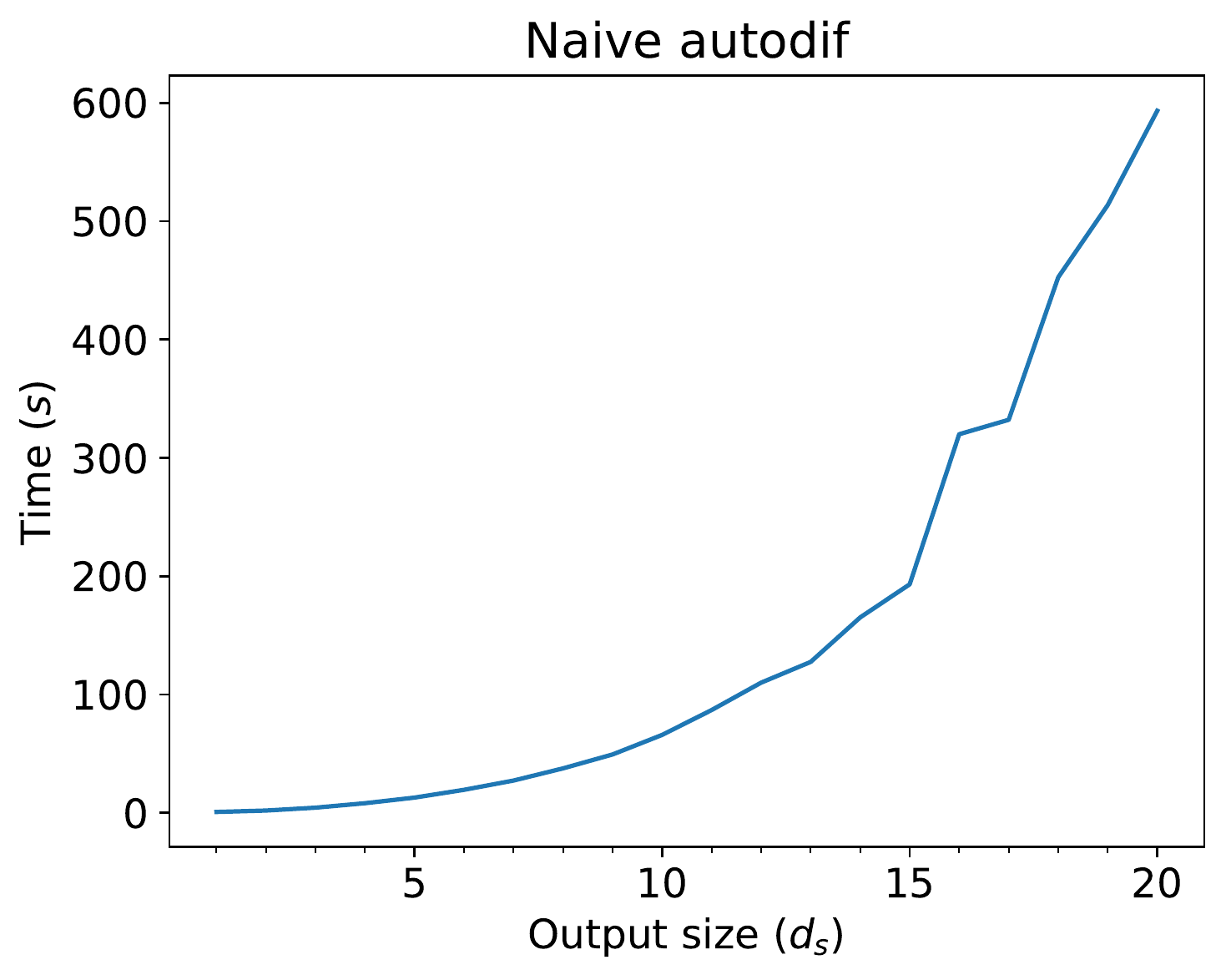}
			
		\end{subfigure}~
		\begin{subfigure}{0.49\textwidth}
			\centering
			\includegraphics[width=1\linewidth]{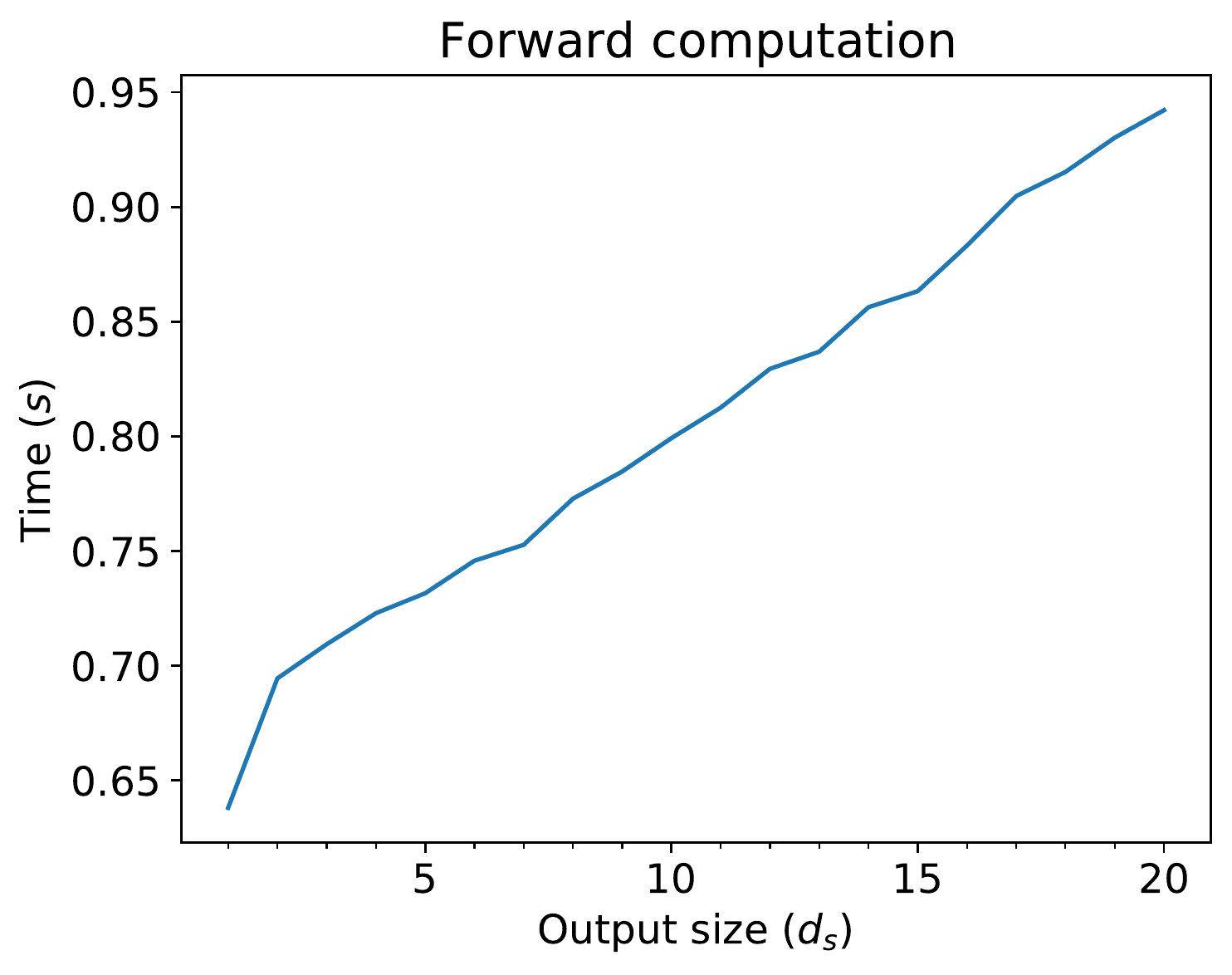}
			
		\end{subfigure}
		\caption{\textbf{Computational complexity for second order derivatives of Neural Network outputs} versus size of the output; we compare the forward computation of derivatives with naive autodifferentiation. Here, input size is fixed to 100, and one single batch of 5000 samples is fed to the network. Computations are done on a CPU machine with 8 cores.}
		\label{fig:NN_derivatives_time}
	\end{figure}

	\subsection{Batch normalization}\label{app:batch_norm}
	
	The exponential family form used as $ p_w $ depends on $ \eta_w(\theta)^T f_w(x) $; if you multiply $ f_w(x) $ with an invertible matrix $ A $ and multiply $ \eta $ with $ (A^T)^{-1} $, the product does not change. In order to remove this additional degeneracy, we use a Batch Normalization (BatchNorm) layer \citep{ioffe2015batch} to normalize the outputs of $ \eta_w(\theta) $. Essentially, BatchNorm rescales the different features to have always the same range across different batches. More specifically, BatchNorm performs the following operation on $ y $:
	\begin{equation}\label{}
		\tilde y = \frac{y -\E[Y]}{\sqrt{\Var[Y] + \epsilon}} * \gamma + b,
	\end{equation}
	where $ \epsilon $ is a small constant used for numerical stability, and $ \gamma $ and $ b $ are two (optionally learnable) sets of constants with dimension equal to $ y $ (set to 1 and 0 respectively by default). During training, the expectation $ \E $ and variance $ \Var $ are estimated over the batch of training samples fed to the Neural Network. In testing mode, BatchNorm rescales the features by using a fixed estimate of $ \E $ and $ \Var $; this estimate is obtained as a running mean over the batches; let $ s(Y) $ represent either the population expectation or variance. When the $ t $-th batch is fed through the network in training mode, the running mean estimate is updated as follows: 
	$$ \hat s_{new}(Y) = (1- p) \cdot \hat s_{old}(Y) + p \cdot s_t(Y), $$
	where $ s_t $ represent the estimate on the current batch, $ \hat s_{old} $ and $ \hat s_{new} $ respectively the old and updated running mean and $ p $ is a \textit{momentum} constant which determines how quickly the running mean changes (the smaller it is, the slower the change of the estimate). 
	
	For instance, if the training of a Neural Network is quite unstable, the running estimate may not be a correct estimate of $ \E $ and $ \Var $, so that the test loss across training epochs may be very spiky, until network training stabilized. To solve this issue, you can either increase $ p $ (so that the running estimate ``forgets'' past information faster) or, alternatively, do a forward pass of the training data set (without computing gradients) before evaluating the test loss, so that the running estimate is more precisely estimated. 
	
	Across this work, we apply BatchNorm to $ y = \eta_w(\theta) $. Moreover, we do not learn the translation parameters $ \gamma, b $, and rather we fix them to be a vector of $ 1$s and $ 0 $s.

	\subsection{Partially Exchangeable Networks}\label{app:PENs}
	
	Partially Exchangeable Networks (PENs) were introduced in \cite{wiqvist2019partially} as a Neural Network architecture that satisfies the probabilistic invariance of Markovian models.
	
	Specifically, let consider the case in which $ x = (x_1, \ldots, x_d)$ comes from a Markovian model of order $ r $, i.e.:
	\begin{equation}\label{}
		\begin{aligned}
			p(x|\theta) &= p(x_1|\theta) p(x_2|x_1;\theta) p(x_3|x_2, x_1;\theta)\prod_{i=4}^d p(x_i|x_1, \ldots, x_{i-1};\theta) \\&= p(x_1|\theta) \prod_{i=2}^d p(x_i|x_{i-r}, \ldots, x_{i-1};\theta).
		\end{aligned}
	\end{equation}
	
	This definition is an extension of the standard Markovianity assumption (of order 1), which corresponds to $ p(x|\theta) = p(x_1|\theta)\prod_{i=2}^{d} p (x_i|x_{i-1};\theta) $, and it means that each element of $ x $ only depends on the last $ r $ elements. When $ r=0 $, this corresponds to i.i.d. assumption.
	
	When a model is $ r $-Markovian, the probability density of an observation $ x $ is invariant to \textit{r}-\textbf{block-switch transformation}, which is defined as follows: 
	
	\begin{definition}\textbf{\textit{r}-\textbf{block-switch transformation} \citep{wiqvist2019partially}}
		Let $ \data_{i:j} $ and $ \data_{k:l} $ be two non-overlapping blocks with $\data_{i:(i+r)} = \data_{k:(k+r)}$ and $\data_{(j-r):j} = \data_{(l-r):l}$. Then, denoting $b = (i,j,k,l) $, with $j-i \ge r$ and $l-k \ge r$: 
		
		\begin{equation*}
			\begin{aligned}
				\data &=   \data_{1:i-1}\ \colorbox{green!50}{$\data_{i:j}$}\tikzmark{a}{} \ \data_{(j+1):(k-1)}\ \tikzmark{b}{\colorbox{red!50}{$ \data_{k:l} $}} \ \data_{(l+1):M}\\	
				T_b^{(r)}(\data) &=   \data_{1:i-1}\ \colorbox{red!50}{$ \data_{k:l} $}\tikzmark{b2} \ \data_{(j+1):(k-1)}\ \tikzmark{a2}{\colorbox{green!50}{$ \data_{i:j} $}}\ \data_{(l+1):M}.
			\end{aligned}
			\begin{tikzpicture}[overlay,remember picture,out=315,in=225,distance=0.4cm]
			\draw[<->,red,shorten >=3pt,shorten <=3pt, ultra thick, opacity=0.5] (b.center) -- (b2.center);
			\draw[<->,green,shorten >=3pt,shorten <=3pt, ultra thick, opacity=0.5] (a.center) -- (a2.center);
			\end{tikzpicture}	
		\end{equation*}
		
		Otherwise, if $\data_{i:(i+r)}\neq \data_{k:(k+d)}$ or $\data_{(j-r):j} \neq \data_{(l-r):l}$, then $T_b^{(r)}(\data) = \data$.
	\end{definition}
	
	For instance, let us consider the case in which the data space is $ \mathcal{X} = \{0,1\ldots 9\}^{16}$, and $ r=2 $. An example of the above transformation is the following:
	
	\begin{align*}
		\data &= 1\ \boxed{\textcolor{red}{7}\ \textcolor{red}{2}\ 3\ 6\ 4\ \textcolor{green}{5}\ \textcolor{green}{8}}\tikzmark{a}{}\ 1\ 7\ \tikzmark{b}\boxed{\textcolor{red}{7}\ \textcolor{red}{2}\ 9\ \textcolor{green}{5}\ \textcolor{green}{8}}\ 1   \\		
		T_{(2,8,11,15)}^{(2)}(\data) &= 1\ \boxed{\textcolor{red}{7}\ \textcolor{red}{2}\ 9\ \textcolor{green}{5}\ \textcolor{green}{8}}\tikzmark{b2}{}\ 1\ 7\ \tikzmark{a2} \boxed{\textcolor{red}{7}\ \textcolor{red}{2}\ 3\ 6\ 4\ \textcolor{green}{5}\ \textcolor{green}{8}}\ 1.
		\begin{tikzpicture}[overlay,remember picture,out=315,in=225,distance=0.4cm]
		\draw[<->,black,shorten >=1pt,shorten <=1pt, ultra thick, opacity=0.75] (b.center) -- (b2.center);
		\draw[<->,black,shorten >=1pt,shorten <=1pt, ultra thick, opacity=0.75] (a.center) -- (a2.center);
		\end{tikzpicture}	
	\end{align*}

	The authors of \cite{wiqvist2019partially} provide a simple Neural Network structure which is invariant to the $ r $-block-switch transformation, motivated by the following theorem:
	\begin{theorem}[\citep{wiqvist2019partially}]
		Let $f: \mathcal{X}^M \to A$ \textit{r}-block-switch invariant. If $\mathcal{X}$ is countable, $\exists\ \phi: \mathcal{X}^{r+1} \to \R$ and $\rho: \mathcal{X}^r \times \R \to A$ such that: 
		\begin{equation}\label{Eq:PEN}
			\forall \data \in \mathcal X^M,\ f(\data) = \rho \left( \data_{1:r},  \sum_{i=1}^{M-r} \phi(\data_{i:(i+r)} ) \right).
		\end{equation}
		
	\end{theorem}
	
	In practice, $\phi$ and $\rho$ are two independent Neural Networks (which we take to be fully connected in our case), giving rise to a PEN of order $ r $. 
	
	In \cite{wiqvist2019partially}, the authors show that the posterior mean of a Markovian variable of order $ r $ needs to be invariant to the $ r $-block-switch transformation. Therefore, this motivates using a PEN for learning a summary statistics as in the approach by \cite{fearnhead_constructing_2012}. Here, we use PEN for parametrizing the statistics in the approximating exponential family as well; this choice implies that the approximating family has the same Markovianity property as the true distribution, as we discuss in the following.
	
	\subsubsection{Results for the exponential family}		
	A deeper connection between \textit{r}-Markovian probability models and \textit{r}-block-switch transformation exists. We can in fact state the following result:
	
	\begin{lemma}
		A probability model $ p(\data|\theta) $ is \textit{r}-Markovian $ \iff $ the function $ \data \mapsto p(\data|\theta) $ is \textit{r}-block-switch invariant, i.e. $ p(\data| \parameter) = p(T_b^{(r)} (\data)| \parameter)  \ \forall\ T_b^{(r)} $.
	\end{lemma}
	\begin{proof}
		The forward direction is straightforward and can be seen by considering the decomposition of an \textit{r}-Markovian model.
		
		The converse direction can be shown by contradiction; assume in fact that $ \data \mapsto p(\data|\theta) $ is \textit{r}-block-switch invariant but not Markovian. As it is not Markovian, $ \exists \data = (\data_i, \data_2, \ldots, \data_n) $ for which $\data_{i:(i+r)} = \data_{k:(k+r)}$ and $\data_{(j-r):j} = \data_{(l-r):l}$ such that, defining $ b=(i,j,k,l) $, $ p(\data| \parameter) \neq p(T_b^{(r)} (\data)| \parameter) $. This is however in contradiction with \textit{r}-block-switch invariance, which leads to our result.
	\end{proof}
	
	In the case where the model we consider has a sufficient statistic, we can write $ p(\data|\theta) = h(\data) g(t(\data)|\theta) $. We get therefore the following corollary, which can be seen by applying the above result: 
	
	\begin{corollary}
		Consider a distribution $ p(\data|\theta) = h(\data) g(t(\data)|\theta) $; if the function $ h(\data) $ and $ t(\data) $ are \textit{r}-block-switch invariant, then $ p(\data|\theta) $ is \textit{r}-Markovian.
	\end{corollary}
	
	Without any further assumptions, this result is not enough to say that the sufficient statistics $ t(\data) $ for a Markovian model is \textit{r}-block-switch invariant; in fact, the choice $ t(x) = x $ always constitutes a sufficient statistic; moreover, in the decomponsition $ p(\data|\theta) = h(\data) g(t(\data)|\theta) $, it may be that the function $ t(x)  $ is not \textit{r}-block-switch invariant but $ g(t(x) |\theta) $ is, or otherwise that the product $ h(x) g(t(x) |\theta) $ is \textit{r}-block-switch invariant even if the individual terms are not. We can however get the following result:
	\begin{lemma}\label{lemma:r-d-b-s_h(x)}
		Consider a \textit{r}-Markovian distribution $ p(\data|\theta) = h(\data) g(t(\data)|\theta) $; if $ \data \mapsto t(\data) $ is not an injection mapping, then $x \mapsto  h(x) $ is \textit{r}-block-switch invariant.
	\end{lemma}
	\begin{proof}
		If $ \data \mapsto t(\data) $ is not an injection, $ \exists x, x' $ such that $ t(x) = t(x')$. If the density is not degenerate, moreover, $ \exists \theta: p(\data|\theta), p(\data'|\theta)>0 $. Therefore, we consider the following ratio: 
		
		\begin{equation}\label{}
			\frac{p(\data|\theta)}{p(\data'|\theta)} = \frac{h(\data)}{h(\data')} \cdot \frac{g(t(\data)|\theta)}{g(t(\data')|\theta)} = \frac{h(\data)}{h(\data')}.
		\end{equation}
		Now, the left hand side is \textit{r}-block-switch invariant with respect to both $ \data $ and $ \data' $ independently, implying that $ h(x) $ is as well.
	\end{proof}
	
	We remark that it does not seem possible in general to say anything about $ t(x) $, as in fact it may be that the function $ t(x)  $ is not \textit{r}-block-switch invariant but $ g(t(x) |\theta) $ is. In the specific case of an exponential family, however, a more specific result can be obtained:
	\begin{lemma}
		Consider an exponential family distribution $ p(\data|\theta) = h(\data) \exp(\eta(\theta)^T  f(\data))/Z(\theta) $  which is \textit{r}-Markovian distribution; if $ \data \mapsto f(\data) $ is not an injection mapping, then $x \mapsto  f(x) $ is \textit{r}-block-switch invariant.
	\end{lemma}
	\begin{proof}
		Without loss of generality, we consider the case in which all elements of $ \eta(\theta) $ are not constant with respect to $ \theta $; if this is not the case, in fact, you can redefine the exponential family by incorporating the elements of $ f(x)  $ corresponding to the constant ones of $ \eta$ in the $ h(x) $ factor. 
		
		Now, $ h(x) $ is \textit{r}-block-switch invariant thanks to Lemma~\ref{lemma:r-d-b-s_h(x)}. Consider now the following decomposition: 
		
		\begin{equation}\label{}
			\log p(x|\theta) = \log h(x) - \log Z(\theta) + \sum_i f_i(x) \eta_i(\theta);
		\end{equation}
		as that needs to be \textit{r}-block-switch invariant for any $ \theta $, this can happen only if each of the $ f_i(x) $ elements are \textit{r}-block-switch invariant.
	\end{proof}
	
	Overall, these results imply that an exponential family in which $ f $ is parametrized with a PEN of order \textit{r} is \textit{r}-Markovian. Moreover, provided that $ f $ is not an injection mapping, all \textit{r}-Markovian exponential families have $ f $ which satisfy the \textit{r}-block-switch invariant property, which is imposed by using a PEN network of order \textit{r}.

	\subsubsection{Forward computation of derivatives for PENs}
	
	We give here the derivation for the forward computation of derivatives with PENs. If we pick here $ \phi $ and $ \rho $ to be fully connected Neural Networks, we can moreover apply the forward computation of derivatives for them and we are able to compute the derivatives for PENs at a much lower cost with respect to using automatic differentiation libraries. 
	
	In Eq.~\eqref{Eq:PEN}, let us denote for brevity $ z = \sum_{i=1}^{M-r} \phi(\data_{i:(i+r)}) $. We are interested now in computing the derivative:
	\begin{equation}\label{}
		\frac{\partial f}{\partial \data_j} = \frac{\partial }{\partial \data_j} \rho(\data_{1:r}, z),
	\end{equation}
	where note that $ z $ depends in general on $ \data_i $. Therefore, in computing the above, we need to compute the derivative with respect to both arguments; let us denote by $ \frac{\partial'}{\partial \data_j} $ the derivative with respect to the first argument. Then, we have: 
	\begin{equation}\label{}
		\frac{\partial f}{\partial \data_j} = \frac{\partial' }{\partial \data_j}\rho(\data_{1:r}, z) \cdot \mathbbm{1}[ j \le r]  + \frac{\partial }{\partial z}\rho(\data_{1:r}, z) \cdot \frac{\partial z}{\partial \data_j},
	\end{equation}
	where the second term is (note that $ z $ and $ \rho  $ are multivariate, so that $ \frac{\partial' }{\partial \data_j}\rho(\data_{1:r}, z)$ is a Jacobian matrix): 
	\begin{equation}\label{}
		\frac{\partial z}{\partial \data_j} = \frac{\partial}{\partial \data_j} \sum_{i=1}^{M-r} \phi(\data_{i:(i+r)}) = \sum_{\substack{i=j-r\\i\ge 1}}^{j} \frac{\partial}{\partial \data_j} \phi(\data_{i:(i+r)}),
	\end{equation}
	where all other terms of the sum disappear as they do not contain $ \data_j $.
	
	Now, we are interested in obtaining the second derivative terms: 
	\begin{equation}\label{}
		\begin{aligned}
			\frac{\partial^2 f}{\partial \data_j^2} &= \frac{\partial }{\partial \data_j} \frac{\partial f}{\partial \data_j} = \frac{\partial }{\partial \data_j} \left[ \frac{\partial' }{\partial \data_j}\rho(\data_{1:r}, z) \cdot \mathbbm{1}[ j \le r]  + \sum_k \frac{\partial }{\partial z_k}\rho(\data_{1:r}, z) \frac{\partial z_k}{\partial \data_j} \right] \\ 
			&=\frac{\partial'^2 }{\partial \data_j^2}\rho(\data_{1:r}, z) \cdot \mathbbm{1}[ j \le r]  + 2 \sum_k \textcolor{red}{\frac{\partial' }{\partial \data_j} \frac{\partial}{\partial z_k} \rho(\data_{1:r}, z)}\cdot \frac{\partial z_k}{\partial \data_j} \cdot \mathbbm{1}[ j \le r] \\ &+  \sum_{k,k'} \textcolor{red}{
				\frac{\partial^2 }{\partial z_k \partial z_{k'}}\rho(\data_{1:r}, z)}  \frac{\partial z_k}{\partial \data_j} \frac{\partial z_{k'}}{\partial \data_j}  + \sum_k  \frac{\partial}{\partial z_k}\rho(\data_{1:r}, z) \frac{\partial^2 z_k}{\partial \data_j^2};
		\end{aligned}
	\end{equation}
	in the above expression, $ \sum_k $ runs over the elements of $ z $ and $ \frac{\partial'^2}{\partial \data_j^2} $ denotes second derivative with respect to the first element. Note that all terms appearing in the above formulas contain first and second derivatives of the Neural Networks $ \phi $ and $ \rho $ with respect to one single input, except for the terms highlighted in red. In order to compute that, obtaining the full hessian matrix of $ \rho $ is required. We remark that the latter can be very large in case the input dimension is large, therefore leading to memory overflow issues.

	\subsection{Exchange MCMC}\label{app:exchange}
	
	For convenience, we describe here the ExchangeMCMC algorithm by \cite{murray2012mcmc}. We consider the task of sampling from a posterior distribution $\pi(\theta|x)$. We can evaluate an unnormalized version of the likelihood $ \tilde p(x|\theta) $, and we denote the normalized version as $ p(x|\theta) = \tilde p(x|\theta) / Z(\theta) $, $ Z(\theta) $ being an intractable normalizing constant. We want to build an MCMC chain by using a proposal distribution $ q(\cdot|\theta; x) $ (which optionally depends on the considered $ x $ as well). Usually, the standard Metropolis acceptance threshold for a proposal $ \theta' $ is defined as: 
	
	\begin{equation}\label{}
		\alpha = \frac{\pi(\theta'|x)q(\theta|\theta';x)}{\pi(\theta|x)q(\theta'|\theta;x)} = \frac{\tilde p(x|\theta')q(\theta|\theta';x) \pi(\theta')}{\tilde p(x|\theta)q(\theta'|\theta;x)\pi(\theta)} \cdot \frac{Z(\theta)}{Z(\theta')},
	\end{equation}
	where the last factor cannot be evaluated, as we do not have access to the normalizing constant.
	
	The ExchangeMCMC algorithm proposed by \cite{murray2012mcmc} bypasses this issue by drawing an auxiliary observation $ x' \sim p(\cdot|\theta') $ and defining the acceptance probability as: 
	\begin{equation}\label{Eq:acc_rate_exchange}
		\alpha = \frac{p(x|\theta')q(\theta|\theta';x) \pi(\theta')}{p(x|\theta)q(\theta'|\theta;x)\pi(\theta)} \cdot\textcolor{blue}{\frac{p(x'|\theta)}{p(x'|\theta')}} = \frac{\tilde p(x|\theta')\textcolor{blue}{\tilde p(x'|\theta)}q(\theta|\theta';x) \pi(\theta')}{\tilde p(x|\theta)\textcolor{blue}{\tilde p(x'|\theta')}q(\theta'|\theta;x)\pi(\theta)} \cdot \frac{\cancel{Z(\theta)}}{\cancel{Z(\theta')}}\cdot \textcolor{blue}{\frac{\cancel{Z(\theta')}}{\cancel{Z(\theta)}}}.
	\end{equation}
	
	Here, all the normalizing constants cancel out, so that the acceptance threshold can be evaluated explicitly, at the expense of drawing a simulation from the likelihood for each MCMC step. \cite{murray2012mcmc} showed that an MCMC chain using the above acceptance rate targets the correct posterior $ \pi(\theta|x) $. The resulting algorithm is given in Algorithm~\ref{alg:ExchangeMCMC}:

	\begin{algorithm} 
		\caption{Original exchangeMCMC algorithm \citep{murray2012mcmc}.}
		\label{alg:ExchangeMCMC}
		\begin{algorithmic}[1]
			\Require Initial $ \theta $, number of iterations $ T $, proposal distribution $ q $, observation $ x $.
			\For{$i=1$ \textbf{to} $T$}
			\State Propose $ \theta' \sim q(\theta'|\theta;x)$
			\State Generate auxiliary observation $x'\sim p(\cdot|\theta')$ 
			\State Compute acceptance threshold $ \alpha $ as in Eq.~\eqref{Eq:acc_rate_exchange}
			\State With probability $ \alpha $, set $ \theta \leftarrow  \theta' $
			\EndFor
		\end{algorithmic}
	\end{algorithm}

	\paragraph{Bridging.} When considering more closely the acceptance rate in Eq.~\eqref{Eq:acc_rate_exchange}, it can be seen that it depends on two ratios: $ \frac{p(x|\theta')}{p(x|\theta)} $ represents how well the proposed parameter value explains the observation with respect to the previous parameter value, while instead $ \frac{p(x'|\theta)}{p(x'|\theta')} $ measures how well the auxiliary variable (generated using $ \theta' $) can be explained with parameter $ \theta $. Therefore, even if the former is large and $ \theta  $ would be a suitable parameter value, $ \alpha $ can still be small if the auxiliary random variable is not explained well by the previous parameter value. This can lead to slow mixing of the chain; to improve on this, \cite{murray2012mcmc} proposed to sample a set of auxiliary variables $ (x'_0, x'_1, \ldots, x'_K) $ from intermediate distributions in the following way\footnote{Differently from the rest of the work, here subscripts do not denote vector components, but rather different auxiliary variables.}: consider a set of densities 
	\begin{equation}\label{}
		\tilde p_k(x	|\theta, \theta') = \tilde p(x|\theta')^{\beta_k} \tilde p(x|\theta)^{1- \beta_k}, \quad \beta_k = \frac{K - k + 1}{K + 1};
	\end{equation}
	$ x'_0 $ is generated from $ p(\cdot|\theta') $ as before, and then each $ x'_k $ is generated from $ R(\cdot | x'_{k-1}; \theta, \theta') $, which denotes a Metropolis-Hastings transition kernel starting from $ x'_{k-1} $ with stationary density $ \tilde p_k(\cdot|\theta, \theta') $. Then, the acceptance rate is modified as follows: 
	\begin{equation}\label{Eq:acc_rate_exchange_bridging}
		\alpha = \frac{\tilde p(x|\theta')q(\theta|\theta';x) \pi(\theta')}{\tilde p(x|\theta)q(\theta'|\theta;x)\pi(\theta)} \cdot \prod_{k=0}^{K} \frac{\tilde p_{k+1}(x'_k|\theta, \theta') }{\tilde p_k(x'_k|\theta, \theta') }.
	\end{equation}
	
	The overall algorithm is given in Algorithm~\ref{alg:ExchangeMCMC_bridging}. Note that $ K=0 $ recovers the original ExchangeMCMC. This procedure generally improves the acceptance rate as it basically introduces a sequence of intermediate updates to the auxiliary data which by smoothening out the difference between the two distributions. 
	
	\begin{algorithm} 
		\caption{ExchangeMCMC algorithm with bridging \citep{murray2012mcmc}.}
		\label{alg:ExchangeMCMC_bridging}
		\begin{algorithmic}[1]
			\Require Initial $ \theta $, number of iterations $ T $, proposal distribution $ q $, number of bridging steps K, observation $ x $.
			\For{$i=1$ \textbf{to} $T$} 
			\State Propose $ \theta' \sim q(\theta'|\theta;x)$
			\State Generate auxiliary observation $x'_0\sim p(\cdot|\theta')$ 
			\For{$k=1$ \textbf{to} $K$} \Comment{Bridging steps}
			\State Generate $ x'_k \sim R(\cdot|x'_{k-1};\theta, \theta') $
			\EndFor	
			\State Compute acceptance threshold $ \alpha $ as in Eq.~\eqref{Eq:acc_rate_exchange_bridging}
			\State With probability $ \alpha $, set $ \theta \leftarrow \theta' $
			\EndFor
		\end{algorithmic}
	\end{algorithm}

	\paragraph{ExchangeMCMC without perfect simulations.}
	
	If, as in the setup considered across this work, we are not able to sample from $ p(\cdot|\theta') $ as it is required in the ExchangeMCMC algorithm (line 3 in Alg.~\ref{alg:ExchangeMCMC}), \cite{murray2012mcmc} suggested to run $ T_{in} $ steps of an MCMC chain on $ x $ targeting $ p(\cdot|\theta') $ at each step of ExchangeMCMC; if $ T_{in} $ is large enough, the last sample can be considered as (approximately) drawn from $ p(\cdot|\theta') $ itself and used in place of the unavailable perfect simulation. In practice, however, this only leads to an approximate ExchangeMCMC algorithm, as at each iteration of the inner chain a finite $ T_{in} $ is used, so that the inner chain would not perfectly converge to its target; for this reason, even an infinitely long outer chain would not target the right posterior for any finite $ T_{in} $. Nonetheless, this approach was shown empirically to work satisfactorily in \cite{caimo2011bayesian, everitt2012bayesian, liang2010double}. Some theoretical guarantees (albeit under strong conditions), are given in in Appendix B by \cite{everitt2012bayesian}, which bounds the total variation distance between target of approximate ExchangeMCMC with finite $ T_{in} $ and the target of the exact one, and shows that they become equal when $ T_{in} \to \infty $. 
	
	In \cite{liang2010double}, they argue that starting the inner chain from the observation value improves convergence; we adapt this approach in our implementation (Algorithm~\ref{alg:ExchangeMCMC_inner_MCMC_generate_true_model}).
	
	\begin{algorithm}
		\caption{ExchangeMCMC algorithm \citep{murray2012mcmc} with inner MCMC.}
		\label{alg:ExchangeMCMC_inner_MCMC_generate_true_model}
		\begin{algorithmic}[1]
			\Require Initial $ \theta $, number of iterations $ T $ and $ T_{in} $, proposal distributions $ q $ and $ q_x $, observation $ x $.
			
			\For{$i=1$ \textbf{to} $T$} \Comment{Outer chain}
			\State Propose $ \theta' \sim q(\theta'|\theta;x)$

			\State Set $x' = x$ \Comment{Start inner chain from the observation}
			
			\For{$j=1$ \textbf{to} $T_{in}$}  \Comment{Inner chain}
			\State Propose $ x'' \sim q_x(x''|x')$
			\State With probability $ \frac{p_w(x''|\theta')q_x(x'|x'')}{p_w(x'|\theta')q_x(x''|x')} $, set $ x' = x'' $
			
			\EndFor
			\State Compute acceptance threshold $ \alpha $ as in Eq.~\eqref{Eq:acc_rate_exchange} \Comment{This uses last point of inner MCMC}
			\State With probability $ \alpha $, set $ \theta \leftarrow  \theta' $
			\EndFor
		\end{algorithmic}
	\end{algorithm}

	Note that it is still possible to run bridging steps after the inner MCMC to sample from $ p(\cdot |\theta') $; the algorithm combining bridging and inner MCMC, which is used across this work, is given in Algorithm~\ref{alg:ExchangeMCMC_overall}).
	
	\begin{algorithm}
		\caption{ExchangeMCMC algorithm \citep{murray2012mcmc} with inner MCMC and bridging.}
		\label{alg:ExchangeMCMC_overall}
		\begin{algorithmic}[1]
			\Require Initial $ \theta $, number of iterations $ T $ and $ T_{in} $, number of bridging steps $ K $, proposal distributions $ q $ and $ q_x $, observation $ x $.
			
			\For{$i=1$ \textbf{to} $T$} \Comment{Outer chain}
			\State Propose $ \theta' \sim q(\theta'|\theta;x)$

			\State Set $x' = x$ \Comment{Start inner chain from the observation}
			
			\For{$j=1$ \textbf{to} $T_{in}$}  \Comment{Inner chain}
			\State Propose $ x'' \sim q_x(x''|x')$
			\State With probability $ \frac{p_w(x''|\theta')q_x(x'|x'')}{p_w(x'|\theta')q_x(x''|x')} $, set $ x' = x'' $
			
			\EndFor
			\State Set $ x'_0 = x' $
			\For{$k=1$ \textbf{to} $K$} \Comment{Bridging steps}
			\State Generate $ x'_k \sim R(\cdot|x'_{k-1};\theta, \theta') $
			\EndFor	
			\State Compute acceptance threshold $ \alpha $ as in Eq.~\eqref{Eq:acc_rate_exchange_bridging}
			\State With probability $ \alpha $, set $ \theta \leftarrow  \theta' $
			\EndFor
		\end{algorithmic}
	\end{algorithm}

	\paragraph{Related algorithms.}
	Algorithms for sampling from doubly-intractable targets which are suitable for parallel computing exist, for instance \cite{caimo2011bayesian} propose a parallel-chain MCMC algorithm, while \cite{everitt2017marginal} build instead an SMC-type algorithm which is also capable of recycling information from past simulations. However, in this work we stick to using ExchangeMCMC, which turned out to be relatively cheap to use and easy to implement.
	
	Finally, we remark that \cite{liang2016adaptive} proposed an algorithm which is inspired from ExchangeMCMC and, in the case of impossible perfect sampling, still targets the right invariant distribution. It works by considering a set of parallel chains targeting $ p(\cdot |\theta^{(i)}) $ for a fixed set of $ \{\theta^{(i)}\} $ and iteratively updating those and the main chain over $ \theta $. The algorithm relies on some assumptions which are probably satisfied in practice (as discussed in \citealp{park2018bayesian}). It also requires some hand-tuning and needs to keep in memory a large amount of data, which may hinder its applicability. For this reason, we do not investigate using that here.

	\subsubsection{Implementation details}\label{app:exchange_implementation}
	
	\paragraph{Acceptance rate tuning.}
	
	In our implementation of ExchangeMCMC, we discarded some burn-in steps to make sure the chain forgets its initial state. During burn-in, moreover, after each window of 100 outer steps, we tuned the proposal sizes considering the acceptance rate in the window, and increasing (respectively decreasing) if that is too large (small) with respect to a chosen interval (see below). Notice that we applied this strategy independently to the proposal size for the outer chain, the inner chain and the bridging step, when the latter is used. When tuning the proposal size for the inner chain or bridging steps, we considered the overall acceptance rate over each window of 100 outer steps. We remark how this was only done during burn-in, so that the convergence properties of the chain were not affected. 
	
	For inner and outer steps as well as bridging, we found that a target acceptance rate in the interval $ [0.2, 0.5] $ lead to good performance; this is consistent with the recommended range for Metropolis-Hastings MCMC \citep{roberts1997weak}.
	
	\paragraph{MCMC on bounded space.}
	When the inner (or outer) MCMC chain is run on a bounded domain, we apply the transformations discussed in Appendix~\ref{app:transformations} to map it to an unbounded domain, and therefore run the MCMC on that space. Notice that the Jacobian factor arising from the transformations has therefore to be taken into account when computing the acceptance rate.

	\section{Details on Neural Networks training}\label{app:experim_details}
	
	For all experiments, we used the \texttt{Pytorch} library \citep{pytorch} to train NNs. 
	In Tables~\ref{Tab:archit_exp_fam_models}, \ref{Tab:archit_AR2}, \ref{Tab:archit_MA2}, \ref{Tab:archit_Lorenz96_small} and \ref{Tab:archit_Lorenz96_large} we report the Neural Network architectures used in the different experiments. FC($ n,m $) denotes a fully connected layer with $ n $ inputs and $ m $ outputs. For the time-series and Lorenz96 experiments, $ \phi_w $ and $ \rho_w $ represent the two Neural Networks used to build the PEN network $ f_w $ as described in Eq.~\eqref{Eq:PEN}, and similarly $ \phi_\beta $ and $ \rho_\beta $ represent the ones used in building $ s_\beta $. Finally, BN($ p $) represents a BatchNorm layer with momentum $ p $ (as described in Appendix~\ref{app:batch_norm}), with $ \gamma, b $, fixed respectively to be vectors of $ 1$s and $ 0 $s. We remark that the momentum value does not impact on the training of the network, but it modifies the evaluation of the test loss, which we use for early stopping, as discussed below.

	\begin{table}[htbp!]
		\centering
		\begin{tabular}{ llll }
			\toprule
			\textit{Network} & $ f_w $ & $ \eta_w $ & $ s_\beta $\\ \midrule
			\textit{Structure} & \parbox[t]{4cm}{FC(10,30)\\FC(30,50)\\FC(50,50)\\FC(50,20)\\FC(20,3)} & \parbox[t]{4cm}{FC(2,15)\\FC(15,30)\\FC(30,30)\\FC(30,15)\\FC(15,2)\\BN(0.9)} & \parbox[t]{4cm}{FC(10,30)\\FC(30,50)\\FC(50,50)\\FC(50,20)\\FC(20,2)} \\
			\bottomrule
		\end{tabular}
		\caption{\textbf{Architectures used for the exponential family models} (Gaussian, Gamma and Beta).}
		\label{Tab:archit_exp_fam_models}
		
	\end{table}

	\begin{table}[htbp!]
		\centering
		\begin{tabular}{ llllll }
			\toprule
			\multirow{2}{*}{\textit{Network}}  & \multicolumn{2}{c}{$ f_w $}  & \multirow{2}{*}{$ \eta_w $} &  \multicolumn{2}{c} {$ s_\beta $} \\ \cmidrule(r){2-3} \cmidrule(r){5-6}
			& $\phi_w$ & $ \rho_w $ & & $ \phi_\beta $ & $ \rho_\beta $ \\ 
			\midrule
			\textit{Structure} & \parbox[t]{2cm}{FC(3,50)\\FC(50,50)\\FC(50,50)\\FC(30,20)} & \parbox[t]{2cm}{FC(22,50)\\FC(50,50)\\FC(50,3)} & \parbox[t]{2cm}{FC(2,15)\\FC(15,30)\\FC(30,30)\\FC(30,15)\\FC(15,2)\\BN(0.9)} & \parbox[t]{2cm}{FC(3,50)\\FC(50,50)\\FC(50,30)\\FC(30,20)} & \parbox[t]{2cm}{FC(22,50)\\FC(50,50)\\FC(50,2)} \\ 		
			\bottomrule
		\end{tabular}
		\caption{\textbf{Architectures used for the AR(2) model.}}
		\label{Tab:archit_AR2}
		
	\end{table}
	
	\begin{table}[htbp!]
		\centering
		\begin{tabular}{ llllll }
			\toprule
			\multirow{2}{*}{\textit{Network}}  & \multicolumn{2}{c}{$ f_w $}  & \multirow{2}{*}{$ \eta_w $}&  \multicolumn{2}{c} {$ s_\beta $} \\ \cmidrule(r){2-3} \cmidrule(r){5-6}
			& $\phi_w$ & $ \rho_w $ & & $ \phi_\beta $ & $ \rho_\beta $ \\ 
			\midrule
			\textit{Structure} & \parbox[t]{2cm}{FC(11,50)\\FC(50,50)\\FC(50,30)\\FC(30,20)} & \parbox[t]{2cm}{FC(30,50)\\FC(50,50)\\FC(50,3)} & \parbox[t]{2cm}{FC(2,15)\\FC(15,30)\\FC(30,30)\\FC(30,15)\\FC(15,2)\\BN(0.9)} & \parbox[t]{2cm}{FC(11,50)\\FC(50,50)\\FC(50,30)\\FC(30,20)} & \parbox[t]{2cm}{FC(30,50)\\FC(50,50)\\FC(50,2)} \\ 		
			\bottomrule
		\end{tabular}
		\caption{\textbf{Architectures used for the MA(2) model.}}
		\label{Tab:archit_MA2}
	\end{table}

	\begin{table}[htbp!]
		\centering
		\begin{tabular}{ llllll }
			\toprule
			\multirow{2}{*}{\textit{Network}}  & \multicolumn{2}{c}{$ f_w $}  & \multirow{2}{*}{$ \eta_w $}&  \multicolumn{2}{c} {$ s_\beta $} \\ \cmidrule(r){2-3} \cmidrule(r){5-6}
			& $\phi_w$ & $ \rho_w $ & & $ \phi_\beta $ & $ \rho_\beta $ \\ 
			\midrule
			\textit{Structure} & \parbox[t]{2cm}{FC(16,50)\\FC(50,100)\\FC(100,50)\\FC(50,20)} & \parbox[t]{2cm}{FC(28,40)\\FC(40,90)\\FC(90,35)\\FC(35,5)} & \parbox[t]{2cm}{FC(4,30)\\FC(30,50)\\FC(50,50)\\FC(50,30)\\FC(30,4)\\BN(0.9)} & \parbox[t]{2cm}{FC(16,50)\\FC(50,100)\\FC(100,50)\\FC(50,20)} & \parbox[t]{2cm}{FC(28,40)\\FC(40,90)\\FC(90,35)\\FC(35,4)} \\ 		
			\bottomrule
		\end{tabular}
		\caption{\textbf{Architectures used for the Lorenz96 model in the small setup.}}
		\label{Tab:archit_Lorenz96_small}
	\end{table}

	\begin{table}[htbp!]
		\centering
		\begin{tabular}{ llllll }
			\toprule
			\multirow{2}{*}{\textit{Network}}  & \multicolumn{2}{c}{$ f_w $}  & \multirow{2}{*}{$ \eta_w $}&  \multicolumn{2}{c} {$ s_\beta $} \\ \cmidrule(r){2-3} \cmidrule(r){5-6}
			& $\phi_w$ & $ \rho_w $ & & $ \phi_\beta $ & $ \rho_\beta $ \\ 
			\midrule
			\textit{Structure} & \parbox[t]{2cm}{FC(80,120)\\FC(120,160)\\FC(160,120)\\FC(120,20)} & \parbox[t]{2cm}{FC(60,80)\\FC(80,100)\\FC(100,80)\\FC(80,5)} & \parbox[t]{2cm}{FC(4,30)\\FC(30,50)\\FC(50,50)\\FC(50,30)\\FC(30,4)\\BN(0.9)} & \parbox[t]{2cm}{FC(80,120)\\FC(120,160)\\FC(160,120)\\FC(120,20)} & \parbox[t]{2cm}{FC(60,80)\\FC(80,100)\\FC(100,80)\\FC(80,4)}\\ 		
			\bottomrule
		\end{tabular}
		\caption{\textbf{Architectures used for the Lorenz96 model in the large setup.}}
		\label{Tab:archit_Lorenz96_large}
	\end{table}

	In all experiments, stochastic gradient descent with a batch size of 1000 samples is used with Adam optimizer \citep{kingMA2014adam}, whose parameters are left to the default values as implemented in \texttt{Pytorch}. Finally, we evaluate the test loss on a test set with same size as the training set at intervals of $ T_\text{check} $ epochs and early-stop training if the test loss increased with respect to the last evaluation. In order to have a better running estimate of the quantities of interest for the BatchNorm layer, before each test epoch we perform a forward pass of the whole training data set, without computing gradients, as discussed in Appendix~\ref{app:batch_norm}. We do not perform early stopping before epoch $ T_{\text{start}} $. The nets are trained for a maximum of $ T $ training epochs with $ N_{train} $ training samples. In some experiments, we used an exponential learning rate scheduler, which decreases progressively the learning rate by multiplying it by a factor $ \zeta<1 $ at each epoch. We fixed $ \zeta=0.99 $. The values of the parameters are reported in Table~\ref{Tab:hyperparam_exp} for all the different models and setups, together with the values of the learning rates (lr) used and whether the scheduler was used or not (Sch).
	
	We remark that, in order to make the comparison fair, the learning rates for the FP experiments were chosen by cross validation with several learning rates choices. Similarly, we hand-picked the best learning rate values for training the Neural Networks with SM and SSM, even if in that case it was not possible, due to computational constraints, to perform a full search on a large set of values of learning rates for all experiments. In fact, the computational cost of training the exponential family approximation with SM or SSM is larger than the cost of learning the summary statistics with the FP approach. Moreover, in this scenario we have two independent learning rates values to tune, as we train two networks simultaneously. For this reason, we did not spend too much time trying different lr values for models for which we already got satisfactory results.
	
	\begin{table}[htbp!]
		\centering
		\begin{tabular}{ llllllll }
			\toprule
			\textit{Model} & \textit{Setup} & \textit{Learning rates} & $ N_{train} $ & $ T $ & $T_{\text{start}} $ & $ T_{\text{check}} $ & \textit{Sch} \\ \midrule
			\multirow{3}{*}{Gaussian} & SM & $ \lr(f_w)=0.0003,\ \lr(\eta_w)=0.003  $  & $ 10^4 $ & 500 & 150 & 10 & Yes\\
			& SSM & $ \lr(f_w)=0.001,\ \lr(\eta_w)=0.001  $ & $ 10^4 $ & 500 & 200  & 10 & Yes\\
			& FP & $ \lr(s_\beta)=0.01 $ & $ 10^4 $ & 1000 & 300 & 25 & No\\ \midrule
			\multirow{3}{*}{Gamma}  & SM & $ \lr(f_w)=0.001,\ \lr(\eta_w)=0.001  $ & $ 10^4 $ & 500 & 200  & 10 & Yes\\
			& SSM & $ \lr(f_w)=0.001,\ \lr(\eta_w)=0.001  $ & $ 10^4 $ & 500 & 200  & 10 & Yes\\
			& FP & $ \lr(s_\beta)=0.001 $ & $ 10^4 $ &1000 & 300 & 25 & Yes\\  \midrule
			\multirow{3}{*}{Beta}  & SM & $ \lr(f_w)=0.001,\ \lr(\eta_w)=0.001  $  & $ 10^4 $ & 500 & 200  & 10 & Yes\\
			& SSM & $ \lr(f_w)=0.001,\ \lr(\eta_w)=0.001  $  & $ 10^4 $ & 500 & 200  & 10 & Yes\\
			& FP & $ \lr(s_\beta)=0.01 $ & $ 10^4 $ &1000 & 250 & 50 & Yes\\  \midrule
			\multirow{3}{*}{AR(2)}  & SM & $ \lr(f_w)=0.001,\ \lr(\eta_w)=0.001  $ & $ 10^4 $& 500 & 100 & 25 & Yes \\
			& SSM & $ \lr(f_w)=0.001,\ \lr(\eta_w)=0.001  $ & $ 10^4 $& 500 & 100 & 25 & Yes \\
			& FP & $ \lr(s_\beta)=0.001 $ & $ 10^4 $ &1000 & 500 & 25 & Yes\\  \midrule
			\multirow{3}{*}{MA(2)} & SM & $ \lr(f_w)=0.001,\ \lr(\eta_w)=0.001  $ & $ 10^4 $ & 500 & 100 & 25 & Yes\\
			& SSM & $ \lr(f_w)=0.001,\ \lr(\eta_w)=0.001  $ & $ 10^4 $& 500 & 100 & 25 & Yes \\
			& FP & $ \lr(s_\beta)=0.001 $ & $ 10^4 $ &1000 & 500 & 25 & Yes\\  \midrule
			\multirow{2}{*}{\parbox[c]{1.5cm}{Lorenz96\\Small}}  & SSM & $ \lr(f_w)=0.001,\ \lr(\eta_w)=0.001  $ & $ 10^4 $ & 1000 & 500 & 50 & Yes \\
			& FP & $ \lr(s_\beta)=0.001 $ & $ 10^4 $ & 1000 & 200 & 25 & Yes\\  \midrule
			\multirow{2}{*}{\parbox[c]{1.5cm}{Lorenz96\\Large}}  & SSM & $ \lr(f_w)=0.001,\ \lr(\eta_w)=0.001  $ & $ 10^4 $ & 1000 & 500 & 50 & Yes \\
			& FP & $ \lr(s_\beta)=0.001 $ & $ 10^4 $ & 1000 & 200 & 25 & Yes\\ 
			\bottomrule
		\end{tabular}
		\caption{\textbf{Hyperparameter values for NN training}: ``FP'' denotes the least squares regression by \cite{fearnhead_constructing_2012, jiang2017learning} used for the ABC-FP experiment, while ``SM'' and ``SSM'' denote respectively exponential family trained with Score Matching and Sliced Score Matching. }
		
		\label{Tab:hyperparam_exp}
	\end{table}

	\newpage
	\FloatBarrier
	\section{Additional experimental results}
	
	\subsection{Exponential family models}\label{app:exp_fam_models}
	
	\subsubsection{Mean Correlation Coefficients for Neural Networks trained with SSM}
	
	We report in Table~\ref{Tab:MCC_toy_SSM} the weak and strong Mean Correlation Coefficient (MCC, Appendix~\ref{app:MCC}) for Neural Networks trained with SSM, for the exponential family models. MCC is a metric in $ [0,1] $, with 1 denoting perfect recovery up to a linear transformation (weak) or permutation (strong). As it can be seen in Table~\ref{Tab:MCC_toy_SSM}, our method leads to values quite close to 1, particularly for the weak MCC, implying that our method is able to recover the embeddings up to a linear transformation, as expected. 
	
	\begin{table}[htbp]
		\centering
		\resizebox{\linewidth}{!}{\begin{tabular}{lx{2.73cm}x{2.73cm}x{2.73cm}x{2.73cm}}
				\toprule
				\textit{Model} & \textit{MCC weak in} & \textit{MCC weak out} & \textit{MCC strong in} & \textit{MCC strong out} \\
				\midrule
				Beta (statistics) & 0.990 & 0.986 & 0.982 & 0.979 \\
				Beta (nat. par.) & 0.987 & 0.989 & 0.983 & 0.985 \\
				\midrule
				Gamma (statistics) & 0.939 & 0.928 & 0.723 & 0.709 \\
				Gamma (nat. par.) & 0.977 & 0.977 & 0.792 & 0.794 \\
				\midrule
				Gaussian (statistics) & 0.874 & 0.844 & 0.623 & 0.638 \\
				Gaussian (nat. par.) & 0.861 & 0.862 & 0.581 & 0.543 \\
				\bottomrule
		\end{tabular}}
		\caption{\textbf{MCC for exponential family models between exact embeddings and those learned with SSM}. We show weak and strong MCC values; MCC is between 0 and 1 and measures how well an embedding is recovered up to permutation and rescaling of its components (strong) or linear transformation (weak); the larger, the better. ``in'' denotes MCC on training data used to find the best transformation, while ``out'' denote MCC on test data. We used 500 samples in both training and test data sets.!}
		\label{Tab:MCC_toy_SSM}
	\end{table}
	
	\subsubsection{Learned and exact embeddings for the exponential family models}
	
	We compare here the exact and learned sufficient statistics and natural parameters of the exponential family models; precisely, we draw samples $ (x^{(j)}, \theta^{(j)} ) $ and then plot the learned statistics $ f_w(x^{(j)}) $ versus the exact one, and similarly for the natural parameters. Figure~\ref{fig:statistics_exp_fam_models} reports the results for Neural Networks trained with SM, while Figure~\ref{fig:statistics_exp_fam_models_SSM} reports the results for Neural Networks trained with SSM.

	\begin{figure}[!htbp]
		\centering
		\begin{subfigure}{0.32\textwidth}
			\centering
			\includegraphics[width=1\linewidth]{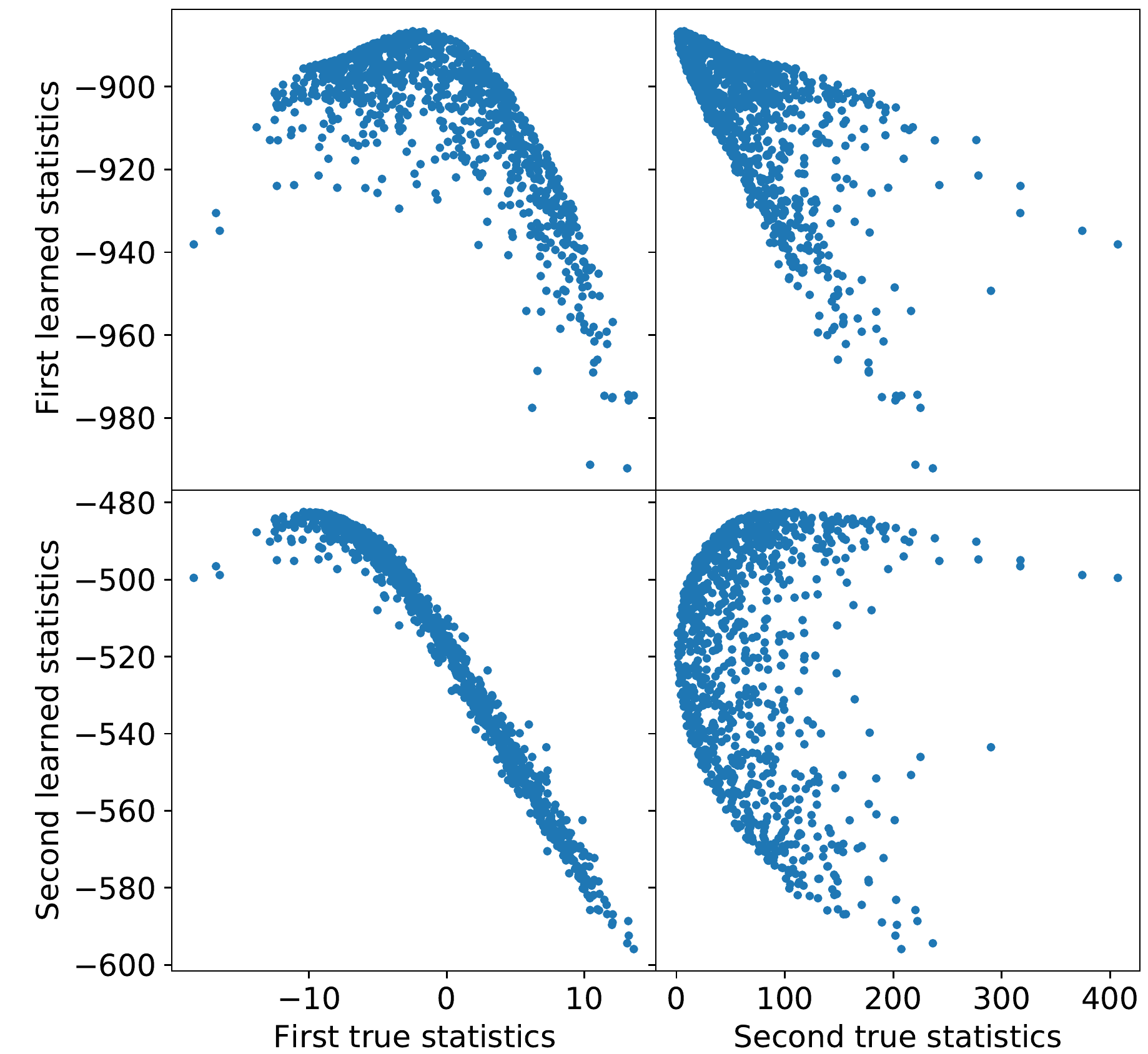}
			\caption{Statistics Gaussian}
		\end{subfigure}~
		\begin{subfigure}{0.32\textwidth}
			\centering
			\includegraphics[width=1\linewidth]{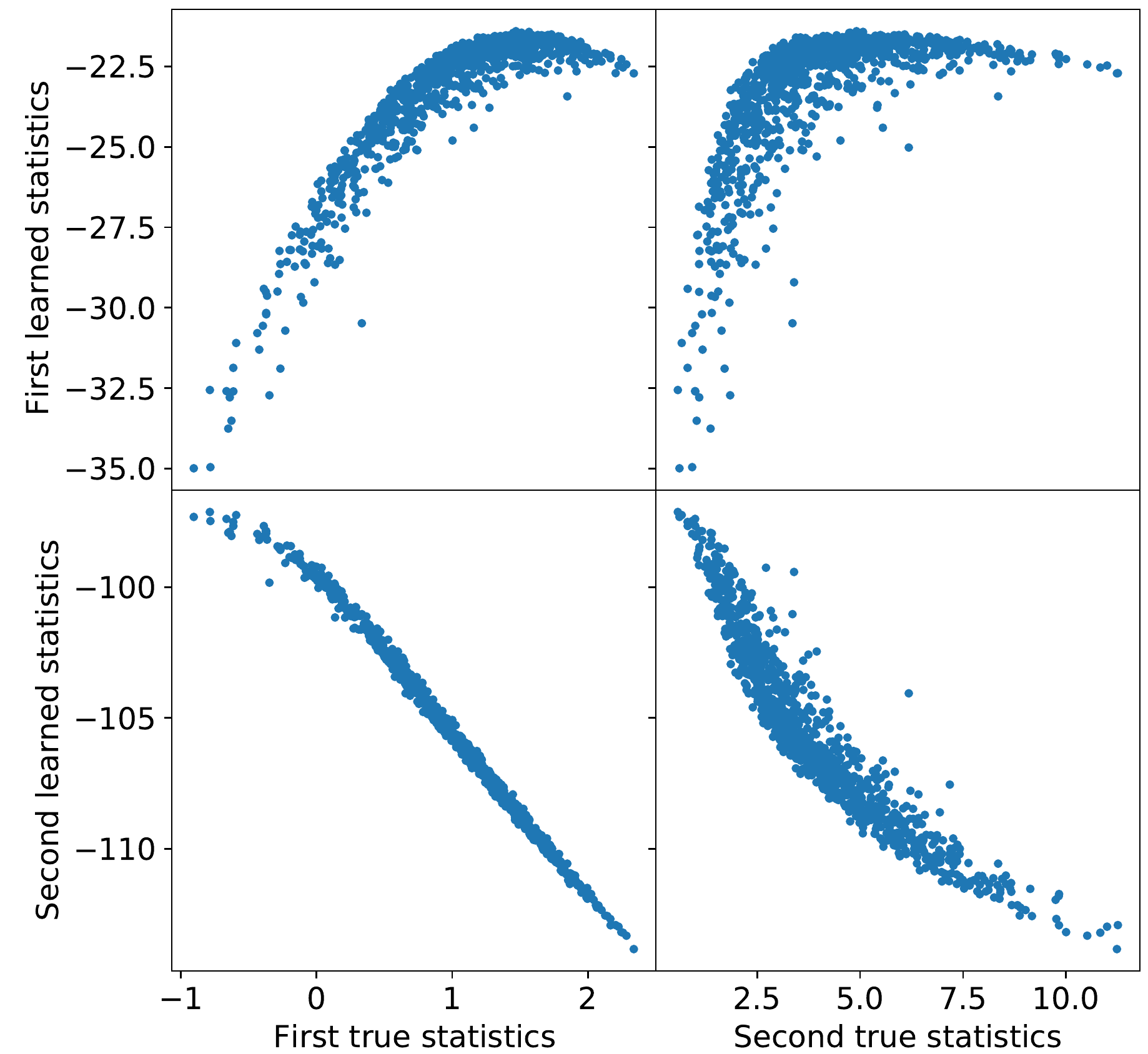}
			\caption{Statistics Gamma}
		\end{subfigure}~	
		\begin{subfigure}{0.32\textwidth}
			\centering
			\includegraphics[width=1\linewidth]{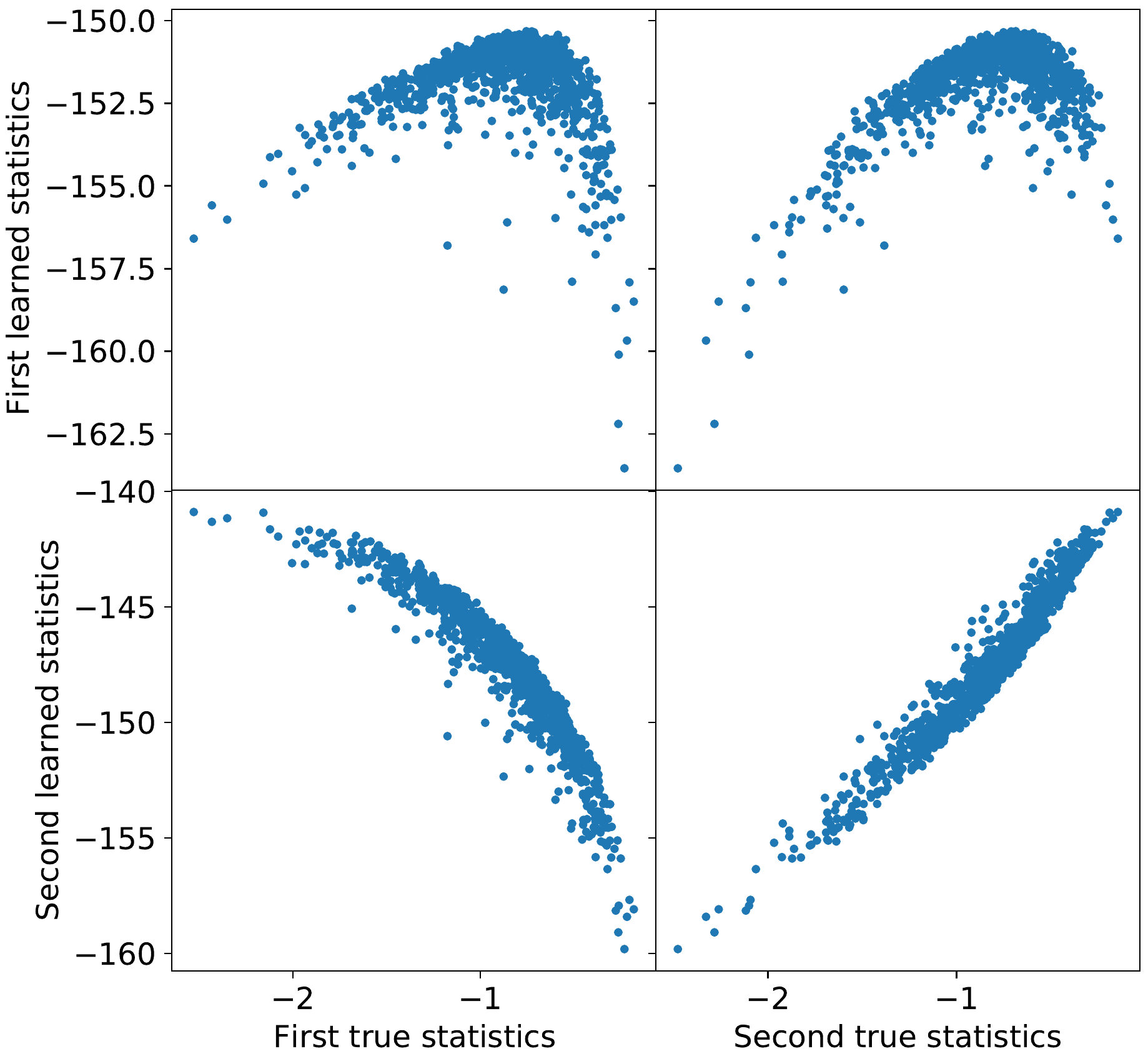}
			\caption{Statistics Beta}
		\end{subfigure}\\
		\begin{subfigure}{0.32\textwidth}
			\centering
			\includegraphics[width=1\linewidth]{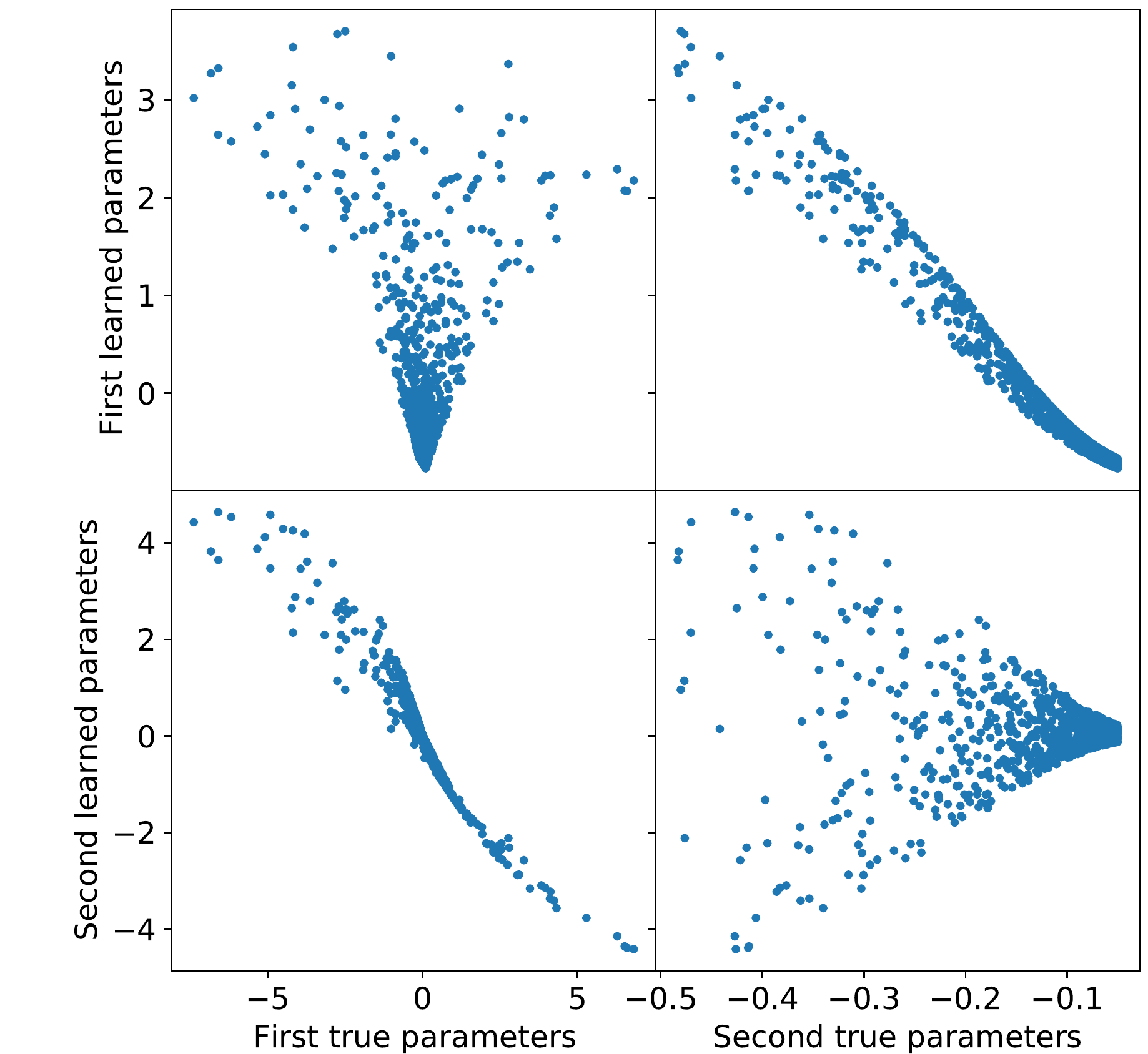}
			\caption{Parameters Gaussian}
		\end{subfigure}~
		\begin{subfigure}{0.32\textwidth}
			\centering
			\includegraphics[width=1\linewidth]{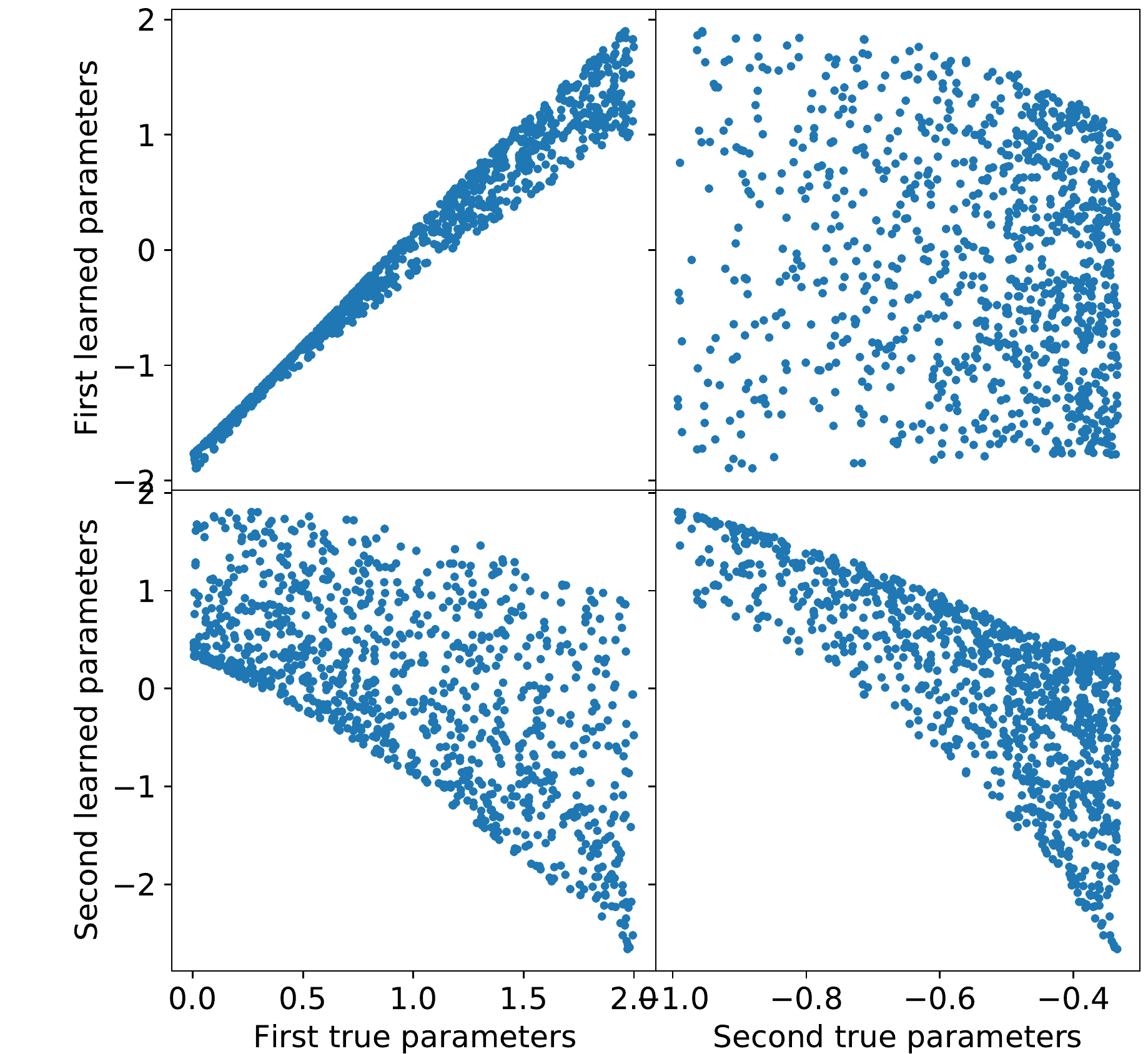}
			\caption{Parameters Gamma}
		\end{subfigure}~
		\begin{subfigure}{0.32\textwidth}
			\centering
			\includegraphics[width=1\linewidth]{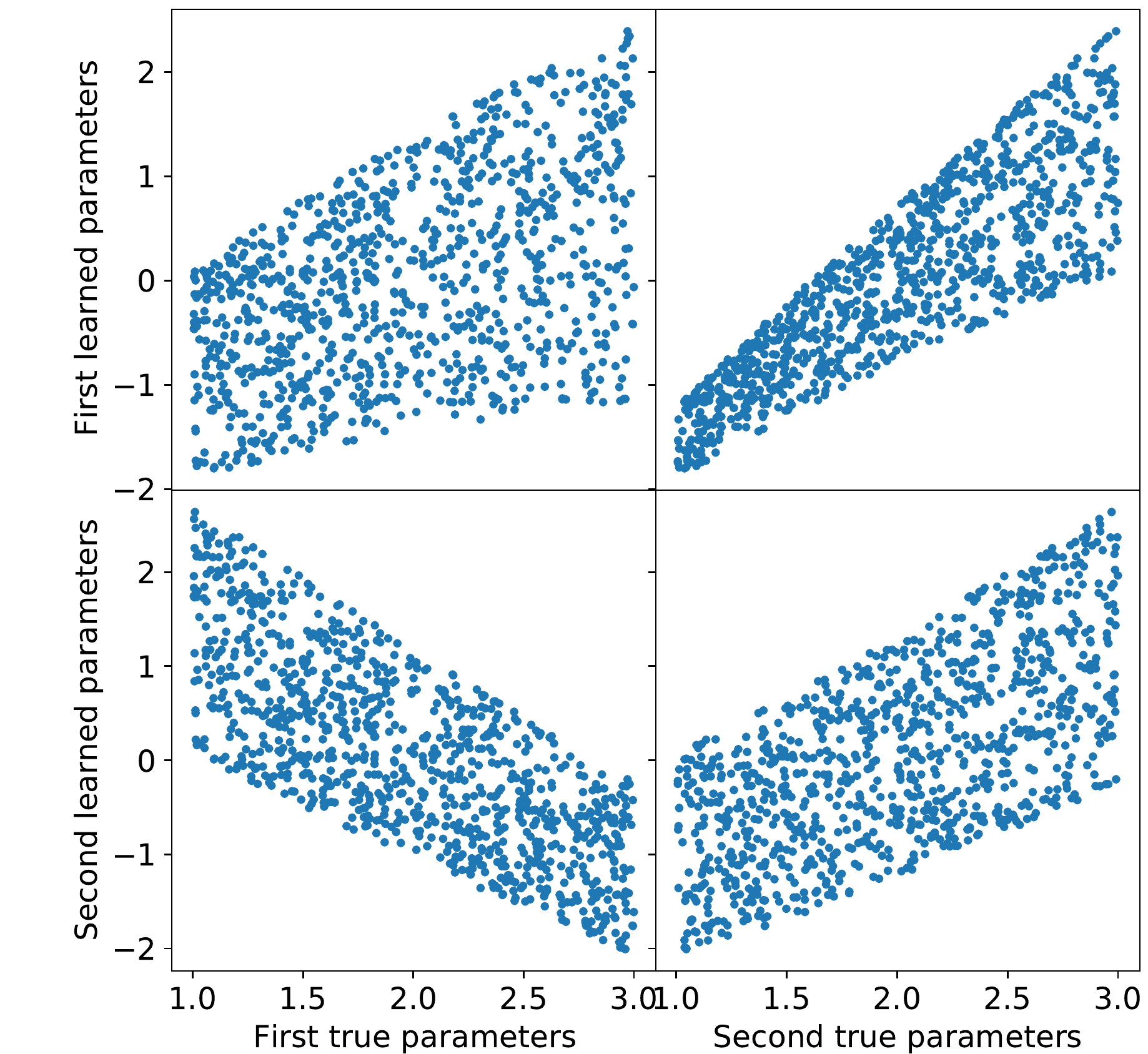}
			\caption{Parameters Beta}
		\end{subfigure}
		\caption{\textbf{Learned and exact embeddings for exponential family models obtained with SM.} 
			Each point represents a different $ x $ (for the statistics) or $ \theta $ (for the natural parameters); 1000 of each were used here.}
		\label{fig:statistics_exp_fam_models}
	\end{figure}
	
	\begin{figure}[!htbp]
		\centering
		\begin{subfigure}{0.32\textwidth}
			\centering
			\includegraphics[width=1\linewidth]{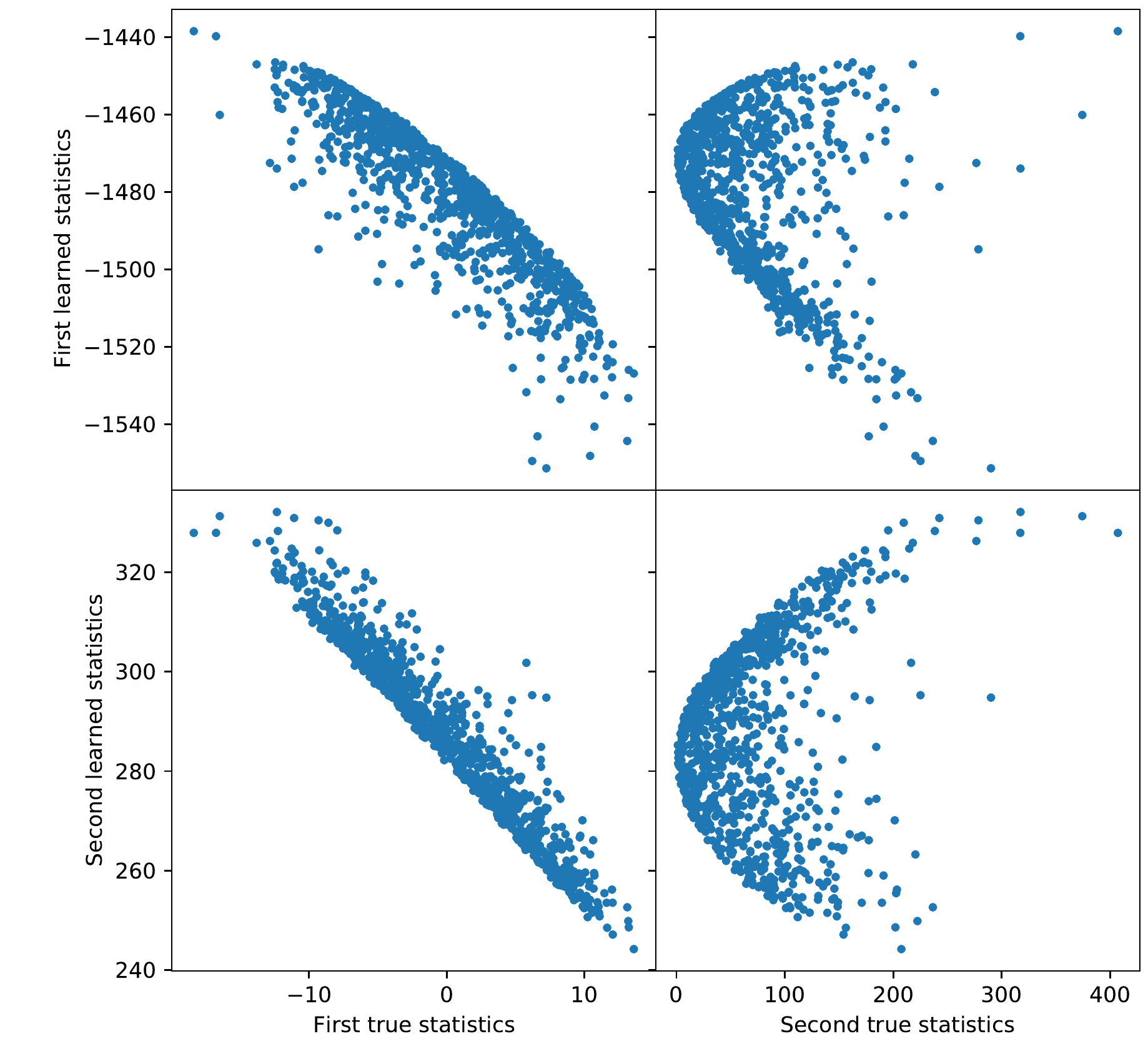}
			\caption{Statistics Gaussian}
		\end{subfigure}~
		\begin{subfigure}{0.32\textwidth}
			\centering
			\includegraphics[width=1\linewidth]{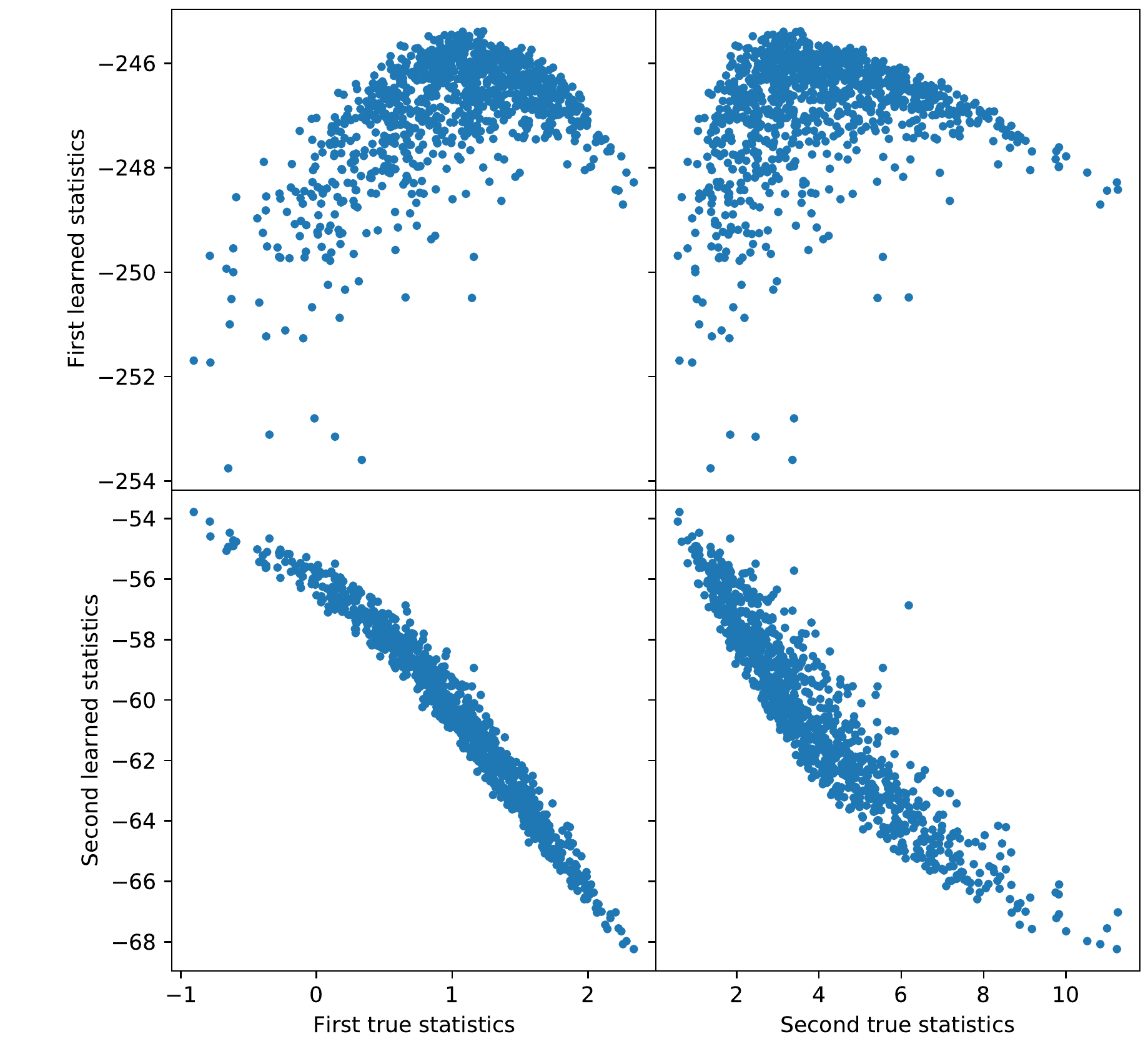}
			\caption{Statistics Gamma}
		\end{subfigure}~	
		\begin{subfigure}{0.32\textwidth}
			\centering
			\includegraphics[width=1\linewidth]{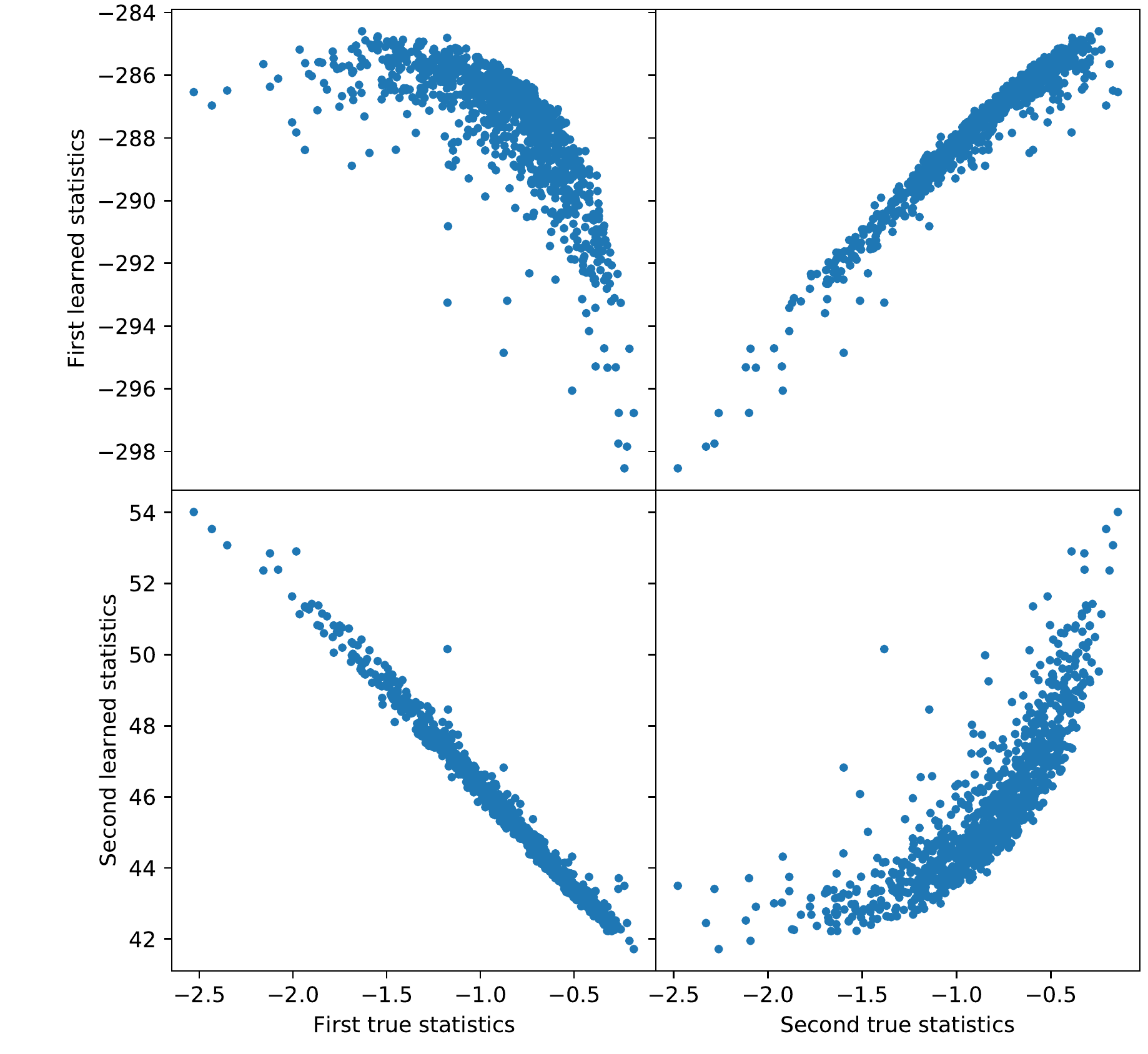}
			\caption{Statistics Beta}
		\end{subfigure}\\
		\begin{subfigure}{0.32\textwidth}
			\centering
			\includegraphics[width=1\linewidth]{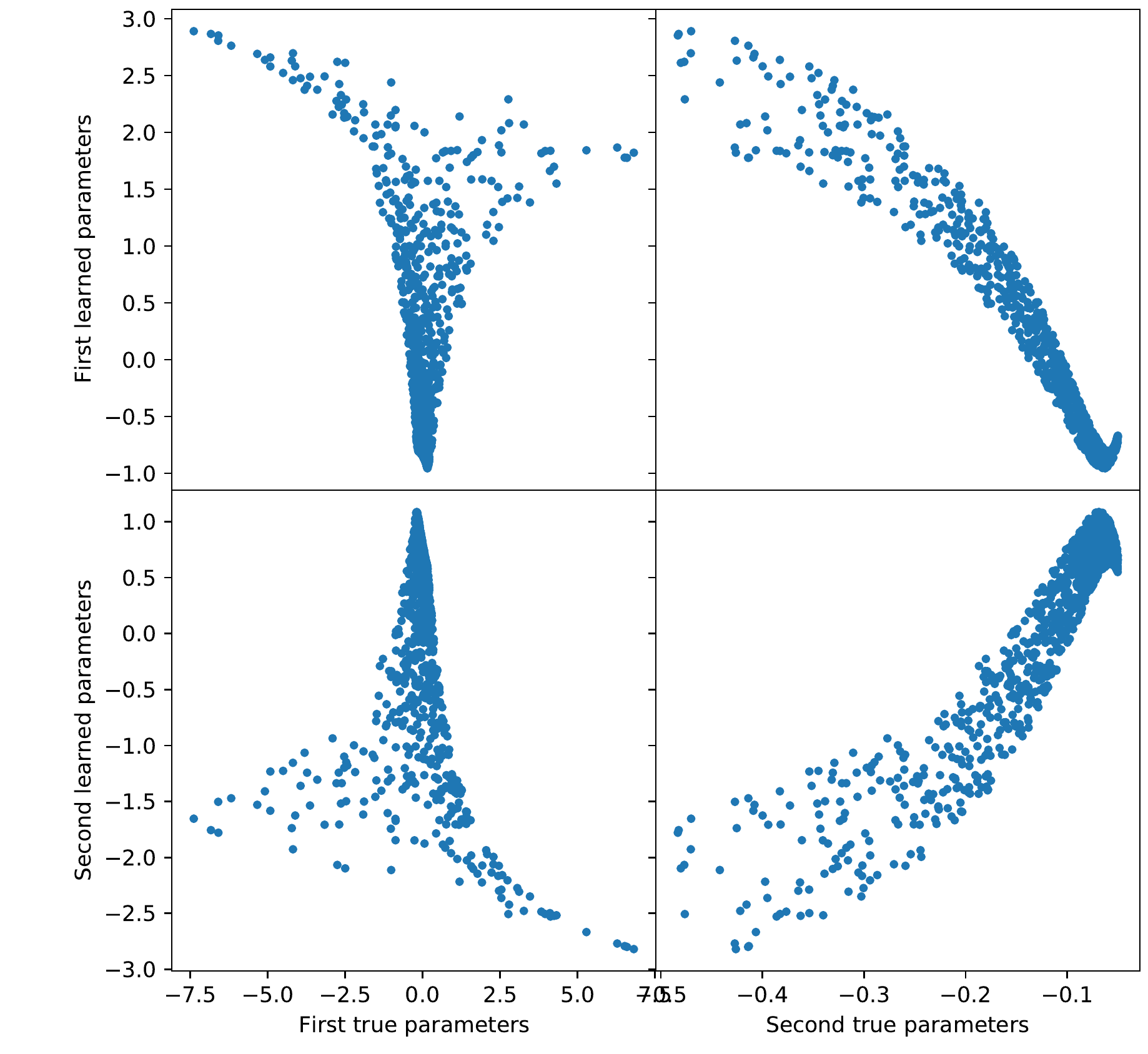}
			\caption{Parameters Gaussian}
		\end{subfigure}~
		\begin{subfigure}{0.32\textwidth}
			\centering
			\includegraphics[width=1\linewidth]{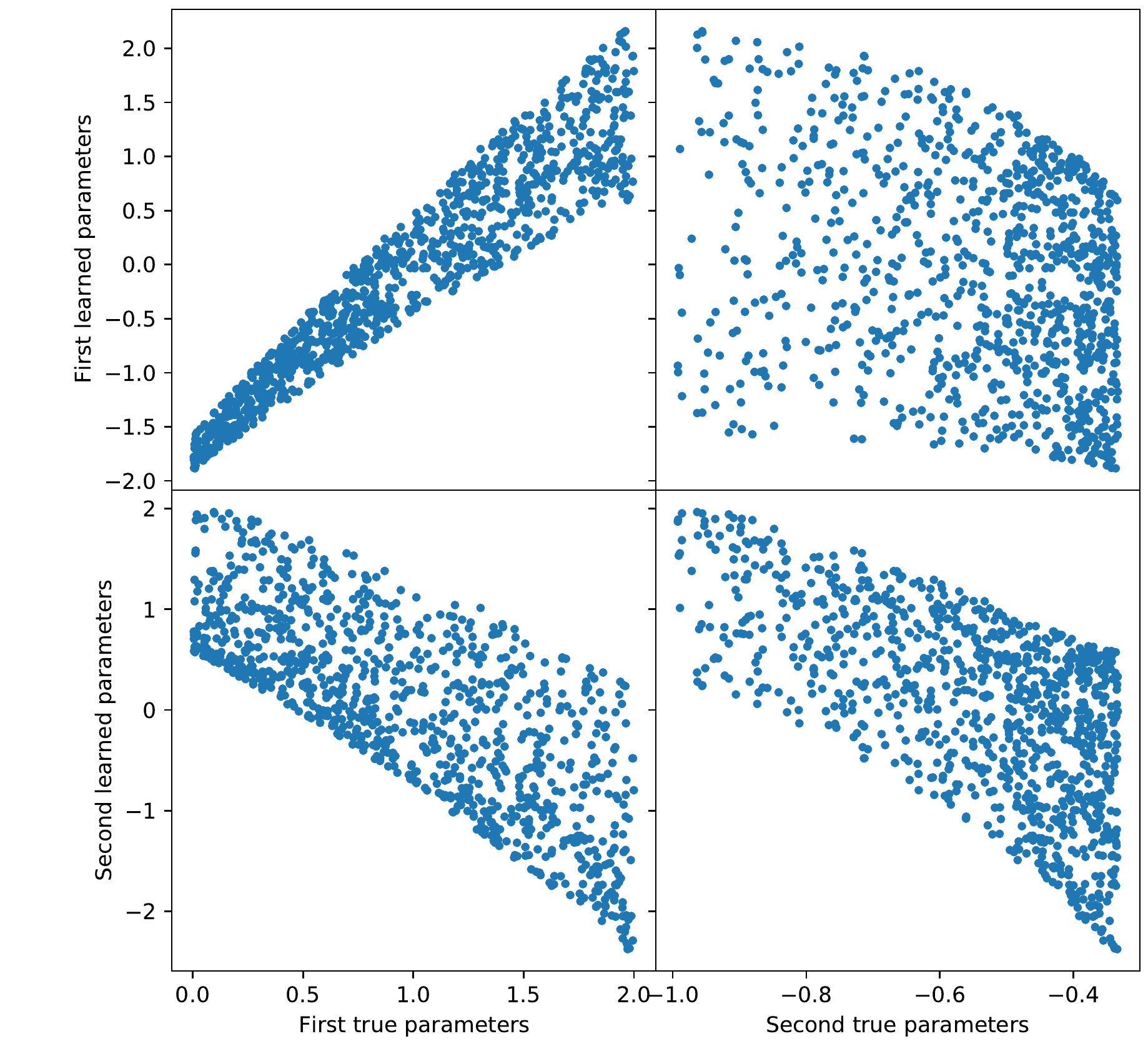}
			\caption{Parameters Gamma}
		\end{subfigure}~
		\begin{subfigure}{0.32\textwidth}
			\centering
			\includegraphics[width=1\linewidth]{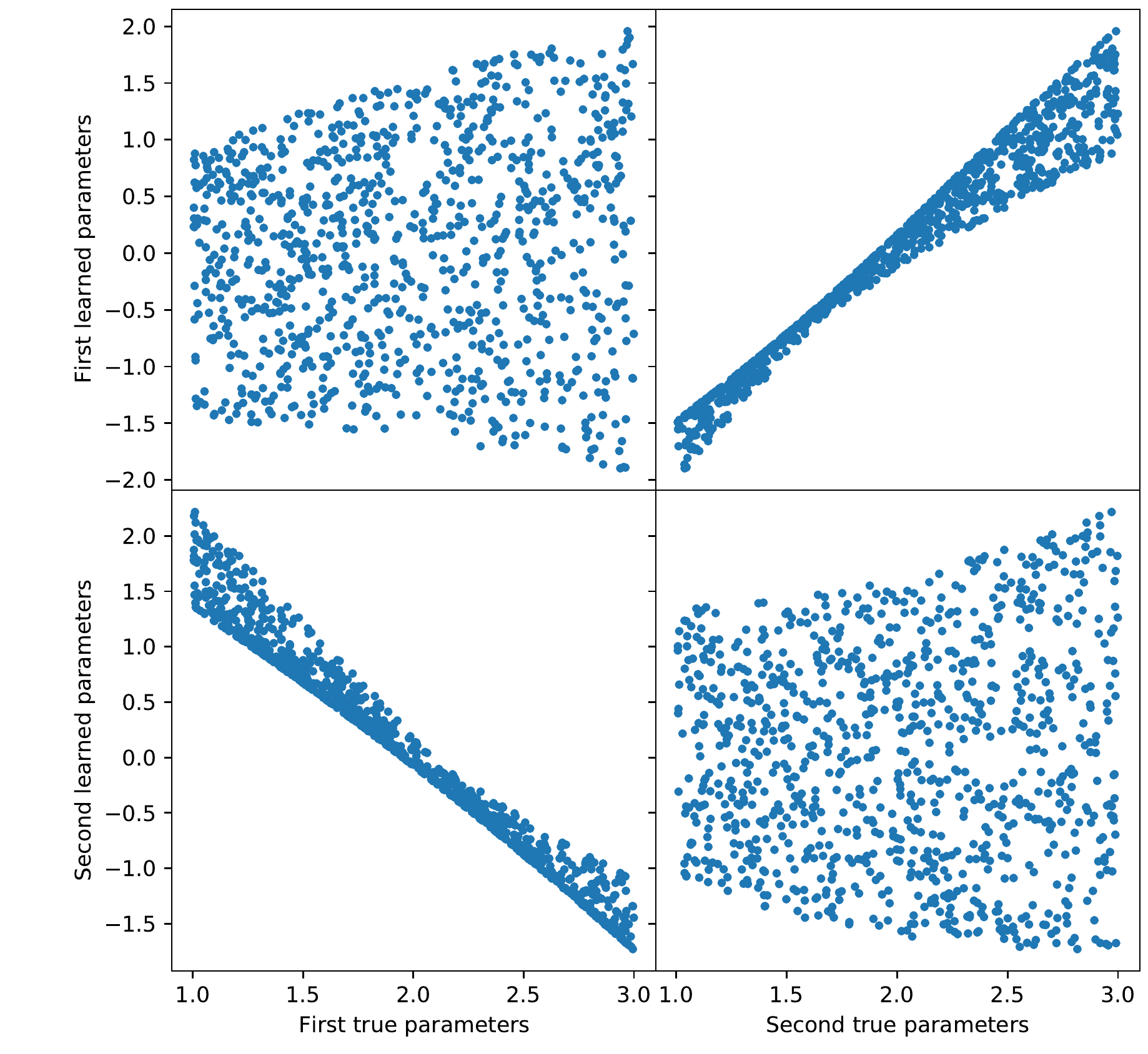}
			\caption{Parameters Beta}
		\end{subfigure}
		\caption{\textbf{Learned and exact embeddings for exponential family models obtained with SSM.} 
			Each point represents a different $ x $ (for the statistics) or $ \theta $ (for the natural parameters); 1000 of each were used here.}
		\label{fig:statistics_exp_fam_models_SSM}	
	\end{figure}

	\newpage
	\FloatBarrier
	\subsubsection{Performance of Exc-SM and Exc-SSM}\label{app:inner_steps_exp_fam_models}
	We study here the performance of Exc-SM and Exc-SSM with different numbers of inner steps in the ExchangeMCMC algorithm (Algorithm~\ref{alg:ExchangeMCMC_overall}) for the Exponential family models. Specifically, we run the inference with 10, 30, 100 and 200 inner MCMC steps, and we evaluate the performance in these 4 cases (Figure~\ref{fig:boxplots_inner_MCMC_step}); considering the different models, we observe that the performance with 30 steps is almost equivalent to the one with 100 and 200, albeit being faster (see the computational time in Table~\ref{Tab:ExcSM_time_exp_fam}). In the main text, we therefore present results using 30 inner MCMC steps.

	\begin{table}[tb]
		\centering
		\begin{tabular}{lcccc}
			\toprule
			\textit{Inner MCMC steps} & 10 & 30 & 100 & 200 \\
			\midrule
			\textit{Time (minutes)} & $ \approx 2 $ & $ \approx 4 $ & $ \approx 16 $& $ \approx 28 $\\
			\bottomrule
		\end{tabular}
		\caption{\textbf{Approximate computational time of Exc-SM with different number of inner MCMC steps for the exponential family models.} These values were obtained by running on a single core.} 
		\label{Tab:ExcSM_time_exp_fam}
	\end{table}

	\begin{figure}[!tb]
		\centering
		\begin{subfigure}{1\textwidth}
			\centering
			\begin{subfigure}{0.20\textwidth}
				\centering
				\includegraphics[width=1\linewidth]{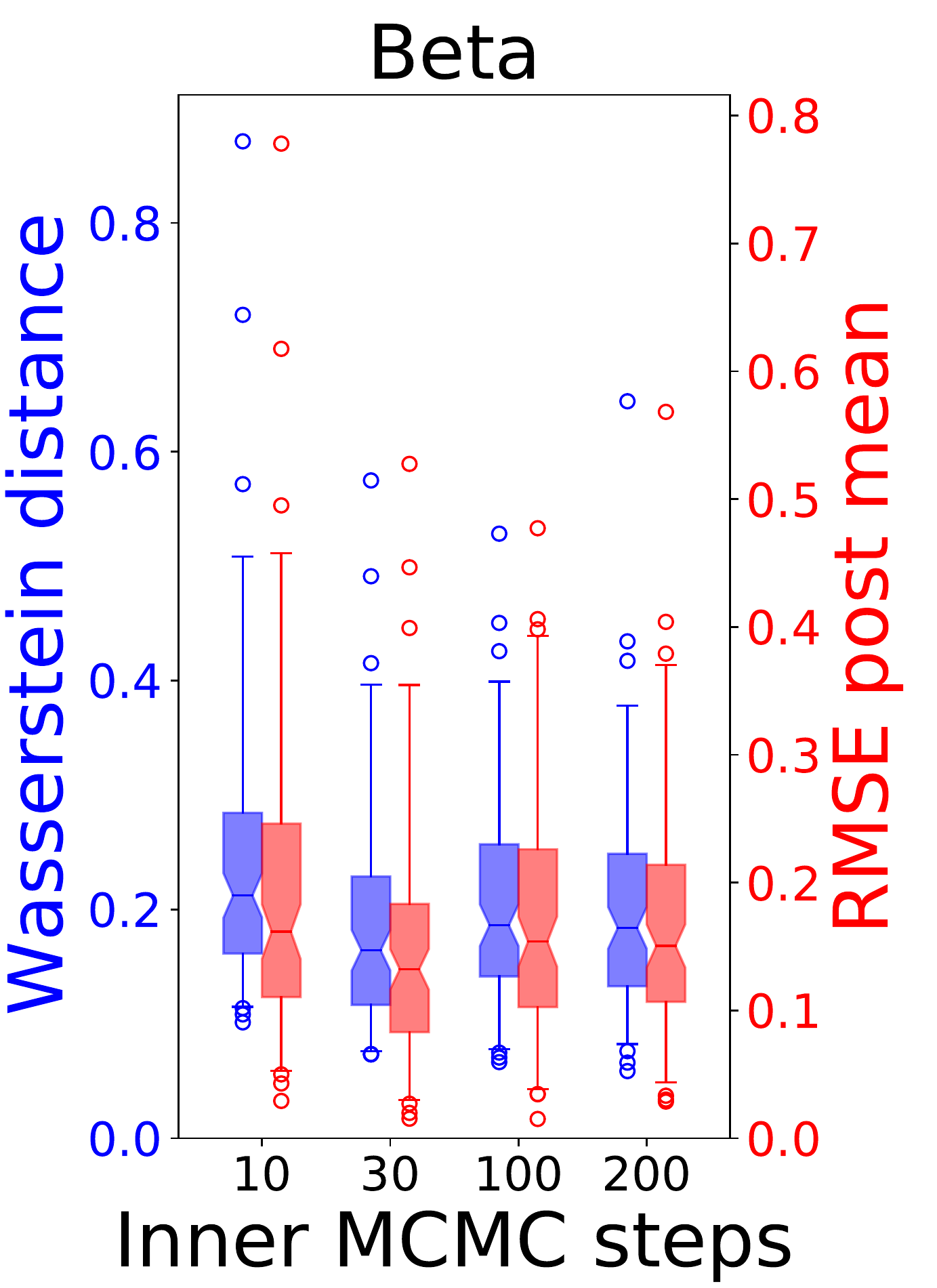}
				
			\end{subfigure}~
			\begin{subfigure}{0.20\textwidth}
				\centering
				\includegraphics[width=1\linewidth]{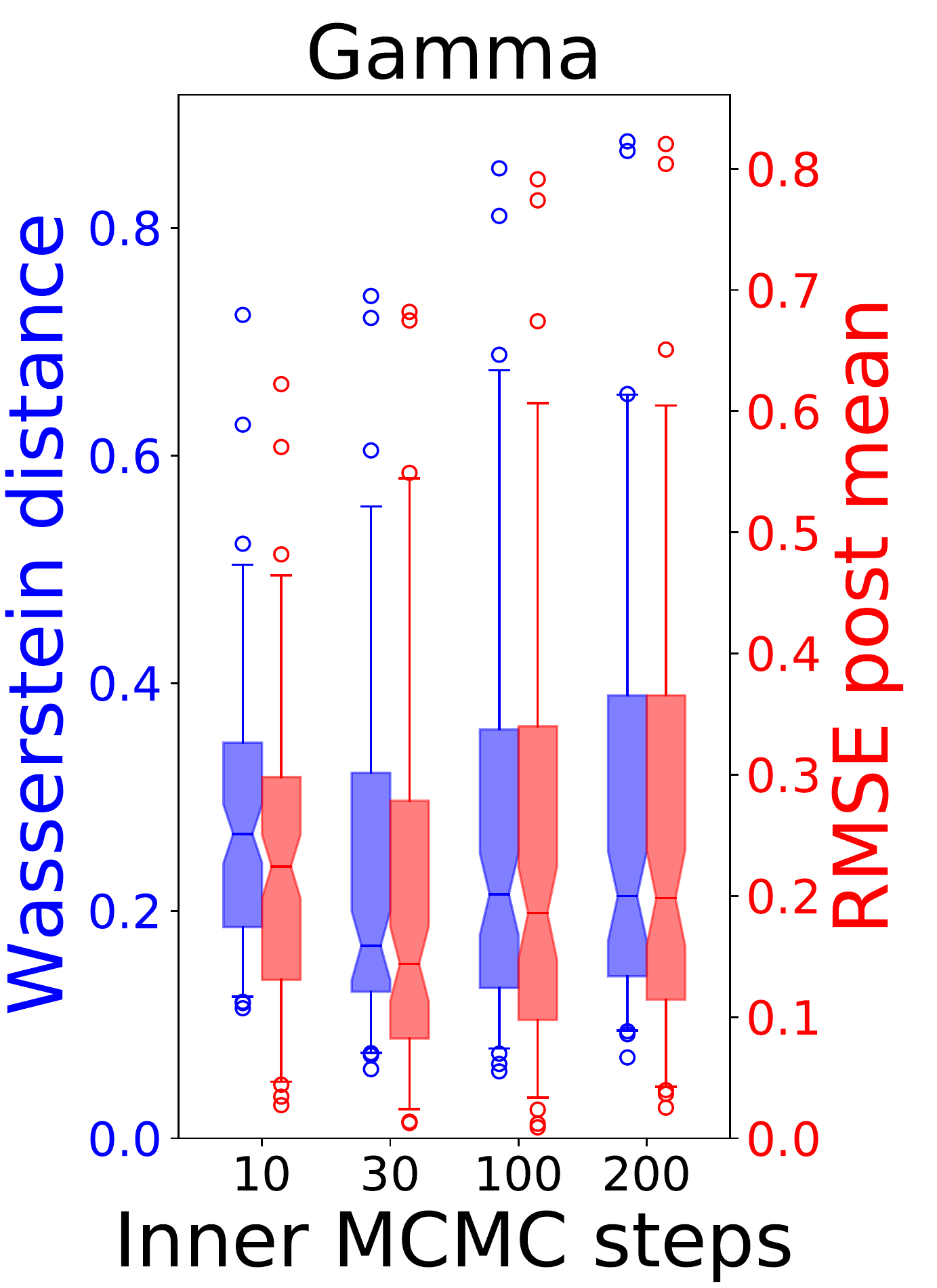}
				
			\end{subfigure}~
			\begin{subfigure}{0.20\textwidth}
				\centering
				\includegraphics[width=1\linewidth]{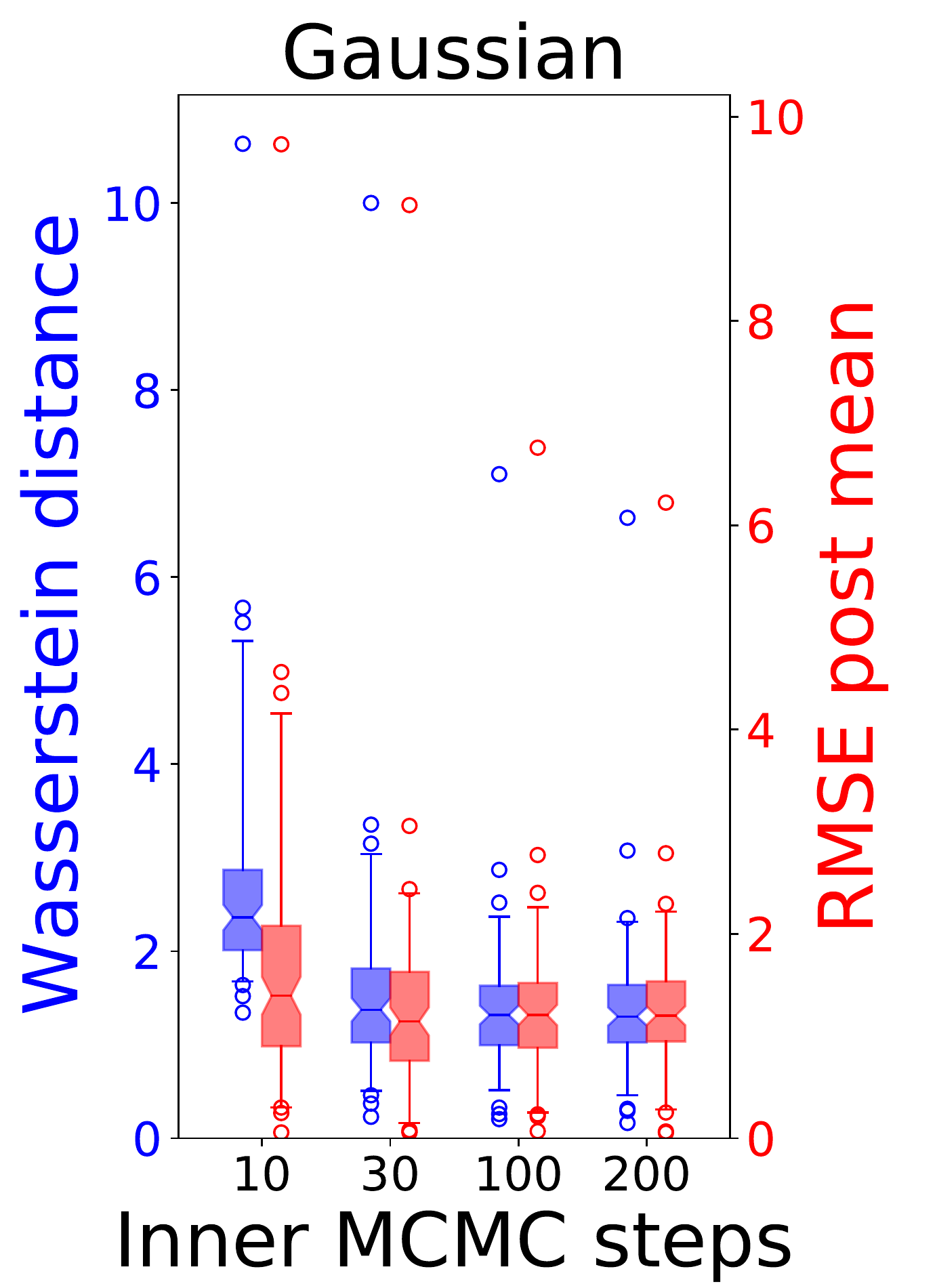}
				
			\end{subfigure}~
			\caption{Exc-SM}
		\end{subfigure}\\
		
		\begin{subfigure}{1\textwidth}
			\centering
			\begin{subfigure}{0.20\textwidth}
				\centering
				\includegraphics[width=1\linewidth]{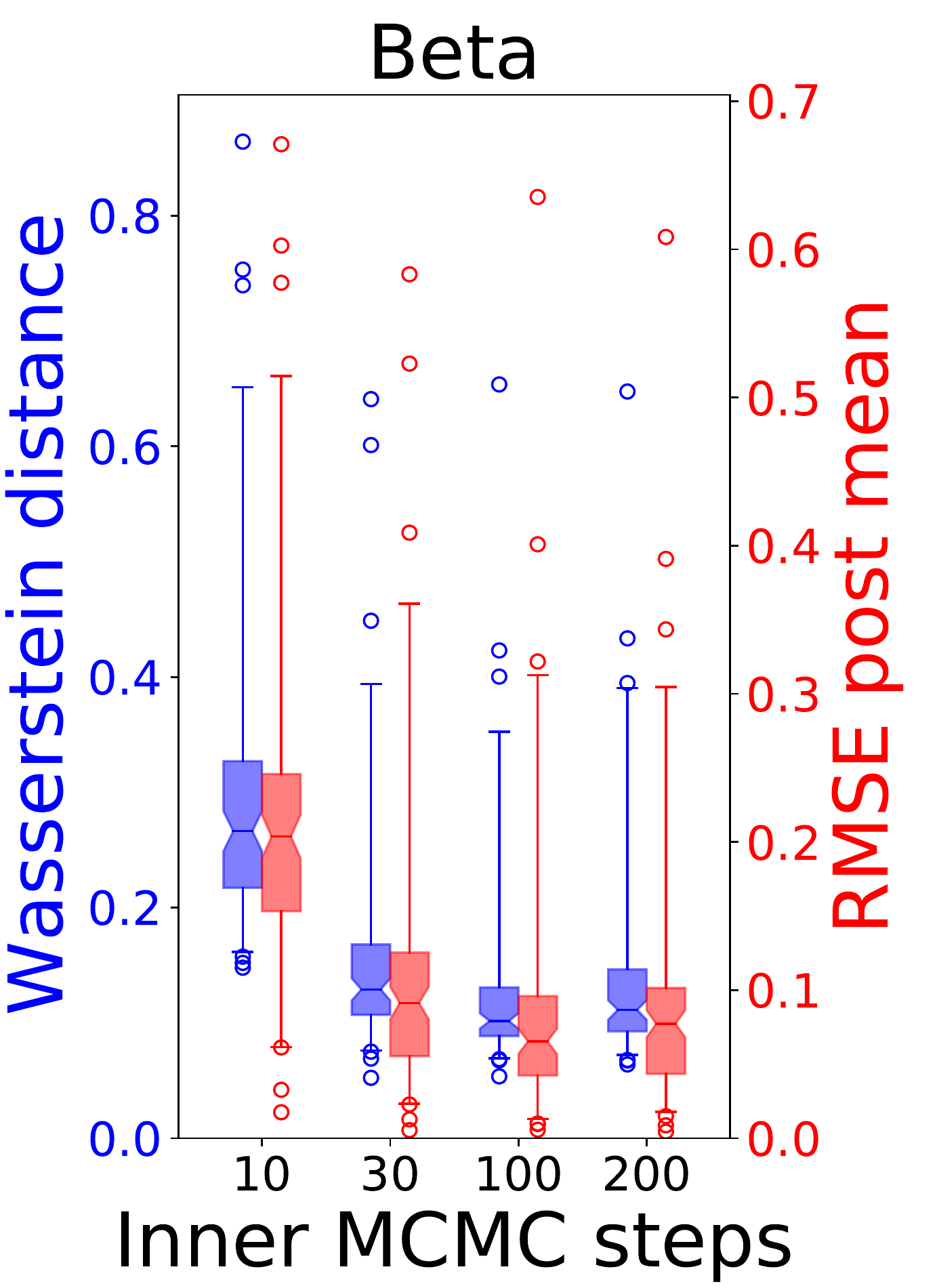}
				
			\end{subfigure}~
			\begin{subfigure}{0.20\textwidth}
				\centering
				\includegraphics[width=1\linewidth]{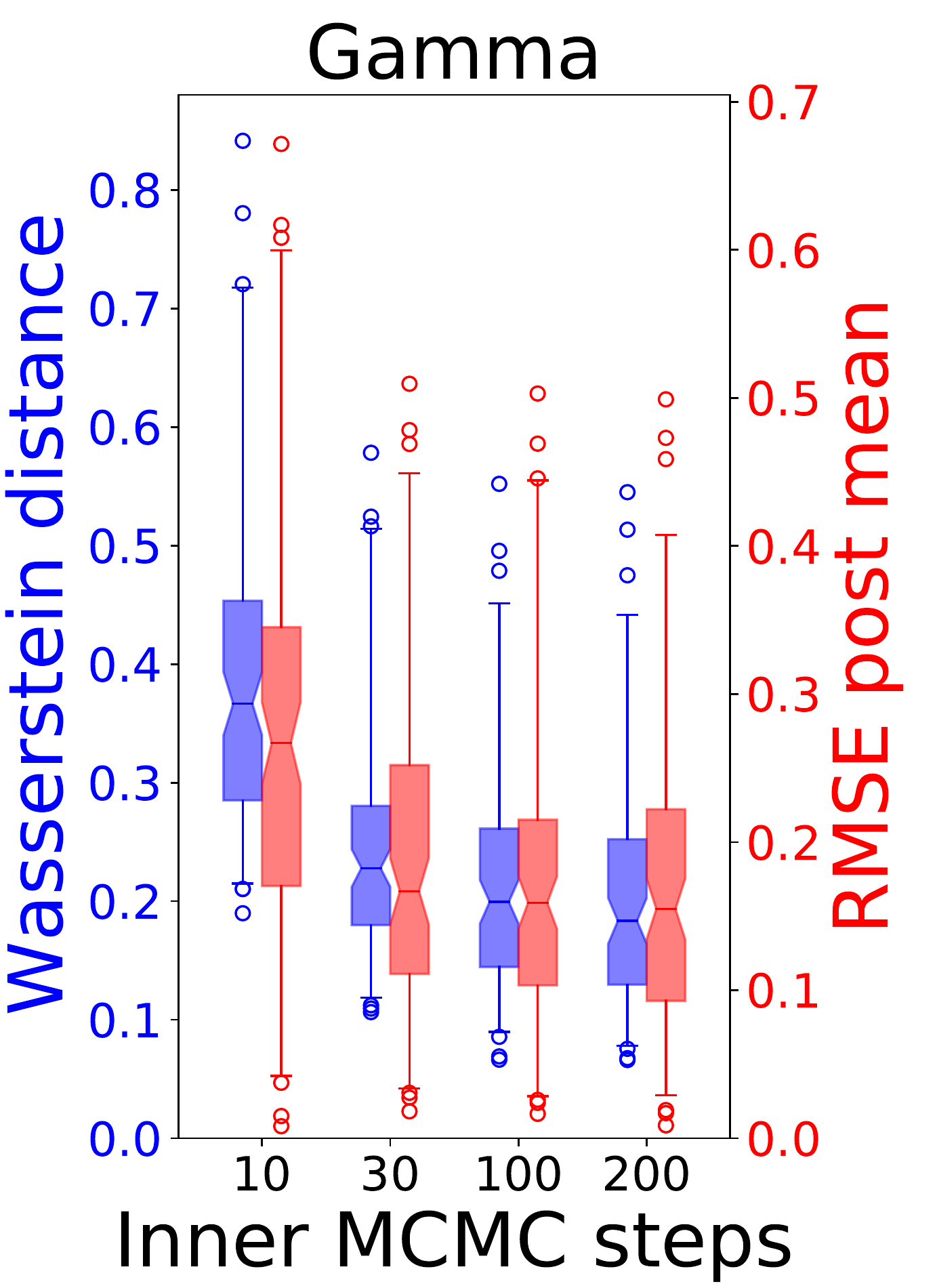}
				
			\end{subfigure}~
			\begin{subfigure}{0.20\textwidth}
				\centering
				\includegraphics[width=1\linewidth]{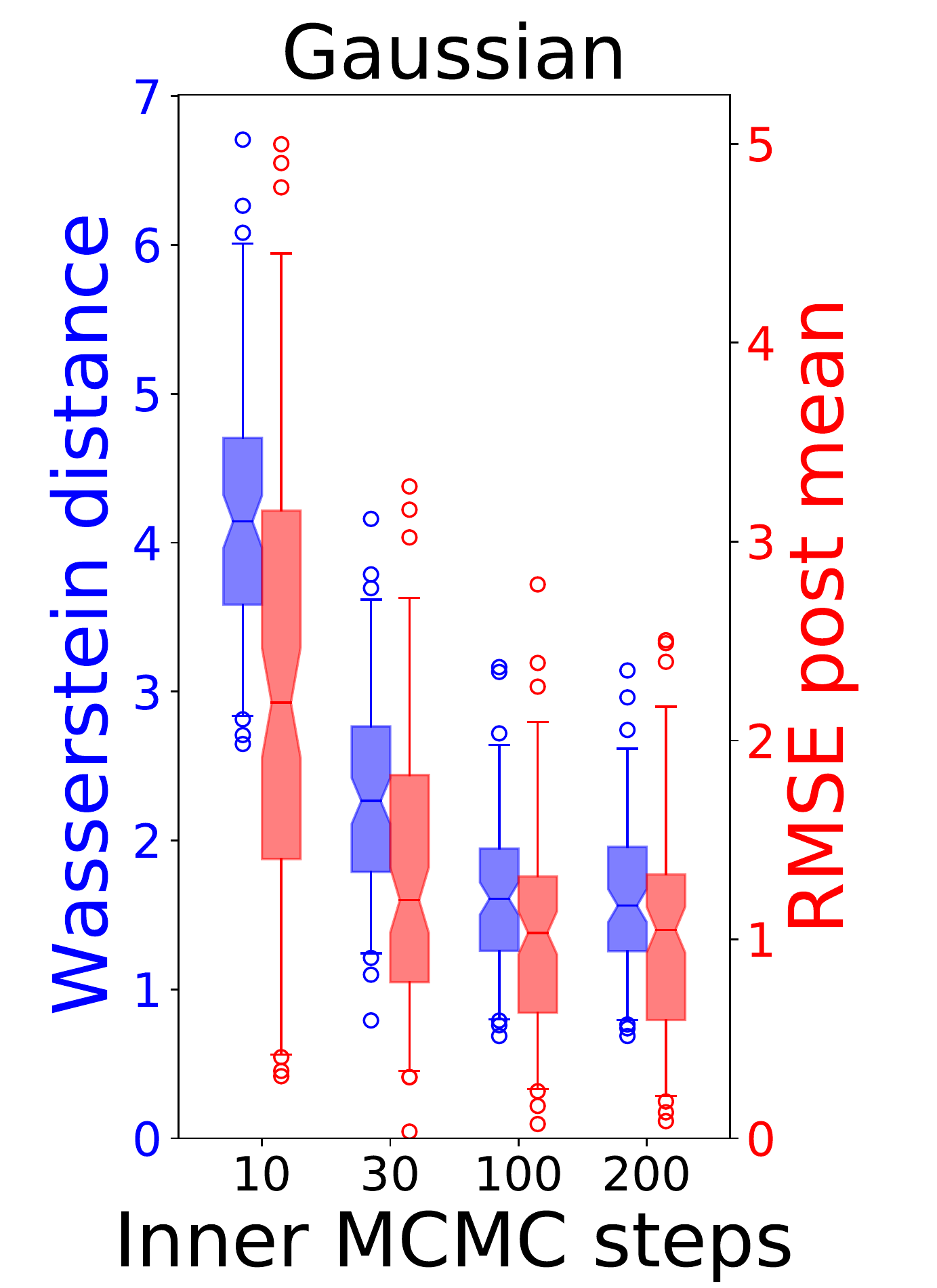}
				
			\end{subfigure}~
			\caption{Exc-SSM}
		\end{subfigure}\\

		\caption{\textbf{Performance of Exc-SM with different number of inner MCMC steps, for exponential family models.} Wasserstein distance from the exact posterior and RMSE between exact and approximate posterior means are reported for 100 observations using boxplots. Boxes span from 1st to 3rd quartile, whiskers span 95\% probability density region and horizontal line denotes median. The numerical values are not comparable across examples, as they depend on the range of parameters.}
		\label{fig:boxplots_inner_MCMC_step}
	\end{figure}

	\subsection{AR(2) and MA(2) models}\label{app:inner_MCMC_arma}

	We study here the performance of Exc-SM and Exc-SSM with different numbers of inner steps in the ExchangeMCMC algorithm (Algorithm~\ref{alg:ExchangeMCMC_overall}) for the AR(2) and MA(2) models. Specifically, we run the inference with 10, 30, 100 and 200 inner MCMC steps, and we evaluate the performance in these 4 cases (Figure~\ref{fig:boxplots_inner_MCMC_step_arma}); considering the different models, we observe that the performance with 30 steps is almost equivalent to the one with 100 and 200, albeit being faster (see the computational time in Table~\ref{Tab:ExcSM_time_timeseries}). In the main text, we therefore present results using 30 inner MCMC steps.

	\begin{table}[tb]
		\centering
		\begin{tabular}{lcccc}
			\toprule
			\textit{Inner MCMC steps} & 10 & 30 & 100 & 200 \\
			\midrule
			\textit{Time (minutes)} & $ \approx 2 $ & $ \approx 5 $ & $ \approx 17 $& $ \approx 29 $\\
			\bottomrule
		\end{tabular}
		\caption{\textbf{Approximate computational time of Exc-SM with different number of inner MCMC steps, for the AR(2) and MA(2) models.} These values were obtained by running on a single core.} 
		\label{Tab:ExcSM_time_timeseries}
	\end{table}

	\begin{figure}[!tb]
		\centering
		\begin{subfigure}{0.48\textwidth}
			\centering
			\begin{subfigure}{0.48\textwidth}
				\centering
				\includegraphics[width=1\linewidth]{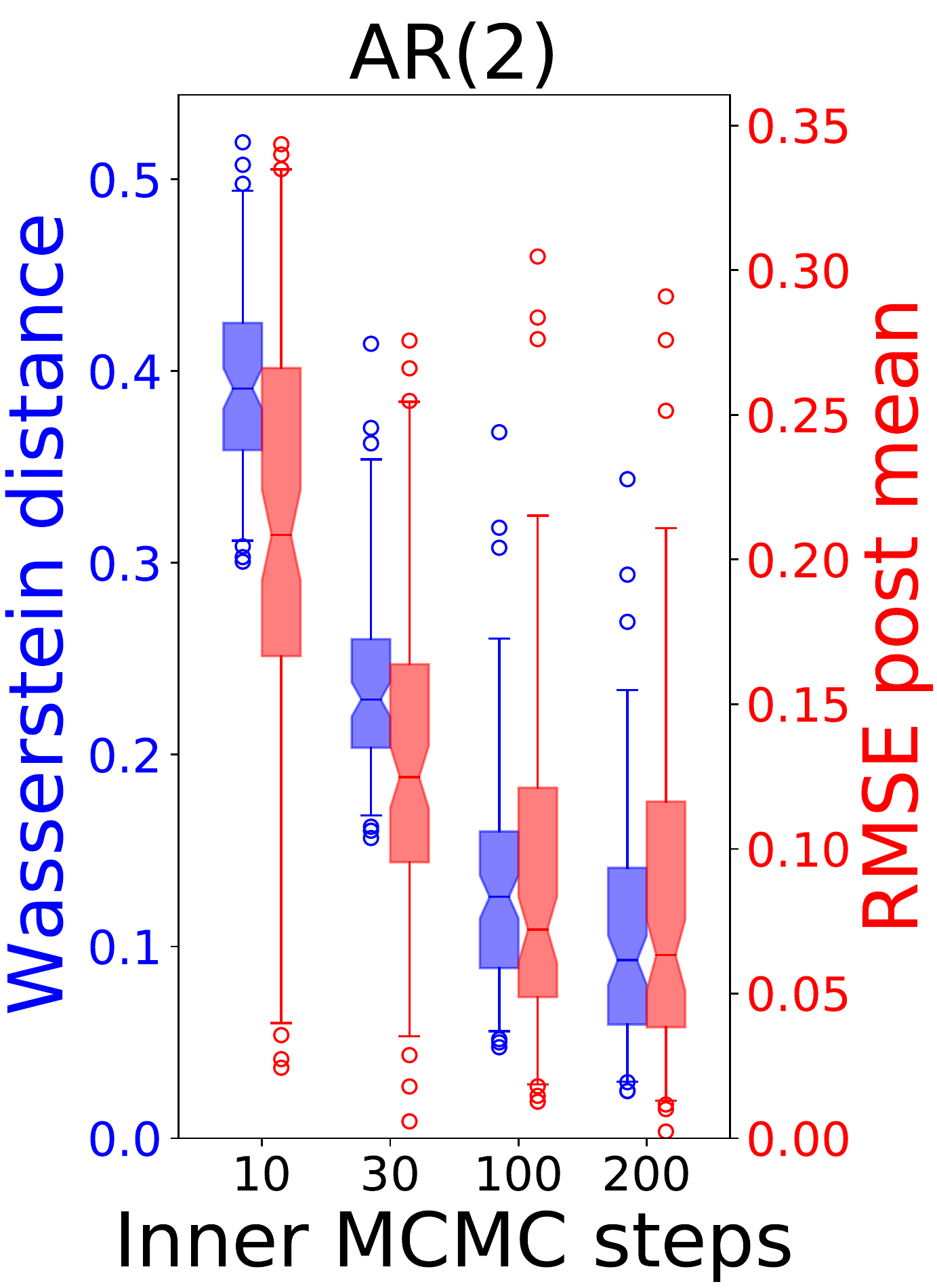}
				
			\end{subfigure}~
			\begin{subfigure}{0.48\textwidth}
				\centering
				\includegraphics[width=1\linewidth]{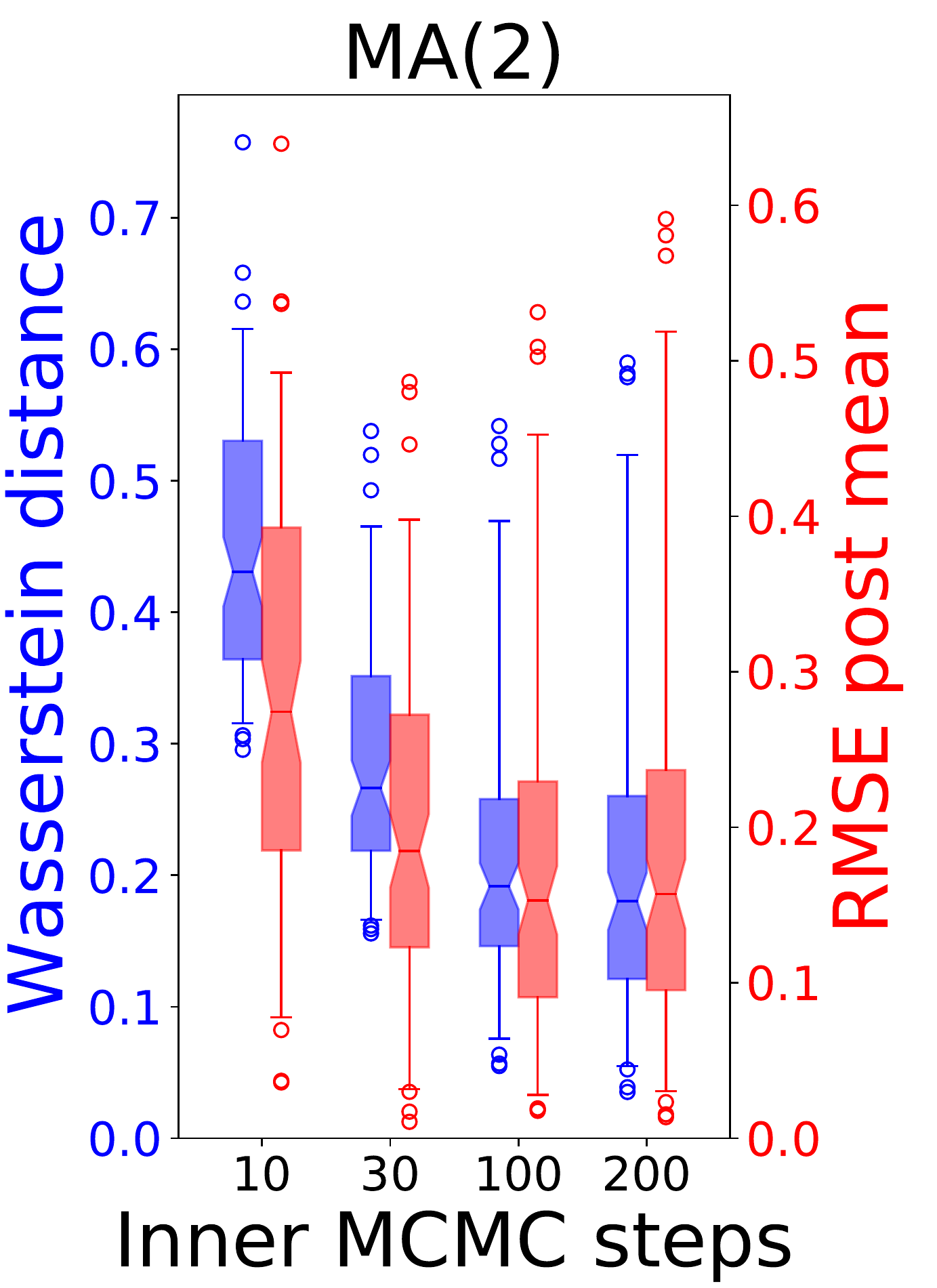}

			\end{subfigure}~
			\caption{Exc-SM}
		\end{subfigure}~
		\begin{subfigure}{0.48\textwidth}
			\centering
			
			\begin{subfigure}{0.48\textwidth}
				\centering
				\includegraphics[width=1\linewidth]{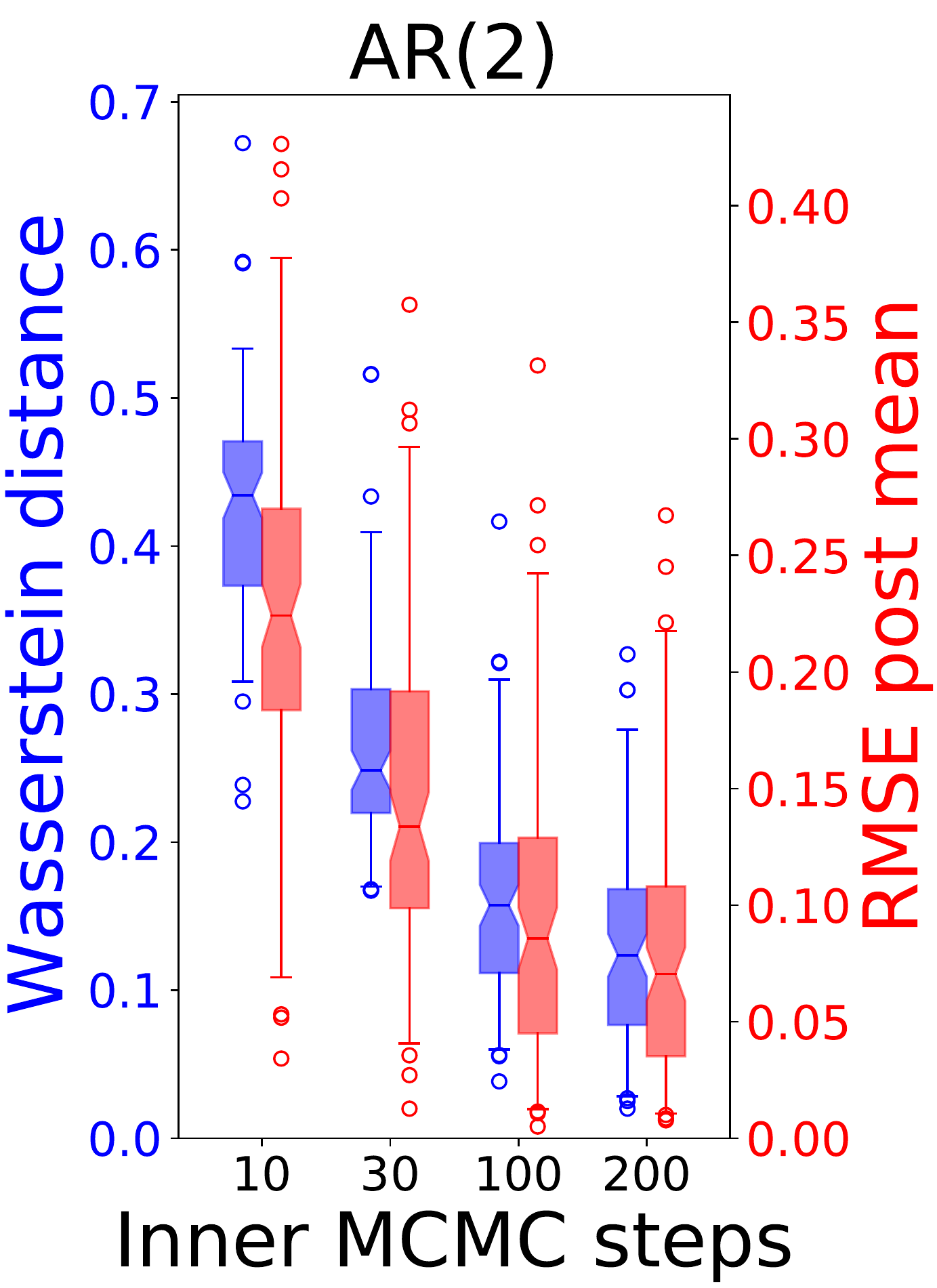}
				
			\end{subfigure}~
			\begin{subfigure}{0.48\textwidth}
				\centering
				\includegraphics[width=1\linewidth]{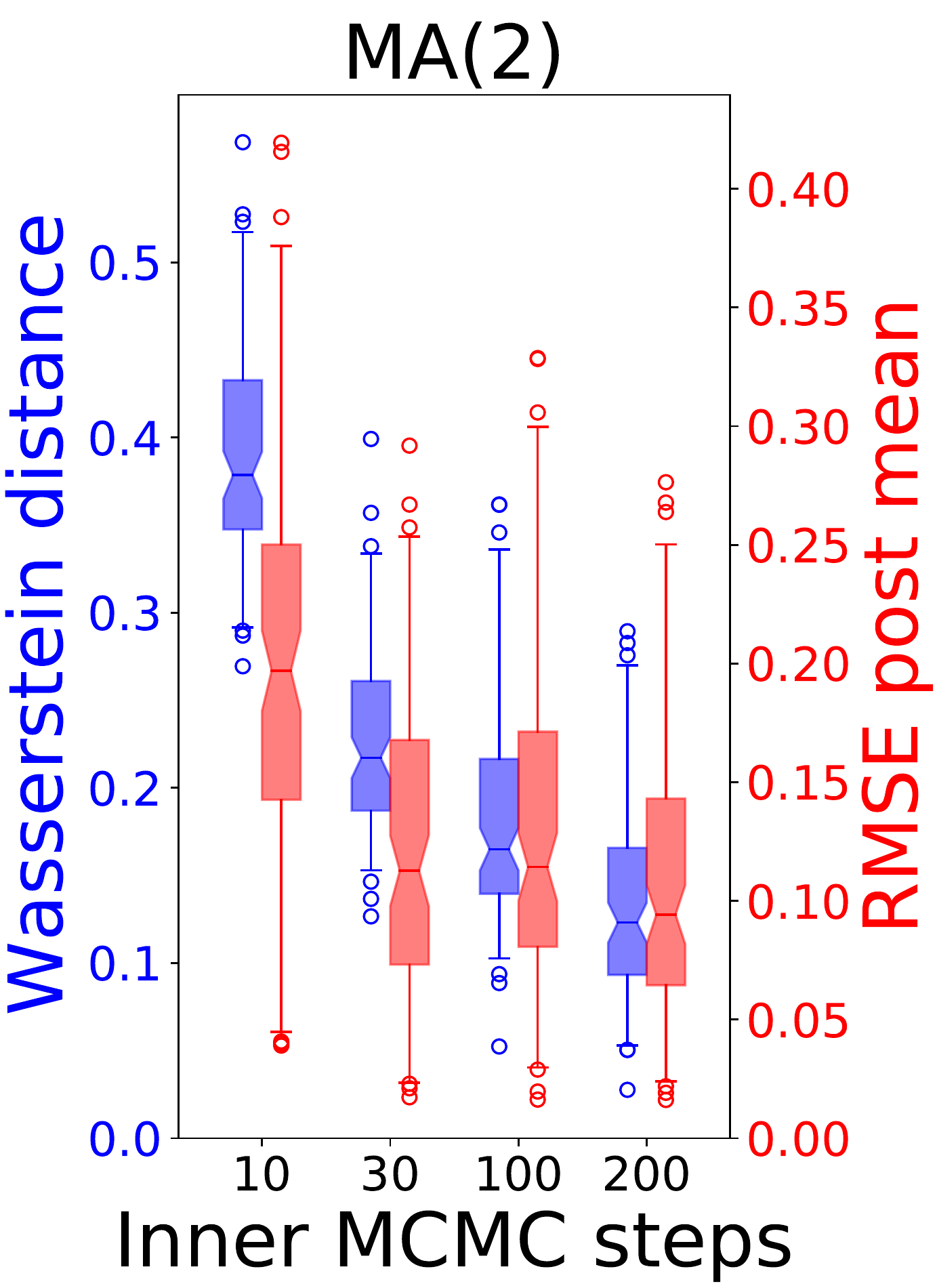}

			\end{subfigure}~
			\caption{Exc-SSM}
		\end{subfigure}~
		
		\caption{\textbf{Performance of Exc-SM with different number of inner MCMC steps, for AR(2) and MA(2) models.} Wasserstein distance from the exact posterior and RMSE between exact and approximate posterior means are reported for 100 observations using boxplots. Boxes span from 1st to 3rd quartile, whiskers span 95\% probability density region and horizontal line denotes median.}
		\label{fig:boxplots_inner_MCMC_step_arma}
	\end{figure}

	\subsection{Simulation cost to reach equivalent performance as Exc-SM}\label{app:evolution_WD_iter}

	Above, we have showed that Exc-SM and Exc-SSM are competitive with the other approaches, even if they require no additional model simulations. Here, we quantify how many model simulations are needed for the other techniques to reach the same performance of Exc-SM in the exponential family and time-series models. The same analysis could be done with respect to Exc-SSM but, as its performance is generally close to the one achieved by Exc-SM, we avoid repeating it. 
	
	We compute therefore the performance of ABC-SM, ABC-SSM, ABC-FP, PMC-SL, PMC-RE at each iteration for all 100 observations and find when their median performance (as quantifies by the Wasserstein distance with respect to the true posterior) becomes comparable or better than the median one achieved by Exc-SM. The number of required simulations are reported in Table~\ref{Table:num_sim}. Notice that some techniques are not able to outperform Exc-SM for some models; we highlight that by a dash in the Table. SL and RE reach similar performance to Exc-SM with one single iteration, when they are able to do so; in fact, we found empirically that the performance of them does not significantly improve with iterations. Still, we remark that one single iteration of SL and RE requires a very large number of model simulations.

	In Table~\ref{Table:num_sim}, we also give the number of simulations required for the preliminary NN training in the methods which use one; further, we compute the overall cost of inference in terms of model simulations (taking into account both NN training and inference) for different number of observations; we remark that, as discussed previously, our method requires no additional model simulations to perform inference after the NNs have been trained. From Table~\ref{Table:num_sim}, it can be seen that, for all models, ABC-SM, ABC-FP, PMC-RE and PMC-SL require a number of simulations larger than the one needed to train the NNs in Exc-SM to reach the performance achieved by Exc-SM, which makes the latter an interesting option for models in which simulations are very expensive.

	\begin{table}[!tbp]
		\centering
		\small
		\begin{tabular}{ lcccccc }
			\toprule
			& \multicolumn{6}{c}{Beta}   \\ \cmidrule[0.1pt](r){2-7}
			& Exc-SM & ABC-SM & ABC-SSM & ABC-FP & PMC-SL & PMC-RE \\ 
			\midrule
			NN training & $2 \cdot 10^4 $ & $  2 \cdot  10^4$ & $2 \cdot  10^4 $ & $2 \cdot  10^4 $ & - & - \\
			Inference & $ 0 $ &  $ 2.5 \cdot 10^4 $ & $ 1.8 \cdot 10^4 $  & - & $1 \cdot 10^5 $ & $1 \cdot10^6 $\\
			\midrule[0.1pt]
			Total 1 obs & $ 2 \cdot 10^4  $ & $ 4.5 \cdot 10^4 $ & $ 3.8\cdot 10^4 $ & - & $1 \cdot 10^5 $ & $1 \cdot10^6 $\\
			Total 3 obs & $ 2 \cdot 10^4  $ & $ 9.5 \cdot 10^4 $ & $ 7.4\cdot 10^4 $ & - & $3 \cdot 10^5 $ & $3 \cdot10^6 $\\
			Total 100 obs & $ 2 \cdot 10^4  $ & $ 2.52\cdot10^6 $ & $ 1.82\cdot 10^6  $ & - & $1 \cdot 10^7 $ & $1 \cdot10^8 $\\
			\bottomrule
			& \multicolumn{6}{c}{Gamma}   \\ \cmidrule[0.1pt](r){2-7}
			& Exc-SM & ABC-SM & ABC-SSM & ABC-FP & PMC-SL & PMC-RE \\ 
			\midrule
			NN training & $2 \cdot 10^4 $ & $  2 \cdot  10^4$ & $2 \cdot  10^4 $ & $2 \cdot  10^4 $ & - & - \\
			Inference & $ 0 $ & $ 4.7\cdot 10^4 $ & $ 2.9 \cdot 10^4 $ & - & $ 1 \cdot10^5 $ & $1 \cdot10^6 $\\
			\midrule[0.1pt]
			Total 1 obs & $ 2 \cdot 10^4  $ &  $ 6.7\cdot 10^4 $ & $ 4.9 \cdot 10^4 $ & - & $ 1 \cdot10^5 $ & $1 \cdot10^6 $\\
			Total 3 obs & $ 2 \cdot 10^4  $ & $ 1.61\cdot 10^5 $ & $ 1.07 \cdot 10^5 $ & - & $ 3 \cdot10^5 $ & $3 \cdot10^6 $\\
			Total 100 obs & $ 2 \cdot 10^4  $ & $ 4.72\cdot 10^6  $ & $ 2.92\cdot 10^6 $ & - & $ 1 \cdot10^7 $ & $1 \cdot10^8 $\\
			\bottomrule
			& \multicolumn{6}{c}{Gaussian}   \\ \cmidrule[0.1pt](r){2-7}
			& Exc-SM & ABC-SM & ABC-SSM & ABC-FP & PMC-SL & PMC-RE \\ 
			\midrule
			NN training & $2 \cdot 10^4 $ & $  2 \cdot  10^4$ & $2 \cdot  10^4 $ & $2 \cdot  10^4 $ & - & - \\
			Inference & $ 0 $ & $ 2.7\cdot 10^4 $ & $ 2.6\cdot 10^4 $  & - & - & $1 \cdot10^6 $\\
			\midrule[0.1pt]
			Total 1 obs & $ 2 \cdot 10^4  $ & $ 4.7\cdot 10^4 $ & $ 4.6\cdot 10^4 $ & - & - & $1 \cdot10^6 $\\
			Total 3 obs & $ 2 \cdot 10^4  $ &$ 1.01\cdot 10^5 $ & $ 9.8\cdot 10^4 $  & - & - & $3 \cdot10^6 $\\
			Total 100 obs & $ 2 \cdot 10^4  $ &  $ 2.72\cdot 10^6 $ & $ 2.62\cdot 10^6 $  & - & - & $1 \cdot10^8 $\\
			\bottomrule
			& \multicolumn{6}{c}{AR2}   \\ \cmidrule[0.1pt](r){2-7}
			& Exc-SM & ABC-SM & ABC-SSM & ABC-FP & PMC-SL & PMC-RE \\ 
			\midrule
			NN training & $2 \cdot 10^4 $ & $  2 \cdot  10^4$ & $2 \cdot  10^4 $ & $2 \cdot  10^4 $ & - & - \\
			Inference & $ 0 $ & $ 3.2\cdot 10^4 $ & $ 3.1\cdot 10^4 $ &$  2.3\cdot 10^4 $  & $ 1 \cdot10^5 $ & - \\
			\midrule[0.1pt]
			Total 1 obs & $ 2 \cdot 10^4  $ & $ 5.2\cdot 10^4$ & $5.1\cdot 10^4$ & $4.3\cdot 10^4 $ & $ 1 \cdot10^5 $ & - \\
			Total 3 obs & $ 2 \cdot 10^4  $ & $ 1.16\cdot 10^5$ & $1.13\cdot 10^5$ & $8.9\cdot 10^4 $ & $ 3 \cdot10^5 $ & -\\
			Total 100 obs & $ 2 \cdot 10^4  $ & $ 3.22\cdot 10^6$ & $3.12\cdot 10^6$ & $2.32\cdot 10^6 $ & $ 1 \cdot 10^7 $ & - \\
			\bottomrule
			& \multicolumn{6}{c}{MA2}   \\ \cmidrule[0.1pt](r){2-7}
			& Exc-SM & ABC-SM & ABC-SSM & ABC-FP & PMC-SL & PMC-RE \\ 
			\midrule
			NN training & $2 \cdot 10^4 $ & $  2 \cdot  10^4$ & $2 \cdot  10^4 $ & $2 \cdot  10^4 $ & - & - \\
			Inference & $ 0 $ & $3.1\cdot 10^4 $& $ 2.1\cdot 10^4 $ & $ 2.0\cdot 10^4 $  & $ 1\cdot 10^5 $ & - \\
			\midrule[0.1pt]
			Total 1 obs & $ 2 \cdot 10^4  $ & $ 5.1\cdot 10^4 $ & $ 4.1\cdot 10^4 $ & $4.0\cdot 10^4$  & $ 1 \cdot10^5 $ & -\\
			Total 3 obs & $ 2 \cdot 10^4  $ & $ 1.13\cdot 10^5 $ & $ 8.3\cdot 10^4 $ & $8.0\cdot 10^4$ & $ 3 \cdot10^5 $ & -\\
			Total 100 obs & $ 2 \cdot 10^4  $ & $ 3.12\cdot 10^6 $ & $ 2.12\cdot 10^6 $ & $2.02\cdot 10^6$  & $ 1 \cdot 10^7 $ &-\\
			\bottomrule
		\end{tabular}
		\caption{\textbf{Model simulations needed for the different techniques,} for both NN training and inference; for ABC-FP, ABC-SM, ABC-SSM, PMC-SL and PMC-RE, we report simulations needed to obtain performance at least as good as Exc-SM; in case the approach does not reach the same performance as Exc-SM, we denote that by a dash. Notice that PMC-SL and PMC-RE do not require NN training before performing inference. We also show the total number of simulations needed to apply the different approaches on 1, 3 and 100 observations, by taking into account NN training and inference steps.}
		\label{Table:num_sim}
		
	\end{table}

	In Figures~\ref{Fig:wass_dist_iter_exp_fam_models} and \ref{Fig:wass_dist_iter_arma_models} we represent the performance attained by ABC-FP, ABC-SM, ABC-SSM, PMC-SL and PMC-RE at each iteration of the iterative algorithm. On the horizontal axis of all plots, we report the number of model simulations corresponding to the iteration of the algorithm.
	
	\begin{figure}[!htbp]
		\centering
		\begin{subfigure}{0.32\textwidth}
			\centering
			\includegraphics[width=1\linewidth]{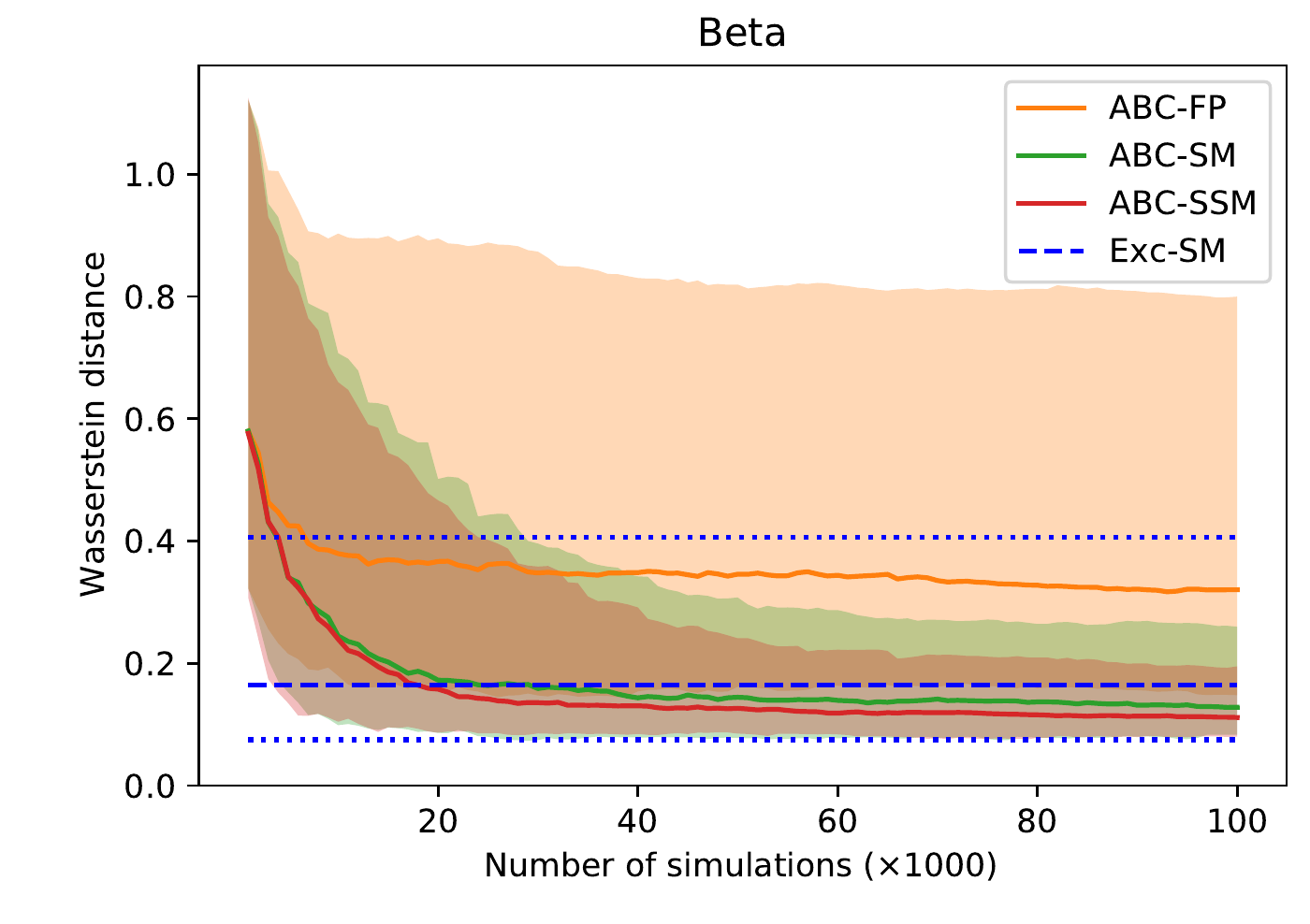}
			
		\end{subfigure}~
		\begin{subfigure}{0.32\textwidth}
			\centering
			\includegraphics[width=1\linewidth]{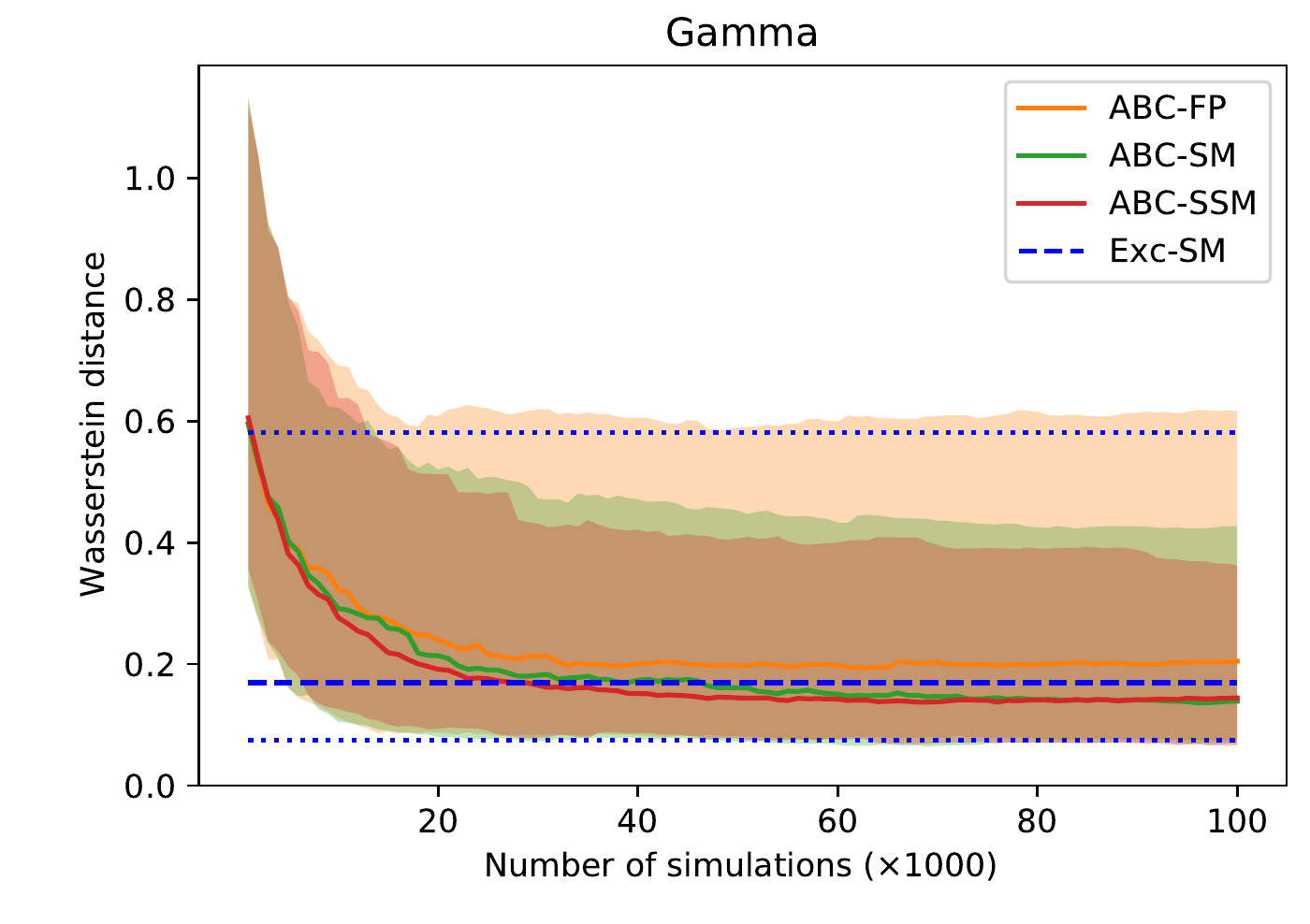}
			
		\end{subfigure}~
		\begin{subfigure}{0.32\textwidth}
			\centering
			\includegraphics[width=1\linewidth]{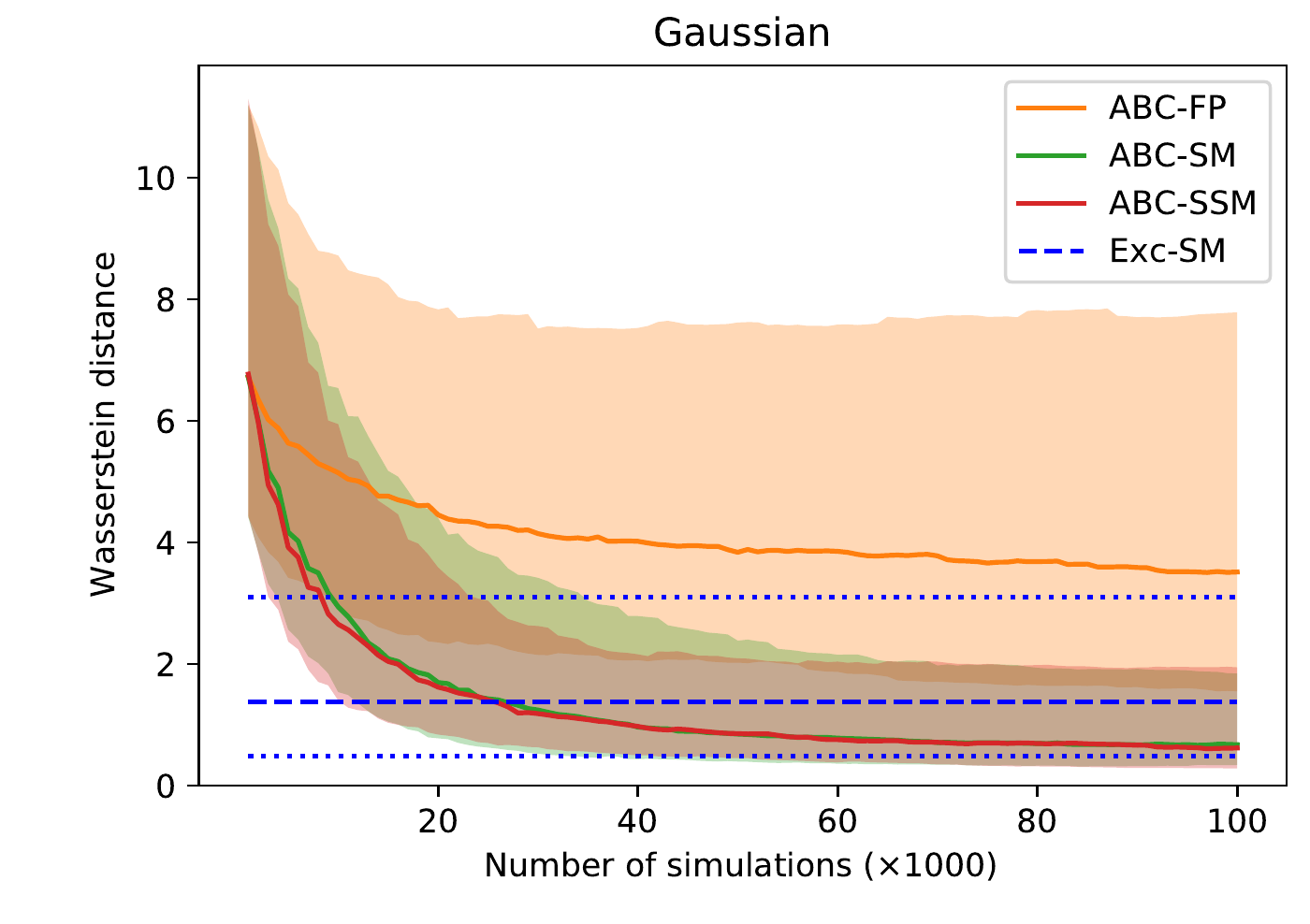}
			
		\end{subfigure}\\
		\begin{subfigure}{0.32\textwidth}
			\centering
			\includegraphics[width=1\linewidth]{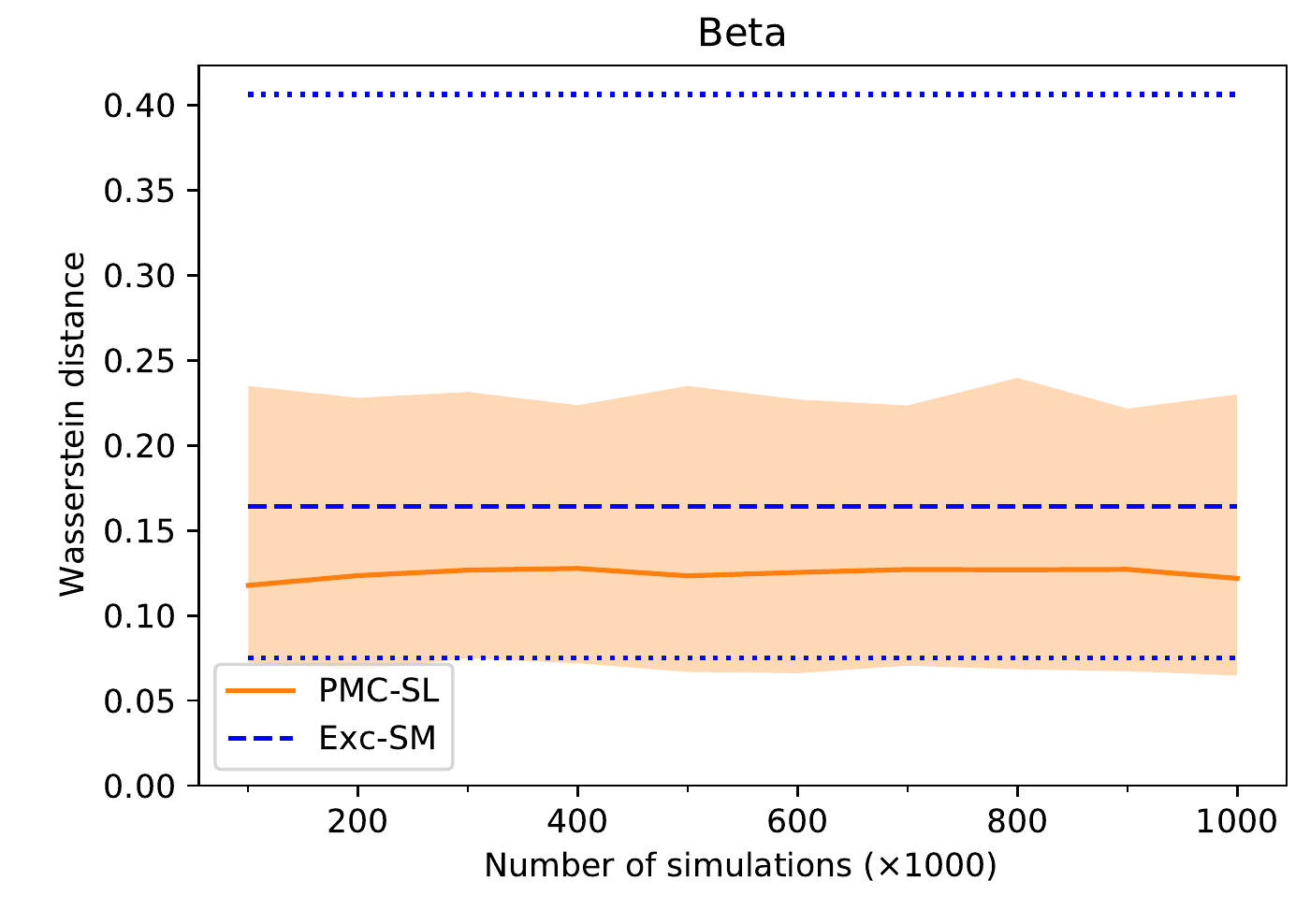}
			
		\end{subfigure}~
		\begin{subfigure}{0.32\textwidth}
			\centering
			\includegraphics[width=1\linewidth]{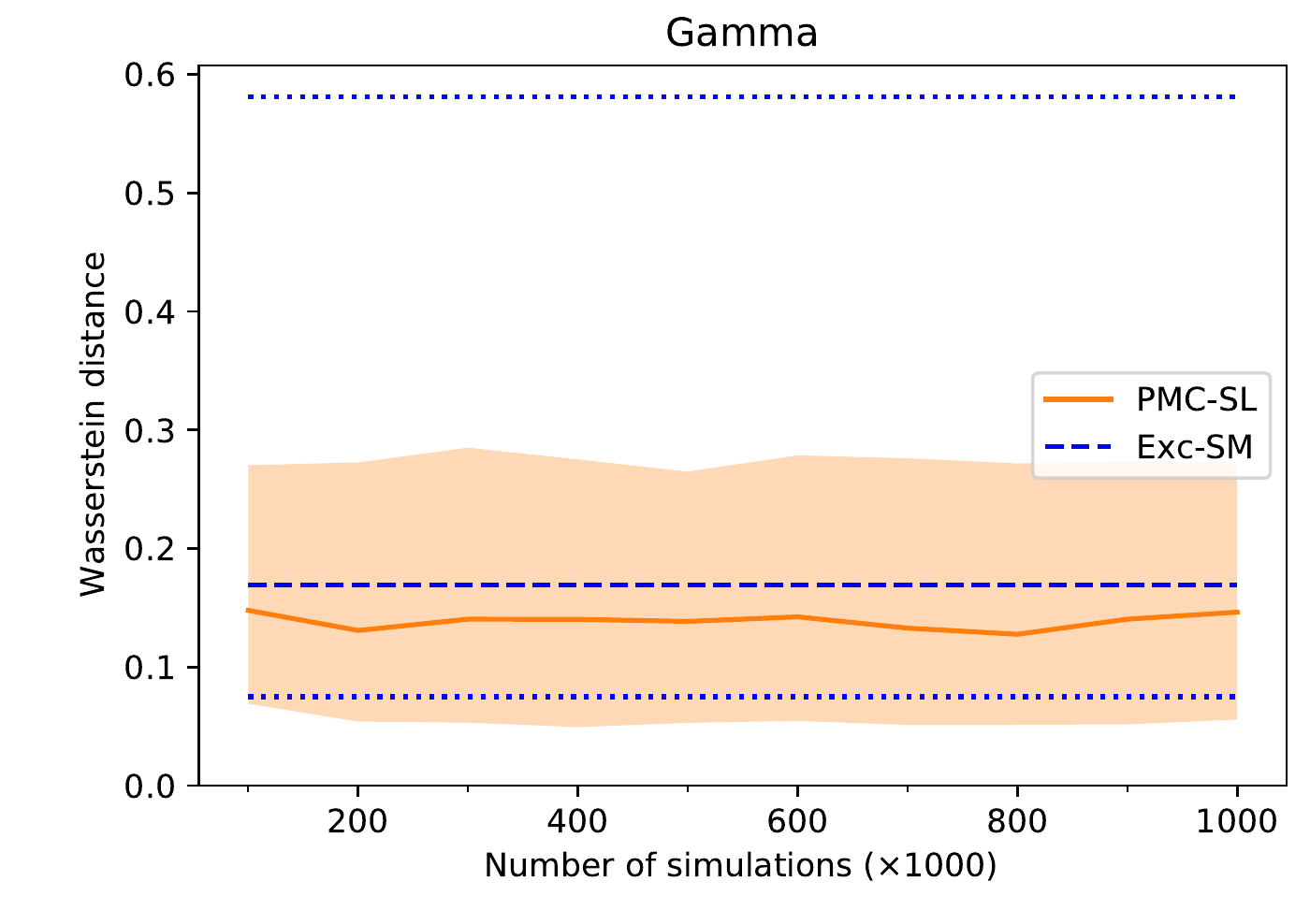}
			
		\end{subfigure}~
		\begin{subfigure}{0.32\textwidth}
			\centering
			\includegraphics[width=1\linewidth]{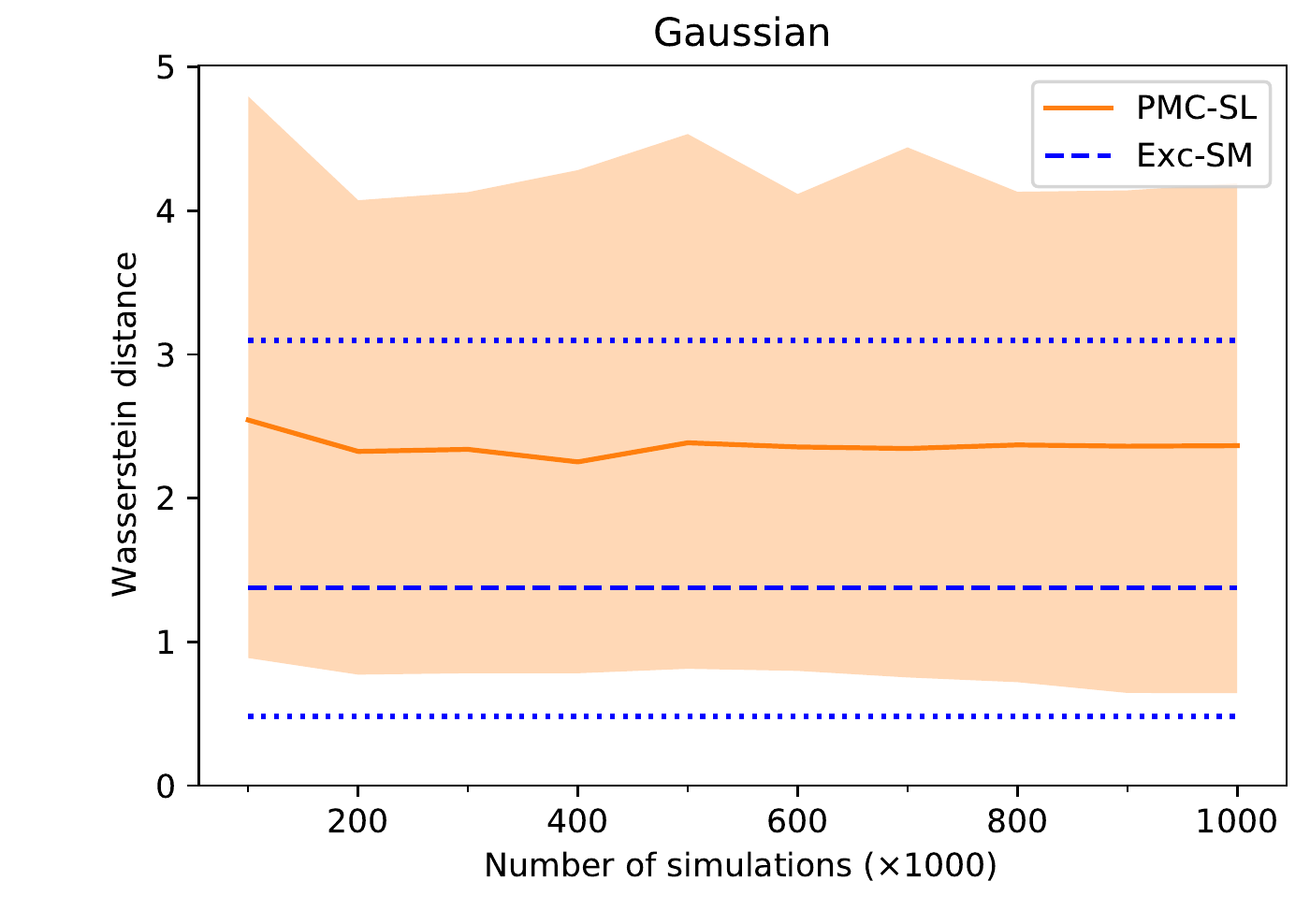}
			
		\end{subfigure}\\
		\begin{subfigure}{0.32\textwidth}
			\centering
			\includegraphics[width=1\linewidth]{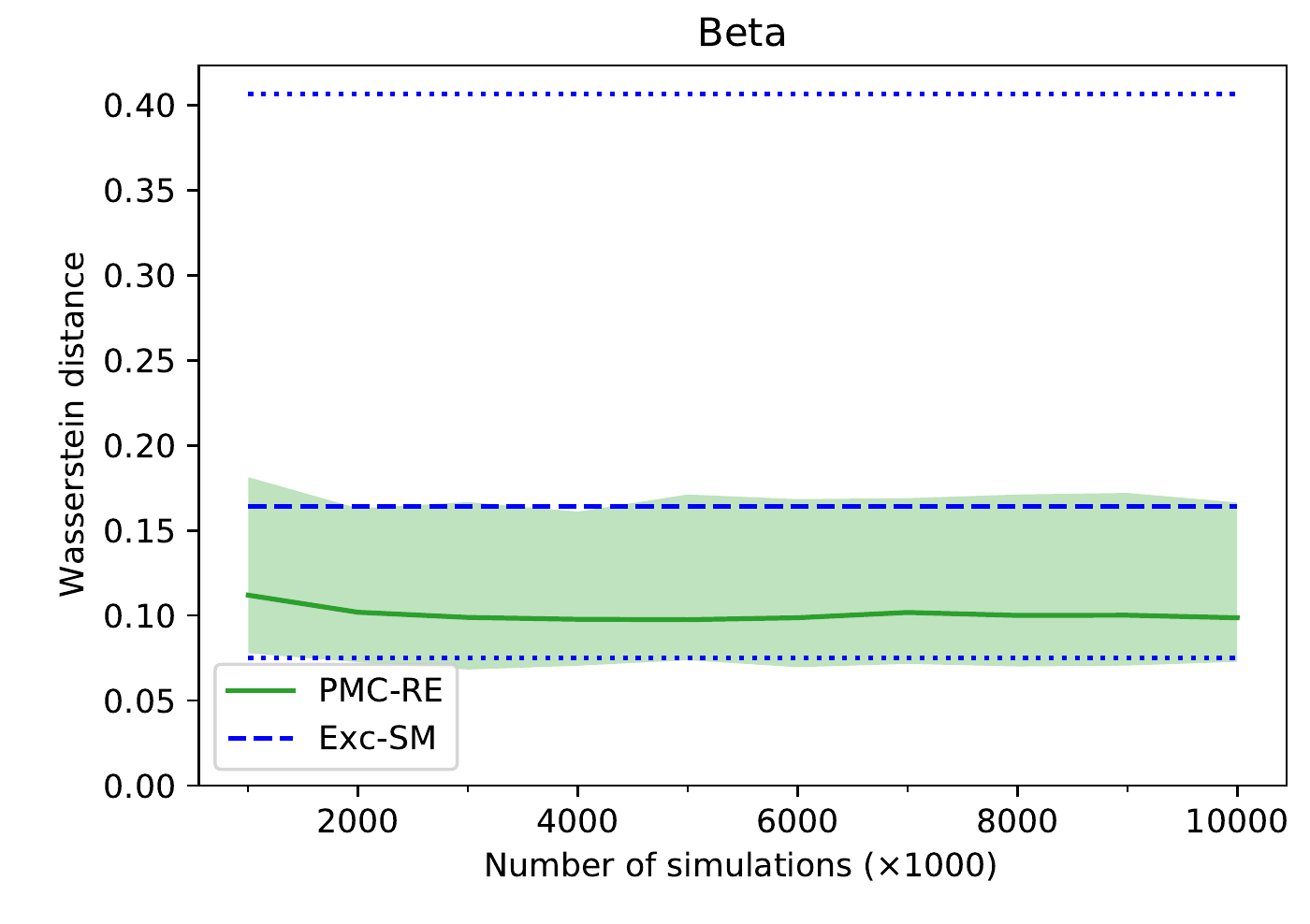}
			
		\end{subfigure}~
		\begin{subfigure}{0.32\textwidth}
			\centering
			\includegraphics[width=1\linewidth]{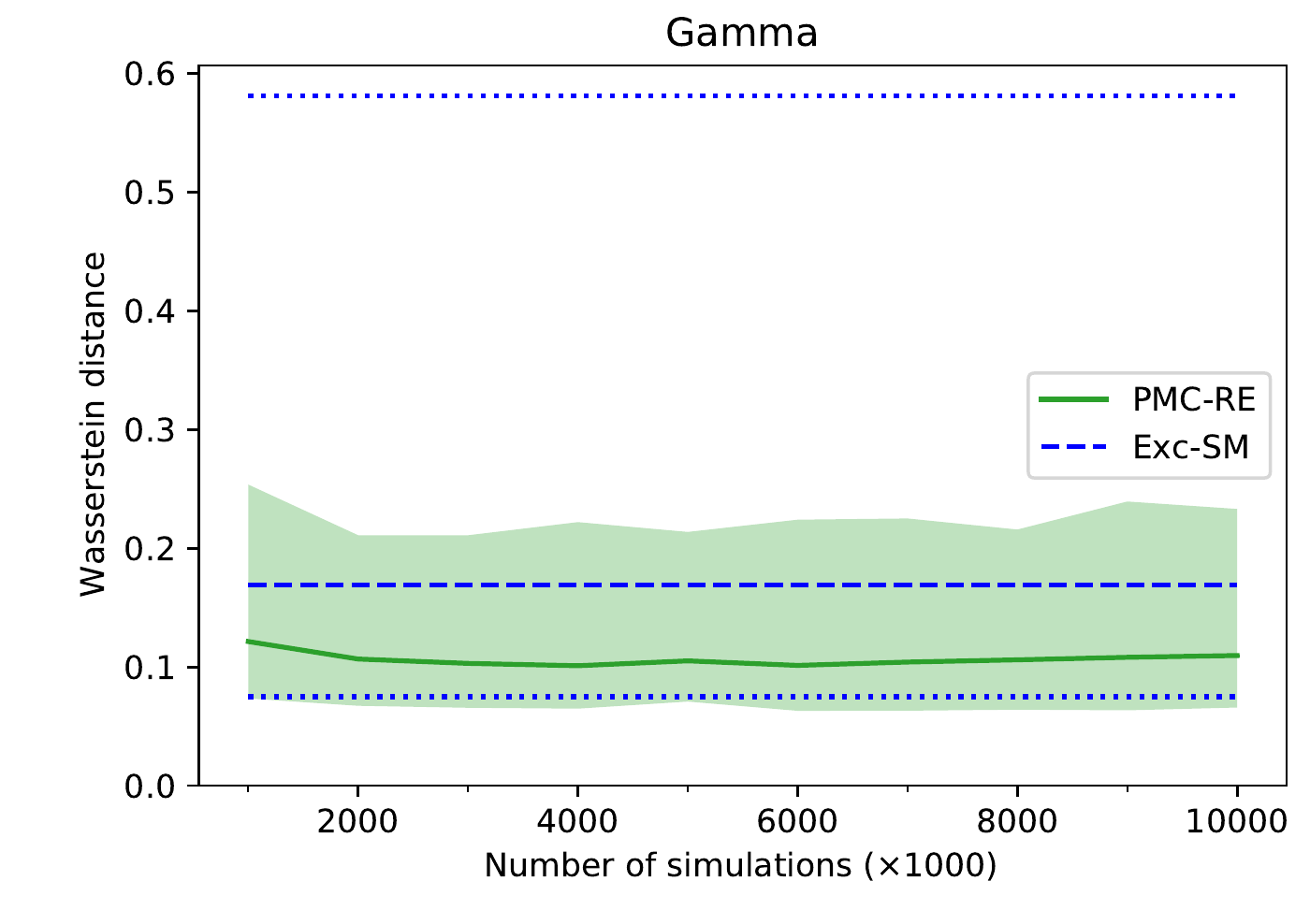}
			
		\end{subfigure}~
		\begin{subfigure}{0.32\textwidth}
			\centering
			\includegraphics[width=1\linewidth]{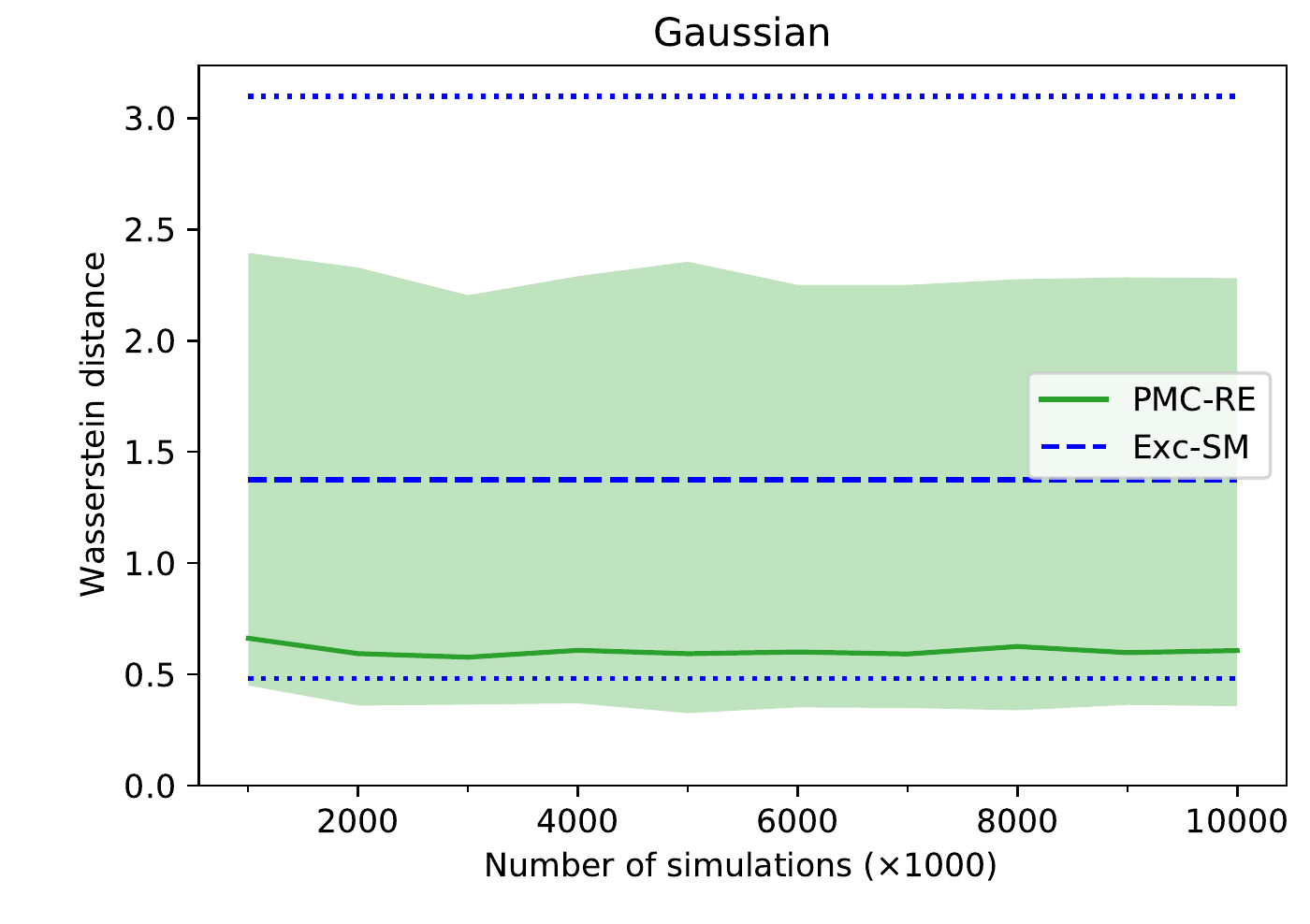}
			
		\end{subfigure}\\
		\caption{\textbf{Wasserstein distance between approximate and exact posterior at different iterations of the sequential algorithms for the exponential family models}, for 100 different observations. The solid line denotes median, while colored regions denote 95\% probability density region; an horizontal line denoting the value obtained with Exc-SM is also represented, 95\% probability density region denoted by dotted horizontal lines. The horizontal axis reports the number of model simulations corresponding to the iteration of the different algorithms. }
		\label{Fig:wass_dist_iter_exp_fam_models}
	\end{figure}

	\begin{figure}[!htbp]
		\centering
		\begin{subfigure}{0.32\textwidth}
			\centering
			\includegraphics[width=1\linewidth]{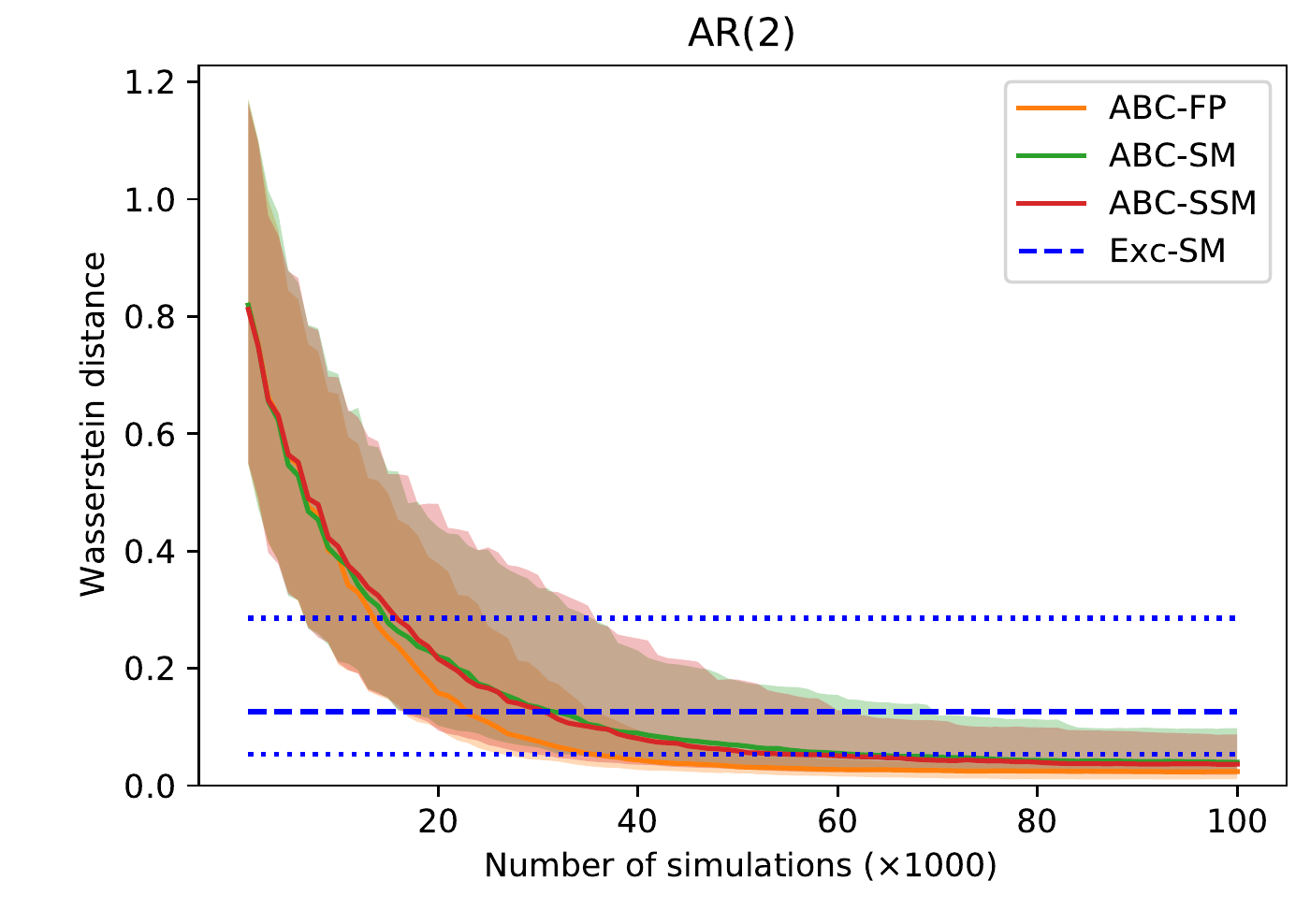}
			
		\end{subfigure}~
		\begin{subfigure}{0.32\textwidth}
			\centering
			\includegraphics[width=1\linewidth]{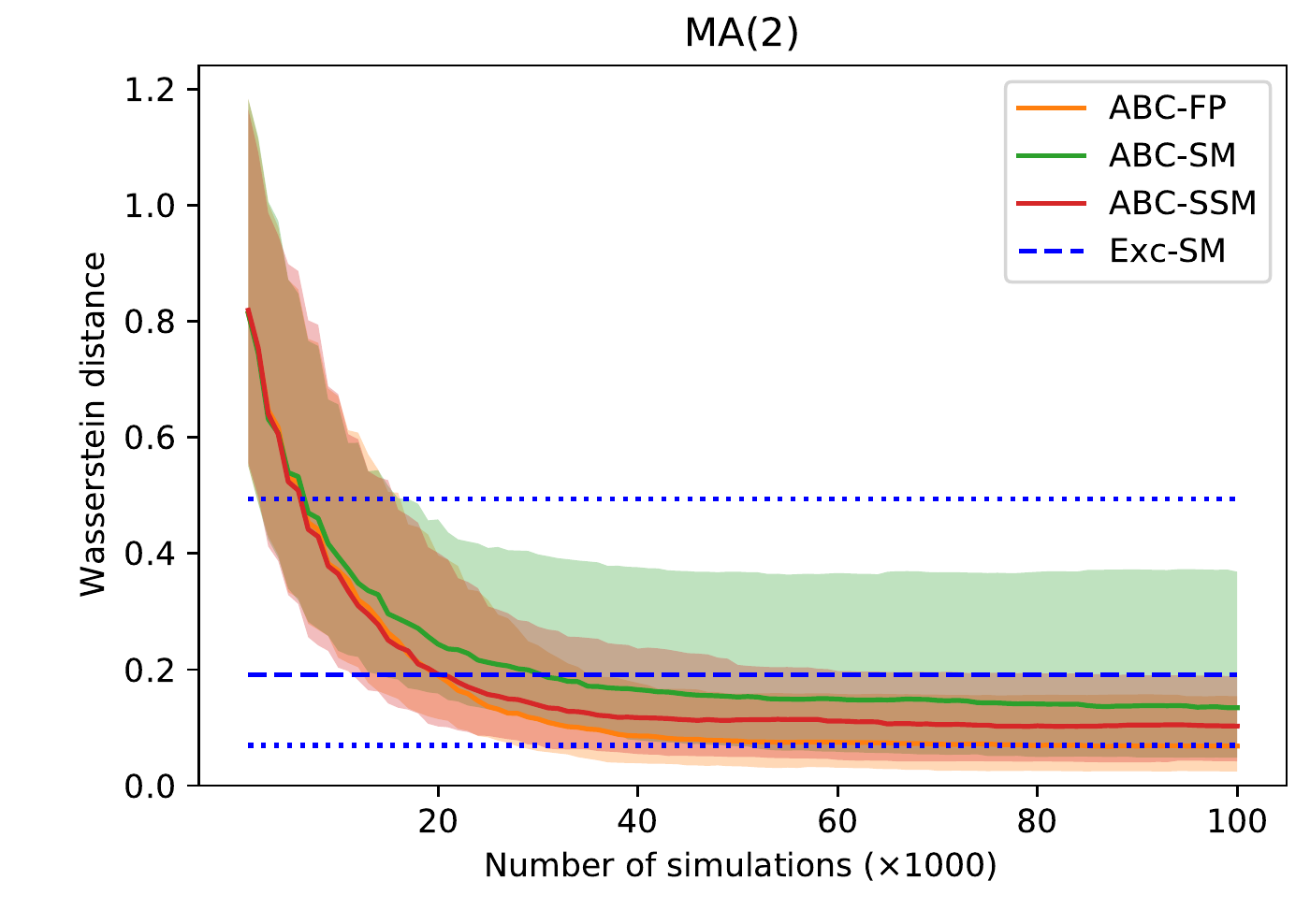}
			
		\end{subfigure}\\
		\begin{subfigure}{0.32\textwidth}
			\centering
			\includegraphics[width=1\linewidth]{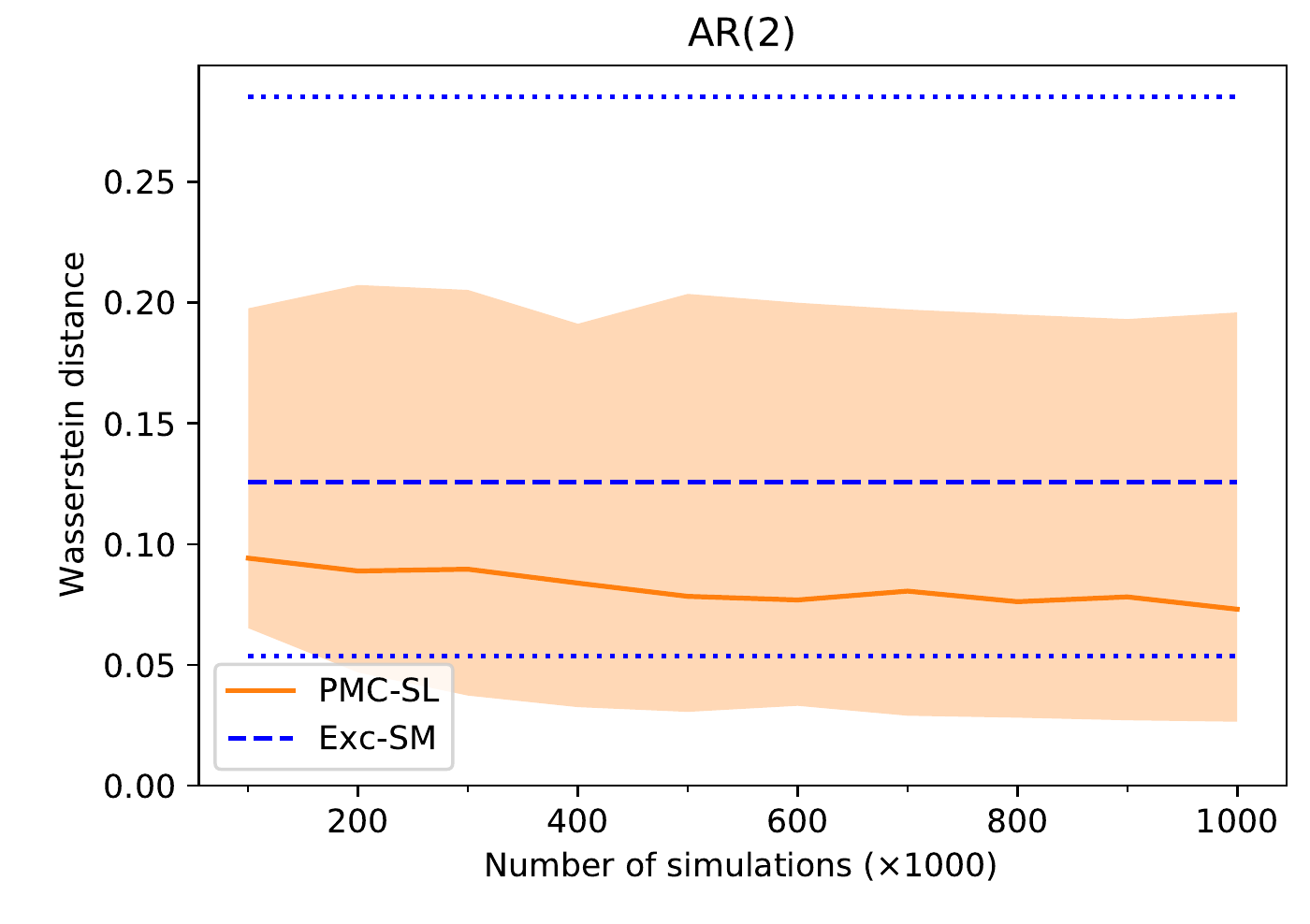}
			
		\end{subfigure}~
		\begin{subfigure}{0.32\textwidth}
			\centering
			\includegraphics[width=1\linewidth]{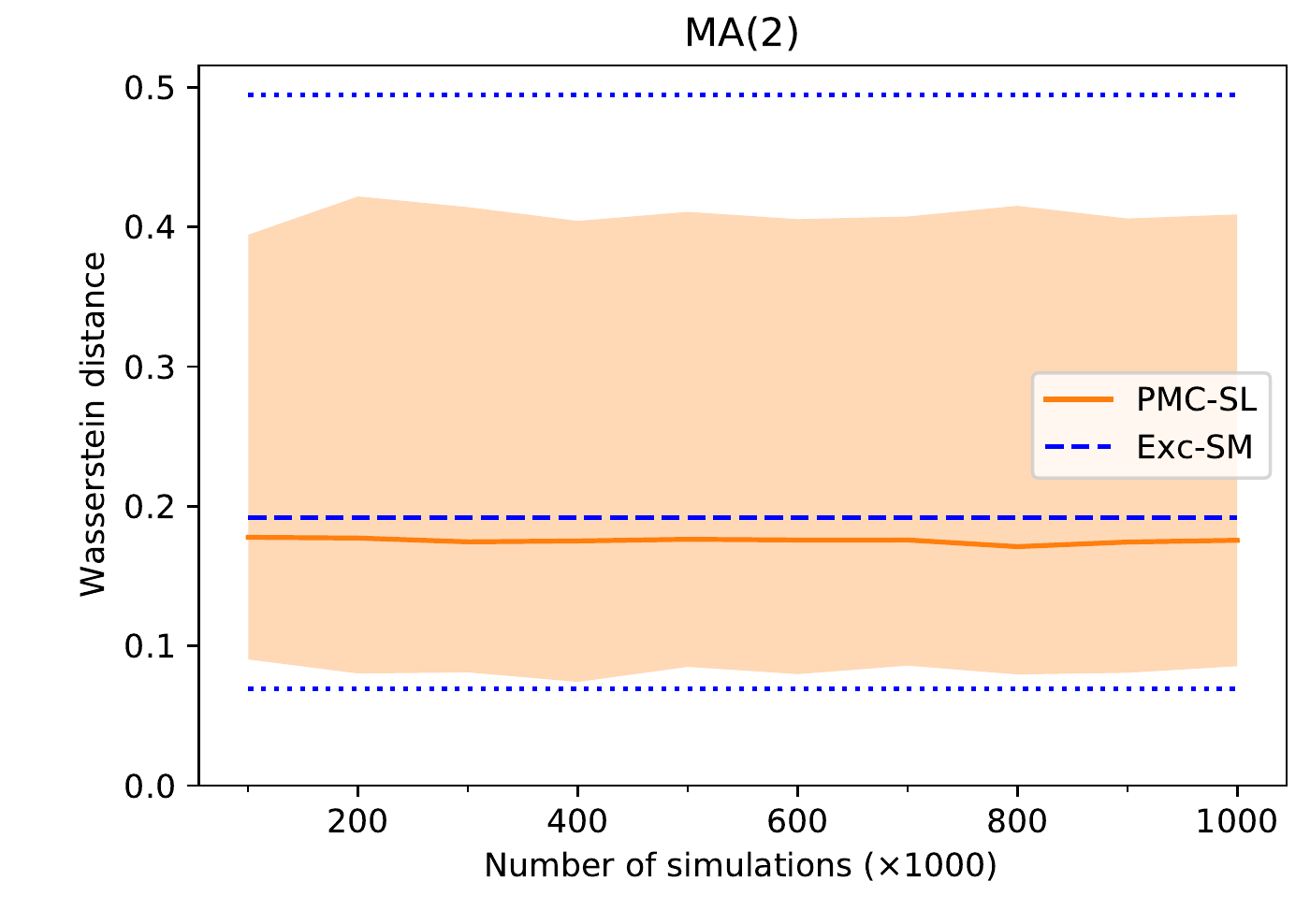}
			
		\end{subfigure}\\
		\begin{subfigure}{0.32\textwidth}
			\centering
			\includegraphics[width=1\linewidth]{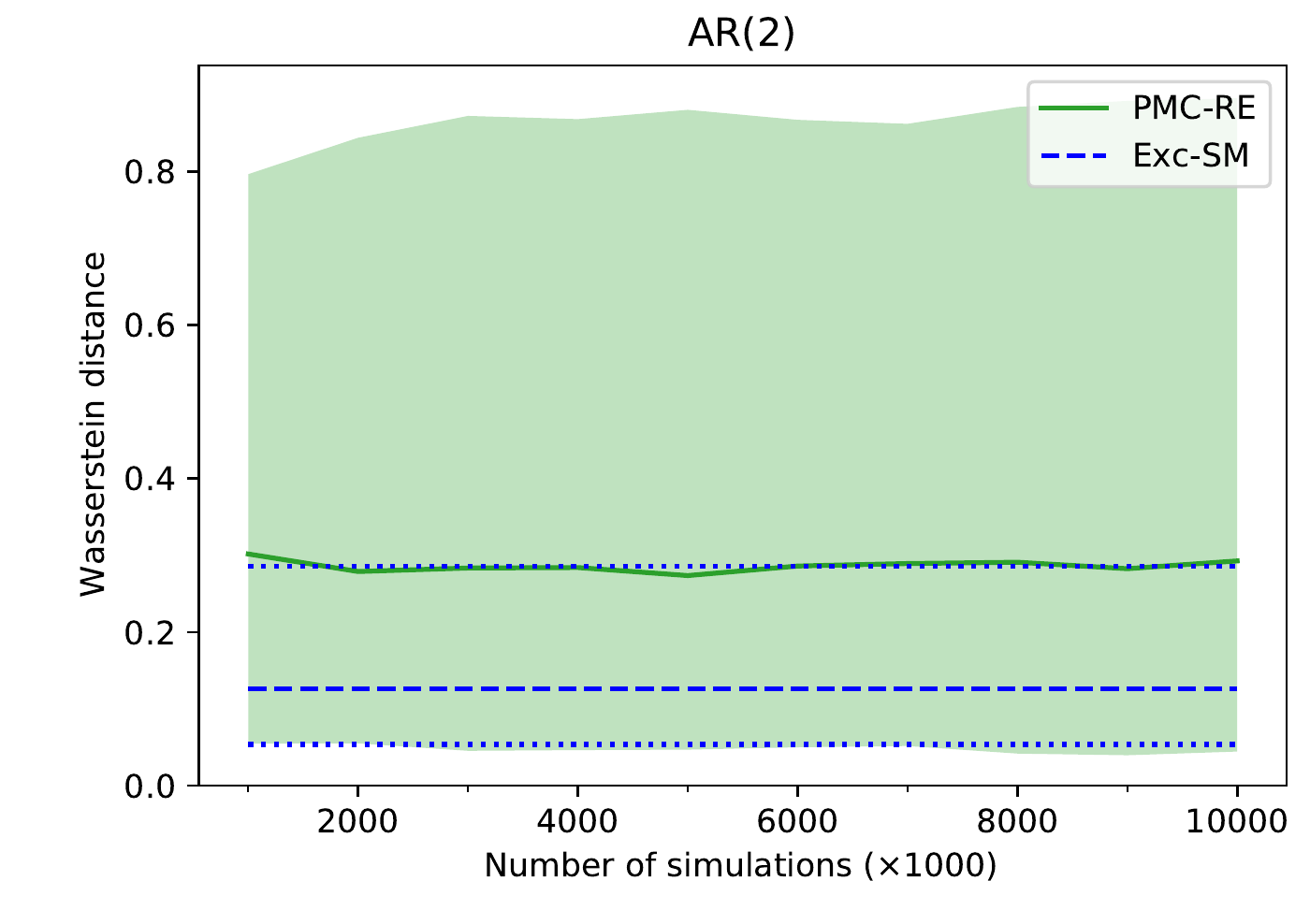}
			
		\end{subfigure}~
		\begin{subfigure}{0.32\textwidth}
			\centering
			\includegraphics[width=1\linewidth]{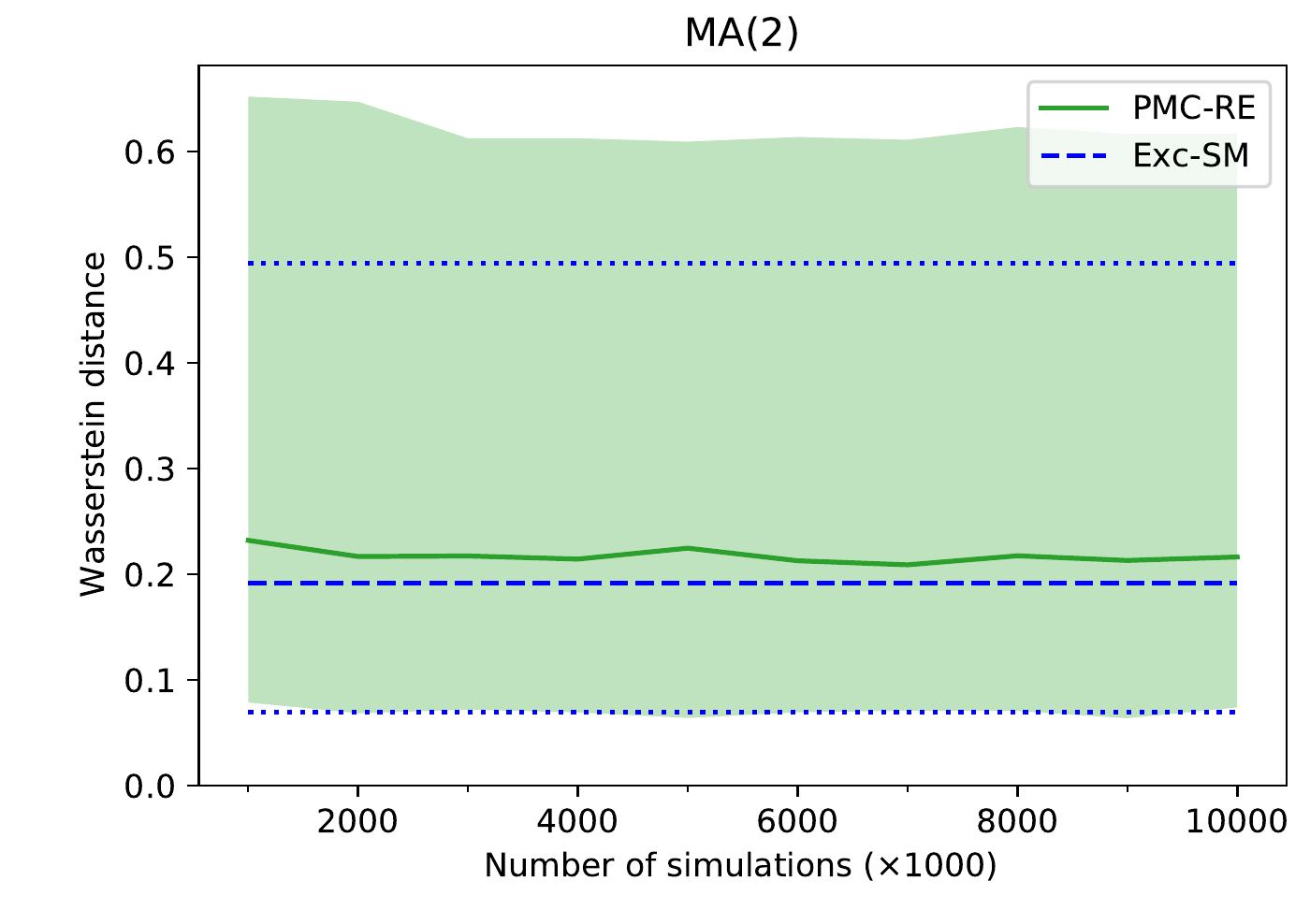}
			
		\end{subfigure}\\
		\caption{\textbf{Wasserstein distance between approximate and exact posterior at different iterations of the sequential algorithms for the AR(2) and MA(2) models}, for 100 different observations. The solid line denotes median, while colored regions denote 95\% probability density region; an horizontal line denoting the value obtained with Exc-SM is also represented, 95\% probability density region denoted by dotted horizontal lines. The horizontal axis reports the number of model simulations corresponding to the iteration of the different algorithms. }
		
		\label{Fig:wass_dist_iter_arma_models}
	\end{figure}

	\subsection{Validation with Scoring Rules for Lorenz96 model}\label{app:Lorenz_validation_results}

	Recall that the Lorenz96 model (Section~\ref{sec:Lorenz}) is a multivariate time-series model. Therefore, we use the Scoring Rules (on which more details are given in Appendix~\ref{app:SRs}) to evaluate the predictive performance at each timestep separately.
	
	First, let us recall the definition of the posterior predictive density: 
	\begin{equation}\label{}
		p(x|\dataObs) = \int p(x|\theta) \pi(\theta|\dataObs),
	\end{equation}
	whose corresponding posterior predictive distribution we denote as $ P (\cdot|\dataObs) $. In the above expression, $ \pi(\theta|\dataObs) $ may represent the posterior obtained with any of the considered methods (Exc-SM, Exc-SSM, ABC-FP, ABC-SM, ABC-SSM). 
	
	Let $ P^{(t)}(\cdot|\dataObs) $ denote the posterior predictive distribution at time $ t $ conditioned on an observation $ \dataObs $, and let $ x^{0,(t)} $ denote the \textit{t}-th timestep of the observation. Then, we are interested in: 
	\begin{equation}\label{}
		\SE(P^{(t)}(\cdot|\dataObs), x^{0,(t)}) \quad \text{ and }\quad   S_k(P^{(t)}(\cdot|\dataObs), x^{0,(t)});
	\end{equation}
	we estimate these using samples from $ P^{(t)}(\cdot|\dataObs) $ with the unbiased estimators discussed in Appendix~\ref{app:SRs}. In the Kernel Score, the bandwidth of the Gaussian kernel is set from simulations as described in Appendix~\ref{app:tuning_gamma}. 
	
	In this way, we obtain a score for an observation $ \dataObs $ at each timestep of the model. This procedure is repeated for 100 different observations; the scores at each timestep are reported in Figure~\ref{fig:Lorenz_validation_timestep}, while the summed scores over timesteps were reported in Figure~\ref{fig:SRs} in the main text. In both Figures, we report the median value and various quantiles over the considered 100 observations. It can be seen that ABC-FP is slightly outperformed by our proposed methods for both the large and small Lorenz96 configuration. Additionally, notice how, in the small configuration, the Energy score has an increasing trend over $ t $, while the Kernel one instead decreases. 	
	
	\begin{figure}[!htbp]
		\centering
		\begin{subfigure}{0.42\textwidth}
			\centering
			\includegraphics[width=1\linewidth]{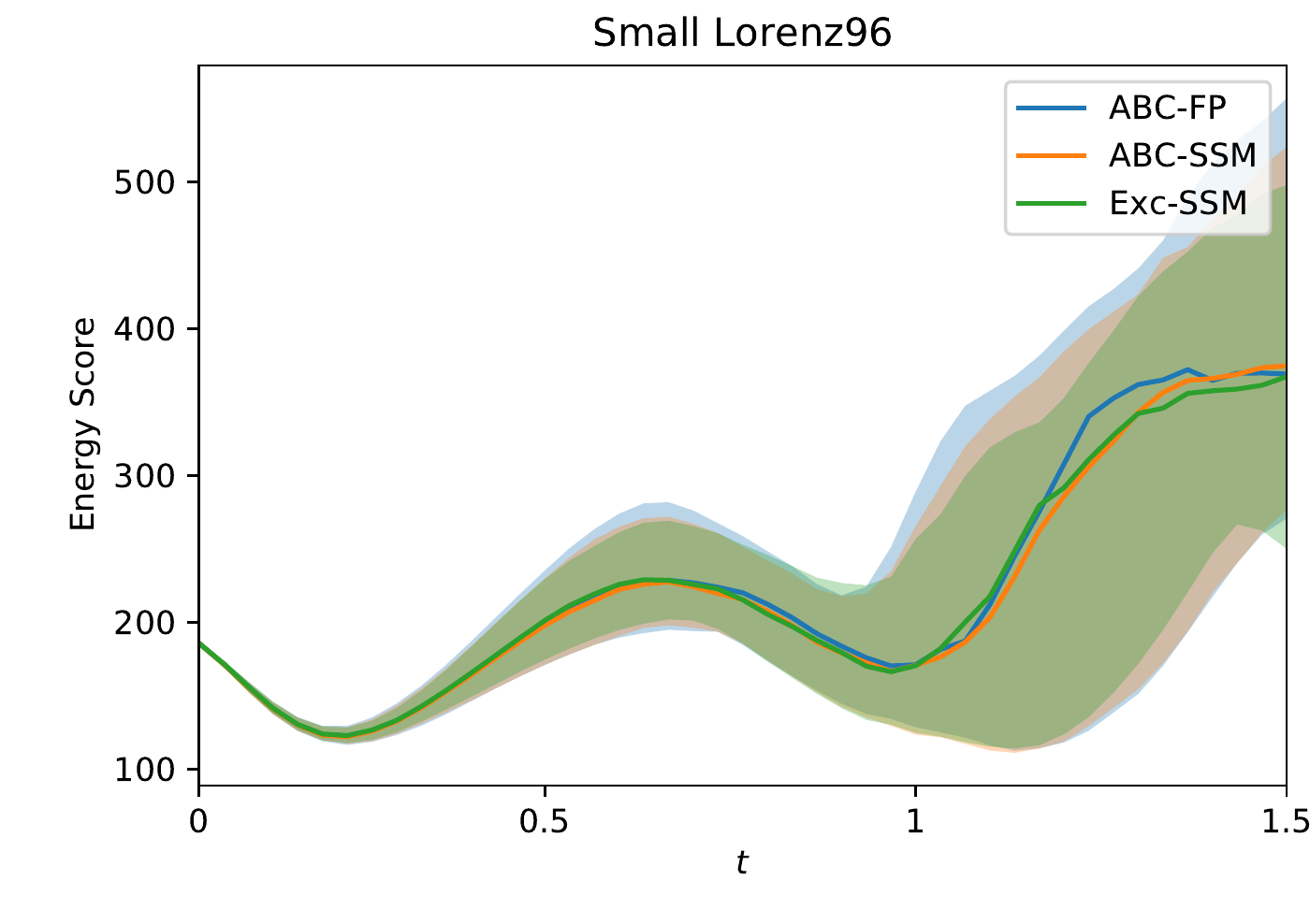}
			
		\end{subfigure}~
		\begin{subfigure}{0.42\textwidth}
			\centering
			\includegraphics[width=1\linewidth]{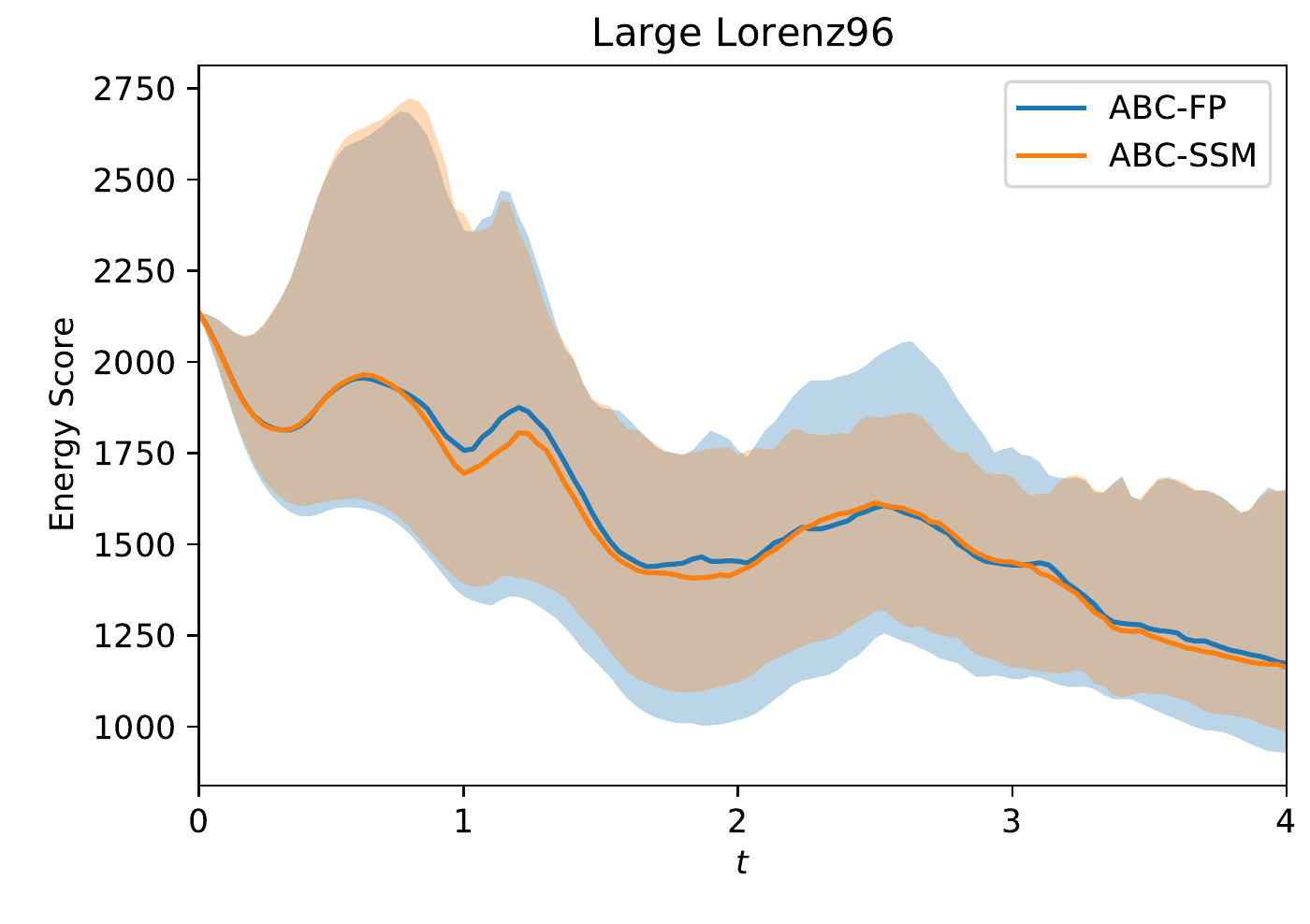}
			
		\end{subfigure}\\
		\begin{subfigure}{0.42\textwidth}
			\centering
			\includegraphics[width=1\linewidth]{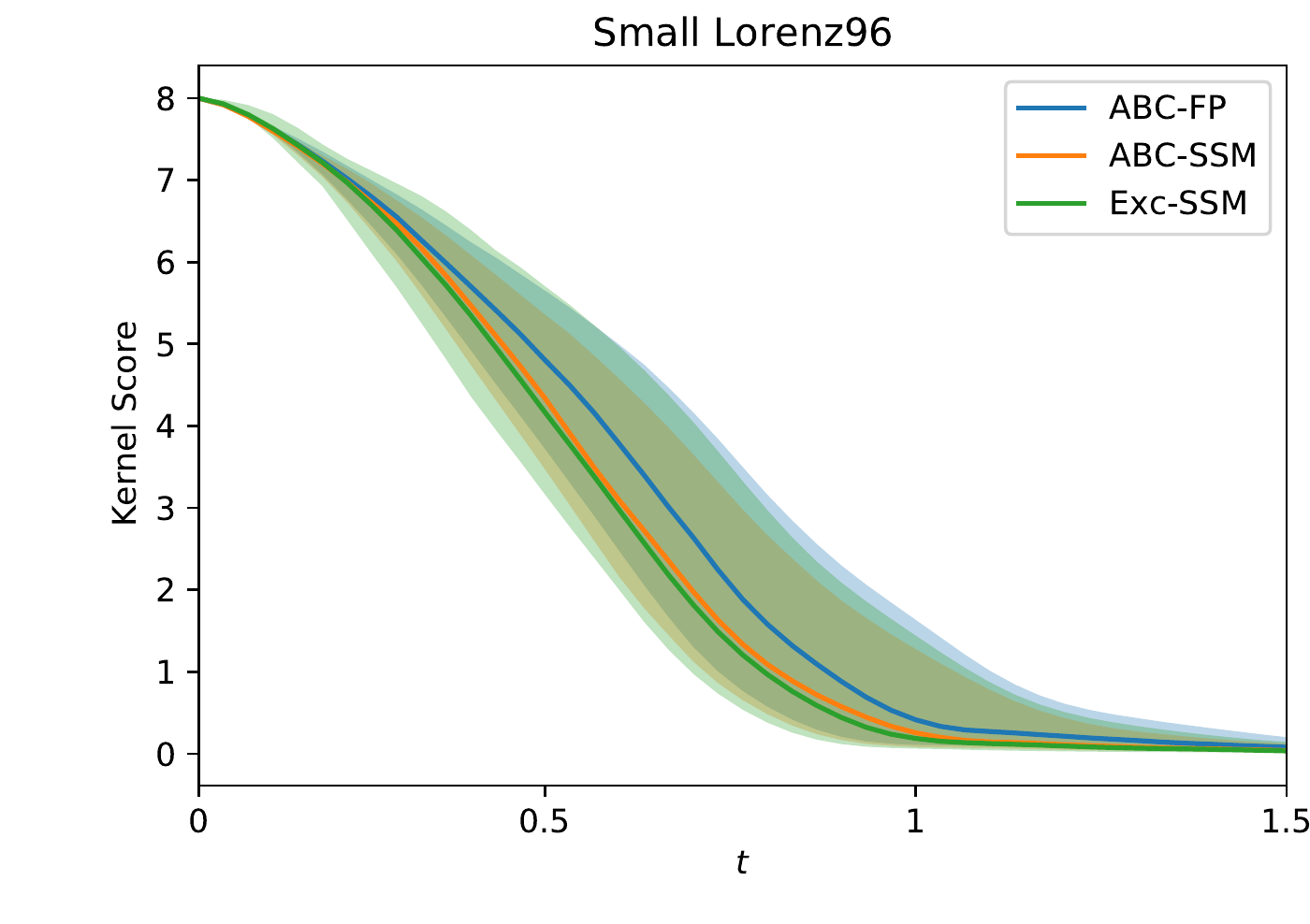}
			
		\end{subfigure}~
		\begin{subfigure}{0.42\textwidth}
			\centering
			\includegraphics[width=1\linewidth]{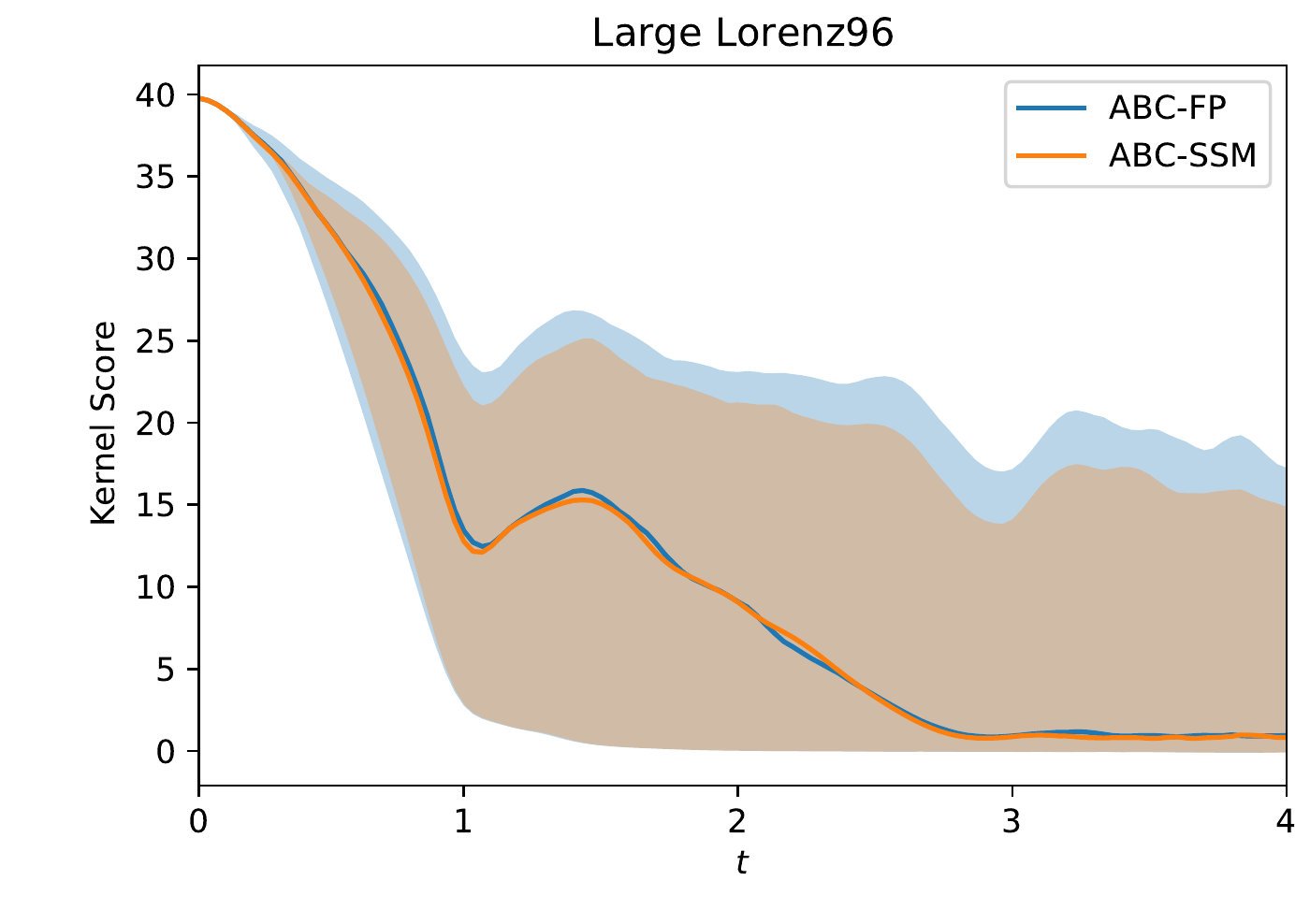}
			
		\end{subfigure}
		\caption{\textbf{Posterior predictive performance of the different methods at each timestep according to the Kernel and Energy Scores;} the smaller, the better. Samples from the posterior predictive were obtained for 100 observations, and both Scoring Rules estimated. The solid lines denote medians over the 100 observations, while colored regions denote 95\% probability density region.}	\label{fig:Lorenz_validation_timestep}
	\end{figure}

	\subsubsection{Setting $ \gamma $ in the kernel Scoring Rule}\label{app:tuning_gamma}
	
	In the kernel score $ S_k $, the Gaussian kernel in Eq.~\eqref{Eq:gau_k} is used across this work. There, the kernel bandwidth $ \gamma $ is a free parameter. In order to ensure comparability between the Scoring Rule values for different observations and inference methods, the value of $ \gamma $ needs to be fixed independently on both.
	
	Inspired by \cite{park2016k2}, we exploit an empirical procedure to set $ \gamma $ for the specific case of multivariate time-series models (of which our Lorenz96 model is an instance, see Section~\ref{sec:Lorenz}). Specifically, we use the following procedure: 
	\begin{enumerate}
		\item Draw a set of parameter values $ \theta_j \sim \pi(\theta) $ and simulations $ \dsim_{j} \sim p(\cdot|\theta_j) $, for $ j=1, \ldots, m $.
		\item Estimate the median of $ \{ ||\dsim_{j}^{(t)} - \dsim_{k}^{(t)} ||_2 \}_{jk}^{m} $ and call it $ \hat \gamma^{(t)} $, for all values of $ t $.
		\item Set the estimate for $ \gamma $ as the median of $  \hat \gamma^{(t)}$ over all considered timesteps $ t $.
	\end{enumerate}
	
	Empirically, we use $ m=1000 $. Note that the above strategy uses medians rather than means as those are more robust to outliers in the estimates. With this method, we obtain $ \gamma\approx1.54  $ for the small version of the Lorenz96 model and $ \gamma=6.38 $ for the large one.

\end{document}